\documentclass{lmcs} 

\keywords{parametric probabilistic automata, compositional verification, assume-guarantee reasoning, robust MDPs}

\usepackage{hyperref}

\usepackage{amsmath}
\usepackage{amssymb}
\usepackage{leftindex}

\usepackage{pifont}

\usepackage{stmaryrd}
\usepackage{textcomp}
 \usepackage{booktabs}

\usepackage{forest}
\usetikzlibrary{positioning, matrix}
\usepackage{tabularray}

\usepackage[thinc]{esdiff} 

\usepackage{cancel}

\usepackage{proof} 

\usepackage{array}

\usepackage{multirow}

\usepackage{diagbox}
\usepackage{thm-restate}

\usepackage{subcaption} 

\usepackage[table]{xcolor}

\usepackage{etoolbox}

\usepackage{wasysym} 

\usepackage{nicefrac}
\usepackage{tabularx}
\newcommand{\ifempty}[3]{\ifthenelse{\equal{#1}{}}{#2}{#3}}

\usepackage{tikz}
\usetikzlibrary{
	arrows.meta, 
	automata, 
	positioning, 
	shapes.geometric,
}

\tikzstyle{labelnode}=[font=\footnotesize]
\tikzstyle{dist}=[circle,  inner sep=0.8pt, solid, draw,fill]
\tikzstyle{ma}=[every label/.style={labelnode}]
\tikzstyle{mdp}=[ma]
\tikzset{ms/.style={draw,rectangle,text centered,minimum size=6mm,text width=#1},ms/.default=3mm}
\tikzset{ps/.style={draw,circle,text centered,minimum size=6mm,text width=#1},ps/.default=3mm}
\tikzset{pswide/.style={draw,ellipse,minimum size=6mm,text centered,text width=#1},pswide/.default=7mm}
\tikzstyle{mtrans}=[->,densely dotted, semithick,labelnode]
\tikzstyle{ptrans}=[->,semithick,labelnode]
\tikzstyle{every initial by arrow}=[inner sep=0pt] 
\tikzset{init/.style={initial #1, initial text={}, initial distance=4mm},init/.default=left}
\newcommand*{\toptionalfrac}[2][]{\ifempty{#1}{#2}{\nicefrac{#1}{#2}}}

\newcommand*{\tprob}[2][]{{\ensuremath{\color{colorProbability}\toptionalfrac[#1]{#2}}}}
\newcommand*{\tact}[2][]{{\ensuremath{\color{colorAction}\toptionalfrac[#1]{#2}}}}
\newcommand*{\tactext}[2][]{{\ensuremath{\color{colorextendedAction}\toptionalfrac[#1]{#2}}}}

\newcommand{\modelType}{pa}
\newcommand{\extended}{true} 
\newcommand{\prodAutomata}{true} 

\newcommand{\mless}{\mathit{mless}}
\newcommand{\comp}{\mathit{cmp}}
\newcommand{\safe}{\mathit{safe}}
\newcommand{\prt}{\mathit{prt}}
\newcommand{\fairWrtRegionModel}[2]{{{\mathit{fair}}_{#2\ }^{#1}}\!\!} 
\newcommand{\fairC}{\fairWrtRegionModel{}{\decomp}} 

 \newcommand{\xmark}{\ding{55}}

\colorlet{colorRate}{orange}
\colorlet{colorProbability}{teal}
\colorlet{colorAction}{magenta}
\colorlet{colorReward}{green}
\colorlet{colorextendedAction}{blue}

\newcommand{\domain}{\mathit{dom}}

\newcommand{\rats}{\mathbb{Q}}
\newcommand{\nats}{\mathbb{N}}
\newcommand{\reals}{\mathbb{R}}

\newcommand{\cupdot}{\mathbin{\mathaccent\cdot\cup}}

\newcommand{\iverson}[1]{\llbracket\,#1\,\rrbracket}

\newcommand{\set}[1]{\{#1\}}
\newcommand{\bigset}[1]{\big\{#1\big\}}

\newcommand{\subDist}[1]{\mathit{SubDist}(#1)}
\newcommand{\dist}[1]{\mathit{Dist}(#1)}
\newcommand{\paramDist}[2]{\mathit{Dist}_{#1}(#2)}

\newcommand{\modelsWrt}[1]{
    \def\temp{#1}
    \ifx\temp\empty%
      \models{}\!
    \else
      \models^{#1}
    \fi
  }

\newcommand{\simulationRelation}[0]{\varrho}

\newcommand{\strSimulationRegion}[1]{
    \preceq_{#1}
  }

\newcommand{\robStrSimulationRegion}[1]{\strSimulationRegion{#1}^{robust}}

\newcommand{\last}[1]{\mathit{last}(#1)}
\newcommand{\pref}[1]{\mathit{pref}(#1)}

\newcommand{\indicatorFct}[1]{\mathbf{1}_{#1}}

\newcommand{\setOfdecompsOf}[1]{2^{#1}}

\newcommand{\decomp}[0]{\mathcal{C}}

\newcommand{\lab}{\mathsf{a}}
\newcommand{\altlab}{\mathsf{b}}
\newcommand{\altaltlab}{\mathsf{c}}

\newcommand{\actSetOf}[1]{
    \def\temp{#1}
    \ifx\temp\empty%
      \mathit{Act}
    \else
      \mathit{Act}_{#1}
    \fi
  }

  \newcommand{\syncOf}[1]{
    \def\temp{#1}
    \ifx\temp\empty%
      L
    \else
      L_{#1}
    \fi
  }

\newcommand{\alphabetOf}[1]{
    \def\temp{#1}
    \ifx\temp\empty%
      \Sigma
    \else
      \Sigma_{#1}
    \fi
  }

\newcommand{\stateSetOf}[1]{
    \def\temp{#1}
    \ifx\temp\empty%
      S 
    \else
      S_{#1}
    \fi
  }
\newcommand{\transRelationOf}[1]{
  \def\temp{#1}
  \ifx\temp\empty%
    \delta 
  \else
    \delta_{#1}
  \fi
}
\newcommand{\transFctOf}[1]{
	\def\temp{#1}
	\ifx\temp\empty%
	\mathbf{P} 
	\else
  \mathbf{P}_{#1}
	\fi
}

\newcommand{\initialOf}[1]{
  \def\temp{#1}
  \ifx\temp\empty%
    {s^{\mathit{init}}}
  \else
    {s^{\mathit{init}}_{#1}}
  \fi
}
\newcommand{\parameterSetOf}[1]{
    \def\temp{#1}
    \ifx\temp\empty%
      V 
    \else
      V_{#1}
    \fi
  }

\newcommand{\idSetOf}[1]{
    \def\temp{#1}
    \ifx\temp\empty%
      ID
    \else
      ID_{#1}
    \fi
  }

\newcommand{\paTupleOf}[1]{\left(\stateSetOf{#1}, \initialOf{#1}, \actSetOf{#1}, \transFctOf{#1},   \syncOf{#1}  \right)} %
\newcommand{\ppaTupleOf}[1]{\left(\stateSetOf{#1}, \initialOf{#1},\parameterSetOf{#1},  \actSetOf{#1} , \transFctOf{#1}, \syncOf{#1} \right)} 
\newcommand{\rpaTupleOf}[1]{\left(\stateSetOf{#1}, \initialOf{#1}, \actSetOf{#1}, \transFctOf{#1},   \syncOf{#1}  \right)} %

\newcommand{\bpAutomatonOf}[1]{\mathcal{A}^{bad}_{#1}}
\newcommand{\QstateSetOf}[1]{
    \def\temp{#1}
    \ifx\temp\empty%
      Q
    \else
      Q_{#1}
    \fi
  }
\newcommand{\qInitialOf}[1]{
      \def\temp{#1}
      \ifx\temp\empty%
        {q^{init}}
      \else
        {q^{init}_{#1}}
      \fi
}

\newcommand{\regionIntersection}[0]{\cap} 
\newcommand{\regionUnion}[0]{\cup} 

\newcommand{\regionIntersectionOf}[3]{
	\def\temp{#2}
	\ifx\temp\empty%
    {{#1} \! \regionIntersection{} \! {#3}}
	\else
    {{#1 \! \regionIntersection{} #2 \! \regionIntersection{} #3} }
	\fi
  }

\newcommand{\stratProjOfToValuation}[3]{
	\def\temp{#3}
	\ifx\temp\empty%
	{\restrOfTo{#1}{#2}}
	\else
	{{\restrOfToVal{#1}{#2}{#3}}}
	\fi}

\newcommand{\regionUnionOf}[3]{
	\def\temp{#2}
	\ifx\temp\empty%
	{{#1} \!\regionUnion\! {#3}}
	\else
	{{#1}\! \regionUnion \!{#2} \regionUnion\! {#3}}
	\fi}

\newcommand{\APOf}[1]{
	\def\temp{#1}
	\ifx\temp\empty%
	{AP}
	\else
	{AP_{#1}}
	\fi}
	
\newcommand{\labellingOf}[2]{
    \def\temp{#2}
      \ifx\temp\empty%
      {StateLab_{#1}}
      \else
       {StateLab_{#1}}\left(#2\right)
      \fi}

\newcommand{\PrOf}[3]{Pr_{#1}^{#2}
    \def\temp{#3}
      \ifx\temp\empty%
      \else
        \!\left( #3 \right)
      \fi}

\newcommand{\paReductionOf}[1]{\mathit{PA}\!\left(#1\right)}
\newcommand{\paReductionOfStrategy}[2]{\mathit{PA}^{#2}\!\left(#1\right)}

\newcommand{\extr}[1]{\mathit{ext}\!\left(#1\right)}
\newcommand{\convh}[1]{\mathit{conv}\!\left(#1\right)}
\newcommand{\gen}[1]{\mathit{gen}\!\left(#1\right)}

\newcommand{\parallelRel}{\mathbin{\parallel^{iRel}}}   
\newcommand{\parallelConv}{\mathbin{\parallel^{conv}}}

\newcommand{\bdownarrow}{\big \downarrow}  
\newcommand{\buparrow}{\big \uparrow}  
\newcommand{\budarrow}{\big \updownarrow}  

\newcommand{\monotonicOnRegionParameter}[5]{
  \def\temp{#3}
  {#2}{\!#1}_{
    {\ifx\temp\empty%
    {#4}
    \else
    {#3}, {#4}
    \fi}
  }^{#5}
}

\newcommand{\ExpTot}[3]{Ex_{#1}^{#2}
	\def\temp{#3}
	\ifx\temp\empty%
	\else
	\!\left( #3 \right)
	\fi}

\newcommand{\probPredicate}[2]{\mathbb{P}_{#1}\!(#2)}  
\newcommand{\expPredicate}[2]{\mathbb{E}_{#1}\!(#2)}
\newcommand{\generalPredicate}[0]{\varphi} 

\newcommand{\infpath}{\pi}
\newcommand{\ppath}{\pi}
\newcommand{\finpath}{\widehat{\pi}}
\newcommand{\infPathsOf}[1]{\mathit{Paths}_{#1}^{\mathit{inf}}}
\newcommand{\finPathsOf}[1]{\mathit{Paths}_{#1}^{\mathit{fin}}}
\newcommand{\liftedpaths}[2]{ #1 \otimes #2 }

\newcommand{\traceOf}[1]{\mathit{tr}({#1})}
\newcommand{\trace}{\rho}

\newcommand{\alphabetExtensionOfTo}[2]{#1{\langle#2\rangle}}

\newcommand{\restrOfToVal}[3]{{#1{\upharpoonright}_{\!#2}^{#3}}}
\newcommand{\restrOfTo}[2]{{#1{\upharpoonright}_{\!#2}}} 

\newcommand{\strategy}{\sigma} 
\newcommand{\nature}{\nu} 
\newcommand{\strategyset}{Str}
\newcommand{\strategysetOf}[2]{\strategyset_{\!#1}^{\!#2}}

\newcommand{\natureSet}{\mathit{Nat}}
\newcommand{\natureSetOf}[2]{\natureSet_{\!#1}^{\!#2}}

\newcommand{\agTriple}[4]{ {#1} \modelsWrt{#2} {#3}  {\rightarrow}  {#4}}  

\newcommand{\agTripleMless}[4]{ {#1} \modelsWrt{#2}_{\mless} {#3}  {\rightarrow}  {#4}}  

\newcommand{\multiobjectiveQuery}[2]{\textsf{#2}^{#1}}

\newcommand{\solutionFctMdpObjective}[2]{\textsf{sol}_{#1}^{#2}}

\newcommand{\cyl}[0]{\textsf{Cyl}}

\newcommand{\ppa}{\mathcal{M}}
\newcommand{\ipa}{\mathcal{I}}
\newcommand{\rpa}{\mathcal{U}}

\newcommand{\pa}{\mathcal{N}}
\newcommand{\product}{{\otimes}}

\newcommand{\regLang}{\mathcal{L}}
\newcommand{\rewFct}{\mathcal{R}}

\newcommand{\valuation}{\valuationVector{v}}
\newcommand{\valuationVector}[1]{\mathpzc{#1}}

\DeclareFontFamily{OT1}{pzc}{}
\DeclareFontShape{OT1}{pzc}{m}{it}{<-> s * [1.10] pzcmi7t}{}
\DeclareMathAlphabet{\mathpzc}{OT1}{pzc}{m}{it}
\newcommand{\region}{{R}}   

\definecolor{myorange}{RGB}{230,159,0}
\definecolor{myblue}{RGB}{55,126,184}
\definecolor{mygreen}{RGB}{77,175,74}
\definecolor{myred}{RGB}{228,26,28}
\definecolor{mypink}{RGB}{204,121,167}
\definecolor{mygray}{rgb}{0.66, 0.66, 0.66}

    \pdfoutput=1 
    \allowdisplaybreaks

\usepackage[nameinlink,noabbrev,capitalise]{cleveref} 
\crefname{rem}{Remark}{Remarks} 
\crefname{lem}{Lemma}{Lemmas} 
\crefname{thm}{Theorem}{Theorems} 
\crefname{cor}{Corollary}{Corollaries} 
\crefname{prop}{Proposition}{Propositions} 
\crefname{defi}{Definition}{Definition} 
\crefname{exa}{Example}{Examples} 

\begin{document}

\title[Compositional Reasoning for PAs with Uncertainty]{Compositional Reasoning\texorpdfstring{\\}{} for Probabilistic Automata with Uncertainty}
\titlecomment{{\lsuper*} This is an extended version of~\cite{Mer+25CONCUR} which appeared at CONCUR 2025.} 

\author[H.~Mertens]{Hannah Mertens\lmcsorcid{0009-0009-6815-3285}}[a]
\author[T.~Quatmann]{Tim Quatmann\lmcsorcid{0000-0002-2843-5511}}[a]
\author[J.\,P.~Katoen]{Joost-Pieter Katoen\lmcsorcid{0000-0002-6143-1926}}[a]

\address{RWTH Aachen University, Aachen, Germany}	
\email{hannah.mertens@cs.rwth-aachen.de, tim.quatmann@cs.rwth-aachen.de, katoen@cs.rwth-aachen.de}  





\begin{abstract}
This paper develops an assume-guarantee (AG) framework for the compositional verification of probabilistic automata (PAs) with \emph{uncertain} transition probabilities.
We study 
\emph{parametric probabilistic automata} (pPAs), where probabilities are given by polynomial functions over a finite set of real-valued parameters 
and \emph{robust probabilistic automata} (rPAs)---a generalisation of \emph{interval probabilistic automata} (iPAs)---where transition probabilities range over potentially uncountable uncertainty sets.

Towards pPAs, an existing AG framework for PAs is lifted to the parametric setting. 
We establish asymmetric, circular, and interleaving proof rules to enable compositional verification of a broad class of multi-objective queries, encompassing probabilistic reachability properties and parametric expected total rewards. 
In addition, we introduce a dedicated AG rule for compositional reasoning about parameter monotonicity.

For convex rPAs and iPAs with history-dependent (memory-full) nature, we establish sound 
AG rules via a reduction to infinite PAs.
We further show that AG reasoning can not straightforwardly be applied to non-convex rPAs, memoryless (once-and-for-all) nature semantics, and the common interval-arithmetic relaxation of parallel composition.

Finally, we develop a simulation-based AG style for pPAs: 
we define strong simulation and robust-strong simulation relations for pPAs 
and derive their corresponding proof rules. 

\end{abstract}

\maketitle

\section{Introduction}\label{sec:intro}

\paragraph{Probabilistic automata.}
Probabilistic model checking~\cite{For+11,Kat16} studies the automated verification of systems with random behaviour.
Applications include network and security protocols, biochemical processes, and planning under uncertainty~\cite{NS06,KNP08,FWHT15}.
Markov decision processes (MDPs) extend Markov chains with nondeterministic choices and thus naturally capture systems that combine probabilistic and controlled behaviour.
Common properties such as extremal reachability probabilities in MDPs can often be verified efficiently in PTIME~\cite{BK08}.
Segala’s probabilistic automata (PAs)~\cite{Seg95} enrich MDPs with a parallel-composition operator that synchronizes components on shared labels and interleaves their independent actions, making them well-suited for modelling concurrent, distributed, and randomized systems.  

\paragraph{Assume-guarantee reasoning.}
In parallel systems, the number of system states grows exponentially with the number of system components.
The resulting \emph{state-space explosion} is an omnipresent challenge when model checking complex systems, often rendering analysis computationally infeasible. 
Compositional verification techniques such as \emph{assume-guarantee} (AG) reasoning \cite{Jon83,Pnu84} address this problem by decomposing the verification task into smaller sub-tasks that consider individual components in isolation. 
This modular verification approach has been successfully applied in various domains, including service-based workflow verification \cite{Bou+16}, large-scale IT systems \cite{Cal+12}, and autonomous systems incorporating deep neural networks \cite{Pas+18,Pas+23}. 
AG reasoning has been developed for PAs by Kwiatkowska et al.~\cite{Kwi+13}: their AG rules reduce compositional reasoning about probabilistic safety, $\omega$-regular properties, and expected-reward objectives to multi-objective model checking on individual components. 

\begin{exa}[Communication System as a PA]\label{ex:high-level-PA}
Consider a communication system, where a \emph{sender} $\mathcal{S}$ broadcasts messages to a \emph{receiver} $\mathcal{R}$ through a \emph{broadcast channel} $\mathcal{B}$. 
The system is modeled by the parallel composition $\mathcal{S} \parallel \mathcal{B} \parallel \mathcal{R}$.
The components are faulty: $\mathcal{S}$ might face a collision, broadcasting in $\mathcal{B}$ might fail due to message loss, and $\mathcal{R}$ might miss broadcasts.
AG reasoning allows to verify the specification without explicitly considering the (potentially large) composition $\mathcal{S} \parallel \mathcal{B} \parallel \mathcal{R}$.
To this end, assume that we have established the following statements:
\begin{itemize}
    \item
    $\mathcal{S} \modelsWrt{} \probPredicate{<0.1}{\mathit{collision}}$---the probability that $\mathcal{S}$ faces a collision is below $0.1$
    \item
    $\agTriple{\mathcal{B}}{}{\probPredicate{<0.1}{\mathit{collision}}}{\probPredicate{\geq 0.8}{\mathit{broadcast}}}$---if $\mathcal{B}$ observes a collision with probability below $0.1$, the message is broadcast with probability at least $0.8$
    \item
    $\agTriple{\mathcal{R}}{}{\probPredicate{\geq 0.8}{\mathit{broadcast}}}{\probPredicate{\geq 0.7}{\mathit{received}}}$---If the message is broadcast with probability at least $0.8$, then $\mathcal{R}$ receives the message with probability at least $0.7$

\end{itemize}
The asymmetric AG rule of Kwiatkowska et al.~\cite[Theorem 1]{Kwi+13} then reasons about the composed system: 
\[
\infer{(\mathcal{S} \parallel \mathcal{B}) \modelsWrt{} \probPredicate{\geq 0.8}{\mathit{broadcast}}}
  {\deduce{\agTriple{\mathcal{B}}{}{\probPredicate{<0.1}{\mathit{collision}}}{\probPredicate{\geq 0.8}{\mathit{broadcast}}}}
   {\mathcal{S}, R \modelsWrt{} \probPredicate{<0.1}{\mathit{collision}}}
  }
\qquad
\infer{(\mathcal{S} \parallel \mathcal{B} \parallel \mathcal{R}) \modelsWrt{} \probPredicate{\geq 0.7}{\mathit{received}}}
  {\deduce{\agTriple{\mathcal{R}}{}{\probPredicate{\geq 0.8}{\mathit{broadcast}}}{\probPredicate{\geq 0.7}{\mathit{received}}}}
   {(\mathcal{S} \parallel \mathcal{B}), R \modelsWrt{} \probPredicate{\geq 0.8}{\mathit{broadcast}}}
  }
\]
\end{exa}
\paragraph{Flavours of uncertainty.}
When the probabilities with which a system evolves are not known exactly---for instance due to modelling abstractions, measurement noise, or variability in the environment---verification results are preferably robust towards perturbations.
We consider two extensions of Segala’s PAs to formalize uncertain transition probabilities:
\begin{itemize}
\item \emph{parametric probabilistic automata (pPAs)}, a compositional variant of parametric MDPs \cite{Daw04,Jun+24}, where probabilities are expressed as polynomials of parameters, and
\item \emph{robust probabilistic automata (rPAs)}, in analogy to robust MDPs~\cite{Iye05,Wi+13}, where each transition is replaced by an uncertainty set of successor distributions; important sub-classes are convex rPAs and interval PAs (iPAs).
\end{itemize}
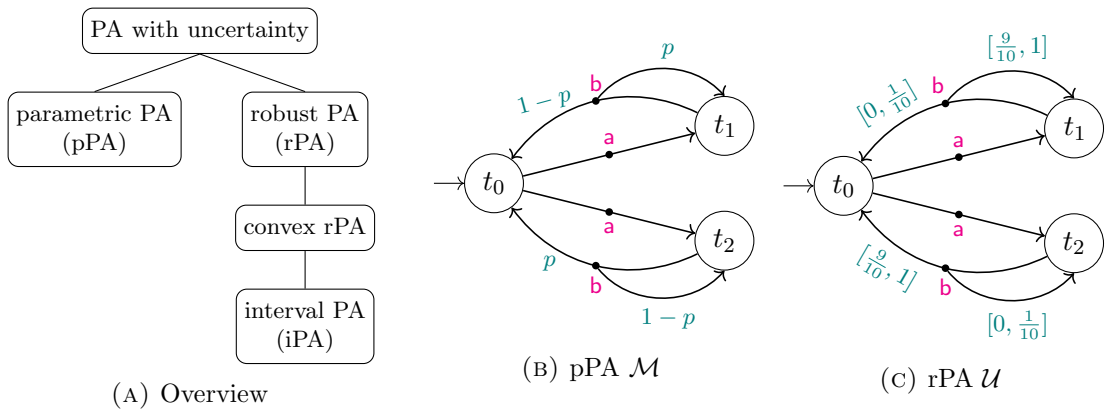
\begin{figure}[t] 
	\begin{subfigure}[c]{.39\textwidth}
		\centering 
        \footnotesize
\begin{forest}
for tree={
    grow=south,
    draw,
    rounded corners,
    align=center,
    parent anchor=south,
    child anchor=north,
    l sep=5mm,
    s sep=8mm
}
[PA with uncertainty
    [parametric PA\\(pPA), name=paramnode]
    [robust PA\\(rPA), name=robnode
        [convex rPA
            [interval PA\\(iPA)]
        ]
    ]
]
\end{forest}		\caption{Overview}\label{fig:overviewIntro}%
	\end{subfigure}\hfill
	\begin{subfigure}[c]{.30\textwidth}
		\centering 	
		\begin{tikzpicture}[mdp]

	\node[ps, init=left] (0)  {$t_0$};

	\node[ps,above right=0.2 and 2.5 of 0] (1)  {$t_{1}$};

    \node[ps,below right=0.2 and 2.5 of 0] (2)  {$t_{2}$};

	\path[ptrans]

	(0) edge[bend right=0] node[pos=0.5,above] {\tact{\lab}} node[dist] (d0a) {} node[pos=0.75,left] {} (1)

    (1) edge[bend right=40] node[pos=0.5,above]  {\tact{\altlab}} node[dist]  (d1b) {} node[pos=0.75,above,rotate=15]{\tprob{1-p}} node[pos=0.5,below] {} (0)

    (d1b) edge[bend left=50] node[above,pos=0.55] {\tprob{p}} node[above left=3pt,pos=0.4] {} (1.north)

	(0) edge[bend left=0] node[pos=0.5,below] {\tact{\lab}} node[dist] (d0a) {} node[pos=0.75,left] {} (2)

    (2) edge[bend left=40] node[pos=0.5,below]  {\tact{\altlab}} node[dist]  (d2b) {} node[pos=0.75,below]{\tprob{p}} node[pos=0.5,below] {} (0)

    (d2b) edge[bend right=50] node[below,pos=0.55] {\tprob{1-p}} node[above left=3pt,pos=0.4] {} (2.south)
    
;
\end{tikzpicture}%
		\caption{pPA $\ppa$}\label{fig:ppaIntro}%
	\end{subfigure}\hfill
	\begin{subfigure}[c]{.30\textwidth} 
		\centering%
		\begin{tikzpicture}[mdp]

	\node[ps, init=left] (0)  {$t_0$};

	\node[ps,above right=0.2 and 2.5 of 0] (1)  {$t_{1}$};

    \node[ps,below right=0.2 and 2.5 of 0] (2)  {$t_{2}$};

	\path[ptrans]

	(0) edge[bend right=0] node[pos=0.5,above] {\tact{\lab}} node[dist] (d0a) {} node[pos=0.75,left] {} (1)

   (1) edge[bend right=40] node[pos=0.5,above,xshift=-1mm]  {\tact{\altlab}} node[dist]  (d1b)  {} node[pos=0.75,above,rotate=30]{\tprob{[0,\frac{1}{10}]}} node[pos=0.5,below] {} (0)

    (d1b) edge[bend left=50] node[above,pos=0.55] {\tprob{[\frac{9}{10},1]}} node[above left=3pt,pos=0.4] {} (1.north)

	(0) edge[bend left=0] node[pos=0.5,below] {\tact{\lab}} node[dist] (d0a) {} node[pos=0.75,left] {} (2)

    (2) edge[bend left=40] node[pos=0.5,below]  {\tact{\altlab}} node[dist]  (d2b) {} node[pos=0.75,below,rotate=-30]{\tprob{[\frac{9}{10},1]}} node[pos=0.5,below] {} (0)

    (d2b) edge[bend right=50] node[below,pos=0.55] {\tprob{[0,\frac{1}{10}]}} node[above left=3pt,pos=0.4] {} (2.south)
    
;
\end{tikzpicture}%
		\caption{rPA $\rpa$}\label{fig:rpaIntro}%
	\end{subfigure} 
	\caption{Overview of uncertain PAs (left), example pPA (middle), and example iPA (right)}\label{fig:paIntro}
\end{figure}
\Cref{fig:overviewIntro} provides an overview of the considered formalisms.

\paragraph{Parametric PAs.}
For pPAs, uncertain model quantities, e.g., the bias of a coin-flip or the probability of a sensor misreading, are represented using polynomials over a set of parameters.
A valuation replaces each parameter by a real value and instantiates the model to a concrete PA, i.e., uncertainty is resolved \emph{globally} by a single valuation that simultaneously fixes all transition probabilities and may induce dependencies between probabilities at different states.
\Cref{fig:ppaIntro} shows an example pPA over a single parameter $p$.
Changing the valuation of $p$ affects the successor state distributions at both states $t_1$ and $t_2$.
The \emph{verification} problem for pPAs asks whether a specification holds under all parameter valuations in a given parameter space and is co-ETR complete for reachability probability objectives~\cite{JKPW21}.
Verifying pPAs is thus significantly more complex compared to PAs without parameters.
\begin{exa}[Communication system as a pPA]\label{ex:comm-pPA}
    We revisit \Cref{ex:high-level-PA}, but now assume that the reliability of $\mathcal{S}$, $\mathcal{B}$, and $\mathcal{R}$ is influenced by parameters and the precise values of these parameters vary depending on network conditions, interference, or other factors.
    
   In a pPA, these uncertainties are given as polynomials in these parameters.
  Our goal is to verify that under all parameter instantiations in a given parameter space $\region$, the message is successfully \emph{received} with at least probability $0.7$, formally denoted by
    \[
    (\mathcal{S} \parallel \mathcal{B} \parallel \mathcal{R}), \region ~\modelsWrt{}~ \probPredicate{\geq 0.7}{\mathit{received}}.
    \]
    We reason about the composed system using the AG rule stated in \Cref{theo:pag_asym_rule} in \Cref{sec:pag}:
    \[
    \infer{(\mathcal{S} \parallel \mathcal{B}), \region \modelsWrt{} \probPredicate{\geq 0.8}{\mathit{broadcast}}}
      {\deduce{\agTriple{\mathcal{B}, \region}{}{\probPredicate{<0.1}{\mathit{collision}}}{\probPredicate{\geq 0.8}{\mathit{broadcast}}}}
       {\mathcal{S}, R \modelsWrt{} \probPredicate{<0.1}{\mathit{collision}}}
      }
    \qquad
    \infer{(\mathcal{S} \parallel \mathcal{B} \parallel \mathcal{R}), \region \modelsWrt{} \probPredicate{\geq 0.7}{\mathit{received}}}
      {\deduce{\agTriple{\mathcal{R}, \region}{}{\probPredicate{\geq 0.8}{\mathit{broadcast}}}{\probPredicate{\geq 0.7}{\mathit{received}}}}
       {(\mathcal{S} \parallel \mathcal{B}), R \modelsWrt{} \probPredicate{\geq 0.8}{\mathit{broadcast}}}
      }
    \]
\end{exa}
In addition to property satisfaction, we also consider compositional reasoning about \emph{monotonicity} in parameters~\cite{Spe+19,Spe+21}: we develop AG-style rules for pPAs that allow one to infer monotonicity of selected parameters in the composite system from the corresponding monotonicity properties of its constituting components. 

\paragraph{Robust PAs.}
In rPAs, uncertainty is represented \emph{locally} by uncertainty sets of successor
distributions for each state–action pair.
A \emph{nature} player reacts to the controller’s action choices by selecting a
concrete distribution from the corresponding
uncertainty set at each step~\cite{Wi+13,Iye05}. 
We adopt the state–action rectangularity assumption, under which the uncertainty sets for different state–action pairs are independent.
An rPA is convex if all uncertainty sets are convex. An interval PA is an rPA where all uncertainty sets can be expressed as intervals.
\Cref{fig:rpaIntro} shows an example iPA. 
The transition probabilities at $t_1$ and $t_2$ can be instantiated independently.

\begin{exa}[Communication system as an rPA]\label{ex:comm-rPA}
Returning to the system of \Cref{ex:high-level-PA}, we now assume that we only know bounds or qualitative constraints on the probabilities of collisions, losses, and missed receptions. %
A robust probabilistic automaton (rPA) model assigns to each state–action pair, an uncertainty set of successor distributions that captures all plausible failure behaviours. %
When the system runs, a controller strategy picks actions as in a PA, while nature, possibly with memory, resolves the uncertainty by choosing a distribution from the corresponding set at each step. %
In contrast to the parametric setting of \Cref{ex:comm-pPA}, where a valuation fixes all probabilities globally before execution, the robust model allows nature to react to the history of collisions and losses when resolving uncertainty. %
AG resoning for such rPAs is particularly usefull, as uncertainty sets can also grow exponentially under parallel composition, exacerbating state-space explosion. %
\end{exa}

\paragraph{Simulation-based AG reasoning.}
As an alternative to property-based reasoning as in~\cite{Kwi+13}, simulation-based compositional reasoning has been suggested in \cite{Cob+03,Pas+08} and applied to PAs in~\cite{Kom+12}. %
The approach formalizes assumptions and guarantees as PAs rather than as logical formulas. A key advantage of this approach is \emph{completeness}: for each fulfilled specification, there is some appropriate assumption that allows proving the specification compositionally.

\paragraph{Motivation.}
\emph{This work introduces a framework for compositional reasoning about parametric and robust PAs.}
The case studies presented by Kwiatkowska et al.~\cite{Kwi+13} demonstrate the practical applicability of AG reasoning within a certain setting. These findings provide strong motivation for extending this approach to the parametric and robust domains.
In this work, we develop the theoretical foundations of an AG reasoning framework for pPAs and rPAs, leveraging results from (non-parametric) PAs~\cite{Kwi+13}. 
Compositional reasoning has great potential and can be crucial in verifying complex uncertain systems that are too large to handle monolithically.

\paragraph{Contributions.}
To the best of our knowledge, \emph{we provide the first framework for compositional reasoning about parametric and robust Markov models}.
Our main contributions are as follows. 
\begin{itemize}
    \item We introduce and formalize pPAs, i.e., compositional probabilistic automata with parametric transitions. 
    \item We provide a conservative extension of \emph{strategy projections}~\cite{Seg95,Kwi+13} to pPAs, including a more natural definition based on conditional probabilities.
    Strategy projections are essential for correctness of compositional reasoning as they allow to link measures of a composed model to measures of its constituting components.
    \item We present rules for assume-guarantee reasoning, generalizing an established framework by Kwiatkowska et al.~\cite{Kwi+13} to the parametric setting.
    The framework applies to $\omega$-regular and expected total reward properties as well as multi-objective combinations thereof.
\item 
We provide a new proof rule for compositional reasoning about monotonicity in pPAs.
Knowing that a measure of interest---either the probability to satisfy an $\omega$-regular specification or an expected total reward---is monotone in one or more parameters can significantly speed up verification~\cite{Jun+24,Spe+19}.
Our rule allows to derive monotonicity w.r.t.\ a composed pPA by only determining monotonicity for its components. 
\item 
We show that the AG rules of~\cite{Kwi+13} cannot be lifted straightforwardly to rPAs when memoryless nature, non-convex uncertainty sets, or interval-arithmetic relaxations of parallel composition are considered.
\item For convex rPAs with history-dependent nature, we recover sound AG rules via a reduction to (possibly uncountably branching) PAs and a convexity-preserving parallel composition operator.
\item Finally, we lift simulation-based AG reasoning of~\cite{Kom+12} to pPAs by introducing strong and \emph{robust strong} simulation preorders and showing how they can be used in sound and complete AG-style rules. 
\end{itemize}

\paragraph{Overview.}
We introduce pPAs in \Cref{sec:preliminaries_math} and discuss strategy projections in \Cref{sec:strat_projections}.
\Cref{sec:verification} outlines properties of interest and \Cref{sec:pag} presents our AG rules for pPAs. 
We outline results for monotonicity in \Cref{sec:pag_mono}. 
\Cref{sec:studying_ag_for_rpa_semantics} presents our results for rPAs. 
In \Cref{sec:AG_sim}, we develop simulation-based AG reasoning for pPAs. 
We discuss related work in \Cref{sec:related_work} and \Cref{sec:conclusion} concludes the paper. 
Proofs omitted in the main part of the paper are provided in \Cref{app:proofs}. 

\paragraph{Extended version.}
This article extends the CONCUR~2025 version~\cite{Mer+25CONCUR} towards AG reasoning for \emph{robust probabilistic automata} (rPAs) as a compositional variant of robust MDPs.
In addition, we present an alternative AG framework for pPAs based on simulation conformance, in line with the approach of~\cite{Kom+12}.

\section{Preliminaries}
\label{sec:preliminaries_math}
For sets $X$ and $Y$, let $f \colon X \hookrightarrow Y$ denote a \emph{partial function} from $X$ to $Y$ with domain $\domain(f) \subseteq X$.
The projection of $f$ to a set $Z$ is written as $\restrOfTo{f}{Z} \colon (X \cap Z) \to Y$.
Iverson brackets $\iverson{\varphi} \in \{0,1\}$ map a Boolean condition $\varphi$ to $1$ if $\varphi$ holds and $0$ otherwise.

$\rats[\parameterSetOf{}]$ denotes the set of (multivariate) \emph{polynomials} over a finite set of real-valued \emph{parameters} $\parameterSetOf{}  = \{p_1, \dots, p_n\}$.  
A (parameter) \emph{valuation} for $\parameterSetOf{}$ is a function $\valuation \colon \parameterSetOf{} \to \reals$. 
Evaluating a polynomial $f \in \rats[\parameterSetOf{{}}]$ at $\valuation$ yields $f[\valuation] \in \reals$. 
A \emph{region} $\region$ for $\parameterSetOf{}$ is a set of valuations.
For $p \in \parameterSetOf{}$, we define the valuation $\valuationVector{e}_p$ with $\valuationVector{e}_p(q) = \iverson{p{=}q}$ for $q \in \parameterSetOf{}$.

A \emph{parametric distribution}\footnote{We use the term parametric \emph{distribution}\textemdash{}rather than parametric \emph{function}\textemdash{}to emphasize that we are typically interested in functions $\mu$ where $\mu[\valuation]$ is a (sub)probability distribution.} for $\parameterSetOf{}$ over a finite set $\stateSetOf{}$ is a function $\mu \colon \stateSetOf{} \to (\rats[\parameterSetOf{}] \cup \reals)$.
Applying valuation $\valuation$ to $\mu$ yields $\mu[\valuation] \colon S \to \reals$ with $\mu[\valuation](s) = \mu(s)[\valuation]$ for all $s \in S$. 
We call $\mu \colon S \to [0,1]$ a \emph{subdistribution} if $\sum_{s \in \stateSetOf{}} \mu(s) \leq 1$ and a \emph{distribution} if $\sum_{s \in \stateSetOf{}} \mu(s) =1$.

The sets of parametric distributions, subdistributions, and distributions over $\stateSetOf{}$ are denoted by $\paramDist{\parameterSetOf{}}{\stateSetOf{}}$, $\subDist{\stateSetOf{}}$, and $\dist{\stateSetOf{}}$, respectively. 
For $s \in \stateSetOf{}$, $\indicatorFct{s} \in \dist{\stateSetOf{}}$ is the \emph{Dirac} distribution with
$\indicatorFct{s}(s') = \iverson{s'{=}s}$.
For sets $S_1,S_2$, the \emph{product} of $\mu_1 \in \paramDist{\parameterSetOf{}}{S_1}$ and $\mu_2 \in \paramDist{\parameterSetOf{}}{S_2}$ is the distribution $\mu_1 {\times} \mu_2 \in \paramDist{\parameterSetOf{}}{S_1\times S_2}$ with $(\mu_1 {\times} \mu_2)(s_1,s_2) = \mu_1(s_1) \cdot \mu_2(s_2)$.

\subsection{Parametric probabilistic automata}
\label{sec:preliminaries_ppa}
We combine probabilistic automata (PA) \cite{Seg+95,Sto+02} with parametric Markov models~\cite{Jun+24}. 
\begin{defi}\label{def:ppa}
A \emph{parametric probabilistic automaton (pPA)} over a finite alphabet $\alphabetOf{}$ is a tuple $\ppa = \ppaTupleOf{}$, where 
$\stateSetOf{}$, $\parameterSetOf{}$, and $\actSetOf{}$ are finite sets of states, parameters, and actions, respectively, $\initialOf{} \in \stateSetOf{}$ is an initial state,
		 $\transFctOf{} \colon  (\stateSetOf{} \times \actSetOf{}) \hookrightarrow \paramDist{\parameterSetOf{}}{\stateSetOf{}}$ is a parametric transition function, and  
	 $\syncOf{} \colon \domain(\transFctOf{}) \to  \alphabetOf{}$ is a labeling function.
\end{defi}
Let $\ppa = \ppaTupleOf{}$ be a pPA. For $s \in S$, 
$\ppa_s$ is the pPA where the initial state is set to~$s$.
We set $\transFctOf{}(s, \alpha, s') = \transFctOf{}(s, \alpha)(s')$ if $(s,\alpha) \in \domain(\transFctOf{})$ and otherwise $\transFctOf{}(s, \alpha, s') = 0$.
$\ppa$ is a \emph{(non-parametric)} PA if $\transFctOf{}(s, \alpha) \in \dist{\stateSetOf{}}$ for all $(s,\alpha) \in \domain(\transFctOf{})$.
In this case, $\transFctOf{}(s,\alpha,s')$ is the probability to transition to successor state $s'$ when action $\alpha$ is selected at state $s$.

The \emph{instantiation} of $\ppa$ at valuation $\valuation$ for $\parameterSetOf{}$ is the pPA 
	$\ppa[\valuation] = \left(\stateSetOf{}, \initialOf{},\emptyset,  \actSetOf{} , \transFctOf{}[\valuation], \syncOf{} \right)$, 
    where $\domain(\transFctOf{}[\valuation]) = \domain(\transFctOf{})$ and $\transFctOf{}[\valuation](s,\alpha) = \transFctOf{}(s,\alpha)[\valuation]$.
If $\ppa[\valuation]$ is a non-parametric PA, we say $\valuation$ is \emph{well-defined} for $\ppa$.
A valuation $\valuation$ is \emph{graph-preserving} for $\ppa$ if it is well-defined and for all $s,  s' \in \stateSetOf{}$ and $\alpha \in \actSetOf{}$: $\transFctOf{}(s,\alpha,s')[\valuation] = 0$ iff $\transFctOf{}(s,\alpha,s') = 0$.
A region $\region$ is well-defined (graph-preserving) if all its valuations $\valuation \in \region$ are well-defined (graph-preserving).

\begin{figure}[t] 
	\begin{subfigure}[c]{.35\textwidth}
		\renewcommand{\modelType}{\ppa}
		\centering 	
		\begin{tikzpicture}[mdp]

	\node[ps, init=left] (0)  {$s_0$};

	\node[ps,below=1 of 0] (1)  {$s_1$};

	\path[ptrans]
	
	(0) edge[loop right] node[pos=0.45,below right] {\tact{\altlab},\tact{\altaltlab}} node[dist] (d1a) {} (0)
	
	(0) edge[bend right=0] node[pos=0.5,right] {\tact{\lab}} node[dist] (d0a) {} node[pos=0.75,left] {\tprob{{\ifthenelse{\equal{\modelType}{pa}}{\frac{9}{10}}{1-p}}}} (1)
	
	(d0a) edge[bend left=60] node[left,pos=0.45] {\tprob{{\ifthenelse{\equal{\modelType}{pa}}{\frac{1}{10}}{p}}}} node[above left=3pt,pos=0.4] {} (0)
	
	(1) edge[loop right] node[pos=0.5,above right] {\tact{\altlab}} node[dist] (d1c) {} node[pos=0.75,below] {\tprob{}} node[pos=0.5,below] {} (1)

	;
\end{tikzpicture}%
		\caption{pPA $\ppa_1$}\label{fig:pPAM1}%
	\end{subfigure}\hfill
	\begin{subfigure}[c]{.55\textwidth} 
		\renewcommand{\extended}{false}
		\renewcommand{\modelType}{\ppa}
		\centering%
		\begin{tikzpicture}[mdp]
	\node[ps, init=left] (t0)  {$t_0$};
	
	\node[ps,above right=0.25 and 2 of t0] (t1)  {$t_1$};

	\node[ps,below right=0.25 and 2 of t0] (t2)  {$t_2$};
	
	\node[ps,right=2 of t1] (t3)  {$t_3$};
	
	\node[ps,right=2 of t2] (t4)  {$t_4$};
	
	\node[dist, right=0.5 of t0] (d0a) {};
		
	\path[ptrans]

		(t0) edge[-] node[pos=1, below] {\tact{\lab}} (d0a)
		(d0a) edge[bend right=0] node[pos=0.5,above,yshift=0.5ex,xshift=-0.5ex] {\tprob{\ifthenelse{\equal{\modelType}{pa}}{\frac{9}{10}}{1-p}}} (t1)
		(d0a) edge[bend right=0] node[below,pos=0.55,yshift=-0.5ex,xshift=-0.5ex] {\tprob{\ifthenelse{\equal{\modelType}{pa}}{\frac{1}{10}}{p}}} node[above,pos=0.4] {} (t2)

		(t1) edge[bend right=0] node[pos=0.25,above] {\tact{\lab}} node[dist, pos=0.25] (d1a) {} node[pos=0.6,above] {\tprob{\ifthenelse{\equal{\modelType}{pa}}{\frac{9}{10}}{q}}} (t3)
		(d1a) edge[bend left=25] node[below,pos=0.25,xshift=-1ex] {\tprob{\ifthenelse{\equal{\modelType}{pa}}{\frac{1}{10}}{1-q}}} node[above,pos=0.4] {} (t4)
		
		(t2) edge[bend right=0] node[pos=0.25,below] {\tact{\altaltlab}} node[dist, pos=0.25] (d1b) {} node[pos=0.6,below] {\tprob{\ifthenelse{\equal{\modelType}{pa}}{\frac{9}{10}}{\frac{9}{10}}}} (t4)
		(d1b) edge[bend right=25] node[above, pos=0.25] {\tprob{\ifthenelse{\equal{\modelType}{pa}}{\frac{1}{10}}{\frac{1}{10}}}} node[above,pos=0.4] {} (t3)

		(t3) edge[loop right] node[pos=0.45,below right] {{\ifthenelse{\equal{\extended}{true}}{\tactext{\altlab},}{}}\tact{ \frownie}} node[dist] (d3frown) {} node[pos=0.15,below] {\tprob{}} node[pos=0.25,below] {} (t3)
		
		(t4) edge[loop right] node[pos=0.5,above right] {\ifthenelse{\equal{\extended}{true}}{\tactext{\altlab},}{}{\tact{\altaltlab}}} node[dist] (d4b) {} node[pos=0.15,below] {\tprob{}} node[pos=0.25,below] {} (t4)

	;	
	
	{\ifthenelse{\equal{\extended}{true}}{\path[ptrans] (t2) edge[loop above] node[pos=0.75,right] {{\ifthenelse{\equal{\extended}{true}}{\tactext{\altlab}}{}}} node[dist] (d1c) {} (t2);}{}}
	
	{\ifthenelse{\equal{\extended}{true}}{\path[ptrans] (t1) edge[loop above] node[pos=0.75,right] {{\ifthenelse{\equal{\extended}{true}}{\tactext{\altlab}}{}}} node[dist] (d1c) {} (t2);}{}}
	
	{\ifthenelse{\equal{\extended}{true}}{\path[ptrans] (t0) edge[loop above] node[pos=0.75,right] {{\ifthenelse{\equal{\extended}{true}}{\tactext{\altlab}}{}}} node[dist] (d1c) {} (t0);}{}}

\end{tikzpicture}%
		\caption{pPA $\ppa_2$}\label{fig:pPAM2}%
	\end{subfigure} 
	\caption{Example pPAs $\ppa_1$ and $\ppa_2$.}\label{fig:ppa1_and_ppa2}
\end{figure}
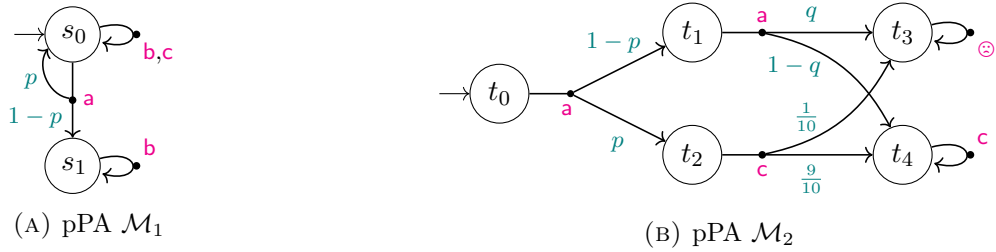
\begin{exa}
\label{ex:pPAM1M2}
	Consider the pPA $\ppa_1$ in \Cref{fig:pPAM1} and the pPA $\ppa_2$ in \Cref{fig:pPAM2}. 
	$\ppa_1 = \ppaTupleOf{1}$, where $\stateSetOf{1} = \{s_0,s_1\}$, $\initialOf{1} = s_0$ and $\parameterSetOf{1}=\{p\}$  
    and alphabet $\alphabetOf{1}=\{\lab,\altlab,\altaltlab\}$. 
    We omit actions in the figure and only display labels on transitions\footnote{%
    Actions and labels are distinct, as the parallel composition (see \Cref{def_ppa_composition}) requires $\actSetOf{i} \cap (\alphabetOf{1} \cup \alphabetOf{2}) = \emptyset$.}. 
    Concretely, for each state $s \in \stateSetOf{1}$ and each label  $a \in \{\lab,\altlab,\altaltlab\}$ that occurs on an outgoing transition of $s$, we define a unique action $\alpha_{s,a} \in \actSetOf{1}$ with $\syncOf{1}(s,\alpha_{s,a}) = a$. 
    In $\ppa_1$ and remaining examples, each state has at most one outgoing transition per label; we therefore simply refer to actions via their labels for readability. 


	Similarly, $\ppa_2 = \ppaTupleOf{2}$,  where $\stateSetOf{2} = \{t_0,\dots, t_4\}$, $\initialOf{2} = t_0$, and $\parameterSetOf{2}=\{p,q\}$. 
    Its alphabet is $\alphabetOf{2}=\{\lab,\altaltlab, \frownie\}$. 
    Again, only the labels are shown in \Cref{fig:pPAM2}; each displayed label corresponds to a unique underlying action. 
\end{exa}

An \emph{infinite path} of $\ppa$ is an alternating sequence $\infpath = s_0, \alpha_0, s_1, \alpha_1, \dots$ of states $s_i \in \stateSetOf{}$ and actions $\alpha_i \in \actSetOf{}$ such that 
$(s_i, \alpha_i) \in \domain(\transFctOf{})$ for all $i \ge 0$.
A finite path of length $n \in \nats$ is a prefix $\finpath = s_0, \alpha_0, \dots, s_n$ of an infinite path, ending in a state $\last{\finpath} = s_n \in \stateSetOf{}$.
$\infPathsOf{\ppa}{}$ and $\finPathsOf{\ppa}{}{}$ are the sets of infinite and finite paths of $\ppa$, respectively.
For a (finite or infinite) path $\infpath \in \infPathsOf{\ppa}{} \cup \finPathsOf{\ppa}{}$, we write $\vert \infpath \vert \in \nats \cup \set{\infty}$ for its length and $\pi[0,j]$ for its prefix of length $j \le \vert \infpath \vert$.
We deliberately allow paths that take transitions with probability 0. 
As a consequence, a path of a pPA $\ppa$ is always also a path of any of its instantiations $\ppa[\valuation]$\textemdash{}even if $\valuation$ is not graph-preserving.

Strategies\textemdash{}also known as policies, schedulers or adversaries\textemdash{}resolve nondeterminism by assigning (sub-)distributions over enabled actions based on the history---i.e., a finite path---observed so far.
We allow for partial strategies that, intuitively, can choose none of the enabled actions to reflect the case that no further transition is executed.
\begin{defi}\label{defi_strategy_ppa}
	A \emph{(partial) strategy} for $\ppa$ is a function $\strategy \colon \finPathsOf{\ppa}{} \to \subDist{\actSetOf{}}$ such that 
	$\strategy(\finpath)(\alpha) > 0 $ implies $(\last{\finpath},\alpha) \in \domain(\transFctOf{})$.
	A strategy $\strategy \colon \finPathsOf{\ppa} \to \dist{\actSetOf{}}$ is called \emph{complete}. 
	The set of all partial and complete strategies on $\ppa$ are denoted by $\strategysetOf{\ppa}{\star}$, where $\star \in \{\prt, \comp\}$, respectively. 
	A \emph{memoryless strategy} only depends on the last state of $\finpath$.
	The set of memoryless strategies on $\ppa$ is denoted by $\strategysetOf{\ppa}{\mless,\star}$.
\end{defi}
For a strategy $\strategy$ for $\ppa$, we may write $\strategy(\finpath, \alpha)$ instead of $\strategy(\finpath)(\alpha)$.
If $\strategy$ is memoryless, we write $\strategy(s_n, \alpha)$ instead of $\strategy(\finpath, \alpha)$, where $s_n = \last{\finpath}$.

A well-defined instantiation $\valuation$ and a strategy $\strategy$ for $\ppa$ yield a purely probabilistic process described by the \emph{(sub)probability measure} 
$\PrOf{\ppa}{\valuation,\strategy}{}$ on the measurable subsets of $\infPathsOf{\ppa}{}$, which is obtained by a standard cylinder set construction~\cite{BK08}: 
\[\cyl(\finpath) = \{ \pi \in \infPathsOf{\ppa}{} \mid \finpath \text{ is a prefix of } \pi \}\] is the \emph{cylinder set} of a \emph{finite} path $\finpath = s_0, \alpha_0,  \dots, s_n$ of $\ppa$ and we set
	\[
	\PrOf{\ppa}{\valuation,\strategy}{\cyl (\finpath)}
	~=~ 
	\iverson{s_0 = \initialOf{}} 
	\cdot \prod_{ i= 0}^{n-1} \strategy(\finpath[0,i], \alpha_i) \cdot \transFctOf{}(s_{i}, \alpha_{i}, s_{i+1})[\valuation].
\] 
This definition extends uniquely to a probability measure on \emph{all} measurable sets of infinite paths. 
We further lift $\PrOf{\ppa}{\valuation,\strategy}{}$ to (sets of) finite paths and write, e.g., $\PrOf{\ppa}{\valuation,\strategy}{\finpath}$ for $\finpath \in \finPathsOf{\ppa}$ or $\PrOf{\ppa}{\valuation,\strategy}{\Pi}$ for $\Pi \subseteq \finPathsOf{\ppa}$---implicitly referring to (unions of) cylinder sets.
If $\ppa$ is a (non-parametric) PA, we may omit $\valuation$ and write $\PrOf{\ppa}{\strategy}{}$.
Well-defined $\valuation$ yields $\PrOf{\ppa[\valuation]}{\strategy}{} = \PrOf{\ppa}{\valuation,\strategy}{}$.
\begin{exa}
    For the pPA $\ppa_2$ from \Cref{fig:pPAM2} and a well-defined valuation $\valuation$, the probability to reach the state $t_3$ under valuation $\valuation$ is $(1-\valuation(p)) \cdot \valuation(q) + \valuation(p) \cdot \frac{1}{10}$. 
\end{exa}
\begin{rem}\label{rem:padef}
Our definition of PA slightly deviates from related work \cite{Seg+95,Kwi+13,Kom+12,LL19}, which commonly define a transition \emph{relation} $\transRelationOf{} \subseteq \stateSetOf{} \times \alphabetOf{} \times \dist{\stateSetOf{}}$ instead of functions $\transFctOf{}$ and $\syncOf{}$.
In our setting, a pair $(s,\alpha) \in \domain(\transFctOf{})$ uniquely identifies both a label $\syncOf{}(s,\alpha) \in \alphabetOf{}$, and a distribution over successor states $\transFctOf{}(s,\alpha) \in \dist{\stateSetOf{}}$, which significantly simplifies formalizations related to pPAs.
In particular, any strategy for $\ppa$ immediately also applies to instantiations of $\ppa$ and vice versa, i.e., we have $\strategysetOf{\ppa}{\star} = \strategysetOf{\ppa[\valuation]}{\star}$ for any valuation $v$. 
On the other hand, our variant does not affect expressiveness of non-parametric PAs as one can convert between the two formalisms. 
\end{rem}
We lift parallel composition of PA~\cite{Seg+95} to pPAs. Composed pPAs synchronize on common transition labels while behaving autonomously on non-common labels.
For simplicity, we assume that composed pPAs consider a common set of parameters $\parameterSetOf{}$.\footnote{If two pPAs have different parameter sets $\parameterSetOf{1} \neq \parameterSetOf{2}$, the assumption can be established by considering $\parameterSetOf{} = \parameterSetOf{1} \cup \parameterSetOf{2}$ instead, potentially adding (unused) parameters to the individual pPAs.}
\begin{defi}[Parallel Composition for pPAs]
\label{def_ppa_composition}
	For $i=1,2$, let $\ppa_{i} = (\stateSetOf{i}, \initialOf{i}, \parameterSetOf{}, \actSetOf{i}, \transFctOf{i}, \syncOf{i})$  be two pPAs over alphabets $\alphabetOf{i}$ with $\actSetOf{i} \cap (\alphabetOf{1} \cup \alphabetOf{2}) = \emptyset$.
The \emph{parallel composition} of $\ppa_1$ and $\ppa_2$ is given by the pPA 
$\ppa_{1} \parallel \ppa_{2} = \big(\stateSetOf{1} \times \stateSetOf{2}, (\initialOf{1}, \initialOf{2}), \parameterSetOf{}, \actSetOf{\parallel}, \transFctOf{\parallel}, \syncOf{\parallel}\big)$
 over $\alphabetOf{1} \cup \alphabetOf{2}$, where 
\begin{itemize}
\item $\actSetOf{\parallel} = (\actSetOf{1} \times \actSetOf{2}) \cupdot (\actSetOf{1} \times \alphabetOf{1} \setminus \alphabetOf{2}) \cupdot (\alphabetOf{2} \setminus \alphabetOf{1} \times \actSetOf{2})$,
\item
for each $(s_1, \alpha_1) \in \domain(\transFctOf{1})$, $(s_2, \alpha_2) \in \domain(\transFctOf{2})$ with $\syncOf{1}(s_1, \alpha_1) = \syncOf{2}(s_2, \alpha_2) \in \alphabetOf{1} \cap \alphabetOf{2}$:
\[\transFctOf{\parallel}((s_1,s_2), (\alpha_1, \alpha_2)) = \transFctOf{1}(s_1,\alpha_1) \times \transFctOf{2}(s_2,\alpha_2) 
\quad\text{and}\quad \syncOf{\parallel}((s_1,s_2), (\alpha_1, \alpha_2)) = \syncOf{1}(s_1, \alpha_1),\]
\item for each $(s_1, \alpha_1) \in \domain(\transFctOf{1})$, $s_2 \in \stateSetOf{2}$ with $\syncOf{1}(s_1, \alpha_1) = \lab_1 \in \alphabetOf{1} \setminus \alphabetOf{2}$: 
\[ \transFctOf{\parallel}((s_1,s_2), (\alpha_1, \lab_1)) = \transFctOf{1}(s_1,\alpha_1) \times \indicatorFct{s_2}
\quad\text{and}\quad \syncOf{\parallel}((s_1,s_2), (\alpha_1, \lab_1)) = \lab_1 = \syncOf{1}(s_1, \alpha_1),\]
\item for each $s_1 \in \stateSetOf{1}$, $(s_2, \alpha_2) \in \domain(\transFctOf{2})$ with $\syncOf{2}(s_2, \alpha_2) = \lab_2 \in \alphabetOf{2} \setminus \alphabetOf{1}$: 
\[ \transFctOf{\parallel}((s_1,s_2), (\lab_2, \alpha_2)) = \indicatorFct{s_1} \times \transFctOf{2}(s_2,\alpha_2)
\quad\text{and}\quad \syncOf{\parallel}((s_1,s_2), (\lab_2, \alpha_2)) = \lab_2 = \syncOf{2}(s_2, \alpha_2).\]
\end{itemize}
\end{defi}
The parallel composition of parametric probabilistic automata (pPAs) is associative, meaning that $(\ppa_{1} \parallel \ppa_{2}) \parallel \ppa_3$ and $\ppa_{1} \parallel (\ppa_{2}\parallel \ppa_3)$ are equivalent up to state renaming. 
Therefore, we denote this composition as $\ppa_{1} \parallel \ppa_{2}\parallel \ppa_3$. 
\renewcommand{\modelType}{\ppa}
\renewcommand{\prodAutomata}{false}
\begin{figure}[!t]
    \centering
	\begin{tikzpicture}[mdp]
	\node[pswide, init=left] (00)  {$s_0, t_0\ifthenelse{\equal{\prodAutomata}{true}}{,p_0}{}$};	
	\node[pswide,right=6 of 00,color=white] (10)  {};
	
	\node[pswide,below right=1.5 and 0.6 of 00] (01)  {$s_0,t_1\ifthenelse{\equal{\prodAutomata}{true}}{,p_0}{}$};
	\node[pswide,below left=1.5 and 0.6 of 00] (02)  {$s_0,t_2\ifthenelse{\equal{\prodAutomata}{true}}{,p_0}{}$};
	
	\node[pswide,below =1.5 of 01] (03)  {$s_0,t_3\ifthenelse{\equal{\prodAutomata}{true}}{,p_0}{}$};
	\node[pswide,below =1.5 of 02] (04)  {$s_0,t_4\ifthenelse{\equal{\prodAutomata}{true}}{,p_0}{}$};

	\node[pswide,below right=1.5 and 0.6 of 10] (11)  {$s_1,t_1\ifthenelse{\equal{\prodAutomata}{true}}{,p_0}{}$};
	\node[pswide,below left=1.5 and 0.6 of 10] (12)  {$s_1,t_2\ifthenelse{\equal{\prodAutomata}{true}}{,p_0}{}$};
	
  	\node[pswide,below =1.5 of 11] (13)  {$s_1,t_3\ifthenelse{\equal{\prodAutomata}{true}}{,p_0}{}$};
	\node[pswide,below =1.5 of 12] (14)  {$s_1,t_4\ifthenelse{\equal{\prodAutomata}{true}}{,p_0}{}$};


	\path[ptrans]

	(00) edge[loop above] node[pos=0.65, right]  {\tact{\altlab}} node[dist] (d00b) {} node[pos=0.5,left] {\tprob{}} node[pos=0.35,below] {} (00)
		
	(03) edge[loop left] node[pos=0.55, below left] {\tact{\altlab\ifthenelse{\equal{\prodAutomata}{true}}{}{,\frownie}}} node[dist] (d01c) {} node[pos=0.5,left] {\tprob{}} node[pos=0.35,below] {} (03)
	
	(02) edge[loop left] node[pos=0.55, below left]  {\tact{\altlab}} node[dist] (d01c) {} node[pos=0.5,left] {\tprob{}} node[pos=0.35,below] {} (03)
	
	(04) edge[loop left] node[pos=0.55, below left]  {\tact{\altlab},\tact{\altaltlab}} node[dist] (d01c) {} node[pos=0.5,left] {\tprob{}} node[pos=0.35,below] {} (04)
	
	(01) edge[loop left] node[pos=0.55, below left]  {\tact{\altlab}} node[dist] (d01c) {} node[pos=0.5,below] {\tprob{}} node[pos=0.35,below] {} (01)
	
	(11) edge[loop left] node[pos=0.55, below left]  {\tact{\altlab}} node[dist] (d01c) {} node[pos=0.5,left] {\tprob{}} node[pos=0.35,below] {} (11)
	
	(12) edge[loop left] node[pos=0.55, below left]  {\tact{\altlab}} node[dist] (d01c) {} node[pos=0.5,left] {\tprob{}} node[pos=0.35,below] {} (12)
	
	(13) edge[loop left] node[pos=0.55, below left]  {\tact{\altlab\ifthenelse{\equal{\prodAutomata}{true}}{}{,\frownie}}} node[dist] (d01c) {} node[pos=0.5,left] {\tprob{}} node[pos=0.35,below] {} (13)
	
	(14) edge[loop left] node[pos=0.55, below left]  {\tact{\altlab}} node[dist] (d01c) {} node[pos=0.5,left] {\tprob{}} node[pos=0.35,below] {} (14)

	(00) edge[bend left=20] node[pos=0.15,above] {\tact{\lab}} node[dist, pos=0.25] (d00a) {} node[pos=0.8,above,yshift=0.8ex,xshift=0.9ex] {\tprob{\ifthenelse{\equal{\modelType}{pa}}{\frac {81}{100}}{(1-p)^2}}} (11)
	(d00a) edge[bend right=0] node[left, pos=0.75] {\tprob{\ifthenelse{\equal{\modelType}{pa}}{\frac{9}{100}}{p\cdot(1-p)}}} node[above,pos=0.75] {} (01)
	(d00a) edge[bend right=10] node[above, pos=0.75] {\tprob{\ifthenelse{\equal{\modelType}{pa}}{\frac{1}{100}}{p^2}}} node[above,pos=0.75] {} (02)
	(d00a) edge[bend left=10] node[right, pos=0.75,yshift=0.1ex,xshift=0.1ex] {\tprob{\ifthenelse{\equal{\modelType}{pa}}{\frac{9}{10}}{(1-p)\cdot p}}} node[above,pos=0.75] {} (12)

	(01) edge[bend left=0] node[pos=0.2,right, yshift=0.2ex] {\tact{\lab}} node[dist, pos=0.3] (d01a) {} node[pos=0.65,left] {\tprob{\ifthenelse{\equal{\modelType}{pa}}{\frac {9}{100}}{p \cdot q}}} (03)
	(d01a) edge[bend left=15] node[above, pos=0.25] {\tprob{\ifthenelse{\equal{\modelType}{pa}}{\frac{81}{100}}{(1-p) \cdot q}}} node[above,pos=0.75] {} (13)
	(d01a) edge[bend right=10] node[above, pos=0.25,rotate=15] {\tprob{\ifthenelse{\equal{\modelType}{pa}}{\frac{1}{100}}{p\cdot(1-q)}}} node[below,pos=0.15] {} (04)
	(d01a) edge[bend left=10] node[below, pos=0.35,rotate=-15] {\tprob{\ifthenelse{\equal{\modelType}{pa}}{\frac{9}{100}}{(1-p)\cdot(1-q)}}} node[above,pos=0.75] {} (14)

	(02) edge[bend right=5] node[pos=0.25,left] {\tact{\altaltlab}} node[dist, pos=0.35] (d02c) {} node[pos=0.58,left] {\tprob{\ifthenelse{\equal{\modelType}{pa}}{\frac{9}{10}}{\frac{9}{10}}}} (04)
	(d02c) edge[bend left=15] node[above, pos=0.25] {\tprob{\ifthenelse{\equal{\modelType}{pa}}{\frac{1}{10}}{\frac{1}{10}}}} node[left,pos=0.75] {} (03)
	;
	{\ifthenelse{\equal{\prodAutomata}{true}}{
			\node[pswide,below =2 of 13, label=-180:{$\{bad_{\regLang_G}\}$}] (13bad)  {$s_1,t_3,p_1$};
			\node[pswide,below =2 of 03, label=-180:{$\{bad_{\regLang_G}\}$}] (03bad)  {$s_0,t_3,p_1$};
			
			\path[ptrans] 
			(13) edge[bend right=0] node[pos=0.25,left] {\tact{\frownie}} node[dist, pos=0.25] (d13frown) {} (13bad)

			(03) edge[bend right=0] node[pos=0.25,left] {\tact{\frownie}} node[dist, pos=0.25] (d03frown) {} (03bad)
			(03bad) edge[loop right] node[pos=0.25, above] {\tact{\lab},\tact{\altlab},\tact{\altaltlab},\tact{\frownie}} node[dist] (d01c) {} node[pos=0.25,above] {\tprob{}} node[pos=0.35,below] {} (03bad)
			(13bad) edge[loop right] node[pos=0.25, above] {\tact{\lab},\tact{\altlab},\tact{\altaltlab},\tact{\frownie}} node[dist] (d01c) {} node[pos=0.25,above] {\tprob{}} node[pos=0.35,below] {} (03bad)
			;
		}{}}
\end{tikzpicture}\caption{Parallel composition of pPAs $\ppa_{1}$ and $\ppa_{2}$ from \Cref{fig:ppa1_and_ppa2}.}\label{fig:pPA_composition}
\end{figure}
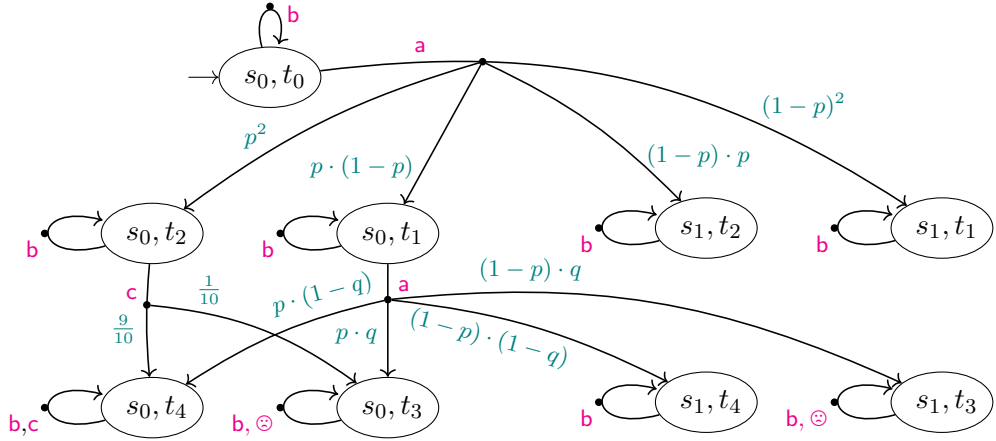
\begin{exa}
	\Cref{fig:pPA_composition} shows the composition of pPAs $\ppa_{1}$ and $\ppa_{2}$ from \cref{fig:ppa1_and_ppa2}. Actions with labels $\lab$ or $\altaltlab$ are synchronized while $\altlab$ and $\frownie$ are asynchronous. 
\end{exa}
Similar to \cite[Section 3.5]{Kwi+13}, we sometimes assume fairness of strategies, meaning that specific sets of labels $\alphabetOf{i} \subseteq \alphabetOf{}$ are visited infinitely often. 
\begin{defi}
    Let $\decomp \subseteq \setOfdecompsOf{\alphabetOf{}}$ and $\valuation$ be a well-defined valuation for $\ppa$ with alphabet $\alphabetOf{}$.  
    A complete strategy $\strategy \in \strategysetOf{\ppa[\valuation]}{\comp}$ is \emph{fair w.r.t. $\decomp$} (denoted $\fairC$) if 
    \[
        \PrOf{\ppa}{\valuation, \strategy}{\bigset{ s_0,\alpha_0,s_1,\alpha_1,\dots  \in \infPathsOf{\ppa}{} \mathbin{\big|} \forall \alphabetOf{i} \in \decomp \colon\, \forall j\in\nats\colon\, \exists k\ge j\colon\, \syncOf{}(s_k,\alpha_k) \in \alphabetOf{i} } }= 1.
    \]  
    The set of all $\fairWrtRegionModel{}{\decomp}$ strategies of $\ppa[\valuation]$ is denoted $\strategysetOf{\ppa[\valuation]}{\fairWrtRegionModel{}{\decomp}}$. 
\end{defi}
Almost-sure repeated reachability in PA only depends on the graph structure, which yields: 
\begin{prop}\label{theo:graph_preserving_fairness}
For any graph-preserving valuations $\valuation, \valuation'$ for $\ppa$ we have $\strategysetOf{\ppa[\valuation]}{\fairC} = \strategysetOf{\ppa[\valuation']}{\fairC}$, i.e., a strategy is $\fairC$ for $\ppa[\valuation]$ iff it is $\fairC$  for $\ppa[\valuation']$.
\end{prop}

\section{Strategy Projections}\label{sec:strat_projections}
In this section, we define the projection of a strategy of a composite pPA $\ppa = \ppa_1 \parallel \ppa_2$ onto a single component $\ppa_i$ for $i=1,2$.
Projections for PA are defined in \cite[Definition 6]{Kwi+13} and originate from \cite[page 65, Definition of Projection]{Seg95}.
They are intuitively used to relate probability measures for $\ppa$ and $\ppa_i$.

The projection of a finite path $\ppath \in \finPathsOf{\ppa}$ onto component $\ppa_i$ is the finite path $\restrOfTo{\ppath}{i}\in \finPathsOf{\ppa_i}$ obtained by restricting $\pi$ to the steps performed by $\ppa_i$.
Formally, $\restrOfTo{(\initialOf{1},\initialOf{2})}{i} = \initialOf{i}$ and for $\ppath = \ppath',(\alpha_1,\alpha_2),(s_1,s_2)$:
\[
\restrOfTo{\ppath}{i} = 
\begin{cases}
\restrOfTo{\ppath'}{i}, \alpha_i, s_i &\text{if } \alpha_i \in \actSetOf{i}\\
\restrOfTo{\ppath'}{i}                &\text{otherwise.} 
\end{cases}
\]
Path projections are neither injective nor surjective, i.e., we might have $\restrOfTo{\ppath}{i	} = \restrOfTo{\ppath'}{i}$ for two distinct paths $\ppath \neq \ppath'$ of $\ppa$, and for some $\ppath_i \in \finPathsOf{\ppa_i}$ there might not be any $\ppath \in \finPathsOf{\ppa}$ with $\ppath_i= \restrOfTo{\ppath}{i}$.
We define the set of paths of $\ppa_1 \parallel \ppa_2$ that are projected to $\ppath_i \in \finPathsOf{\ppa_i}$ as 
\[
\liftedpaths{\ppath_i}{\ppa_{3-i}} = \bigset{\ppath \in \finPathsOf{\ppa} \mid \ppath_i= \restrOfTo{\ppath}{i}}. 
\]
We first focus on strategy projections for non-parametric PA.
Then, we lift our notions to the parametric setting. 

\subsection{Projections for non-parametric PAs}
In the following, we fix the parallel composition $\pa = {\pa_1 \parallel \pa_2} = \paTupleOf{\parallel}$ of two PAs $\pa_1$ and $\pa_2$ with $\pa_i = \paTupleOf{i}$ and alphabets $\alphabetOf{i}$ for $i=1,2$.

\begin{defi}
\label{def:projectionStrategyPA}
The \emph{projection} of a strategy $\strategy \in \strategysetOf{\pa}{\prt}$ to $\pa_{i}$ is the strategy 
    $\stratProjOfToValuation{\strategy}{i}{\pa}\in\strategysetOf{\pa_i}{\prt}$,
where for $\ppath_i \in \finPathsOf{\pa_i}$ and $\alpha_i \in \actSetOf{i}$:
    \begin{align*}
		\stratProjOfToValuation{\strategy}{i}{\pa}(\ppath_i , \alpha_i) 
		= 
		\begin{cases}\displaystyle
			\frac{\sum_{s_i \in \stateSetOf{i}}
					\PrOf{\pa}{\strategy}{ \liftedpaths{(\ppath_i,\alpha_i,s_i)}{\pa_{3-i}}}
			}{\PrOf{\pa}{\strategy}{\liftedpaths{\ppath_i}{\pa_{3-i}}}}
				& \text{if } \PrOf{\pa}{\strategy}{\liftedpaths{\ppath_i}{\pa_{3-i}}} > 0 \\
			0 & \text{otherwise.}
		\end{cases}
	\end{align*}
\end{defi}
If $\pa$ is clear, we simply write $\stratProjOfToValuation{\strategy}{i}{}$ instead of $\stratProjOfToValuation{\strategy}{i}{\pa}$.
\Cref{lem:projection:alternative,lem:projection:measure} below yield that \Cref{def:projectionStrategyPA} is equivalent to the projection defined in \cite[Def.\ 6]{Kwi+13}.
We argue that our variant is more intuitive since the numerator and denominator of the fraction consider the same probability measure $\PrOf{\pa}{\strategy}{}$.
Intuitively, $\stratProjOfToValuation{\strategy}{i}{}(\ppath_i , \alpha_i)$ coincides with the conditional probability that---under $\PrOf{\pa}{\strategy}{}$ and given that a path $\ppath$ with projection $\restrOfTo{\ppath}{i} = \ppath_i$ is observed---the next action of component $\pa_i$ is $\alpha_i$.
We now provide an alternative characterization for the numerator given in \Cref{def:projectionStrategyPA}.
\begin{restatable}{lem}{projectionAlternative}
\label{lem:projection:alternative}
For $\strategy \in \strategysetOf{\pa}{\prt}$, $\ppath_i \in \finPathsOf{\pa_{i}}$, and $\alpha_i \in \actSetOf{i}$:
\[
\sum_{s_i \in \stateSetOf{i}} \PrOf{\pa}{\strategy}{\liftedpaths{(\ppath_i,\alpha_i,s_i)}{\pa_{3-i}}} 
~=~\sum_{\ppath \in (\liftedpaths{\ppath_i}{\pa_{3-i}})}~\sum_{\substack{(\hat\alpha_1,\hat\alpha_2) \in \actSetOf{\parallel},\,\hat\alpha_i = \alpha_i}}  \PrOf{\pa}{\strategy}{\ppath} \cdot \strategy(\ppath,(\hat\alpha_1,\hat\alpha_2)).
\]
\end{restatable}
The next lemma is the key observation for strategy projections as it connects the probability measure $\PrOf{\pa}{\strategy}{}$ for $\pa$ and $\PrOf{\pa_i}{\stratProjOfToValuation{\strategy}{i}{}}{}$ for each component $\pa_i$.
\begin{restatable}{lem}{projectionMeasure}
\label{lem:projection:measure}
For $\strategy \in \strategysetOf{\pa}{\prt}$ and $\ppath_i \in \finPathsOf{\pa_i}$: 
$\PrOf{\pa_i}{\stratProjOfToValuation{\strategy}{i}{}}{\ppath_i} = \PrOf{\pa}{\strategy}{\liftedpaths{\ppath_i}{\pa_{3-i}}}$.
\end{restatable}

The following result lifts \cite[Lemma 2]{Kwi+13} to our setting and states that fair (and therefore also complete) strategies have fair and complete projections.
\begin{lem}
	\label{theo:strategy_projection_partial_fair_preserved}
Let $\decomp_1 \subseteq \setOfdecompsOf{\alphabetOf{1}}$ and $\decomp_2 \subseteq \setOfdecompsOf{\alphabetOf{2}}$. 
If $\strategy \in \strategysetOf{\pa}{\fairWrtRegionModel{}{\decomp_1 \cup \decomp_2}}$, then $\stratProjOfToValuation{\strategy}{i}{} \in \strategysetOf{\pa_i}{\fairWrtRegionModel{}{\decomp_i}}$. 
\end{lem}
The converse of \Cref{theo:strategy_projection_partial_fair_preserved} does not hold:
The projection $\restrOfTo{\strategy}{i}$ might not be a complete strategy for $\pa_i$, even though $\strategy$ is complete for $\pa$.
Furthermore, $\restrOfTo{\strategy}{i}$ might not be memoryless, even if $\strategy$ is memoryless.
The following example shows both cases.
\begin{exa}
\label{ex:projectionPA}
    Consider the pPA \( \ppa =  \ppa_1 \parallel \ppa_2 \) from \Cref{fig:pPA_composition}, using \( \ppa_1 \) and \( \ppa_2 \) depicted in \Cref{fig:ppa1_and_ppa2}.  
    We fix the valuation \( \valuation \) with \( \valuation(p) = \valuation(q) = 0.1 \) and set $\pa = \ppa[\valuation]$ and $\pa_i = \ppa_i[\valuation]$ for $i=1,2$.
   Let \( \strategy \in \strategysetOf{\pa}{\mless,\comp} \) be a strategy that always selects \( \lab \), \( \altaltlab \), or \( \frownie \) with probability 1 when available; otherwise, it chooses \( \altlab \).
    
    We compute the projection \( {\stratProjOfToValuation{\strategy}{2}{}} \) to the PA \( \ppa_2[\valuation] \): 
    \begin{itemize}
        \item \( {\stratProjOfToValuation{\strategy}{2}{}}(t_0,\lab) = 1 \)
        \item \( {\stratProjOfToValuation{\strategy}{2}{}}((t_0,\lab, t_2), \altaltlab) 
        = 
         \frac{(\valuation(p))^2 \cdot \frac{9}{10}  + (\valuation(p))^2 \cdot \frac{1}{10}  }{\valuation(p)^2 + (1-\valuation(p)) \cdot \valuation(p)}
         =
        \frac{(\valuation(p))^2}{\valuation(p)} = \valuation(p) = 0.1 \)
        \item \( {\stratProjOfToValuation{\strategy}{2}{}}((t_0,\lab, t_2, \altaltlab, t_3), \frownie) = \frac{(\valuation(p))^2 \cdot \frac{1}{10}}{\valuation(p) \cdot \frac{1}{10}} = \valuation(p) = 0.1 \)
    \end{itemize}
    Similarly,  \( {\stratProjOfToValuation{\strategy}{2}{}}((t_0, \lab, t_1), \lab) = \valuation(p) = 0.1 \) and \( {\stratProjOfToValuation{\strategy}{2}{}}((t_0, \lab, t_1, \lab, t_3), \frownie) = \valuation(p) = 0.1 \). 
    We observe that the projection is not a complete strategy as the transition labeled $\frownie$ is the only available transition in $t_3$ for $\ppa_2$. 
    Moreover, the projection is not memoryless, because, for example, ${\stratProjOfToValuation{\strategy}{2}{}}((t_0, \lab, t_2, \altaltlab, t_4), \altaltlab) = \valuation(p) = 0.1$
    is not equal to ${\stratProjOfToValuation{\strategy}{2}{}}((t_0, \lab, t_1, \lab, t_4), \altaltlab) = 1.$
\end{exa}
	
\subsection{Projections for parametric PAs}
We now lift strategy projections to the parametric case. 
We fix the parallel composition $\ppa = \ppa_1 \parallel \ppa_2$ of two pPAs $\ppa_1$ and $\ppa_2$ with a common set of parameters
$\parameterSetOf{}$. Let $i=1,2$ and let $\valuation_i \colon \parameterSetOf{} \to \reals$ be a well-defined valuation for $\ppa_i$.
We define strategy projections for pPAs in terms of the instantiated PAs.
\begin{defi}
\label{def:projectionStrategy}
The \emph{projection} of a strategy $\strategy \in \strategysetOf{\ppa}{\prt}$ to pPA $\ppa_{i}$ w.r.t. $\valuation_1,\valuation_2$ is the strategy $\stratProjOfToValuation{\strategy}{i}{\valuation_1,\valuation_2} \in \strategysetOf{\ppa_i}{\prt}$, defined by
$\stratProjOfToValuation{\strategy}{i}{\valuation_1,\valuation_2} = \stratProjOfToValuation{\strategy}{i}{\ppa_1[\valuation_1] \parallel \ppa_2[\valuation_2]}$.
\end{defi}
If the valuations coincide, i.e., $\valuation_1=\valuation_2=\valuation$, we write $\stratProjOfToValuation{\strategy}{i}{\valuation}$ instead of $\stratProjOfToValuation{\strategy}{i}{\valuation_1,\valuation_2}$.
Strategy projections depend on the parameter instantiations as the following example illustrates.
Such strategies are also referred to as \emph{(parameter) dependent strategies}~\cite[Def.\ {2.6}]{Spel23}. 
\begin{exa}
\label{ex:projectionPPA}
    Consider \( \ppa_1 \parallel \ppa_2 \) from \Cref{fig:pPA_composition} and the strategy \( \strategy \) as in \Cref{ex:projectionPA}. 
    Let \( \valuation_1 \) and \( \valuation_2 \) be the valuations used for \( \ppa_1 \) and \( \ppa_2 \), respectively, with \( \valuation_1(p) = \valuation_1(q) = 0.1 \), and \( \valuation_2(p) = \valuation_2(q) = 0.9 \). 
    For the projection \( \stratProjOfToValuation{\strategy}{2}{\valuation_1, \valuation_2} \) to \( \ppa_2 \), we obtain:
    \begin{itemize}
        \item \( {\stratProjOfToValuation{\strategy}{2}{\valuation_1,\valuation_2}}(t_0,\lab) = 1 \)
        \item \( {\stratProjOfToValuation{\strategy}{2}{\valuation_1,\valuation_2}}((t_0,\lab, t_2), \altaltlab) = \frac{\valuation_1(p)\cdot \valuation_2(p)}{\valuation_2(p)} = \valuation_1(p) = 0.1 \)
        \item \( {\stratProjOfToValuation{\strategy}{2}{\valuation_1,\valuation_2}}((t_0,\lab, t_2, \altaltlab, t_3), \frownie) = \frac{\valuation_1(p) \cdot \valuation_2(p) \cdot \frac{1}{10}}{\valuation_2(p) \cdot \frac{1}{10}} = \valuation_1(p) = 0.1 \), 
        \item \( {\stratProjOfToValuation{\strategy}{2}{\valuation_1,\valuation_2}}((t_0, \lab, t_1), \lab) = \valuation_1(p) = 0.1 \), and 
        \item \( {\stratProjOfToValuation{\strategy}{2}{\valuation_1,\valuation_2}}((t_0, \lab, t_1, \lab, t_3), \frownie) = \valuation_1(p) = 0.1 \).
    \end{itemize}
    Note that the resulting projection coincides with the one computed in \Cref{ex:projectionPA}, where both components were instantiated using \( \valuation_1 \), i.e., $
    \stratProjOfToValuation{\strategy}{2}{\valuation_1, \valuation_2} = \stratProjOfToValuation{\strategy}{2}{\valuation_1} = \stratProjOfToValuation{\strategy}{2}{\valuation_1, \valuation_1}
    $.
\end{exa}
The next lemma states that---when restricting to valuations that yield the same non-zero transitions---the strategy projection to $\ppa_i$ only depends on the parameter instantiation applied to $\ppa_{3-i}$.
This observation is the key insight for the correctness of the proof rule in \Cref{theo:pag_mono_rule}, which enables compositional reasoning about monotonicity.
\begin{restatable}{lem}{changeValuationProjection}
\label{theo:change_valuation_projection}
For $i=1,2$, and well-defined valuations $\valuation_i, \valuation_i' \colon \parameterSetOf{} \to \reals$ for $\ppa_i$ such that $\transFctOf{i}(s_i,\alpha_i,s_i')[\valuation_i] = 0$ iff $\transFctOf{i}(s_i,\alpha_i,s_i')[\valuation_i'] = 0$, we have:
$\stratProjOfToValuation{\strategy}{1}{\valuation_1, \valuation_2} = \stratProjOfToValuation{\strategy}{1}{\valuation_1', \valuation_2}$
and 
$\stratProjOfToValuation{\strategy}{2}{\valuation_1, \valuation_2} = \stratProjOfToValuation{\strategy}{2}{\valuation_1, \valuation_2'}$.
\end{restatable}

\section{Verification of Objectives for pPAs}\label{sec:verification}
We define properties of interest for pPA verification.
Let $\alphabetOf{}$ be a finite alphabet.
$\alphabetOf{}^{\infty} = \alphabetOf{}^* \cup \alphabetOf{}^{\omega}$ denotes the set of all finite and infinite words over $\alphabetOf{}$.
For a word $\trace = \lab_0, \lab_1, \dots \in \alphabetOf{}^\infty$ and another alphabet $\widehat{\alphabetOf{}}$, let $\restrOfTo{\trace}{\widehat{\alphabetOf{}}} \in \widehat{\alphabetOf{}}^\infty$ denote the projection of $\trace$ onto $\widehat{\alphabetOf{}}$---obtained by dropping all $\lab_i \in \alphabetOf{} \setminus \widehat{\alphabetOf{}}$ from $\trace$.
The restriction $\restrOfTo{\trace}{\widehat{\alphabetOf{}}}$ can be finite, even if $\trace$ is infinite.

We fix a pPA $\ppa = \ppaTupleOf{}$.
The trace of $\infpath = s_0, \alpha_0, s_1, \alpha_1, \dots \in \infPathsOf{\ppa}$ is the sequence $\traceOf{\infpath} = \syncOf{}(s_0, \alpha_0), \syncOf{}(s_1, \alpha_1), \dots $ of transition labels. 
The \emph{probability of a language} $\regLang \subseteq \alphabetOf{}^\infty$ at a well-defined valuation $\valuation$ under strategy $\strategy$ of $\ppa$ is given by
	\[
	\PrOf{\ppa}{\valuation,\strategy}{\regLang} = \PrOf{\ppa}{\valuation,\strategy}{\{ \pi \in \infPathsOf{\ppa}{} \mid \restrOfTo{\traceOf{\pi}}{\alphabetOf{}} \in \regLang  \}}. 
	\]

We also consider (parametric) \emph{expected total reward properties}. Let $\parameterSetOf{}$ be a set of parameters.
A \emph{reward function} $\rewFct \colon \alphabetOf{} \to \rats[{\parameterSetOf{}}] \cup \reals_{\geq 0}$ over $\alphabetOf{}$ assigns a (potentially parametric) reward to each symbol $\lab \in \alphabetOf{}$.
\emph{Instantiation} of $\rewFct$ at a valuation $\valuation \colon \parameterSetOf{} \to \reals$ yields $\rewFct[\valuation]$ with $\rewFct[\valuation](\lab) = \rewFct(\lab)[\valuation]$ for all $\lab \in \alphabetOf{}$. Valuation $\valuation$ is \emph{well-defined} for $\rewFct$ if $\rewFct[\valuation] \colon  \alphabetOf{} \to \reals_{\geq 0}$.
In this case, the accumulated reward for a word $\trace = \lab_0, \lab_1, \dots \in \alphabetOf{}^\infty$ is given by $\rewFct[\valuation](\trace) = \sum_{i=0}^{|\trace|} \rewFct[\valuation](\lab_i) \in \reals_{\geq 0} \cup \set{\infty}$.

When applied to a pPA $\ppa$, a reward function $\rewFct$ assigns the reward $\rewFct(\syncOf{}(s,\alpha))$ to the enabled state-action-pairs $(s,\alpha) \in \domain(\transFctOf{})$ with $\syncOf{}(s,\alpha) \in \alphabetOf{}$.
For a well-defined valuation $\valuation$ for $\ppa$ and $\rewFct$, we define the \emph{expected total reward} under strategy $\strategy$ as
\[\ExpTot{\ppa}{\valuation, \strategy}{\rewFct}= \textstyle \int_{\pi \in \infPathsOf{\ppa}{}}^{} \rewFct[\valuation]\big(\restrOfTo{\traceOf{\pi}}{\alphabetOf{}}\big) \,d \PrOf{\ppa}{\valuation, \strategy}{\pi}.\]

We consider probabilistic and reward-based objectives as well as their multi-objective combinations.
\begin{defi}
\label{defi:objectives}
For ${\sim} \in \{>, \geq, <, \leq\}$, $p \in [0,1]$, and $r \in \reals_{\geq 0}$, we denote 
\begin{itemize}
    \item a \emph{probabilistic objective} over $\regLang \subseteq \alphabetOf{}^\infty$ by $\probPredicate{\sim p}{\regLang}$ and
    \item a \emph{reward objective} over $\rewFct \colon \alphabetOf{} \to \rats[{\parameterSetOf{}}] \cup \reals_{\geq 0}$ by $\expPredicate{\sim r}{\rewFct}$. 
    \end{itemize}
Their satisfaction for a well-defined valuation $\valuation$ and strategy $\strategy$ is defined by
	\[
		\ppa, \valuation, \strategy  \modelsWrt{} \probPredicate{\sim p}{\regLang}  \  \Leftrightarrow  \  \PrOf{\ppa}{\valuation,\strategy}{\regLang} \sim p 
        \quad  \text{ and}  \quad
        \ppa, \valuation, \strategy \modelsWrt{} \expPredicate{\sim r}{\rewFct}  \  \Leftrightarrow  \  \ExpTot{\ppa}{\valuation,\strategy}{\rewFct} \sim r . 
	\]
\end{defi}
Let $\generalPredicate{} \in \set{\probPredicate{\sim p}{\regLang}, \expPredicate{\sim r}{\rewFct}}$ refer to a (probabilistic or reward) objective. 
If neither $\ppa$ nor $\generalPredicate{}$ consider any parameters, we may drop the valuation from the notation and just write $\ppa, \strategy  \modelsWrt{} \generalPredicate{}$.
We lift the satisfaction relation $\modelsWrt{}$ to regions, i.e., sets of valuations.

\begin{defi}
\label{def:satregion}
	Let $\star \in \{\prt, \comp\} \cup \set{\fairC \mid \decomp \subseteq \setOfdecompsOf{\alphabetOf{\ppa}}}$.  
	For objective $\generalPredicate$
    and well-defined region $\region$ for $\ppa$---and $\rewFct$ if $\generalPredicate{} = \expPredicate{\sim r}{\rewFct}$---the \emph{region satisfaction relation} $\modelsWrt{\star}$ is given by: 
	\[
		\ppa, \region \modelsWrt{\star} \generalPredicate{} \quad \Leftrightarrow \quad \forall \valuation \in \region:  \forall \strategy \in \strategysetOf{\ppa[\valuation]}{\star}: \ppa, \valuation, \strategy \modelsWrt{} \generalPredicate{}. 
	\]
    Satisfaction under memoryless strategies---denoted by $\modelsWrt{\mless,\star}$---is defined similarly.
\end{defi}

\begin{rem}
For $\star \in \{\prt, \comp\}$ and any well-defined valuation $\valuation$ we have $\strategysetOf{\ppa[\valuation]}{\star} = \strategysetOf{\ppa}{\star}$.
Thus, for well-defined $\region$, we can swap the quantifiers in \Cref{def:satregion}: 
	\[
		\ppa, \region \modelsWrt{\star} \generalPredicate{} 
		\quad \Leftrightarrow \quad  \forall \strategy \in \strategysetOf{\ppa}{\star}: 
		\forall \valuation \in \region: \ppa, \valuation, \strategy \modelsWrt{} \generalPredicate{}. 
	\]
	However, this is not the case for \emph{fair} strategies and regions that are not graph-preserving: 
A strategy that is not $\fairC$ for $\ppa$ (under graph-preserving instantiations) might be $\fairC$ for $\ppa[\valuation]$ if $\valuation$ is not graph-preserving, because states that violate the fairness condition might not be reachable in $\ppa[\valuation]$.
	For a \emph{graph-preserving} region $\region$ and all $\valuation \in \region$, we have $\strategysetOf{\ppa}{\fairC} = \strategysetOf{\ppa[\valuation]}{\fairC}$. 
	Thus, we can swap quantifiers as above.
\end{rem}

Our framework also handles conjunctions of multiple objectives.
\begin{defi}
\label{defi_mo_sat}
A \emph{multi-objective} query (mo-query) is a set  $\multiobjectiveQuery{}{X}=\{ \generalPredicate{}_1, \dots , \generalPredicate{}_n\}$ of $n$ probabilistic or reward objectives with
	$
		\ppa, \valuation, \strategy  \modelsWrt{} \multiobjectiveQuery{}{X}  \  \Leftrightarrow  \  \ppa,\valuation, \strategy \modelsWrt{} \generalPredicate{}_i$  for all $\generalPredicate{}_i \in \multiobjectiveQuery{}{X}
	$.
\end{defi}
The conjunction of two mo-queries is a union of sets: $\multiobjectiveQuery{}{X}_1 \land \multiobjectiveQuery{}{X}_2 = \multiobjectiveQuery{}{X}_1 \cup \multiobjectiveQuery{}{X}_2$.
We lift objective satisfaction for regions (\Cref{def:satregion}) to mo-queries in a straightforward way.

\begin{rem}
    \Cref{app:equivalence_partial_complete_measures,app:equivalence_partial_complete_qmo} show that model checking under partial strategies in $\ppa$ for probabilistic properties, rewards, and multi-objective queries reduces to model checking under complete strategies in a modified pPA, denoted $\ppa_\tau$. 
    This result extends \cite[Proposition 2]{Kwi+13} to pPA. 
\end{rem}
We consider \emph{safety objectives} as a special type of probabilistic objectives.
\begin{defi}
\label{def:safety_objective}
        $\probPredicate{\geq p}{\regLang}$ is a \emph{safety objective}\footnote{Note that safety objectives contain all finite prefixes of words in $\regLang$, i.e., they are prefix-closed. 
        This is different in \cite{Kwi+13}, where only infinite words are considered, leading to technical problems, see \Cref{sec:explain_differences}.
        } 
        if  $\regLang$ can be characterized by a DFA $\bpAutomatonOf{\regLang}$ accepting a language of finite words (bad prefixes):
        \[
        \regLang = \{ w \in  \alphabetOf{}^{\infty} \mid \text{no prefix of $w$ is accepted by } \bpAutomatonOf{\regLang} \}.
        \]
\end{defi}
A mo-query is called \emph{safe}, denoted $\multiobjectiveQuery{\safe}{X}$, if each $\generalPredicate{}_i$ is a probabilistic safety objective. 

For PA, computing the probability for a safety objective reduces to maximal reachability properties in the PA-DFA product \cite[Lemma 1]{Kwi+13}. 
This result can be lifted to pPA in a straightforward manner; see \Cref{theo:prob_safety_pPA_prop_reach_poduct_bijection} in \Cref{app:proofs_of_verification}. 
For reachability and safety objectives, it is equivalent to quantify over complete or partial strategies.
\Cref{lemma_safety_partial_vs_complete} lifts \cite[Proposition 1]{Kwi+13} to pPA. 
\begin{restatable}{lem}{restatableSafetyPartVsComplete}\label{lemma_safety_partial_vs_complete}
    Let $\ppa$ be a pPA, let $\region$ be a well-defined region and let $\probPredicate{\geq p}{\regLang}$ be a safety objective. 
    It holds that: 
    $\ppa, \region \modelsWrt{\comp} \probPredicate{\geq p}{\regLang}  \Leftrightarrow 
    \ppa, \region \modelsWrt{\prt} \probPredicate{\geq p}{\regLang}$.
    Same for $\modelsWrt{\mless,\star}$.
\end{restatable}

\subsection{Preservation under projection}\label{sec:projection_properties}
We generalize the result from \cite[Lemma 3]{Kwi+13}---originally stated in \cite[Lemma 7.2.6]{Seg95}---to the parametric setting. In particular, we show that probabilistic and reward properties in a composed pPA $\ppa = \ppa_1 \parallel \ppa_2$ under a strategy $\strategy$ are preserved when projecting to either component $\ppa_i$ over $\alphabetOf{i}$, assuming a well-defined valuation $\valuation$. 
\begin{restatable}{thm}{restatableLemmaThreeNewPre}
    \label{theo:lemma3NewDependentPre}
    \label{lemma_3_aDependent}
    \label{lemma_3_cDependent}
		For $i=1,2$, let $\regLang$ be a language over $\alphabetOf{i}$ and $\rewFct$ be a reward function over $\alphabetOf{i}$. 
		Then, for a  well-defined valuation $\valuation$: 
			\begin{align*}
				& \PrOf{{{\ppa_i}}}{\valuation, {\stratProjOfToValuation{\strategy}{i}{\valuation}}}{\regLang}
				 =  
				\PrOf{{{{\ppa_1}} \parallel \ppa_2}}{\valuation, \strategy}{\regLang}
				\quad  \text{ and } \quad 
				\ExpTot{{\ppa_i}}{\valuation, {\stratProjOfToValuation{\strategy}{i}{\valuation}}}{\rewFct}
				 =  
				\ExpTot{{\ppa_1} \parallel {\ppa_2}}{\valuation, \strategy}{\rewFct}
			\end{align*} 	
\end{restatable}
\begin{exa}\label{ex:projection_same_prob}
    Let $\regLang = \{ w \in \{\lab,\altaltlab,\frownie\}^{\infty} \mid \vert w \vert_{\frownie} > 0 \}$ be the language of finite and infinite words in which $\frownie$ occurs. 
    Reconsider the strategy $\strategy$ of $\ppa_1 \parallel \ppa_2$ from \Cref{ex:projectionPPA} and the projection ${\stratProjOfToValuation{\strategy}{{2}}{\valuation}}$ to $\ppa_2$. 
    We have 
    $\PrOf{\ppa_1 \parallel \ppa_2 }{\valuation, \strategy}{\regLang} 
    = \PrOf{\ppa_2}{\valuation, {\stratProjOfToValuation{\strategy}{2}{\valuation}}}{\regLang} 
    = (\valuation(p))^2 \cdot \frac{1}{10} + \valuation(p) \cdot(1-\valuation(p)) \cdot \valuation(q)$. 
\end{exa}
Theorem~\ref{theo:lemma3NewDependentPre} assumes that the property only involves action labels from a single component $\ppa_i$.
To allow objectives over arbitrary alphabets $\alphabetOf{}$, we can add a self-loop labeled $\lab$ at every state, for each label $\lab \notin \alphabetOf{i}$.
\begin{defi}[Alphabet Extension for pPAs]
\label{def_alphabet_extension}
Let $\ppa = \ppaTupleOf{}$ be a pPA over $\alphabetOf{\ppa}$ and let $\alphabetOf{}$ be an alphabet with $\actSetOf{} \cap (\alphabetOf{} \setminus \alphabetOf{\ppa}) = \emptyset$.
The \emph{alphabet extension} of $\ppa$ with respect to $\alphabetOf{}$ is the pPA
$\alphabetExtensionOfTo{\ppa}{\alphabetOf{}} = (\stateSetOf{}, \initialOf{}, \parameterSetOf{}, \actSetOf{} \cupdot (\alphabetOf{} \setminus \alphabetOf{\ppa}), \transFctOf{\alphabetOf{}}, \syncOf{\alphabetOf{}})$ over alphabet $\alphabetOf{\ppa} \cup \alphabetOf{}$, where
\begin{itemize}
\item $\transFctOf{\alphabetOf{}}(s,\alpha) = \transFctOf{}(s,\alpha)$ and $\syncOf{\alphabetOf{}}(s,\alpha) = \syncOf{}(s,\alpha)$ for all $(s,\alpha) \in \domain(\transFctOf{})$ and
\item $\transFctOf{\alphabetOf{}}(s,\lab) = \indicatorFct{s}$ and $\syncOf{\alphabetOf{}}(s,\lab) = \lab$ for all $s \in \stateSetOf{}$ and $\lab \in \alphabetOf{}\setminus\alphabetOf{\ppa}$.
\end{itemize}
\end{defi}
\begin{figure}[t]
    	\renewcommand{\modelType}{\ppa}
    	\renewcommand{\extended}{true} 
        \centering
    	{\begin{tikzpicture}[mdp]
	\node[ps, init=left] (t0)  {$t_0$};
	
	\node[ps,above right=0.25 and 2 of t0] (t1)  {$t_1$};

	\node[ps,below right=0.25 and 2 of t0] (t2)  {$t_2$};
	
	\node[ps,right=2 of t1] (t3)  {$t_3$};
	
	\node[ps,right=2 of t2] (t4)  {$t_4$};
	
	\node[dist, right=0.5 of t0] (d0a) {};
		
	\path[ptrans]

		(t0) edge[-] node[pos=1, below] {\tact{\lab}} (d0a)
		(d0a) edge[bend right=0] node[pos=0.5,above,yshift=0.5ex,xshift=-0.5ex] {\tprob{\ifthenelse{\equal{\modelType}{pa}}{\frac{9}{10}}{1-p}}} (t1)
		(d0a) edge[bend right=0] node[below,pos=0.55,yshift=-0.5ex,xshift=-0.5ex] {\tprob{\ifthenelse{\equal{\modelType}{pa}}{\frac{1}{10}}{p}}} node[above,pos=0.4] {} (t2)

		(t1) edge[bend right=0] node[pos=0.25,above] {\tact{\lab}} node[dist, pos=0.25] (d1a) {} node[pos=0.6,above] {\tprob{\ifthenelse{\equal{\modelType}{pa}}{\frac{9}{10}}{q}}} (t3)
		(d1a) edge[bend left=25] node[below,pos=0.25,xshift=-1ex] {\tprob{\ifthenelse{\equal{\modelType}{pa}}{\frac{1}{10}}{1-q}}} node[above,pos=0.4] {} (t4)
		
		(t2) edge[bend right=0] node[pos=0.25,below] {\tact{\altaltlab}} node[dist, pos=0.25] (d1b) {} node[pos=0.6,below] {\tprob{\ifthenelse{\equal{\modelType}{pa}}{\frac{9}{10}}{\frac{9}{10}}}} (t4)
		(d1b) edge[bend right=25] node[above, pos=0.25] {\tprob{\ifthenelse{\equal{\modelType}{pa}}{\frac{1}{10}}{\frac{1}{10}}}} node[above,pos=0.4] {} (t3)

		(t3) edge[loop right] node[pos=0.45,below right] {{\ifthenelse{\equal{\extended}{true}}{\tactext{\altlab},}{}}\tact{ \frownie}} node[dist] (d3frown) {} node[pos=0.15,below] {\tprob{}} node[pos=0.25,below] {} (t3)
		
		(t4) edge[loop right] node[pos=0.5,above right] {\ifthenelse{\equal{\extended}{true}}{\tactext{\altlab},}{}{\tact{\altaltlab}}} node[dist] (d4b) {} node[pos=0.15,below] {\tprob{}} node[pos=0.25,below] {} (t4)

	;	
	
	{\ifthenelse{\equal{\extended}{true}}{\path[ptrans] (t2) edge[loop above] node[pos=0.75,right] {{\ifthenelse{\equal{\extended}{true}}{\tactext{\altlab}}{}}} node[dist] (d1c) {} (t2);}{}}
	
	{\ifthenelse{\equal{\extended}{true}}{\path[ptrans] (t1) edge[loop above] node[pos=0.75,right] {{\ifthenelse{\equal{\extended}{true}}{\tactext{\altlab}}{}}} node[dist] (d1c) {} (t2);}{}}
	
	{\ifthenelse{\equal{\extended}{true}}{\path[ptrans] (t0) edge[loop above] node[pos=0.75,right] {{\ifthenelse{\equal{\extended}{true}}{\tactext{\altlab}}{}}} node[dist] (d1c) {} (t0);}{}}

\end{tikzpicture}}
    	\caption{Alphabet extension $\alphabetExtensionOfTo{\ppa_2}{\{\lab,\altlab\}}$ of the pPA $\ppa_2$ from \Cref{fig:pPAM2} to the alphabet $\{\lab, \altlab\}$.
        }%
    	\label{fig:ppaExtension}
    \end{figure}
\begin{exa}
 \Cref{fig:ppaExtension} shows $\alphabetExtensionOfTo{\ppa_2}{\{\lab,\altlab\}}$ for pPA $\ppa_2$ over $\alphabetOf{2} = \set{\lab,\altaltlab,\frownie}$ from \Cref{fig:pPAM2}.
 Transitions with label $\lab$ remain unchanged as $\lab \in \alphabetOf{2}$, but an additional self-loop with action $\altlab\not \in \alphabetOf{2}$ is added to every state. 
\end{exa}
We now lift \cite[Lemma 3]{Kwi+13} to the parametric setting, covering properties and mo-queries over an alphabet that is not necessarily shared by $\ppa_i$: 
\begin{restatable}{thm}{restatableLemmaThreeNew}
    \label{theo:lemma3NewDependent}
		Let $\alphabetOf{} \subseteq \alphabetOf{\ppa_1\parallel \ppa_2}$, and $\sigma$ be a strategy for ${{\alphabetExtensionOfTo{\ppa_1}{\alphabetOf{}}} \parallel {\alphabetExtensionOfTo{\ppa_2}{\alphabetOf{}}}}$. 
		Let $\regLang$ be a language over $\alphabetOf{}$ and $\rewFct$ be a reward function over $\alphabetOf{}$. 
		Then, for well-defined valuation $\valuation$:      
        \label{lemma_3_bDependent}
        \label{lemma_3_dDependent}
			\begin{align*}
				& \PrOf{{\alphabetExtensionOfTo{\ppa_i}{\alphabetOf{}}}}{\valuation, {\stratProjOfToValuation{\strategy}{i}{\valuation}}}{\regLang}
				 =  
				\PrOf{{{\alphabetExtensionOfTo{\ppa_1}{\alphabetOf{}}} \parallel {\alphabetExtensionOfTo{\ppa_2}{\alphabetOf{}}}}}{\valuation, \strategy}{\regLang}
				\quad  \text{ and } \quad 
				\ExpTot{{\alphabetExtensionOfTo{\ppa_i}{\alphabetOf{}}}}{\valuation, {\stratProjOfToValuation{\strategy}{i}{\valuation}}}{\rewFct}
				 =  
				\ExpTot{{{\alphabetExtensionOfTo{\ppa_1}{\alphabetOf{}}} \parallel {\alphabetExtensionOfTo{\ppa_2}{\alphabetOf{}}}}}{\valuation, \strategy}{\rewFct}
			\end{align*} 	  
        Let  $\multiobjectiveQuery{}{X}$ be a mo-query over $\alphabetOf{}$. 
	    Then, for any well-defined valuation $\valuation$: 
          \label{theo:lemma3NewMulti}
         \label{lemma_3_eDependent} 
        \label{lemma_3_fDependent} 
		\begin{align*}	
			& { \alphabetExtensionOfTo{\ppa_i}{\alphabetOf{} }}, \valuation, { {\stratProjOfToValuation{\strategy}{i}{\valuation}} } 
            \modelsWrt{}  \multiobjectiveQuery{}{X}
			\quad {\Leftrightarrow} \quad 
			({{\alphabetExtensionOfTo{\ppa_1}{\alphabetOf{}}} \parallel {\alphabetExtensionOfTo{\ppa_2}{\alphabetOf{}}}}), \valuation, {\strategy} \modelsWrt{} \multiobjectiveQuery{}{X} 
		\end{align*}	
\end{restatable}
\begin{rem}\label{rem:extension}
    Since alphabet extensions add self-loop transitions for new labels, and thus do not change the system's state, the pPAs ${{\alphabetExtensionOfTo{\ppa_1}{\alphabetOf{}}} \parallel {\alphabetExtensionOfTo{\ppa_2}{\alphabetOf{}}}}$ and $\ppa_1 \parallel \ppa_2$ satisfy the same properties and mo-queries over the alphabet 
    $\alphabetOf{} \subseteq \alphabetOf{\ppa_1 \parallel \ppa_2}$. 
\end{rem}
\Cref{theo:lemma3NewDependentPre,theo:lemma3NewDependent} play a key role in the proof of the AG framework for reasoning about mo-queries and monotonicity, which will be established in \Cref{sec:pag,sec:pag_mono}.

\section{Assume-Guarantee Reasoning for pPA}\label{sec:pag}
Kwiatkowska et al.\ \cite{Kwi+13} introduced assume-guarantee (AG) reasoning proof rules for PA. 
This section extends their proof rules to the parametric setting. 
We first generalize the concept of AG triples to pPAs in \Cref{sec:triples_for_pPA}. 
Then, we extend the asymmetric and circular proof rule in \Cref{sec:AG_for_pPA}. 
Additional proof rules from \cite{Kwi+13} are presented in \Cref{app:pag_extension}.

\subsection{Assume-guarantee triples for pPA}\label{sec:triples_for_pPA}
We extend compositional reasoning to the parametric setting by generalizing assume-guarantee (AG) triples. Intuitively, an AG triple states that if a component satisfies an assumption, it also satisfies the guarantee under the same strategy and valuation. 
\begin{defi}
\label{def_pag_triple}
    The \emph{assume-guarantee triple} for pPA $\ppa$, (parametric) mo-query $\multiobjectiveQuery{}{A}$ (assumption) and (parametric) mo-query $\multiobjectiveQuery{}{G}$ (guarantee), well-defined region $\region$, and $\star \in \{\comp, \prt, \fairC\}$ is 
	\begin{align*}
	\agTriple{\ppa, \region}{\star}{\multiobjectiveQuery{}{A}}{\multiobjectiveQuery{}{G}} \quad 
	& \Leftrightarrow \quad 
    \left(\forall \valuation \in \region : \forall \strategy \in \strategysetOf{\ppa[\valuation]}{\star} : \quad  \ppa, \valuation, {\strategy} \models \multiobjectiveQuery{}{A} \quad \rightarrow  \quad \ppa, \valuation, {\strategy} \models \multiobjectiveQuery{}{G} \right)
	\end{align*}
\end{defi}

\subsection{Assume-guarantee rules for pPA}
\label{sec:AG_for_pPA}
We present AG proof rules for the compositional verification of parametric probabilistic automata (pPAs). 
In the remainder of this section, we fix two pPAs $\ppa_1$ and $\ppa_2$ with alphabets $\alphabetOf{1}$, and $\alphabetOf{2}$, respectively. Further, let $\region_i$ be a well-defined region for $\ppa_i$. 

First, we establish the asymmetric proof rule for safety and mo-queries---based on \cite[Theorem 1 and 2]{Kwi+13}, respectively---for pPA. %
\begin{restatable}[Asymmetric Rule]{thm}{restatableAsymRule}\label{theo:pag_asym_rule}
    Let $\multiobjectiveQuery{}{A}$ and $\multiobjectiveQuery{}{G}$ be mo-queries over $\alphabetOf{\multiobjectiveQuery{}{A}} \subseteq \alphabetOf{1}$ and 
    $\alphabetOf{\multiobjectiveQuery{}{G}} \subseteq \alphabetOf{2} \cup \alphabetOf{\multiobjectiveQuery{}{A}}$, respectively. 
    Let $\decomp_1 \subseteq \setOfdecompsOf{\alphabetOf{1}}$ and $\decomp_2 \subseteq \setOfdecompsOf{\alphabetOf{2} \cup \alphabetOf{\multiobjectiveQuery{}{A}}}$. 
    Then, the two proof rules hold:
    
     \begin{tabularx}{\linewidth}{p{0.4\linewidth}p{0.4\linewidth}} 
                 $
                 \infer{
                     \ppa_1 \parallel \ppa_2,  \regionIntersectionOf{\region_1}{}{\region_2} \modelsWrt{\comp} \multiobjectiveQuery{\safe}{G}
                 } 
                 {
                     \deduce{
                         \agTriple{\alphabetExtensionOfTo{\ppa_2}{\alphabetOf{\multiobjectiveQuery{\safe}{A}}}, \region_2 }{\prt}{\multiobjectiveQuery{\safe}{A}}{\multiobjectiveQuery{\safe}{G}} 
                     }{
                                \deduce{
                                    \ppa_1, \region_1 \modelsWrt{\comp} \multiobjectiveQuery{\safe}{A}
                                    }{}
                    }
                 }
                 $
         & 
          \hspace{30pt}
             $
                 \infer{
                     \ppa_1 \parallel \ppa_2,   \regionIntersectionOf{\region_1}{}{\region_2}
                     \modelsWrt{\fairWrtRegionModel{}{
                     \decomp_1 \cup  \decomp_2}} \multiobjectiveQuery{}{G}
                 }
                 {
                    \deduce{
                        \agTriple{\alphabetExtensionOfTo{\ppa_2}{\alphabetOf{\multiobjectiveQuery{}{A}}}, \region_2}{\fairWrtRegionModel{}{\decomp_2}}{\multiobjectiveQuery{}{A}}{\multiobjectiveQuery{}{G}} 
                    }{
                        \deduce{
                            \ppa_1, \region_1 \modelsWrt{\fairWrtRegionModel{}{\decomp_1}} \multiobjectiveQuery{}{A}
                        }{}
                    }
                 }
              $
     \end{tabularx}
     
\end{restatable}
\begin{proof}[Proof sketch]
Let $\valuation \in \regionIntersectionOf{\region_1}{}{\region_2}$, and $\strategy$ be a strategy for the composed pPA $\ppa_1 \parallel \ppa_2$. 
To prove validity of the rule, we need to show that $\ppa_1 \parallel \ppa_2$ instantiated with $\valuation$ satisfies $\multiobjectiveQuery{\safe}{G}$. 
\begin{enumerate}
    \item Since $\valuation \in \region_1$, the first premise implies $\ppa_1, \valuation \modelsWrt{\comp} \multiobjectiveQuery{\safe}{A}$, which is equivalent to $\ppa_1, \valuation \modelsWrt{\prt} \multiobjectiveQuery{\safe}{A}$ by \Cref{lemma_safety_partial_vs_complete}. 
    This implies that $\ppa_1$ under the partial strategy $\stratProjOfToValuation{\strategy}{\ppa_1}{\valuation}$ also satisfies $\multiobjectiveQuery{\safe}{A}$. 
    Since strategies and their projections satisfy the same properties (as shown in \Cref{theo:lemma3NewDependent}), we conclude that $\ppa_1 \parallel \ppa_2$ instantiated at $\valuation$ under the strategy $\strategy$ satisfies $\multiobjectiveQuery{\safe}{A}$. 
    \item Then, \Cref{theo:lemma3NewDependent} implies that the instantiated pPA $\alphabetExtensionOfTo{\ppa_2}{\alphabetOf{\multiobjectiveQuery{\safe}{A}}}[\valuation]$ also satisfies $\multiobjectiveQuery{\safe}{A}$ under the strategy $\stratProjOfToValuation{\strategy}{\alphabetExtensionOfTo{\ppa_2}{\alphabetOf{\multiobjectiveQuery{\safe}{A}}}}{\valuation}$. 
    \item As $\valuation \in \region_2$, the second premise implies that $\alphabetExtensionOfTo{\ppa_2}{\alphabetOf{\multiobjectiveQuery{\safe}{A}}}[\valuation]$ under the strategy $\stratProjOfToValuation{\strategy}{\alphabetExtensionOfTo{\ppa_2}{\alphabetOf{\multiobjectiveQuery{\safe}{A}}}}{\valuation}$ satisfies $\multiobjectiveQuery{\safe}{G}$. 
    Again,  \Cref{theo:lemma3NewDependent} implies that $(\ppa_1 \parallel \ppa_2)[\valuation]$ under the strategy $\strategy$ satisfies $\multiobjectiveQuery{\safe}{G}$.
\end{enumerate}
Thus, we conclude that $\ppa_1 \parallel \ppa_2$ instantiated at $\valuation$ under $\strategy$ satisfies $\multiobjectiveQuery{\safe}{G}$. 
The rule on the right holds by a similar reasoning, where, in addition, \Cref{theo:strategy_projection_partial_fair_preserved} is used to establish that projections of fair strategies remain fair. 
\end{proof}

\begin{exa}
    We illustrate the left proof rule from \Cref{theo:pag_asym_rule} for the pPA $\ppa_1 \parallel \ppa_2$ in \Cref{fig:pPA_composition}---composed of the pPAs $\ppa_1$ and $\ppa_2$ depicted in \Cref{fig:ppa1_and_ppa2}---and w.r.t.\ 
    $\multiobjectiveQuery{}{G} = \probPredicate{\geq 0.9}{\regLang_G}$, 
    where $\regLang_G = \{ w \in \{\lab, \altlab, \altaltlab, \frownie\}^{\infty} \mid \vert w \vert_{\frownie} = 0 \}$. 
    Let $\multiobjectiveQuery{}{A} = \{\probPredicate{\geq 0.9}{\regLang_A}\}$, where $\regLang_A = \{ w \in \{\lab, \altlab\}^{\infty} \mid  \vert  w  \vert_{\lab} \leq 1  \}$. 
    The pPA $\alphabetExtensionOfTo{\ppa_2}{\{\lab,\altlab\}}$ is depicted in \Cref{fig:ppaExtension}. 
    For the premises of the proof rule, we obtain that the (largest) region $\region_1$ for which $\ppa_1, \region_1 \modelsWrt{\comp} \probPredicate{\geq 0.9}{\regLang_A}$ is $\region_1 = \{\valuation \colon \{p,q\} \to \reals \mid \valuation(p) \in [0,0.1]\}$ and 
    the (largest) region $\region_2$ for which $\agTriple{\alphabetExtensionOfTo{\ppa_2}{\{\lab,\altlab\}}, \region_2}{\prt}{\multiobjectiveQuery{}{A}}{\multiobjectiveQuery{}{G}}$ holds, is 
    $\region_2= \{\valuation \colon \{p,q\} \to \reals \mid  \valuation(p) \in [0,1) \land \valuation(q) \in [0,1-\valuation(p)] \}$.
    %
    The intersection $\regionIntersectionOf{\region_1}{}{\region_2}$---for which $\ppa_1 \parallel \ppa_2,\regionIntersectionOf{\region_1}{}{\region_2} \modelsWrt{\prt}  \multiobjectiveQuery{}{G}$ holds by \Cref{theo:pag_asymN_rule}---contains all valuations with $\valuation(p) \in [0,0.1], \valuation(q) \in [0,1]$. 
    The (largest) region $\region$ for which $\ppa_1 \parallel \ppa_2, \region \models^{\prt}  \multiobjectiveQuery{}{G}$ is 
    $\region = \{\valuation \colon \{p,q\} \to \reals \mid (\valuation(p) \in [0, \frac{1}{9}] \cup \{1\} \land \valuation(q) \in [0,1]) \lor (\valuation(p) \in (\frac{1}{9}, 1),  \valuation(q) \in [0,\frac{\valuation(p) +1}{10 \cdot \valuation(p) }]) \}$. 
    This satisfies $(\regionIntersectionOf{\region_1}{}{\region_2}) \subset \region$. 
\end{exa}
   The proof rules in \Cref{theo:pag_asym_rule} can be extended to systems with more than two components, as detailed in \Cref{theo:pag_asymN_rule} in \Cref{app:pag_extension}. 
   Next, we lift the circular proof rule given in \cite[Theorem 5]{Kwi+13} to pPAs: 
\begin{restatable}[Circular Rule]{thm}{restatableCircRule}\label{theo:pag_circ_rule}
    Let $\multiobjectiveQuery{}{A}_1$,$\multiobjectiveQuery{}{A}_2$ and $\multiobjectiveQuery{}{G}$ be (parametric) mo-queries over $\alphabetOf{\multiobjectiveQuery{}{A}_1} \subseteq \alphabetOf{2}$, $\alphabetOf{\multiobjectiveQuery{}{A}_2} \subseteq \alphabetOf{1} \cup \alphabetOf{\multiobjectiveQuery{}{A}_1}$ and  
    $\alphabetOf{\multiobjectiveQuery{}{G}} \subseteq \alphabetOf{2} \cup \alphabetOf{\multiobjectiveQuery{}{A}_2}$, respectively. 
    Let $\decomp_i \in \setOfdecompsOf{\alphabetOf{i}\cup \alphabetOf{\multiobjectiveQuery{}{A}_{i}}}$ for $i \in \{1,2\}$, and $\decomp_3 \in \setOfdecompsOf{\alphabetOf{2}}$. 
    Then: 
    
    \begin{tabularx}{\linewidth}{p{0.45\linewidth} p{0.45\linewidth}} 

                $
                \infer{
                    \ppa_1 \parallel \ppa_2, \regionIntersectionOf{\region_1}{\region_2}{\region_3} \modelsWrt{\comp} \multiobjectiveQuery{\safe}{G}
                } 
                {
                    \deduce{
                        \deduce{
                            \ppa_2, \region_3 \modelsWrt{\comp} \multiobjectiveQuery{\safe}{A}_1
                            }{
                                \agTriple{\alphabetExtensionOfTo{\ppa_2}{\alphabetOf{\multiobjectiveQuery{\safe}{A}_{2}}}, \region_2}{\prt}{\multiobjectiveQuery{\safe}{A}_{2}}{\multiobjectiveQuery{\safe}{G}} 
                            }
                    }{\agTriple{\alphabetExtensionOfTo{\ppa_1}{\alphabetOf{\multiobjectiveQuery{\safe}{A}_{1}}}, \region_1 }{\prt}{\multiobjectiveQuery{\safe}{A}_{1}}{\multiobjectiveQuery{\safe}{A}_2} 
                    }
                }
                $
        & 
        \hspace{8pt}
               $
                \infer{
                    \ppa_1 \parallel \ppa_2, \regionIntersectionOf{\region_1}{\region_2}{\region_3} \modelsWrt{\fairWrtRegionModel{}{\decomp_1 \cup \decomp_2 \cup \decomp_3}} \multiobjectiveQuery{}{G}
                } 
                {
                    \deduce{
                        \deduce{
                            \ppa_2, \region_3 \modelsWrt{\fairWrtRegionModel{}{\decomp_3}} \multiobjectiveQuery{}{A}_1
                            }{
                                \agTriple{\alphabetExtensionOfTo{\ppa_2}{\alphabetOf{\multiobjectiveQuery{}{A}_{2}}}, \region_2}{\fairWrtRegionModel{}{\decomp_2}}{\multiobjectiveQuery{}{A}_{2}}{\multiobjectiveQuery{}{G}} 
                            }
                    }{\agTriple{\alphabetExtensionOfTo{\ppa_1}{\alphabetOf{\multiobjectiveQuery{}{A}_{1}}}, \region_1 }{\fairWrtRegionModel{}{\decomp_1}}{\multiobjectiveQuery{}{A}_{1}}{\multiobjectiveQuery{}{A}_2} 
                    }
                }
                $
    \end{tabularx}
\end{restatable}
\begin{proof}[Proof sketch]
Similar to \Cref{theo:pag_asym_rule}, the proof of the circular rules makes use of \Cref{theo:lemma3NewDependent}, which establishes that the composition under a strategy satisfies the same properties as the individual components under their corresponding projections.
For safety, \Cref{lemma_safety_partial_vs_complete} allows us to verify the condition for complete strategies rather than partial strategies in the third premise. 
For fairness, \Cref{theo:strategy_projection_partial_fair_preserved} ensures that strategy projections remain fair. 
\end{proof}
\begin{rem}
    The inclusion of fairness in the premises of the right rules in \Cref{theo:pag_asym_rule} and \Cref{theo:pag_circ_rule} enables recursive application  and thus supports the compositional verification of systems with more than two components. 
    In the case of a single application of one of the rules, it is sufficient to verify with respect to complete strategies, which, while a stronger condition, simplifies the verification process.
\end{rem}

\section{Compositional Reasoning about Monotonicity}\label{sec:pag_mono}
Exploiting monotonicity can significantly enhance the efficiency of parameter synthesis~\cite{Spe+21}. 
However, determining monotonicity is computationally hard\footnote{For deterministic pPA (Markov chains) determining monotonicity is coETR-hard~\cite[Sec.\ 3.4]{Spel23}.} and it would be beneficial to determine monotonicity in a compositional way. 
Additionally, monotonicity in composed pPAs is challenging due to the complexities introduced by parameter dependencies and interactions among components. 
While we focus on global monotonicity, the following results can be extended to \emph{local monotonicity}, which considers only the first transition from a given state. See \cite[Definitions 4.4 and 4.5]{Spel23}.

The probability of a language or the expected total reward for a pPA $\ppa$ can be viewed as a function\textemdash{}called \emph{solution function}\textemdash{}that maps a well-defined parameter valuation to the corresponding probability or expected total reward, respectively \cite[Definition~4.7]{Jun20}. 
\begin{defi}
    Let $\ppa$ be a pPA over $\alphabetOf{}$, let $\strategy \in \strategysetOf{\ppa}{}$ and let $\region$ be a well-defined region.
    The \emph{solution function} for $\ppa$ and language $\regLang \subseteq \alphabetOf{}^\infty$ 
    is $\solutionFctMdpObjective{\ppa, \strategy}{\PrOf{}{}{\regLang}} \colon \region  \to [0,1]$, where $\solutionFctMdpObjective{\ppa, \strategy}{\PrOf{}{}{\regLang}}(\valuation) = 
    \PrOf{\ppa}{\valuation,\strategy}{\regLang}$. 
    The solution function for $\ppa$ and a reward function $\rewFct$ over $\alphabetOf{}$ is $\solutionFctMdpObjective{\ppa, \strategy}{\ExpTot{}{}{\rewFct}} \colon \region  \to \reals_{\geq 0}$, where $\solutionFctMdpObjective{\ppa, \strategy}{\ExpTot{}{}{\rewFct}}(\valuation) = 
    \ExpTot{\ppa}{\valuation,\strategy}{\rewFct}$. 
\end{defi}
When referring to a solution function without specifying whether it pertains to probabilities or expected rewards, we simply write $\solutionFctMdpObjective{\ppa, \strategy}{}$. 

\begin{exa}
\label{ex:solfct}
    Consider the pPA $\ppa_1\parallel\ppa_2$ in \Cref{fig:pPA_composition} and the region $\region =\{\valuation \colon \{p,q\} \to [0,1]\}$ which is well-defined for $\ppa_1 \parallel \ppa_2$. 
    Let $\regLang = \{ w \in \{\lab,\altaltlab,\frownie\}^{\infty} \mid \vert w \vert_{\frownie} = 0 \}$ be the language of words over $\{\lab,\altlab,\altaltlab, \frownie\}$ that do not contain $\frownie$. 
    Let $\strategy$ be the complete strategy of $\ppa_1\parallel\ppa_2$ from \Cref{ex:projectionPPA}, which always selects an action labeled $\lab,\altaltlab$ or $\frownie$ with probability 1 whenever any of them is enabled; otherwise, it chooses $\altlab$ with probability 1. 
    The solution function $\solutionFctMdpObjective{\ppa_1 \parallel \ppa_2, \strategy}{\PrOf{}{}{\regLang}} \colon \region \to [0,1]$ is defined by 
    $\solutionFctMdpObjective{\ppa_1 \parallel \ppa_2, \strategy}{\PrOf{}{}{\regLang}}(p,q) = 
    1- \left(p^2 \cdot \frac{1}{10} + p \cdot(1-p) \cdot (p\cdot q + (1-p) \cdot q)\right) 
   = 1- \left(p^2 \cdot \frac{1}{10} + (p-p^2) \cdot q\right) 
    $. 
\end{exa}

\label{sec:monotonicity}
We extend the standard notion of monotonicity~\cite{Spel23} by differentiating between different strategy classes, including complete, partial, and fair strategies. 
\begin{defi}
    \label{def:monotonicity}
    Let $\strategy$ be a strategy of $\ppa$. A solution function $\solutionFctMdpObjective{\ppa,\strategy}{}$ is \emph{monotonically increasing} in $p\in \parameterSetOf{}$ on region $\region$\textemdash{}denoted $\monotonicOnRegionParameter{\buparrow}{\solutionFctMdpObjective{\ppa, {\strategy}}{}}{p}{\region}{}$\textemdash{}if for all $\valuation, \valuation_{+} \in \region$ with $\valuation_{+}(q) = \valuation(q) + x \cdot \iverson{p{=}q}$ for $q \in \parameterSetOf{}$ and some $x \geq 0$, we have: 
    ${\solutionFctMdpObjective{{\ppa},{\strategy}}{}}{}(\valuation) \leq {\solutionFctMdpObjective{{\ppa},{\strategy}}{}}{}(\valuation_{+}).$

    For $\star \in \{ \prt, \comp\}$, 
    we write $\monotonicOnRegionParameter{\buparrow}{\solutionFctMdpObjective{\ppa}{}}{p}{\region}{\star}$ if $\monotonicOnRegionParameter{\buparrow}{\solutionFctMdpObjective{\ppa, {\strategy}}{}}{p}{\region}{}$ for all $\strategy \in \strategysetOf{\ppa}{\star}$.
    If $\region$ is graph-preserving, we write $\monotonicOnRegionParameter{\buparrow}{\solutionFctMdpObjective{\ppa}{}}{p}{\region}{\fairWrtRegionModel{}{\decomp}}$ if $\monotonicOnRegionParameter{\buparrow}{\solutionFctMdpObjective{\ppa, {\strategy}}{}}{p}{\region}{}$ holds for all fair strategies $\strategy \in \strategysetOf{\ppa[\valuation]}{\fairWrtRegionModel{}{\decomp}}$, $\valuation \in \region$.

    \noindent Notations $\monotonicOnRegionParameter{\bdownarrow}{\solutionFctMdpObjective{\ppa, {\strategy}}{}}{p}{\region}{}$ and $\monotonicOnRegionParameter{\bdownarrow}{\solutionFctMdpObjective{\ppa}{}}{p}{\region}{\star}$ for \emph{monotonically decreasing} $\solutionFctMdpObjective{\ppa,\strategy}{}$ are defined analogously. 
\end{defi}
We require the region to be graph-preserving when defining monotonicity w.r.t.\ fair strategies. 
This ensures that for any two valuations, $\valuation, \valuation_+$, we have $\strategysetOf{\ppa[\valuation]}{\fairWrtRegionModel{}{\decomp}}= \strategysetOf{\ppa[\valuation_+]}{\fairWrtRegionModel{}{\decomp}}$; see \Cref{theo:graph_preserving_fairness}. 
\begin{rem}
        Monotonicity for partial strategies w.r.t.\ general properties is equivalent to monotonicity for complete strategies in a modified pPA; see \Cref{app:equivalence_partial_complete_mono}. 
\end{rem}

The following theorem states that  monotonicity of a composed system can be verified by analyzing its individual components. 
\begin{restatable}[Monotonicity]{thm}{restatablMonoRule}\label{theo:pag_mono_rule}
    Let $\ppa_1$, $\ppa_2$ be pPAs with alphabets $\alphabetOf{1}$ and $\alphabetOf{2}$ and $\region_i$ be a graph-preserving region for $\ppa_i$. 
    Let $\solutionFctMdpObjective{}{} \in \{\solutionFctMdpObjective{}{\PrOf{}{}{\regLang}}, \solutionFctMdpObjective{}{\ExpTot{}{}{\rewFct}} \}$ 
	be a solution function w.r.t.\ the language $\regLang$ or reward function $\rewFct$ over $\alphabetOf{} \subseteq (\alphabetOf{1}\cup \alphabetOf{2})$ and 
    let $\budarrow \in \{\buparrow, \bdownarrow\}$. 
    Let $\decomp_i \subseteq \setOfdecompsOf{\alphabetOf{1} \cup {\alphabetOf{}}}$. 
    Then the following two proof rules hold:

    \begin{tabularx}{\linewidth}{p{0.45\linewidth}p{0.45\linewidth}} 
            \hspace{20pt}
                $
                \infer{
                    \monotonicOnRegionParameter{\budarrow}{\solutionFctMdpObjective{\ppa_1 \parallel \ppa_2}{}}{p}{\regionIntersectionOf{\region_1}{}{\region_2}}{\prt}
                } 
                {
                    \deduce{\monotonicOnRegionParameter{\budarrow}{\solutionFctMdpObjective{\alphabetExtensionOfTo{\ppa_2}{\alphabetOf{}}}{}}{p}{\region_2}{\prt}
                    }{
                        \monotonicOnRegionParameter{\budarrow}{\solutionFctMdpObjective{\alphabetExtensionOfTo{\ppa_1}{\alphabetOf{}}}{}}{p}{\region_1}{\prt} 
                    }
                }
                $
        & 
       \hspace{20pt}
       $
       \infer{
           \monotonicOnRegionParameter{\budarrow}{\solutionFctMdpObjective{\ppa_1 \parallel \ppa_2}{}}{p}{\regionIntersectionOf{\region_1}{}{\region_2}}{\fairWrtRegionModel{}{\decomp_1 \cup \decomp_2}}
       } 
       {
           \deduce{\monotonicOnRegionParameter{\budarrow}{\solutionFctMdpObjective{\alphabetExtensionOfTo{\ppa_2}{\alphabetOf{}}}{}}{p}{\region_2}{\fairWrtRegionModel{}{\decomp_2}} 
           }{
               \monotonicOnRegionParameter{\budarrow}{\solutionFctMdpObjective{\alphabetExtensionOfTo{\ppa_1}{\alphabetOf{}}}{}}{p}{\region_1}{\fairWrtRegionModel{}{\decomp_1}}  
           }
       }
       $
    \end{tabularx} 
\end{restatable}

\begin{proof}
We show the premises imply $\monotonicOnRegionParameter{\budarrow}{\solutionFctMdpObjective{\alphabetExtensionOfTo{\ppa_1}{\alphabetOf{}} \parallel \alphabetExtensionOfTo{\ppa_2}{\alphabetOf{}}}{}}{p}{\regionIntersectionOf{\region_1}{}{\region_2}}{\star}$ for $\star \in \set{\prt,\fairWrtRegionModel{}{\decomp_1\cup\decomp_2}}$, which directly implies that $\monotonicOnRegionParameter{\budarrow}{\solutionFctMdpObjective{\ppa_1 \parallel \ppa_2}{}}{p}{\regionIntersectionOf{\region_1}{}{\region_2}}{\star}$ holds; see \Cref{rem:extension}.
We focus on the left rule, i.e., $\star = \prt$.
\begin{rem}
       The proof for $\star = \fairWrtRegionModel{}{\decomp_1\cup\decomp_2}$ is similar but additionally requires \Cref{theo:strategy_projection_partial_fair_preserved}; see \Cref{app:proofs_of_pag_mono}. 
\end{rem}
We further consider $\budarrow  = \buparrow$. 
The case $\budarrow =\bdownarrow$ follows analogously. 
Our proof is by contradiction.
        Assume that the premises hold but $\monotonicOnRegionParameter{\buparrow}{\solutionFctMdpObjective{\alphabetExtensionOfTo{\ppa_1}{\alphabetOf{}} \parallel \alphabetExtensionOfTo{\ppa_2}{\alphabetOf{}}}{}}{p}{\region_1 \cap \region_2}{\star}$ does not hold. 
        Thus, there is a strategy $\strategy \in  \strategysetOf{(\alphabetExtensionOfTo{\ppa_1}{\alphabetOf{}} \parallel \alphabetExtensionOfTo{\ppa_2}{\alphabetOf{}})}{\star}$ and 
        valuations $\valuation, \valuation_+ \in \region_1 \cap \region_2$
        with $\valuation_{+}(q) = \valuation(q) + x \cdot \iverson{p{=}q}$ for $q \in \parameterSetOf{}$ and some $x \geq 0$ and 
        \begin{align}\label{proof:contradiction}
            & {\solutionFctMdpObjective{{(\alphabetExtensionOfTo{\ppa_1}{\alphabetOf{}} \parallel \alphabetExtensionOfTo{\ppa_2}{\alphabetOf{}})},{\strategy}}{}}{}(\valuation) > 	
            {\solutionFctMdpObjective{{(\alphabetExtensionOfTo{\ppa_1}{\alphabetOf{}} \parallel \alphabetExtensionOfTo{\ppa_2}{\alphabetOf{}})},{\strategy}}{}}{}(\valuation_{+}). 
        \end{align}
        \Cref{theo:lemma3NewDependent} yields 
        $ {\solutionFctMdpObjective{{(\alphabetExtensionOfTo{\ppa_1}{\alphabetOf{}} \parallel \alphabetExtensionOfTo{\ppa_2}{\alphabetOf{}})},{\strategy}}{}}{}(\valuation) 
        = {\solutionFctMdpObjective{{\alphabetExtensionOfTo{\ppa_1}{\alphabetOf{}}},{{\stratProjOfToValuation{\strategy}{1}{\valuation, \valuation}}}}{}}{}(\valuation)$. 
        We  have $\transFctOf{\alphabetExtensionOfTo{\ppa_1}{\alphabetOf{}}}(s,\alpha,s')[\valuation] = 0$ iff $\transFctOf{\alphabetExtensionOfTo{\ppa_1}{\alphabetOf{}}}(s,\alpha,s')[\valuation_+] = 0$ as $\valuation$ and $\valuation_+$ are graph preserving for $\ppa_1$. 
        Thus, we can apply \Cref{theo:change_valuation_projection} and obtain 
        \begin{align*}
            {\solutionFctMdpObjective{{\alphabetExtensionOfTo{\ppa_1}{\alphabetOf{}}},{{\stratProjOfToValuation{\strategy}{1}{\valuation, \valuation}}}}{}}{}(\valuation) 
            & = {\solutionFctMdpObjective{{\alphabetExtensionOfTo{\ppa_1}{\alphabetOf{}}},{{\stratProjOfToValuation{\strategy}{1}{\valuation_+, \valuation}}}}{}}{}(\valuation) 
            \tag{by \Cref{theo:change_valuation_projection}}\\
            & \leq  {\solutionFctMdpObjective{{\alphabetExtensionOfTo{\ppa_1}{\alphabetOf{}}},{{\stratProjOfToValuation{\strategy}{1}{\valuation_+, \valuation}}}}{}}{}(\valuation_+)
             \tag{ ${\stratProjOfToValuation{\strategy}{1}{\valuation_+, \valuation}}\in  \strategysetOf{\alphabetExtensionOfTo{\ppa_1}{\alphabetOf{}}}{\prt}$,  $\monotonicOnRegionParameter{\buparrow}{\solutionFctMdpObjective{\alphabetExtensionOfTo{\ppa_1}{\alphabetOf{}}}{}}{p}{\region_1}{\prt}$} \\
            & = {\solutionFctMdpObjective{{(\alphabetExtensionOfTo{\ppa_1}{\alphabetOf{}} \parallel \alphabetExtensionOfTo{\ppa_2}{\alphabetOf{}}[\valuation])},{\strategy}}{}}{} (\valuation_+)
             \tag{by \Cref{theo:lemma3NewDependent}}\\
        \end{align*}
        We observe that 
        \[\big(\alphabetExtensionOfTo{\ppa_1}{\alphabetOf{}} \parallel \alphabetExtensionOfTo{\ppa_2}{\alphabetOf{}}[\valuation]\big)[\valuation_+]
        ~=~\big(\alphabetExtensionOfTo{\ppa_1}{\alphabetOf{}}[\valuation_+] \parallel \alphabetExtensionOfTo{\ppa_2}{\alphabetOf{}}[\valuation]\big)
        ~=~\big(\alphabetExtensionOfTo{\ppa_1}{\alphabetOf{}}[\valuation_+] \parallel \alphabetExtensionOfTo{\ppa_2}{\alphabetOf{}}\big)[\valuation].\]
        Consequently,
        ${\solutionFctMdpObjective{{(\alphabetExtensionOfTo{\ppa_1}{\alphabetOf{}} \parallel \alphabetExtensionOfTo{\ppa_2}{\alphabetOf{}}[\valuation])},{\strategy}}{}}{} (\valuation_+) = {\solutionFctMdpObjective{{(\alphabetExtensionOfTo{\ppa_1}{\alphabetOf{}}[\valuation_+] \parallel \alphabetExtensionOfTo{\ppa_2}{\alphabetOf{}}},{\strategy}}{}}{} (\valuation)$. 
        By a similar reasoning as above, we obtain    
        \begin{align*}
            {\solutionFctMdpObjective{{(\alphabetExtensionOfTo{\ppa_1}{\alphabetOf{}}[\valuation_+] \parallel \alphabetExtensionOfTo{\ppa_2}{\alphabetOf{}})},{\strategy}}{}}{} (\valuation)
            &= {\solutionFctMdpObjective{{\alphabetExtensionOfTo{\ppa_2}{\alphabetOf{}}}, {{\stratProjOfToValuation{\strategy}{2}{\valuation_+, \valuation}}}}{}}{}(\valuation)
              \tag{by \Cref{theo:lemma3NewDependent}} 
            \\
            & \leq  {\solutionFctMdpObjective{{\alphabetExtensionOfTo{\ppa_2}{\alphabetOf{}}},{{\stratProjOfToValuation{\strategy}{2}{\valuation_+, \valuation}}}}{}}{}(\valuation_+)
             \tag{ ${\stratProjOfToValuation{\strategy}{2}{\valuation_+, \valuation}}\in  \strategysetOf{\alphabetExtensionOfTo{\ppa_2}{\alphabetOf{}}}{\prt}$,  $\monotonicOnRegionParameter{\buparrow}{\solutionFctMdpObjective{\alphabetExtensionOfTo{\ppa_2}{\alphabetOf{}}}{}}{p}{\region_2}{\prt}$} \\
            & = {\solutionFctMdpObjective{{\alphabetExtensionOfTo{\ppa_2}{\alphabetOf{}}}, {{\stratProjOfToValuation{\strategy}{2}{\valuation_+, \valuation_+}}}}{}}{}(\valuation_+) 
         \tag{by \Cref{theo:change_valuation_projection} }\\
        &= {\solutionFctMdpObjective{{(\alphabetExtensionOfTo{\ppa_1}{\alphabetOf{}} \parallel \alphabetExtensionOfTo{\ppa_2}{\alphabetOf{}})},{\strategy}}{}}{}(\valuation_+)  \tag{by \Cref{theo:lemma3NewDependent}}
    \end{align*}
        Thus, $ {\solutionFctMdpObjective{{(\alphabetExtensionOfTo{\ppa_1}{\alphabetOf{}} \parallel \alphabetExtensionOfTo{\ppa_2}{\alphabetOf{}})},{\strategy}}{}}{}(\valuation) \leq {\solutionFctMdpObjective{{(\alphabetExtensionOfTo{\ppa_1}{\alphabetOf{}} \parallel \alphabetExtensionOfTo{\ppa_2}{\alphabetOf{}})},{\strategy}}{}}{}(\valuation_+) $, violating \Cref{proof:contradiction}. 
        \end{proof}

\begin{exa}\label{ex:mono_compositional}
    Reconsider the pPA $\ppa_1\parallel\ppa_2$ in \Cref{fig:pPA_composition}, the region $\region =\{\valuation \colon \{p,q\} \to [0,1]\}$, and the language $\regLang = \{ w \in \{a,c,\frownie\}^{\infty} \mid \vert w \vert_{\frownie} = 0 \}$ from \Cref{ex:solfct}. 
    The pPA $\ppa_1\parallel\ppa_2$ is composed of the pPAs $\ppa_1$ and $\ppa_2$ shown in \Cref{fig:ppa1_and_ppa2}. 
    The region $\region$ is well-defined for $\ppa_1$ and $\ppa_2$. 
    We check whether $\solutionFctMdpObjective{\ppa_1 \parallel \ppa_2}{\PrOf{}{}{\regLang}}$ is monotonic in $q$ on $\region$ via \Cref{theo:pag_mono_rule}. 
    Since the premises $\monotonicOnRegionParameter{\bdownarrow}{\solutionFctMdpObjective{\alphabetExtensionOfTo{\ppa_i}{\alphabetOf{\regLang}}}{\PrOf{}{}{\regLang}}}{q}{\region}{\prt}$ for $i \in \{1,2\}$ are satisfied, we conclude $\monotonicOnRegionParameter{\bdownarrow}{\solutionFctMdpObjective{\ppa_1 \parallel \ppa_2}{\PrOf{}{}{\regLang}}}{q}{\region}{\prt}$. 
\end{exa}

\section{AG Reasoning for Robust PAs}
\label{sec:studying_ag_for_rpa_semantics}
We now study AG reasoning for \emph{robust} probabilistic automata (rPAs), a compositional variant of robust MDPs~\cite{NG05,Wi+13,Iye05,Su+25}. 
In contrast to parametric PAs, where a valuation determines all parameters once and for all, rPAs feature a nature player that, after each action choice, selects a single successor distribution from a (possibly uncountable) uncertainty set of distributions.
%
We adopt semantic choices that---while standard in the robust MDP literature---conceptually differ from (parametric) PA semantics: 
\begin{itemize}
  \item Nature in rPAs may be \emph{memoryless} (once-and-for-all semantics) or
        \emph{memory-full} (history-dependent), whereas strategies in the PA AG
        framework~\cite{Kwi+13} are always memory-full and parameter valuations for pPAs are memoryless.
  \item Nature is \emph{non-probabilistic}, while PA strategies are probabilistic; 
  for non-convex uncertainty sets, these two notions are not equivalent. 
\end{itemize}
We analyse AG reasoning under different rPA semantics and identify those cases in which the PA-based AG rules of Kwiatkowska et al.~\cite{Kwi+13} can straightforwardly be lifted to rPAs. 
\Cref{sec:rpa_prelim} presents preliminary notions for rPAs.
In \Cref{sec:agr_rpa_discussion}, we examine the applicability of AG rules for general rPAs and present counterexamples to the asymmetric AG rule for (i) memoryless nature, and (ii) non-convex uncertainty sets. 
\Cref{sec:AG_for_rpa_conv} restricts to convex rPAs with memory-full nature and identifies a relation to (possibly uncountably branching) PAs. 
We recover compositional AG rules by introducing a convexity-preserving composition operator $\parallelConv$.
Finally,~\Cref{sec:interval_arithm_relaxation} considers rPAs with interval transitions and a common interval-arithmetic relaxation of the parallel composition (e.g.,~\cite{Kwi+11,Has+16}) that preserves the interval structure while introducing spurious distributions~\cite{SchnitzerAP26}.
We show that, under this relaxed composition, the
AG framework of~\cite{Kwi+13}---in particular the asymmetric rule---cannot be lifted directly.

\subsection{Robust probabilistic automata} 
\label{sec:rpa_prelim}
We introduce \emph{robust probabilistic automata (rPAs)} by combining rMDPs (e.g., ~\cite[Definition~3.25]{Badings25Thesis}) and Segala’s PAs~\cite{Seg95}.
The transition function and the action structure follow our definition of pPAs, see \Cref{def:ppa,rem:padef}. 
\begin{defi}\label{def:rpa}
  A \emph{robust probabilistic automaton} (rPA) over finite alphabet $\alphabetOf{}$ is a tuple $\rpa = \rpaTupleOf{\rpa}{}$, where $\stateSetOf{\rpa}$, $\actSetOf{\rpa}$, $\initialOf{\rpa} \in \stateSetOf{\rpa}$, and $\syncOf{\rpa} \colon \domain(\transFctOf{\rpa}) \to \alphabetOf{\rpa}$ are as for pPAs and the partial function
$
\transFctOf{\rpa} \colon (\stateSetOf{\rpa} \times \actSetOf{\rpa}) \hookrightarrow 2^{\dist{\stateSetOf{\rpa}}},
$
  is an \emph{uncertain transition function} mapping each state–action pair $(s,\alpha) \in \domain(\transFctOf{\rpa})$ to a non-empty set $\transFctOf{s,\alpha} \subseteq \dist{\stateSetOf{\rpa}}$ of successor distributions.
\end{defi}
We may drop the subscript $\rpa$, if it is clear from the context. 
For a state $s \in \stateSetOf{}$ and an action $\alpha \in \actSetOf{}$, the set $\transFctOf{}(s,\alpha) \subseteq \dist{\stateSetOf{}}$ is called an \emph{uncertainty set}.  
As we assume a finite state space $\stateSetOf{}$, all probability distributions in  $\dist{\stateSetOf{}}$ have discrete, finite support. 
\begin{rem}\label{rem:rectangular} 
  The above definition assumes that the uncertainty sets
  $\transFctOf{}(s,\alpha) \subseteq \dist{\stateSetOf{}}$ for different
  state–action pairs $(s,\alpha)$ are assumed to be independent---following the standard $(s,\alpha)$ rectangularity assumption for robust or parametric MDPs in the literature~\cite{Wi+13,Badings25Thesis}. 
\end{rem}

A set $D \subseteq  \dist{\stateSetOf{\rpa}}$ is \emph{convex} if it is equal to its convex hull, i.e., $D = \convh{D}$, where
\[
\convh{D} ~=~ \Big\{ \sum_{i=1}^k \lambda(i) \cdot \mu_i ~\Big|~ k \in \mathbb{N}\,,~\{\mu_1, \dots, \mu_k\} \subseteq D\,,~\lambda \in \dist{\{1,\dots,k\}}  \Big\}.
\]
The set of \emph{extreme points} $\extr{D}$ of a convex set $D$ is the set of distributions that can not be expressed as a convex combination of the remaining elements of $D$, i.e.,  
\[
\extr{D} ~=~ \Big\{ \mu \in D ~\Big|~ \convh{D \setminus \{\mu\}} \neq D \Big\}.
\]
Convex set $D$ is \emph{polytopic}, if it coincides with the convex hull of its finitely many extreme points, i.e., $|\extr{D}|<\infty$ and  $D=\convh{\extr{D}}$.
\begin{defi}
\label{def:convex_rpa}
    An rPA is called \emph{convex} (\emph{polytopic}), if for every $(s,\alpha) \in \domain(\transFctOf{})$ the uncertainty set $\transFctOf{}(s,\alpha) \subseteq \dist{\stateSetOf{}}$ is convex (polytopic). 
\end{defi}
\begin{figure}[t] 
	\begin{subfigure}[c]{.35\textwidth}
		\centering 	
		\begin{tikzpicture}[mdp]

	\node[ps, init=left] (0)  {$s_0$};

	\node[ps,right=2.5 of 0] (1)  {$s_1$};

	\path[ptrans]

	(0) edge[bend right=0] node[pos=0.5,above] {\tact{\lab}} node[dist] (d0a) {} node[pos=0.75,above] {\tprob{[\frac{1}{2},1]}} (1)

    (d0a) edge[bend left=60] node[left,pos=0.45] {\tprob{}} node[below,pos=0.6] {\tprob{[0,\frac{1}{2}]}} (0)
	
	(1) edge[loop above] node[pos=0.5,above right] {\tact{\lab},\tact{\altlab}} node[dist] (d1c) {} node[pos=0.75,below] {\tprob{}} node[pos=0.5,below] {} (1)

	;
\end{tikzpicture}%
		\caption{rPA $\rpa_1$}\label{fig:ipa_1}%
	\end{subfigure}\hfill
	\begin{subfigure}[c]{.45\textwidth} 
		\centering%
		\begin{tikzpicture}[mdp]
	\node[ps, init=left] (t0)  {$t_0$};
	
	\node[ps,above right=0.25 and 2 of t0] (t1)  {$t_1$};

	\node[ps,below right=0.25 and 2 of t0] (t2)  {$t_2$};

	\node[dist, right=0.3 of t0] (d0a) {};
		
	\path[ptrans]
	
		(t0) edge[-] node[pos=1,above] {\tact{\lab}} (d0a)
		(d0a) edge[bend right=0] node[pos=0.5,above,yshift=0.5ex,xshift=-0.5ex,rotate=25] {\tprob{[\frac{1}{10},\frac{9}{10}]}} (t1)
		(d0a) edge[bend right=0] node[below,pos=0.55,yshift=-0.5ex,xshift=-0.5ex,rotate=-25] {\tprob{[\frac{1}{10},\frac{9}{10}]}} node[above,pos=0.4] {} (t2)

        (t1) edge[bend right=60] node[pos=0.5,above] {\tact{\altaltlab}} node[dist] (d1a) {} node[pos=0.75,left] {\tprob{}} (t0)

        (t2) edge[loop right] node[pos=0.5,above right] {\tact{\altlab},\tact{\frownie}} node[dist] (d1c) {} node[pos=0.75,below] {\tprob{}} node[pos=0.5,below] {} (t2)
	;

\end{tikzpicture}%
		\caption{rPA $\rpa_2$}\label{fig:ipa_2}%
	\end{subfigure} 
	\caption{Example robust PAs $\rpa_1$ and $\rpa_2$.}\label{fig:ipas}
\end{figure}
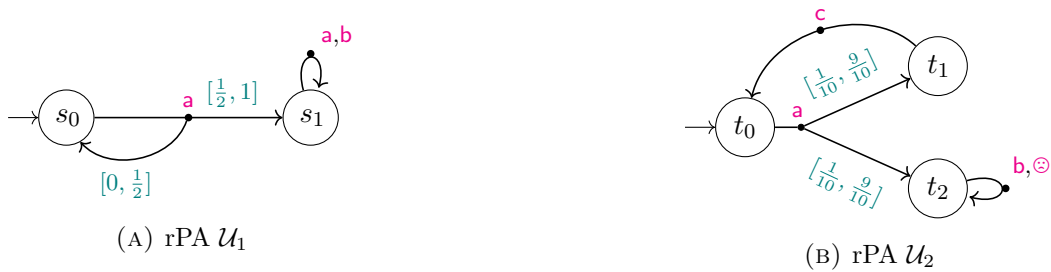
\begin{exa}\label{exa:rpas}
    Consider the rPAs $\rpa_1$ and $\rpa_2$ depicted in \Cref{fig:ipa_1,fig:ipa_2}, respectively.
    The alphabet of $\rpa_1$ is $\alphabetOf{1} = \{\lab, \altlab \}$ and for $\rpa_2$ it is $\alphabetOf{2} = \{\lab, \altlab, \altaltlab, \frownie \}$. 
    
    In all rPA figures, we display only transition labels and omit the underlying actions, since in the depicted rPAs the actions are uniquely determined by the transition label. 
    By misuse of notation, we often write $\transFctOf{}(s,\lab)$ as shorthand for $\transFctOf{}(s,\alpha)$, where $\alpha$ is the unique action with $\syncOf{}(s,\alpha) = \lab$. 
 
    The uncertainty sets in \Cref{fig:ipas} are given in interval form. 
    Informally, an annotation of the form $[\ell,u]$ next to a transition from state $s$ to a successor state $s'$ indicates that for every distribution $\mu$ in the corresponding uncertainty set we have $\ell \leq \mu(s') \leq u$; the probabilities for all successors of that transition must sum to~1. 
    For instance, in $\rpa_1$ the uncertainty set $\transFctOf{}(s_0,\lab)$ contains exactly those mappings $\mu \colon \{s_0,s_1\} \to [0,1]$ with $\mu(s_0)+\mu(s_1)=1$ (i.e., $\mu \in \dist{\{s_0,s_1\}}$) and for which $0  \leq  \mu(s_0)  \leq  \frac{1}{2}$ and $\frac{1}{2} \leq  \mu(s_1) \leq 1$.

    All uncertainty sets in $\rpa_1$ and $\rpa_2$ are polytopic uncertainty sets. 
    In particular, both $\rpa_1$ and $\rpa_2$ are \emph{interval probabilistic automata} (iPAs), a special type of rPAs in which each uncertainty set is defined by lower and upper bounds on transition probabilities.
    We further elaborate on iPAs in \Cref{sec:interval_arithm_relaxation}. 
\end{exa}

Infinite and finite paths of an rPA are defined analogously to pPAs: 
An \emph{infinite path} of $\rpa$ is an alternating sequence $\infpath = s_0,\alpha_0,s_1,\alpha_1,\ldots$ of states $s_i \in \stateSetOf{}$ and actions $\alpha_i \in \actSetOf{}$ such that $s_0 = \initialOf{}$ and $(s_i,\alpha_i) \in \domain(\transFctOf{})$ for all $i \ge 0$. 
A \emph{finite path} of length $n \in \nats$ is a prefix $\finpath = s_0,\alpha_0,\ldots,s_n$ of an infinite path, ending in a state $\last{\finpath} = s_n$. 
For a (finite or infinite) path $\finpath$ we write $\vert\finpath\vert \in \nats \cup \{\infty\}$ for its length and $\pi[0,j]$ for its prefix of length $j \le \vert\finpath\vert$. 
We write $\infPathsOf{\rpa}$ and $\finPathsOf{\rpa}$ for the sets of infinite and finite paths of $\rpa$, respectively. 

As for PAs, strategies 
resolve nondeterminism over actions by assigning (sub)distributions to the enabled actions as a function of the history. 
We allow partial strategies that, intuitively, may assign probability~$0$ to all enabled actions, modelling that no further transition is taken. 

The key difference to PAs is that in rPAs the probabilistic choice between successor states is not fixed by the transition function: once a strategy has selected an action, a separate (non-probabilistic) nature player chooses a concrete successor distribution from the corresponding uncertainty set. 
\begin{defi}
	A  \emph{nature} for $\rpa$ is a function $\nature \colon \finPathsOf{\rpa}{}\times \actSetOf{} \to \dist{\stateSetOf{}}$ such that $\nature(\ppath, \alpha) \in \transFctOf{}(\last{\ppath}, \alpha)$. 
    The set of all natures on $\rpa$ is denoted by $\natureSetOf{\rpa}{}$.
    
    A \emph{memoryless nature} only depends on the last state of the path, i.e., $\nature(\ppath,\alpha) = \nature(\ppath',\alpha)$ whenever $\last{\ppath} = \last{\ppath'}$.  
    The set of memoryless natures on $\rpa$ is denoted by $\natureSetOf{\rpa}{\mless}$.
\end{defi}
We simply write $\natureSetOf{}{}$ or $\natureSetOf{}{\mless}$ when the rPA $\rpa$ is clear from the context. 
\begin{rem}\label{rem:nature_non_probabilistic}
  In our setting, nature is \emph{non-probabilistic}: for each $(\ppath,\alpha)$ it deterministically
  selects a single distribution from $\transFctOf{}(\last{\ppath},\alpha)$. 
    Probabilistic nature formalisations do not gain expressiveness:
  Allowing nature to choose a distribution from an uncertainty set according to another probability distribution is equivalent to replacing all uncertainty sets with their convex hull.
\end{rem}

\begin{rem}\label{rem:nature_semantics_variants}
 In the literature, a memoryless nature that fixes one distribution per state–action pair is referred to as \emph{once-and-for-all} semantics~\cite{Su+25,Bac+21}. 
 Further variants are possible, such as (dynamic) Markovian natures \cite{Badings25Thesis,Su+25} that additionally depend on the length of the path, or more general semantics that allow branching and merging of states~\cite{Bac+21}. 
 We focus on the stepwise setting with memoryless or memory-full (history-dependent) nature and leave other variants to future work. 
\end{rem}

We lift various established notions from (p)PAs to rPAs.
Strategies for rPAs are defined analogously to strategies for (p)PAs~(\Cref{defi_strategy_ppa}). 
We write $\strategysetOf{\rpa}{\prt}$ and $\strategysetOf{\rpa}{\comp}$ for the sets of partial and complete strategies on $\rpa$, respectively, and use $\strategysetOf{\rpa}{\star}$ when $\star \in \{\prt,\comp\}$. 
The \emph{(sub-)probability measure} $\PrOf{\rpa}{\nature,\strategy}{}$ for a fixed nature $\nature$ and a strategy $\strategy$ for $\rpa$ on the measurable subsets of $\infPathsOf{\rpa}{}$ is obtained by a standard cylinder-set construction~\cite{BK08}.
For a finite path $\finpath = s_0,\alpha_0,\ldots,s_n$ of $\rpa$, we set 
	\[
	\PrOf{\rpa}{\nature,\strategy}{\cyl (\finpath)}
	~=~ 
	\iverson{s_0 = \initialOf{}} 
	\cdot \prod_{ i= 0}^{n-1} \strategy(\finpath[0,i], \alpha_i) \cdot \nature{}(\finpath[0,i], \alpha_{i})(s_{i+1}).
\] 
This definition extends uniquely to a (sub-)probability measure on all measurable subsets of $\infPathsOf{\rpa}{}$.
\begin{rem}
The possible uncountability of the uncertainty set does not cause measure-theoretic difficulties: by definition, all successor distributions $\nature(\finpath,\alpha)$ and all strategy choices $\strategy(\finpath)$ are discrete distributions and thus have at most countable support. 
\end{rem}
We lift $\PrOf{\rpa}{\nature,\strategy}{}$ to (sets of) finite paths and write, e.g., $\PrOf{\rpa}{\nature,\strategy}{\finpath}$ for $\finpath \in \finPathsOf{\rpa}$ or $\PrOf{\rpa}{\nature,\strategy}{\Pi}$ for $\Pi \subseteq \finPathsOf{\rpa}$, implicitly referring to (unions of) cylinder sets.

Probabilistic objectives for rPAs and their satisfaction are defined analogously to \Cref{defi:objectives}. 
Satisfaction of an objective $\probPredicate{\sim p}{\regLang}$ with ${\sim} \in \{>, \geq, <, \leq\}$, $p \in [0,1]$, and $r \in \reals_{\geq 0}$, $\regLang \subseteq \alphabetOf{}^\infty$, by rPA $\rpa$ nature $\nature$ and strategy $\strategy$ is defined as
	\[
		\rpa, \nature, \strategy  \modelsWrt{} \probPredicate{\sim p}{\regLang}  \  \Leftrightarrow  \  \PrOf{\rpa}{\nature,\strategy}{\regLang} \sim p. 
	\]
We distinguish two nature semantics: \emph{memory-full} and \emph{memoryless} nature. 
\begin{defi}
\label{def:satrPA} 
	Let $\star \in \{\prt, \comp\} $ and let $\generalPredicate$ be an objective. 
	The \emph{satisfaction relation} $\modelsWrt{\star}_{}$ for rPA $\rpa$ is given by: 
	\[
		\rpa \modelsWrt{\star} \generalPredicate{} \quad \Leftrightarrow \quad \forall \nature \in \natureSetOf{\rpa}{}:  \forall \strategy \in \strategysetOf{\rpa}{\star}: \rpa, \nature, \strategy \modelsWrt{} \generalPredicate{}. 
	\]
   Satisfaction under memoryless nature---denoted $\modelsWrt{\star}_{\mless}$---is defined analogously w.r.t.\ $\natureSetOf{}{\mless}$. 
Satisfaction of multi-objective properties is defined as for pPAs (\Cref{defi_mo_sat}).
\end{defi}

We denote AG triples for rPAs as 
\begin{align*}
\agTriple{\rpa}{\star}{\multiobjectiveQuery{}{A}}{\multiobjectiveQuery{}{G}}
	\quad
	& \Leftrightarrow \quad 
    \left(\forall \nature \in \natureSetOf{\rpa}{} : \forall \strategy \in \strategysetOf{\rpa}{\star} : \quad  \rpa, \nature, {\strategy} \models \multiobjectiveQuery{}{A} \quad \rightarrow  \quad \rpa, \nature, {\strategy} \models \multiobjectiveQuery{}{G} \right),
\end{align*}
for $\star \in \{\comp,\prt\}$. 
AG triples for memoryless nature are denoted 
$\agTripleMless{\rpa}{\star}{\multiobjectiveQuery{}{A}}{\multiobjectiveQuery{}{G}}$
and are defined analogously, but with respect to the set of memoryless natures $\natureSetOf{\rpa}{\mless}$.

\begin{defi}[Alphabet Extension for rPAs]\label{def_alphabet_extension_rPA}
    Let $\rpa = \rpaTupleOf{}$ be a rPA over $\alphabetOf{\rpa}$ and let $\alphabetOf{}$ be an alphabet with $\actSetOf{} \cap (\alphabetOf{} \setminus \alphabetOf{\rpa}) = \emptyset$.
    The \emph{alphabet extension} of $\rpa$ with respect to $\alphabetOf{}$ is the rPA
    $\alphabetExtensionOfTo{\rpa}{\alphabetOf{}} = (\stateSetOf{}, \initialOf{}, \actSetOf{} \cupdot (\alphabetOf{} \setminus \alphabetOf{\rpa}), \transFctOf{\alphabetOf{}}, \syncOf{\alphabetOf{}})$ over alphabet $\alphabetOf{\rpa} \cup \alphabetOf{}$, where
    \begin{itemize}
    \item $\transFctOf{\alphabetOf{}}(s,\alpha) = \transFctOf{}(s,\alpha)$ and $\syncOf{\alphabetOf{}}(s,\alpha) = \syncOf{}(s,\alpha)$ for all $(s,\alpha) \in \domain(\transFctOf{})$ and
    \item $\transFctOf{\alphabetOf{}}(s,\lab) = \{\indicatorFct{s} \}$ and $\syncOf{\alphabetOf{}}(s,\lab) = \lab$ for all $s \in \stateSetOf{}$ and $\lab \in \alphabetOf{}\setminus\alphabetOf{\rpa}$.
    \end{itemize}
\end{defi}

Parallel composition for rPAs is defined analogously to the pPA composition in \Cref{def_ppa_composition}; the difference is that the transition function $\transFctOf{\parallel}$ is given by products of uncertainty sets rather than by products of single successor distributions. 
\begin{defi}[Parallel Composition for rPAs]
\label{def_rpa_composition}
	For $i=1,2$, let $\rpa_{i} = (\stateSetOf{i}, \initialOf{i}, \actSetOf{i}, \transFctOf{i}, \syncOf{i})$ be rPAs over alphabets $\alphabetOf{i}$ with $\actSetOf{i} \cap (\alphabetOf{1} \cup \alphabetOf{2}) = \emptyset$.
        The \emph{parallel composition} of $\rpa_1$ and $\rpa_2$ is given by the rPA 
        $\rpa_{1} \parallel \rpa_{2} = \big(\stateSetOf{1} \times \stateSetOf{2}, (\initialOf{1}, \initialOf{2}), \actSetOf{\parallel}, \transFctOf{\parallel}, \syncOf{\parallel}\big)$
         over $\alphabetOf{1} \cup \alphabetOf{2}$, where $\actSetOf{\parallel}$ and $\syncOf{\parallel}$ are defined as in \Cref{def_ppa_composition} and
        \begin{itemize}
            \item 
            for each $(s_1, \alpha_1) \in \domain(\transFctOf{1})$, $(s_2, \alpha_2) \in \domain(\transFctOf{2})$ with $\syncOf{1}(s_1, \alpha_1) = \syncOf{2}(s_2, \alpha_2) \in \alphabetOf{1} \cap \alphabetOf{2}$:
            \[
            \transFctOf{\parallel}((s_1,s_2), (\alpha_1, \alpha_2)) =  \{ \mu_1 \times \mu_2 \mid \mu_1 \in \transFctOf{1}(s_1,\alpha_1), \mu_2 \in \transFctOf{2}(s_2,\alpha_2) \}
            ,\]
            \item for each $(s_1, \alpha_1) \in \domain(\transFctOf{1})$, $s_2 \in \stateSetOf{2}$ with $\syncOf{1}(s_1, \alpha_1) = \lab_1 \in \alphabetOf{1} \setminus \alphabetOf{2}$: 
            \[ \transFctOf{\parallel}((s_1,s_2), (\alpha_1, \lab_1)) = \{\mu_1 \times \indicatorFct{s_2} \mid \mu_1 \in \transFctOf{1}(s_1,\alpha_1) \}
            ,\]
            \item for each $s_1 \in \stateSetOf{1}$, $(s_2, \alpha_2) \in \domain(\transFctOf{2})$ with $\syncOf{2}(s_2, \alpha_2) = \lab_2 \in \alphabetOf{2} \setminus \alphabetOf{1}$: 
            \[ \transFctOf{\parallel}((s_1,s_2), (\lab_2, \alpha_2)) = \{\indicatorFct{s_1} \times  \mu_2 \mid \mu_2 \in \transFctOf{2}(s_2,\alpha_2) \}
           .\]
        \end{itemize}
\end{defi}

\begin{figure}[t]
  \centering%
  \begin{tikzpicture}[mdp]
	\node[pswide, init=left] (s0t0)  {$s_0,t_0$};
    	
        \node[dist, right=0.5 of s0t0] (d0a) {};

	\node[pswide,above right=1.5 and 1 of s0t0] (s0t1)  {$s_0,t_1$};

    \node[pswide,above right=0.5 and 3.5 of s0t0] (s1t1)  {$s_1,t_1$};

    \node[pswide,below right=0.5 and 3.5 of s0t0] (s1t2)  {$s_1,t_2$};

    \node[pswide,below right=1.5 and 1 of s0t0] (s0t2)  {$s_0,t_2$};%

    \node[pswide, below right=0.5 and 3.5 of s1t1] (s1t0)  {$s_1,t_0$};
    \node[dist, left=0.5 of s1t0] (d10a) {};

	\path[ptrans]
	
		 (s0t0) edge[-] node[pos=0.8,above] {\tact{\lab}} (d0a)
		 (d0a) edge[bend left=20] node[pos=0.55,above,yshift=-0.0ex,xshift=0.0ex,rotate=70] {\tprob{p\cdot q}} (s0t1)
		(d0a) edge[bend left=20] node[pos=0.55,below,yshift=-0.0ex,xshift=-0.5ex,rotate=20] {\tprob{(1-p)\cdot q}} (s1t1)
        (d0a) edge[bend right=20] node[pos=0.55,below,yshift=0.0ex,xshift=-0.0ex,rotate=-70] {\tprob{p\cdot (1-q)}} (s0t2)
		(d0a) edge[bend right=20] node[pos=0.55,above,yshift=0.5ex,xshift=-0.5ex,rotate=-20] {\tprob{(1-p)\cdot (1-q)}} (s1t2)
        
        (s1t1.north east) edge[bend left=20] node[pos=0.5,above] {\tact{\altaltlab}} node[dist] (d11a) {} node[pos=0.75,left] {\tprob{}} (s1t0)

          (s0t1) edge[bend right=50] node[pos=0.5,above] {\tact{\altaltlab}} node[dist] (d01a) {} node[pos=0.75,left] {\tprob{}} (s0t0)
	

            (s1t0) edge[-] node[above,pos=1] {\tact{\lab}} (d10a)
            (d10a) edge[bend left=20] node[pos=0.5,above,yshift=0.0ex,xshift=-0.5ex,rotate=20] {\tprob{(1-q')}} (s1t2)
		      (d10a) edge[bend right=20] node[pos=0.5,below,yshift=-0.0ex,xshift=-0.5ex,rotate=0] {\tprob{q'}} (s1t1)

        	(s0t2) edge[loop above] node[pos=0.7, right] {\tact{\frownie}} node[dist] (d1c) {} node[pos=0.75,below] {\tprob{}} node[pos=0.5,below] {} (s0t2)

            (s1t2) edge[loop below] node[pos=0.5,below right] {\tact{\altlab}, \tact{\frownie}} node[dist] (d1c) {} node[pos=0.75,below] {\tprob{}} node[pos=0.5,below] {} (s1t2)
   
	;

\end{tikzpicture}%
  \caption{Parallel composition $\rpa_1 \parallel \rpa_2$ of the rPAs $\rpa_1$ and $\rpa_2$ from \Cref{fig:ipas}. 
  The uncertainty sets in $\rpa_1 \parallel \rpa_2$ are determined by the uncertainty sets of the individual components: $p \in [0,\frac{1}{2}]$ (from $\rpa_1$) and $q,q' \in [\frac{1}{10},\frac{9}{10}]$ (from $\rpa_2$).
  }
  \label{fig:rpa_12}%
\end{figure}
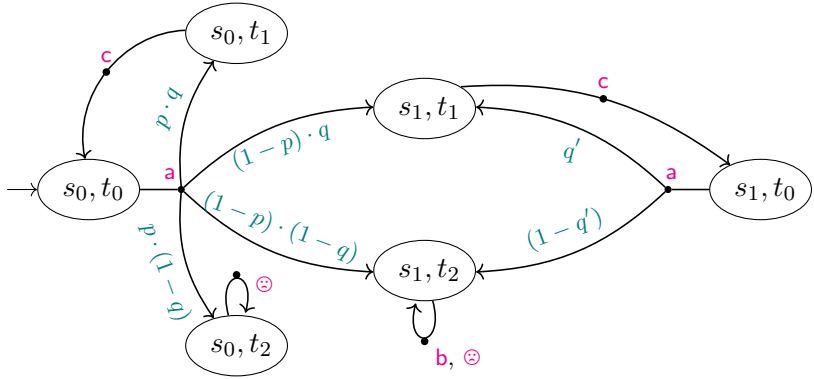
\begin{exa}
\label{ex:rPA_comp}
  The parallel composition of the rPAs $\rpa_1$ and $\rpa_2$ from \Cref{fig:ipas}\,
  is the rPA $\rpa_1 \parallel \rpa_2$ depicted in \Cref{fig:rpa_12}. 

  The uncertainty sets of $\rpa_1 \parallel \rpa_2$ are determined by the same interval bounds as in the components: for $\rpa_1$, the distributions in $\transFctOf{1}(s_0,\lab)$
  are of the form $\mu_1(s_0) = p,$ $\mu_1(s_1) = 1-p$ for some $p \in [0,\frac{1}{2}], $ and for $\rpa_2$, the distributions in $\transFctOf{2}(t_0,\lab)$ are of the form $\mu_2(t_1) = q, \mu_2(t_2) = 1-q$ for some $q \in [\frac{1}{10},\frac{9}{10}].$ 
    In the composition, each uncertainty set $\transFctOf{\parallel}((s_1,s_2),\cdot)$ is obtained by taking products $\mu_1 \times \mu_2$ of such component distributions. 
  Thus, the uncertainty sets of $\rpa_1 \parallel \rpa_2$ are  described by the variables $p \in [0,\frac{1}{2}]$ and $q,q' \in [\frac{1}{10},\frac{9}{10}]$, inherited from $\rpa_1$ and $\rpa_2$. 
When assuming memoryless nature, the rPA is equivalent to a corresponding pPA over parameters $p$, $q$ and $q'$ and an appropriate region. 
In our compositional setting, however, this correspondence is of limited use: memoryless nature may choose different distributions at different state–action pairs in the composition that share the same state of one component (e.g., at $(s_0,t_0)$ and $(s_1,t_0)$), whereas in a pPA a valuation assigns a single global value to each parameter.
\end{exa}
The following example shows that the parallel composition of polytopic (or convex) rPAs is, in general, not polytopic (or convex). 
This has also been observed, e.g., by Schnitzer et al.~\cite{SchnitzerAP26}. 
\begin{exa}\label{ex:non_convex_comp}
    Consider again the polytopic (convex) rPAs $\rpa_1$ and $\rpa_2$ from~\Cref{fig:ipas} and their parallel composition $\rpa_1 \parallel \rpa_2$ from~\Cref{fig:rpa_12}. 
    We show that the uncertainty set $\transFctOf{\parallel}((s_0,t_0),(\lab,\lab))$ is not convex, i.e., 
    $\transFctOf{\parallel}((s_0,t_0),(\lab,\lab)) \subsetneq \convh{\transFctOf{\parallel}((s_0,t_0),(\lab,\lab))}$:
  \begin{itemize}
      \item 
    In $\rpa_1$, the uncertainty set $\transFctOf{1}(s_0,\lab)$ contains distributions $\mu_1,\mu_1' \in \dist{\{s_0,s_1\}}$ with $ \mu_1(s_0) = 0, \mu_1(s_1) = 1,$ and $ \mu_1'(s_0) = \frac{1}{2}$, $\mu_1'(s_1) = \frac{1}{2}.$
    \item
    In $\rpa_2$, the uncertainty set $\transFctOf{2}(t_0,\lab)$ contains distributions
  $\mu_2,\mu_2' \in \dist{\{t_1,t_2\}}$ with $\mu_2(t_1) = \frac{1}{10},  \mu_2(t_2) = \frac{9}{10},$ and $ \mu_2'(t_1) = \frac{9}{10},$ $\mu_2'(t_2) = \frac{1}{10}.$
\item
  Hence the product distributions $\mu_{12}  = \mu_1 \times \mu_2$ and $ \mu_{12}' = \mu_1' \times \mu_2'$ both belong to $\transFctOf{\parallel}((s_0,t_0),(\lab,\lab))$.
\item
  However, the convex combination $\mu_{\mathrm{conv}} = \frac{1}{4}\cdot \mu_{12} + \frac{3}{4}\cdot \mu_{12}' \in \convh{\transFctOf{\parallel}((s_0,t_0),(\lab,\lab))}$ is not contained in $\transFctOf{\parallel}((s_0,t_0),(\lab,\lab))$ and thus the parallel composition is not polytopic (convex). 
  
    To see why $\mu_{\mathrm{conv}} \notin \transFctOf{\parallel}((s_0,t_0),(\lab,\lab))$, assume for contraposition that there exist $\hat{\mu}_1 \in \transFctOf{1}(s_0,\lab)$ and $\hat{\mu}_2 \in \transFctOf{2}(t_0,\lab)$ such that $\mu_{\mathrm{conv}} = \hat{\mu}_1 \times \hat{\mu}_2.$ 
    By the probabilities assigned by $\mu_{\mathrm{conv}}$, we obtain:
        \begin{align*}
        \hat{\mu}_1(s_0)  \cdot  \hat{\mu}_2(t_1)  &= \mu_{conv}((s_0,t_1)) =  \frac{1}{4}\cdot \underbrace{\mu_{12}(s_0,t_1)}_{=0 \cdot \frac{1}{10}} + \frac{3}{4}\cdot \underbrace{\mu_{12}'(s_0,t_1)}_{= \frac{1}{2} \cdot \frac{9}{10}} = \frac{27}{80} \\
         \hat{\mu}_1(s_0)  \cdot  \hat{\mu}_2(t_2)  &= \mu_{conv}((s_0,t_2)) =  \frac{1}{4}\cdot \underbrace{\mu_{12}(s_0,t_2)}_{=0 \cdot \frac{9}{10}} + \frac{3}{4}\cdot \underbrace{\mu_{12}'(s_0,t_2)}_{= \frac{1}{2} \cdot \frac{1}{10}} =\frac{3}{80} 
        \end{align*}
    Both distributions are uniquely determined by
    $\hat{\mu}_1(s_0)$ and $\hat{\mu}_2(t_1)$, with
    $\hat{\mu}_1(s_1) = 1 - \hat{\mu}_1(s_0)$, and $\hat{\mu}_2(t_2) = 1 - \hat{\mu}_2(t_1).$ 
    Together, this gives $\hat{\mu}_1(s_0) = \frac{3}{8}$ 
    and $\hat{\mu}_2(t_1) = \frac{9}{10}$. 
    However, this contradicts 
    \[
    \hat{\mu}_1(s_1) \cdot \hat{\mu}_2(t_1) = \frac{5}{8} \cdot \frac{9}{10} = \frac{9}{16} \neq \mu_{conv}((s_1,t_1)) =  \frac{1}{4}\cdot \underbrace{\mu_{12}(s_1,t_1)}_{=1 \cdot \frac{1}{10}} + \frac{3}{4}\cdot \underbrace{\mu_{12}'(s_1,t_1)}_{= \frac{1}{2} \cdot \frac{9}{10}} =\frac{29}{80}.
    \]
   
\end{itemize}

\end{exa}

\subsection{Limitations for general rPAs}
\label{sec:agr_rpa_discussion}
We investigate assume-guarantee reasoning for rPAs under parallel composition~$\parallel$ as in \Cref{def_rpa_composition}.
Our aim is to investigate whether the assume–guarantee framework of Kwiatkowska et al.~\cite{Kwi+13} remains sound in the robust setting. 
We focus on the asymmetric AG rule for (multi-objective) safety queries from
\cite[Theorem~1]{Kwi+13}:
                 \begin{align}\label{eq:asym_goal_rule_rpa}
                 \infer{
                     \rpa_1 \parallel \rpa_2 \modelsWrt{\comp} \multiobjectiveQuery{}{G}
                 } 
                 {
                     \deduce{
                         \agTriple{\alphabetExtensionOfTo{\rpa_2}{\alphabetOf{\multiobjectiveQuery{}{A}}}}{\prt}{\multiobjectiveQuery{}{A}}{\multiobjectiveQuery{}{G}} 
                     }{
                                \deduce{
                                    \rpa_1 \modelsWrt{\comp} \multiobjectiveQuery{}{A}
                                    }{}
                    }
                 }
                 \end{align}
where $\multiobjectiveQuery{}{A}$ and $\multiobjectiveQuery{}{G}$ are safety mo-queries with $\alphabetOf{\multiobjectiveQuery{}{A}} \subseteq \alphabetOf{1}$ and $\alphabetOf{\multiobjectiveQuery{}{G}} \subseteq \alphabetOf{2} \cup \alphabetOf{\multiobjectiveQuery{}{A}}$.

We show that the above rule is not sound for general rPAs.
For \emph{memoryless} nature ($\modelsWrt{}_{\mless}$), the rule already
        fails for polytopic (and thus convex) rPAs; see \Cref{sec:mless_nature_rpa}.
For \emph{memory-full} nature, the rule fails for \emph{non-convex} rPAs;
        see \Cref{sec:non_convex_rpa}.
For the remaining cases---i.e., convex (including polytopic) rPAs with memory-full nature---we can recover AG reasoning proof rules as outlined in \Cref{sec:AG_for_rpa_conv}.

\subsubsection{Memoryless nature}
\label{sec:mless_nature_rpa}
We show that the asymmetric AG proof rule from \eqref{eq:asym_goal_rule_rpa} is not sound when nature is restricted to be memoryless, i.e., using $\modelsWrt{}_{\mless}$, even for polytopic (convex) rPAs. 
\begin{exa}
\label{exa:mless_not_ok}
  Consider the polytopic rPAs $\rpa_1$ and $\rpa_2$ from \Cref{fig:ipas}; both are iPAs (and thus polytopic rPAs). 
  Let $\rpa_1 \parallel \rpa_2$ be their parallel
  composition as in \Cref{fig:rpa_12} (which is not a polytopic rPA; see~\Cref{ex:non_convex_comp}).

  We instantiate the asymmetric AG rule with the following safety mo-queries:
  $
\multiobjectiveQuery{}{A} = \probPredicate{\geq 0}{\regLang_A},$ where $\regLang_A = \{\lab, \altlab\}^\infty$, i.e., $\multiobjectiveQuery{}{A}$ is trivially satisfied by $\rpa_1$, and $\multiobjectiveQuery{}{G} = \probPredicate{\geq 0.25}{\regLang_G}, $ where 
$\regLang_G
      = \{w \in \{\lab,\altlab,\altaltlab, \frownie \}^\infty
                  \mid w \text { does not start with } \lab \altaltlab \lab \frownie  
         \}.$

    Under memoryless-nature semantics, we have:
      \begin{itemize}
    \item $\rpa_1 \modelsWrt{\comp}_{\mless} \multiobjectiveQuery{}{A}$, since $\multiobjectiveQuery{}{A}$ is
          trivially true, and
    \item
$\agTripleMless{\alphabetExtensionOfTo{\rpa_2}{\alphabetOf{\multiobjectiveQuery{}{A}}}}{\prt}{\multiobjectiveQuery{}{A}}{\multiobjectiveQuery{}{G}}$, since in $\rpa_2,$ the maximal probability---among all memoryless natures and partial strategies---of starting with the prefix $\lab \altaltlab \lab \frownie$ is $0.25$, which implies that $\multiobjectiveQuery{}{G}$ is satisfied.
      In particular, for any memoryless nature $\nature$ and (partial) strategy $\strategy$ of $\rpa_2 = \alphabetExtensionOfTo{\rpa_2}{\alphabetOf{\multiobjectiveQuery{}{A}}}$, if $\multiobjectiveQuery{}{A}$ holds under $\nature$ and $\strategy$, then so does $\multiobjectiveQuery{}{G}$.
  \end{itemize}
  Thus both premises of the asymmetric rule \eqref{eq:asym_goal_rule_rpa} are satisfied under $\modelsWrt{}_{\mless}$. 
   However, there exists a complete strategy $\strategy$ and a memoryless nature $\nature$ such that $\rpa_1 \parallel \rpa_2, \nature, \strategy \not\modelsWrt{} \multiobjectiveQuery{}{G}$:
  \begin{itemize}
    \item $\strategy$ selects with probability 1
    the action with label $\frownie $ in $(s_1,t_2)$ and the unique available action in the remaining states.
    \item $\nature$ chooses in $(s_0,t_0)$ after action with label $\lab$ a distribution $\mu$ with $\mu((s_1,t_1)) = \frac{9}{10}$, 
    $\mu((s_1,t_2)) = \frac{1}{10}$, and $\mu((s_0,t_i)) = 0$, for $i \in {1,2}$. 
    In $(s_1,t_0)$ after seeing label $\lab$, nature chooses a distribution $\mu'$ with $\mu'((s_1,t_2)) = \frac{9}{10}$ and $\mu'((s_1,t_1)) = \frac{1}{10}$. 
  \end{itemize}
  For this strategy and nature, the probability of violating $\regLang_G$ (i.e., starting with the bad prefix $\lab \altaltlab \lab \frownie$) is $\frac{81}{100}$.  
  Hence, the probability of satisfying $\regLang_G$ is strictly less than $0.25$, i.e., $\rpa_1 \parallel \rpa_2 \not\modelsWrt{\comp}_{\mless} \multiobjectiveQuery{}{G}.$ 
\end{exa}
Intuitively, the problem is that a memoryless nature in $\rpa_2$, must fix a single distribution for each state–action pair (in particular for $t_0$ in $\rpa_2$), 
whereas in the composition $\rpa_1 \parallel \rpa_2$ the state–action pairs $(s_0,t_0)$ and
$(s_1,t_0)$ are distinct.  
Nature may choose different distributions at these two composite states, which yields behaviours that cannot be captured by a single memoryless nature on the component $\rpa_2$. 

  To summarize, the asymmetric AG rule is not sound under memoryless nature, even when restricted to convex (or polytopic) rPAs. In fact, the above example also shows that the memoryless-nature AG rule fails for convex rPAs even when using the convexity-preserving composition $\parallelConv$ introduced in \Cref{def_rpa_conv_composition}.

\subsubsection{Non-convex rPAs}
\label{sec:non_convex_rpa}
We now consider the asymmetric AG rule \eqref{eq:asym_goal_rule_rpa} for rPAs $\rpa_1$ and $\rpa_2$ whose uncertainty sets are not assumed to be convex, under memory-full nature. 
In this setting, the reduction to PAs no longer preserves behaviour in general, and the PA-based AG framework does not carry over as illustrated by the following example. 
\begin{exa}
\label{exa:nonconvexNotOk}
    Reconsider the rPAs $\rpa_1$ and $\rpa_2$ depicted in \Cref{fig:ipas}. 
    We define $\rpa_1'$ and $\rpa_2'$ by modifying the transition uncertainty functions as follows: 
    $\transFctOf{\rpa_1'}(s_0, \lab) = \{\mu_1 \colon s_0,s_1 \mapsto \frac{1}{2} \}$
    and 
    $\transFctOf{\rpa_2'}(t_0, \lab) = \{\mu_{2,1} \colon t_1 \mapsto \frac{9}{10}, t_2 \mapsto \frac{1}{10}\,,~\mu_{2,2} \colon t_1 = \frac{1}{10}, t_2 = \frac{9}{10} \}$.
    The remaining transitions are as in \Cref{fig:ipas}. 
    $\rpa_2'$ is not convex.
    The parallel composition $\rpa_1' \parallel \rpa_2'$ is the rPA in \Cref{fig:rpa_12}, where the uncertainty sets in $\rpa_1' \parallel \rpa_2'$ are now determined by $p = \frac{1}{2}$ (from $\rpa_1'$) and $q,q' \in \{\frac{1}{10},\frac{9}{10}\}$ (from $\rpa_2'$). 

  We instantiate the asymmetric AG rule with the following safety mo-queries:
  $ \multiobjectiveQuery{}{A} = \probPredicate{\geq 0.4}{\regLang},$ and $\multiobjectiveQuery{}{G} = \probPredicate{\geq 0.8}{\regLang}, $ where  $\regLang
      = \{w \in \{\lab,\altlab \}^\infty
                  \mid w \text { does not start with } \lab \altlab 
         \}.$
    We have:
      \begin{itemize}
    \item $\rpa_1' \modelsWrt{\comp} \multiobjectiveQuery{}{A}$,
    since under the unique strategy and nature of $\rpa_1'$ the prefix $\lab\altlab$ is seen with probability 0.5,
    and 
    \item $\agTripleMless{\alphabetExtensionOfTo{\rpa_2'}{\alphabetOf{\multiobjectiveQuery{}{A}}}}{\prt}{\multiobjectiveQuery{}{A}}{\multiobjectiveQuery{}{G}}$, since, in order for the assumption $\multiobjectiveQuery{}{A}$ to hold in $\rpa_2'$, any nature must resolve the uncertainty at $t_0$ after taking the action labeled $\lab$ by choosing the distribution with $\mu(t_1) = \frac{9}{10}$ (rather than $\frac{1}{10}$). 
       Under this choice, the probability of producing the prefix $\lab\altlab$ is at most $0.1$, so $\multiobjectiveQuery{}{G}$ is satisfied. 
  \end{itemize}
Both premises of the asymmetric rule are satisfied. 
However, $\rpa_1' \parallel \rpa_2', \nature, \strategy \not\modelsWrt{} \multiobjectiveQuery{}{G}$, for the following complete strategy $\strategy$ and nature $\nature$:
  \begin{itemize}
     \item strategy $\strategy$ selects the action with label $\lab$ in state $(s_0,t_0)$, and the action with label $\altlab$ in $(s_1,t_2)$, all with probability 1. 
     \item nature $\nature$ chooses a distribution $\mu$ with $\mu((s_1,t_2)) = \frac{9}{20}$ in $(s_0,t_0)$ after seeing label $\lab$. 
   \end{itemize}
   For this strategy $\strategy$ and nature $\nature$, the probability that the trace starts with $\lab \altlab$ is $\frac{9}{20} = 0.45$, i.e., $\strategy$ and $\nature$ witness violation of  $\multiobjectiveQuery{}{G} = \probPredicate{\geq 0.8}{\regLang}$. 

Note that for the convex rPA $\rpa_2$ from \Cref{fig:ipa_2}, the second premise fails: 
$\agTriple{\alphabetExtensionOfTo{\rpa_2}{\alphabetOf{\multiobjectiveQuery{}{A}}}}{\prt}{\multiobjectiveQuery{}{A}}{\multiobjectiveQuery{}{G}}$ does not hold. 
For example, in state $t_0$, after the strategy has chosen the action labeled $\lab$, nature can select a distribution $\mu$ with $\mu(t_1) = \frac{7}{10}$ and $\mu(t_2) = \frac{3}{10}$. 
Under this choice, the assumption $\multiobjectiveQuery{}{A}$ is satisfied, but the guarantee $\multiobjectiveQuery{}{G}$ with threshold $0.8$ is violated.
\end{exa}
For PAs, the proof of the asymmetric AG rule in~\cite{Kwi+13} relies on the strategies’ ability to randomise over successor distributions in order to define strategy projections (see \Cref{def:projectionStrategyPA}).
In the rPA setting, nature is non-probabilistic.
For non-convex uncertainty sets $\transFctOf{}(s,\alpha)$, it might not be possible to adequatly project all natures of a composed rPA to natures of its individual components---preventing compositional reasoning for general rPAs.

\subsection{Compositionality of convex rPAs with memory-full nature}\label{sec:AG_for_rpa_conv}

For rPAs with \emph{convex} uncertainty sets, probabilistic extensions of the nature formalism are as expressive as non-probabilistic variants: any distribution $\tilde\mu \in \dist{\transFctOf{}(s,\alpha)}$ over distributions from an uncertainty set $\transFctOf{}(s,\alpha)$  can be flattened into a distribution $\mu \in \convh{\transFctOf{}(s,\alpha)} = \transFctOf{}(s,\alpha)$.
Intuitively, once probabilistic and memory-full natures can be assumed, \emph{the conceptual difference between (r)PA strategies and rPA natures vanishes.}


We define a \emph{PA-reduction} that transforms a convex rPA under memory-full nature to an equivalent (possibly uncountably branching) PA by lifting nature choices to strategy choices.
The construction is analogous to a standard reduction from polytopic rMDPs~\cite{NG05,Wi+13,Iye05,Badings25Thesis} to stochastic games.
Let the \emph{generator} of a convex uncertainty set $\transFctOf{}(s,\alpha)$ be given by
\[
\gen{\transFctOf{}(s,\alpha)} ~=~
\begin{cases}
\extr{\transFctOf{}(s,\alpha)} & \text{if } \transFctOf{}(s,\alpha) = \convh{\extr{\transFctOf{}(s,\alpha)}}\\
\transFctOf{}(s,\alpha) & \text{otherwise.}
\end{cases}
\]
We have $\convh{\gen{\transFctOf{}(s,\alpha)}} = \transFctOf{}(s,\alpha)$ for any convex $\transFctOf{}(s,\alpha)$.
The set $\gen{\transFctOf{}(s,\alpha)}$ is finite iff $\transFctOf{}(s,\alpha)$ is polytopic.
If $\transFctOf{}(s,\alpha)$ is an open set, we have $\transFctOf{}(s,\alpha) \neq \convh{\extr{\transFctOf{}(s,\alpha)}}$ and $\gen{\transFctOf{}(s,\alpha)} = \transFctOf{}(s,\alpha)$ is uncountable infinite.
\begin{defi}[PA-Reduction for rPAs]
\label{defi_PA_reduction_rpa_paths}
Let $\rpa = \rpaTupleOf{}{}$ be a convex rPA.
The \emph{PA-reduction} of a convex rPA $\rpa  = \rpaTupleOf{}{}$ is the PA $\paReductionOf{\rpa} = \left(\stateSetOf{}, \initialOf{}, \actSetOf{\paReductionOf{\rpa}}, \transFctOf{\paReductionOf{\rpa}},   \syncOf{\paReductionOf{\rpa}} \right)$, where 
$\actSetOf{\paReductionOf{\rpa}} =  \{ (\alpha, \mu) \mid \exists (s, \alpha) \in \domain(\transFctOf{}): \mu \in \gen{\transFctOf{}(s, \alpha)} \} $ is the (possibly uncountable) set of actions and for each $(s,\alpha) \in \domain(\transFctOf{})$ and each $\mu \in \gen{\transFctOf{}(s,\alpha)}$, we set
$\transFctOf{\paReductionOf{\rpa}}(s, (\alpha, \mu)) = \mu$ and
$\syncOf{\paReductionOf{\rpa}}(s,(\alpha,\mu)) = \syncOf{}(s,\alpha).
$ 
\end{defi}
$\paReductionOf{\rpa}$ is a finite-state, finite-action PA iff $\rpa$ is polytopic.
There is a correspondence between strategy-nature pairs $(\strategy,\nature)$ of rPA $\rpa$ and strategies $\widehat\strategy$ of $\paReductionOf{\rpa}$ yielding a semantic equivalence: $\rpa$ (under memor-full nature) and $\paReductionOf{\rpa}$ satisfy the same objectives, multi-objective properties, and AG triples.
Further details in \Cref{ap:proof_for_AG_for_rPA}.

For the remainder of this section, we fix two convex rPAs $\rpa_1$ and $\rpa_2$ with alphabets $\alphabetOf{1}$ and $\alphabetOf{2}$, respectively.
Our observations so far enable a reduction to AG reasoning for (standard) PAs.
For example, when combining the asymmetric AG rule of~\cite{Kwi+13} with the above-mentioned equivalences of $\rpa_i$ and $\paReductionOf{\rpa_i}$ for $i=1,2$, we obtain the rule:
 \begin{align*}\label{eq:asym_goal_rule_rpa}
 \infer{
     \paReductionOf{\rpa_1} \parallel \paReductionOf{\rpa_2} \modelsWrt{\comp} \multiobjectiveQuery{}{G}
 } 
 {
     \deduce{
         \agTriple{\alphabetExtensionOfTo{\rpa_2}{\alphabetOf{\multiobjectiveQuery{}{A}}}}{\prt}{\multiobjectiveQuery{}{A}}{\multiobjectiveQuery{}{G}} 
     }{
                \deduce{
                    \rpa_1 \modelsWrt{\comp} \multiobjectiveQuery{}{A}
                    }{}
    }
 }
 \end{align*}
To reason about the composed rPA $\rpa_1 \parallel \rpa_2$, it remains to relate it to the PA $\paReductionOf{\rpa_1} \parallel \paReductionOf{\rpa_2}$.
However, the two composed models are not equivalent in general:
the randomised strategy choices of the PA $\paReductionOf{\rpa_1} \parallel \paReductionOf{\rpa_2}$ are inherently convex whereas parallel composition of rPAs does not preserve convexity (cf. \Cref{ex:non_convex_comp}).
The latter further prevents compositional reasoning with more than two convex rPAs.
We therefore introduce \emph{convex parallel composition}~$\parallelConv$. 
Given $\rpa_1 \parallel \rpa_2$, the rPA $\rpa_1 \parallelConv \rpa_2$ is obtained by replacing, for each state–action pair, the uncertainty set of the composition by its convex hull.
\begin{defi}[Convex Parallel Composition for rPAs]
\label{def_rpa_conv_composition}
Given rPAs $\rpa_1$ and $\rpa_2$, let $\rpa_1 \parallel \rpa_2  = (\stateSetOf{1} \times \stateSetOf{2}, (\initialOf{1}, \initialOf{2}), \actSetOf{\parallel}, \transFctOf{\parallel}, \syncOf{\parallel})$ be as in \Cref{def_rpa_composition}.
  The \emph{convex parallel composition} of $\rpa_1$ and $\rpa_2$ is the rPA
  $\rpa_1 \parallelConv \rpa_2
      = (\stateSetOf{1} \times \stateSetOf{2}, (\initialOf{1}, \initialOf{2}),
         \actSetOf{\parallel}, \transFctOf{\parallelConv}, \syncOf{\parallel})$,
  where $\transFctOf{\parallelConv}(s,\alpha) = \convh{\transFctOf{\parallel}(s,\alpha)} $ for all $(s,\alpha) \in \domain(\transFctOf{\parallel})$.
\end{defi}
Convex parallel composition is associative, in the sense that
$(\rpa_1 \parallelConv \rpa_2) \parallelConv \rpa_3$ and
$\rpa_1 \parallelConv (\rpa_2 \parallelConv \rpa_3)$ are equivalent up to renaming of states.
We therefore write $\rpa_1 \parallelConv \rpa_2 \parallelConv \rpa_3$ without parentheses.
Every uncertainty set of $\rpa_1 \parallel \rpa_2$ is a subset of an uncertainty set of $\rpa_1 \parallelConv \rpa_2$. 
Hence, $\parallelConv$ over-approximates the standard composition $\parallel$ and any (multi-objective) property/AG triple satisfied by $\rpa_1 \parallelConv \ldots \parallelConv \rpa_n$ immediately transfers to the
standard parallel composition $\rpa_1 \parallel \ldots \parallel \rpa_n$. 
The following lemma shows that the PA reduction commutes with the convex parallel composition w.r.t.\ satisfaction of multi-objective queries. 
\begin{restatable}{lem}{restatablePAConvConvPA}\label{thm:pa_conv_eq_conv_pa}
Let $\star \in \{\prt,\comp\}$ and let $\rpa_1$ and $\rpa_2$ be convex rPAs. 
Then, for any mo-query $\multiobjectiveQuery{}{X}$ over $\alphabetOf{1} \cup \alphabetOf{2}$:  
\[ 
 \paReductionOf{\rpa_1} \parallel \paReductionOf{\rpa_2} \modelsWrt{\star} \multiobjectiveQuery{}{X}  \quad  \Leftrightarrow  \quad   \paReductionOf{\rpa_1 \parallelConv \rpa_2} \modelsWrt{\star} \multiobjectiveQuery{}{X}
\]
\end{restatable}

We derive AG rules for convex rPAs under the convex parallel composition $\parallelConv$, focusing on safety multi-objective queries.\footnote{Extensions to fairness and more general quantitative objectives require a careful treatment of nature in order to preserve fairness properties under the reduction to PAs and are left for future work.}
Full proofs are given in \Cref{ap:proof_for_AG_for_rPA}. 
We lift the asymmetric AG rule from \Cref{theo:pag_asym_rule} (cf.~\cite[Theorem~1]{Kwi+13}) and the circular AG rule~\cite[Theorem~5]{Kwi+13} to convex rPAs. 
\begin{restatable}[Asymmetric and Circular Rule]{thm}{restatableAsymRuleRPA}\label{theo_rpa_asym_rule}
  Let $\multiobjectiveQuery{}{A}$ and $\multiobjectiveQuery{}{G}$ be safety
  mo-queries over alphabets
  $\alphabetOf{\multiobjectiveQuery{}{A}} \subseteq \alphabetOf{1}$ and
  $\alphabetOf{\multiobjectiveQuery{}{G}} \subseteq
    \alphabetOf{2} \cup \alphabetOf{\multiobjectiveQuery{}{A}}$, respectively. 
    Then the following proof rules hold:

    \medskip \hspace{30pt}
      \begin{tabularx}{\linewidth}{p{0.45\linewidth} p{0.45\linewidth}} 
             \raisebox{0.7ex}{$ 
                \infer{
      \rpa_1 \parallelConv \rpa_2 \modelsWrt{\comp} \multiobjectiveQuery{}{G}
    }{
      \deduce{
        \agTriple{\alphabetExtensionOfTo{\rpa_2}{\alphabetOf{\multiobjectiveQuery{}{A}}}}{\prt}
                 {\multiobjectiveQuery{}{A}}{\multiobjectiveQuery{}{G}}
      }{
        \rpa_1 \modelsWrt{\comp} \multiobjectiveQuery{}{A}
      }
    }$}
        & 
        \hspace{8pt}
         {$ 
               \infer{
      \rpa_1 \parallelConv \rpa_2 \modelsWrt{\comp} \multiobjectiveQuery{}{G}
    }{
      \deduce{
        \deduce{
          \rpa_2 \modelsWrt{\comp} \multiobjectiveQuery{}{A}_1
        }{
          \agTriple{\alphabetExtensionOfTo{\rpa_2}{\alphabetOf{\multiobjectiveQuery{}{A}_2}}}{\prt}
                   {\multiobjectiveQuery{}{A}_2}{\multiobjectiveQuery{}{G}}
        }
      }{
        \agTriple{\alphabetExtensionOfTo{\rpa_1}{\alphabetOf{\multiobjectiveQuery{}{A}_1}}}{\prt}
                 {\multiobjectiveQuery{}{A}_1}{\multiobjectiveQuery{}{A}_2}
      }
    }
                $}
    \end{tabularx}
    
\end{restatable}
\begin{proof}[Proof sketch]
Following the explanations above, the proof is based on a reduction to the PA setting of~\cite{Kwi+13}. 
We reduce convex rPAs to (possibly uncountable branching) PAs (see \Cref{defi_PA_reduction_rpa_paths}). 
Under this reduction, the premises for $\rpa_1$ and $\rpa_2$ induce the corresponding premises required by the asymmetric (circular) AG rule for PAs from~\cite{Kwi+13}, i.e., $\paReductionOf{\rpa_1}$ and $\paReductionOf{\rpa_2}$ (see \Cref{forallpaexistsrpa}), yielding  $\paReductionOf{\rpa_1} \parallel \paReductionOf{\rpa_2} \modelsWrt{\comp}  \multiobjectiveQuery{}{G}$. 
By \Cref{thm:pa_conv_eq_conv_pa} follows that $\paReductionOf{\rpa_1 \parallelConv \rpa_2} \modelsWrt{\star} \multiobjectiveQuery{}{G}$. 
Then, we lift this back to convex rPAs (see \Cref{theo:forallrPAexistsPA}), obtaining $\rpa_1 \parallelConv \rpa_2 \modelsWrt{\comp} \multiobjectiveQuery{}{G}$. 
\end{proof}

As observed above, $\parallelConv$ over-approximates $\parallel$. 
Consequently, if $\rpa_1 \parallelConv \rpa_2 \modelsWrt{\comp} \multiobjectiveQuery{}{G}$,
then also $\rpa_1 \parallel \rpa_2 \modelsWrt{\comp} \multiobjectiveQuery{}{G}$.
The remaining AG rules from \Cref{theo:pag_asym_rule,defi_mo_sat} can be lifted to convex rPAs in a similar manner; see \Cref{ap:rPA_conv_rules} for details.

\begin{rem}
  The premises of the asymmetric and the circular rule can be verified directly on the PA-reduction.
  This is particularly useful for polytopic rPAs, where the PA-reduction is finite
  (see, e.g.,~\cite{NG05,Wi+13,Iye05,Su+25,Badings25Thesis}).
\end{rem}

\subsection{Interval PAs}
\label{sec:interval_arithm_relaxation}
\emph{Interval probabilistic automata} (iPAs) are a special case of convex rPAs, in which each uncertainty set is specified by interval bounds on the transition probabilities.
We lift the iMDP definition of Badings~\cite[Definition~3.27]{Badings25Thesis} to the
compositional setting: 
\begin{defi}
\label{defi_ipas} 
  An \emph{interval probabilistic automaton} (iPA) over finite alphabet $\alphabetOf{}$
  is an rPA $\ipa = \rpaTupleOf{}{}$ as in \Cref{def:rpa}, such that for all $(s,\alpha) \in \domain(\transFctOf{})$ there exist  $\ell^{(s,\alpha,s')}, u^{(s,\alpha,s')} \in [0,1] \text{ for each } s' \in \stateSetOf{}$ such that
\[
    \transFctOf{}(s,\alpha)  ~=~ \big\{ \mu \in \dist{\stateSetOf{}} ~\big| 
    \forall s' \in \stateSetOf{}:\ell^{(s,\alpha,s')} \leq \mu(s') \leq u^{(s,\alpha,s')} \big\}.
\]
\end{defi}
\begin{exa}
\label{ex:ipas}
    The rPAs in \Cref{fig:ipas}---described in \Cref{exa:rpas}---are iPAs. 
    Each transition annotated with an interval $[\ell,u]$ represents the uncertainty set of all distributions whose probability for the corresponding successor lies between $\ell$ and $u$. 
    Their parallel composition---depicted in \Cref{fig:ipa_12,ex:rPA_comp}---is not an iPA. 
\end{exa}
The class of iPAs is neither closed under $\parallel$ nor $\parallelConv$. 
The literature (e.g.,~\cite{Kwi+11}) often employs an interval-arithmetic relaxation $\parallelRel$ of the parallel composition. 
This relaxation preserves the iPA property---$\ipa_1 \parallelRel \ipa_2$ is again an iPA---but introduces spurious distributions~\cite{SchnitzerAP26}. 
We adopt this relaxation to our setting, and show by example that the AG framework of~\cite{Kwi+13}---precisely the asymmetric rule---cannot be lifted directly to this relaxed composition.

\begin{defi}[Interval-relaxation Composition]\label{def_ipa_composition}
Given iPAs $\ipa_1$ and $\ipa_2$, let $\ipa_1 \parallel \ipa_2  = (\stateSetOf{1} \times \stateSetOf{2}, (\initialOf{1}, \initialOf{2}), \actSetOf{\parallel}, \transFctOf{\parallel}, \syncOf{\parallel})$ be their parallel composition rPA as in \Cref{def_rpa_composition}.
  The \emph{interval-arithmetic parallel composition} of $\ipa_1$ and $\ipa_2$ is the iPA
  $\ipa_1 \parallelRel \ipa_2
      = (\stateSetOf{1} \times \stateSetOf{2}, (\initialOf{1}, \initialOf{2}),
         \actSetOf{\parallel}, \transFctOf{\parallelRel}, \syncOf{\parallel})$,
  where for all $(s,\alpha) \in \domain(\transFctOf{\parallel})$:
  \begin{align*}
   \transFctOf{\parallelRel}(s,\alpha) ~=~
   \Big\{
   \mu \in \dist{\stateSetOf{1} \times \stateSetOf{2}}
~\Big|&~
\forall s' = (s_1',s_2') \in \stateSetOf{1} \times \stateSetOf{2}\colon\\
&~\min_{\mu_\parallel \in \transFctOf{\parallel}(s,\alpha)}  \mu_\parallel(s') ~\le~ \mu(s) ~\le~ \max_{\mu_\parallel \in \transFctOf{\parallel}(s,\alpha)} \mu_\parallel(s')
   \Big\}.
  \end{align*}
\end{defi}

\begin{figure}[t] 
  \centering%
  \begin{tikzpicture}[mdp]
	\node[pswide, init=left] (s0t0)  {$s_0,t_0$};
    	
        \node[dist, right=0.5 of s0t0] (d0a) {};

	\node[pswide,above right=1.5 and 1 of s0t0] (s0t1)  {$s_0,t_1$};

    \node[pswide,above right=0.5 and 3.5 of s0t0] (s1t1)  {$s_1,t_1$};

    \node[pswide,below right=0.5 and 3.5 of s0t0] (s1t2)  {$s_1,t_2$};

    \node[pswide,below right=1.5 and 1 of s0t0] (s0t2)  {$s_0,t_2$};%

    \node[pswide, below right=0.5 and 3.5 of s1t1] (s1t0)  {$s_1,t_0$};
    \node[dist, left=0.5 of s1t0] (d10a) {};

	\path[ptrans]
	
		 (s0t0) edge[-] node[pos=0.8,above] {\tact{\lab}} (d0a)
		 (d0a) edge[bend left=20] node[pos=0.55,above,yshift=-0.0ex,xshift=0.0ex,rotate=70] {\tprob{[0,\frac{9}{20}]}} (s0t1)
		(d0a) edge[bend left=20] node[pos=0.55,below,yshift=-0.0ex,xshift=-0.5ex,rotate=20] {\tprob{[\frac{1}{20},\frac{9}{10}]}} (s1t1)
        (d0a) edge[bend right=20] node[pos=0.55,below,yshift=0.0ex,xshift=-0.0ex,rotate=-70] {\tprob{[0,\frac{9}{20}]}} (s0t2)
		(d0a) edge[bend right=20] node[pos=0.55,above,yshift=0.5ex,xshift=-0.5ex,rotate=-20] {\tprob{[\frac{1}{20},\frac{9}{10}]}} (s1t2)
        
        (s1t1.north east) edge[bend left=20] node[pos=0.5,above] {\tact{\altaltlab}} node[dist] (d11a) {} node[pos=0.75,left] {\tprob{}} (s1t0)

          (s0t1) edge[bend right=50] node[pos=0.5,above] {\tact{\altaltlab}} node[dist] (d01a) {} node[pos=0.75,left] {\tprob{}} (s0t0)
	

            (s1t0) edge[-] node[above,pos=1] {\tact{\lab}} (d10a)
            (d10a) edge[bend left=20] node[pos=0.5,above,yshift=0.0ex,xshift=-0.5ex,rotate=20] {\tprob{[\frac{1}{10},\frac{9}{10}]}} (s1t2)
		      (d10a) edge[bend right=20] node[pos=0.5,below,yshift=-0.0ex,xshift=-0.5ex,rotate=-20] {\tprob{[\frac{1}{10},\frac{9}{10}]}} (s1t1)

        	(s0t2) edge[loop above] node[pos=0.7, right] {\tact{\frownie}} node[dist] (d1c) {} node[pos=0.75,below] {\tprob{}} node[pos=0.5,below] {} (s0t2)

            (s1t2) edge[loop below] node[pos=0.5,below right] {\tact{\altlab}, \tact{\frownie}} node[dist] (d1c) {} node[pos=0.75,below] {\tprob{}} node[pos=0.5,below] {} (s1t2)
	;

\end{tikzpicture}
  \caption{Interval-arithmetic relaxation $\rpa_1 \parallelRel \rpa_2$ of the iPAs $\rpa_1$ and $\rpa_2$ from \Cref{fig:ipas}.}
  \label{fig:ipa_12}%
\end{figure}
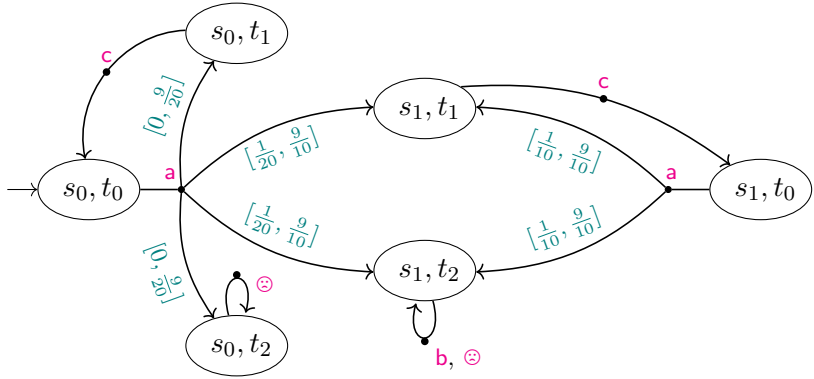

\begin{exa}
\label{ex:interval_comp}
  \Cref{fig:ipa_12} shows the interval-arithmetic relaxation composition 
  $\rpa_1 \parallelRel \rpa_2$ of the iPAs $\rpa_1$ and $\rpa_2$
  from \Cref{fig:ipas}, synchronising on their common labels $\lab$ and $\altlab$. 
\end{exa}

The asymmetric AG rule for safety (multi-objective) queries from \cite[Theorem~1]{Kwi+13}, instantiated in our setting for iPAs under $\parallelRel$, would read as follows:  
\begin{align}\label{eq:asym_rule_ipa}
\infer{
    \ipa_1 \parallelRel \ipa_2 \modelsWrt{\comp} \multiobjectiveQuery{}{G}
  }{
    \deduce{
      \agTriple{\alphabetExtensionOfTo{\ipa_2}{\alphabetOf{\multiobjectiveQuery{}{A}}}}{\prt}
               {\multiobjectiveQuery{}{A}}{\multiobjectiveQuery{}{G}}
    }{
      \ipa_1 \modelsWrt{\comp} \multiobjectiveQuery{}{A}
    }
  } 
\end{align}
This AG rule is not sound for iPAs under composition $\parallelRel$, as there exist counterexamples for both memoryless and memory-full nature, see \Cref{{ex:interval_non_AG}}. 
\begin{exa}\label{ex:interval_non_AG}
    For $\rpa_1$, $\rpa_2$, and $\rpa_1 \parallelRel \rpa_2$ as in \Cref{ex:interval_comp}, consider the safety multi-objective queries
    $\multiobjectiveQuery{}{A} = \probPredicate{\geq 0}{\regLang_A}$, for $\regLang_A = \{\lab, \altlab\}^\infty,$ (which is trivially satisfied) and
    $\multiobjectiveQuery{}{G} = \probPredicate{\geq 0.1}{\regLang_G}, $ for $
    \regLang_G = \{ w \in \{\lab,\altlab,\altaltlab, \frownie\}^\infty \mid w \text{ does not contain any } \altaltlab \}$. 
The premises of the AG rule are satisfied for $\rpa_1$ and $\rpa_2$:
  \begin{itemize}
    \item $\rpa_1 \modelsWrt{\comp} \multiobjectiveQuery{}{A}$, since $\regLang_A$ contains
          all words over $\{\lab,\altlab\}$,
    \item
      $\agTriple{\alphabetExtensionOfTo{\rpa_2}{\alphabetOf{\multiobjectiveQuery{}{A}}}}{\prt}{\multiobjectiveQuery{}{A}}{\multiobjectiveQuery{}{G}}$ holds, i.e., for 
      $\rpa_2 = \alphabetExtensionOfTo{\rpa_2}{\alphabetOf{\multiobjectiveQuery{}{A}}}$ any (partial) strategy and nature satisfying $\multiobjectiveQuery{}{A}$ also satisfies
      $\multiobjectiveQuery{}{G}$.
  \end{itemize}
    However, there exists a complete strategy $\strategy$ and a nature $\nature$ such that $\rpa_1 \parallelRel \rpa_2, \nature, \strategy \not\modelsWrt{} \multiobjectiveQuery{}{G}$: 
  \begin{itemize}
    \item $\strategy$ always selects the action with label $\lab$ in state $(s_0,t_0)$,
          and, upon reaching a state with second component $t_1$, always selects the
          action with label $\altaltlab$.
    \item (memoryless) nature $\nature$ chooses at $(s_0,t_0),\lab$ the distribution
          $\mu \colon 
          (s_0,t_1) \mapsto \frac{1}{20},
          (s_1,t_1) \mapsto \frac{9}{10},
          (s_1,t_2) \mapsto \frac{1}{20}
          $.
  \end{itemize}
  Then, the probability of eventually seeing the label $\altaltlab$ is $\frac{19}{20}$, so the probability of satisfying $\regLang_G$ (i.e., of never  seeing $\altaltlab$) is $\frac{1}{20} < 0.1$. 
  Hence, $\rpa_1 \parallelRel \rpa_2 \not \modelsWrt{\comp} \multiobjectiveQuery{}{G}.$
\end{exa}
    Intuitively, the interval relaxation $\parallelRel$ discards all dependencies between successor state probabilities which allows for uncertainty resolutions that can not be projected to the individual components.
    As a consequence, the AG framework of~\cite{Kwi+13} cannot be applied directly to iPAs under the relaxed composition~$\parallelRel$. 
    Since nature in the counterexample above is memoryless, the rule is also unsound under the memoryless semantics~$\modelsWrt{}_{\mless}$.

\section{Simulation-Based Assume-Guarantee Reasoning for pPAs}
\label{sec:AG_sim}
We develop a sound and complete AG proof rule for pPAs based on a \emph{simulation relation}, following the work of Komuravelli et al.~\cite{Kom+12} for PAs, which builds upon the simulation-based AG frameworks of~\cite{Bob+08,Pas+08}. 
While the quantitative AG framework of Kwiatkowska et al.~\cite{Kwi+13} considered so far focuses on multi-objective properties, the simulation-based approach is \emph{model-based}: assumptions and guarantees are given as PAs, and reasoning proceeds via a simulation relation between the automata rather than via property satisfaction. 

Concretely, the approach of~\cite{Kom+12} employs strong simulation~\cite{Seg+95} (denoted $\strSimulationRegion{}$) and establishes the following asymmetric AG rule for PAs:
\begin{align*}
\infer{
         \pa_1 \parallel \pa_2 \strSimulationRegion{} \pa_G
       }{
         \deduce{\pa_2 \parallel \pa_A \strSimulationRegion{} \pa_G}
                {\pa_1 \strSimulationRegion{} \pa_A}
       }
\end{align*}
Here, $\pa_1$ and $\pa_2$ are system components, $\pa_A$ is an assumption on the environment of $\pa_1$, and $\pa_G$ models the desired behaviour of the overall system, i.e., of the composition. 
The conclusion $\pa_1 \parallel \pa_2 \strSimulationRegion{} \pa_G$ ensures that $\pa_G$ strongly simulates the composite system. 
Preservation of specifications is then obtained indirectly from the fact that strong simulation preserves a suitable logic fragment, e.g., a safety fragment of PCTL~\cite{CV10}; that is, if $\pa \strSimulationRegion{} \pa_\mathit{spec}$ holds and $\pa_\mathit{spec} \models \Psi$, then also $\pa \models \Psi$.

We lift simulation-based AG reasoning to pPAs. 
The asymmetric rule depends on two properties of the simulation relation:
\begin{enumerate*}
  \item[(i)] it is a \emph{preorder}, i.e., reflexive and transitive, and
  \item[(ii)] it is \emph{compositional} with respect to parallel composition.
\end{enumerate*}
We introduce two variants of simulation relations for pPAs---\emph{strong simulation} and \emph{robust-strong simulation}---which satisfy (i) and (ii). 
Classical strong simulation for PAs~\cite{Seg+95} is lifted to the pPA setting by defining simulation with respect to a region of valuations and thus capturing parameter dependencies across components.
The main result of this section is a simulation-based AG proof rule for pPAs (cf.~\Cref{theo:Asym_simulation_ppa}). 

\subsection{AG reasoning via strong simulation for PAs}
\label{sec:simulation_AG_PA}
We first recall strong simulation for PAs as introduced by Segala and Lynch~\cite{Seg+95} and the asymmetric simulation-based AG rule of Komuravelli et al.~\cite{Kom+12}.
For this section, let $\pa_i=\paTupleOf{i}$ over alphabet $\alphabetOf{i}$ for $i \in \{1,2\}$ be two PAs.

For non-probabilistic labeled transition systems, simulation is defined in terms of successor states: a state $s_2$ simulates a state $s_1$ if, for each successor of $s_1$, there is a matching successor of $s_2$ that again simulates it. 
In the probabilistic setting, transitions lead to \emph{distributions} over successor states, and the simulation relation has to be lifted from states to distributions. 
Segala and Lynch formalise this by means of weight functions~\cite[Definition~15]{Seg+95}. 
We use an equivalent characterisation that is given in~\cite[Lemma 4.2.1]{Zha09}: 
\begin{defi}\label{theo:alt_charact_dist_sim}
  For relation $\simulationRelation{} \subseteq \stateSetOf{1} \times \stateSetOf{2}$ and $\mu_i \in \dist{\stateSetOf{i}}{}$ ($i=1,2$),
  let $\mu_1 \sqsubseteq_{\simulationRelation{}} \mu_2$ iff for every $A \subseteq \stateSetOf{1}$ we have $\mu_1(A) \leq \mu_2(\simulationRelation(A)),$ where $\simulationRelation(A) = \{s \in \stateSetOf{2} \mid \exists s' \in A: (s',s) \in \simulationRelation{} \}$. 
\end{defi}
Intuitively, $\mu_1 \sqsubseteq_{\simulationRelation{}} \mu_2$  means that for every subset $A$ of the domain of $\mu_1$, the probability mass that $\mu_1$ assigns to $A$ does not exceed the mass that $\mu_2$ assigns to the set of states related to $A$ via $\simulationRelation{}$.
\begin{defi}\label{def:strong_sim}
A \emph{strong simulation} between $\pa_1$ and $\pa_2$ is a relation $\simulationRelation \subseteq \stateSetOf{1}  \times \stateSetOf{2} $ such that  
    $(\initialOf{1},\initialOf{2}) \in \simulationRelation$ and
    for every $(s_1, s_2) \in \simulationRelation$ and 
    $(s_1, \alpha_1) \in \domain(\transFctOf{1})$, there is 
    $(s_2, \alpha_2) \in \domain(\transFctOf{2})$ 
   such that $\syncOf{1}(s_1, \alpha_1) = \syncOf{2}(s_2, \alpha_2)$ 
   and $\transFctOf{1}(s_1,\alpha_1) \sqsubseteq_{\simulationRelation} \transFctOf{2}(s_2,\alpha_2)$.

$\pa_2$ strongly simulates $\pa_1$, written $\pa_1 \strSimulationRegion{} \pa_2$, iff there is a strong simulation between $\pa_1$ and $\pa_2$.
\end{defi}

\begin{rem}\label{rem:preservation_sim_PA}
Strong simulation preserves satisfaction of safety PCTL formulas~\cite[Lemma~2.11]{CV10}: 
if $\pa_1 \strSimulationRegion{} \pa_2$, then for every safety PCTL formula $\Psi$---for example 
$\Psi = \probPredicate{\geq 0.9}{\Box \frownie }$---we have that 
$\pa_2 \models \Psi$ implies $\pa_1 \models \Psi$.%
\end{rem}
\begin{lem}[{\cite[Proposition~17]{Seg+95}}]\label{theo:preorder_compositional}
  Strong simulation $\strSimulationRegion{}$ is a preorder, i.e., reflexive and transitive. Moreover, if $\pa_1 \strSimulationRegion{} \pa_2$ and $\alphabetOf{2} \subseteq \alphabetOf{1}$, then for every PA $\pa$: $\pa_1 \parallel \pa  \strSimulationRegion{} \pa_2 \parallel \pa$.
\end{lem}

\begin{restatable}[{\cite[Theorem~1]{Kom+12}}]{thm}{restatablSimRulePA}
\label{theo:Asym_simulation_pa}
The simulation-based AG rule is sound and complete: 
\[
\infer{
           \pa_1 \parallel \pa_2 \strSimulationRegion{} \pa_G
         }{
           \deduce{\pa_2 \parallel \pa_A \strSimulationRegion{} \pa_G}
                  {\pa_1 \strSimulationRegion{} \pa_A}
         }
\]
\end{restatable}
 
\begin{proof}
Towards completeness, it has to be shown that whenever $\pa_1 \parallel \pa_2 \strSimulationRegion{} \pa_G$ holds, there exists an assumption PA $\pa_A$ such that both premises are satisfied---which follows directly by instantiating the assumption with $\pa_A = \pa_1$ and by commutativity of $\parallel$. 
 
For soundness, assume the premises hold. 
 From $\pa_1 \strSimulationRegion{} \pa_A$ and $\alphabetOf{A} \subseteq \alphabetOf{1}$, compositionality of strong simulation (\Cref{theo:preorder_compositional}) yields
 \[
 \pa_1 \parallel \pa_2 ~\strSimulationRegion{}~
 \pa_A \parallel \pa_2 
 ~=~ \pa_2 \parallel \pa_A.
 \]  
The second premise $\pa_2 \parallel \pa_A \strSimulationRegion{} \pa_G$ and transitivity of $\strSimulationRegion{}$ then imply $\pa_1 \parallel \pa_2 \strSimulationRegion{} \pa_G.$
\end{proof}

Intuitively, $\pa_1$ and $\pa_2$ are system components, $\pa_A$ is an assumption on the environment of $\pa_1$, and $\pa_G$ is a guarantee specification for the composite system. 
The rule states that if $\pa_1$ conforms (via strong simulation) to the assumption $\pa_A$, and $\pa_2 \parallel \pa_A$ simulates the guarantee $\pa_G$, then the actual composition $\pa_1 \parallel \pa_2$ also conforms to $\pa_G$. 
Via \Cref{rem:preservation_sim_PA}, the guarantee $\pa_G$ can be given as a safety PCTL formula $\Psi$: if $\pa_1 \strSimulationRegion{} \pa_A$ and $(\pa_2 \parallel \pa_A) \models \Psi$ for some safety PCTL formula $\Psi$, then $\pa_1 \parallel \pa_2 \models \Psi$ as well.%

\subsection{Strong and robust-strong simulation for pPAs}
\label{sec:strong_robust_sim_pPA}
We lift strong simulation from PAs to pPAs by parameterising the simulation relation with a region of valuations. 
We consider two variants of simulation for pPAs: \emph{strong simulation} and a stricter \emph{robust-strong simulation}. 
The former allows the witnessing simulation relation to depend on the chosen valuation, whereas the latter requires a single relation that works uniformly for all valuations in the region. 
We then show that both relations are preorders and compositional. 
Let $\ppa_i = \ppaTupleOf{i}$ for $i =1,2$ be two pPAs over $\alphabetOf{i}$ and let $\region$ be a region that is well-defined for $\ppa_1$ and $\ppa_2$.
\begin{defi}
\label{def:strong_sim_ppa}
$\ppa_2$ \emph{strongly simulates} $\ppa_1$ w.r.t.\ region $\region$, written $\ppa_1 \strSimulationRegion{\region} \ppa_2$, iff
$\ppa_1[\valuation] \strSimulationRegion{} \ppa_2[\valuation]$
for all valuations $\valuation\in\region$.
\end{defi}

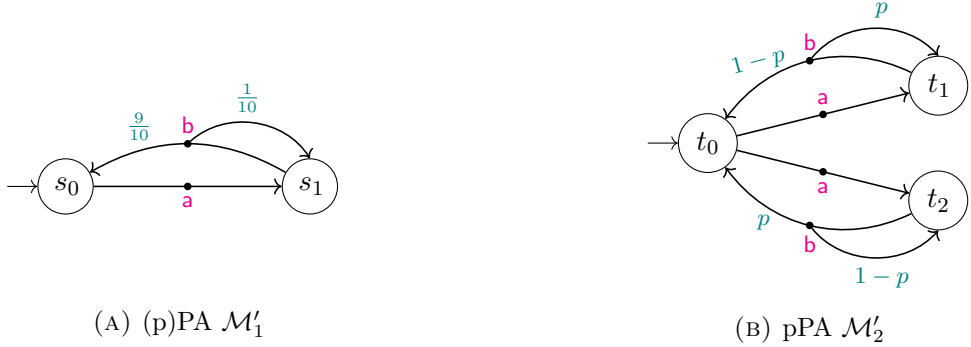
\begin{figure}[t]
  \centering
  \begin{subfigure}[c]{.45\textwidth}
    \centering
    \vspace{1cm}
    \begin{tikzpicture}[mdp]

	\node[ps, init=left] (0)  {$s_0$};

	\node[ps,right=2.5 of 0] (1)  {$s_1$};

	\path[ptrans]

	(0) edge[bend right=0] node[pos=0.5,below] {\tact{\lab}} node[dist] (d0a) {} node[pos=0.75,left] {} (1)

    (1) edge[bend right=30] node[pos=0.5,above]  {\tact{\altlab}} node[dist] (d1b) {} node[pos=0.75,above]{\tprob{{\frac{9}{10}}}} node[pos=0.5,below] {} (0)

    (d1b) edge[bend left=50] node[above,pos=0.45] {\tprob{{\frac{1}{10}}}} node[above left=3pt,pos=0.4] {} (1.north)

	;
\end{tikzpicture}
    \vspace{1cm}
    \caption{(p)PA $\ppa_1'$}\label{fig:sim_pPA_M1}
  \end{subfigure}\hfill
  \begin{subfigure}[c]{.45\textwidth}
    \centering
    \begin{tikzpicture}[mdp]

	\node[ps, init=left] (0)  {$t_0$};

	\node[ps,above right=0.2 and 2.5 of 0] (1)  {$t_{1}$};

    \node[ps,below right=0.2 and 2.5 of 0] (2)  {$t_{2}$};

	\path[ptrans]

	(0) edge[bend right=0] node[pos=0.5,above] {\tact{\lab}} node[dist] (d0a) {} node[pos=0.75,left] {} (1)

    (1) edge[bend right=40] node[pos=0.5,above]  {\tact{\altlab}} node[dist]  (d1b) {} node[pos=0.75,above,rotate=15]{\tprob{1-p}} node[pos=0.5,below] {} (0)

    (d1b) edge[bend left=50] node[above,pos=0.55] {\tprob{p}} node[above left=3pt,pos=0.4] {} (1.north)

	(0) edge[bend left=0] node[pos=0.5,below] {\tact{\lab}} node[dist] (d0a) {} node[pos=0.75,left] {} (2)

    (2) edge[bend left=40] node[pos=0.5,below]  {\tact{\altlab}} node[dist]  (d2b) {} node[pos=0.75,below]{\tprob{p}} node[pos=0.5,below] {} (0)

    (d2b) edge[bend right=50] node[below,pos=0.55] {\tprob{1-p}} node[above left=3pt,pos=0.4] {} (2.south)
    
;
\end{tikzpicture}
    \caption{pPA $\ppa_2'$}\label{fig:sim_pPA_M2}
  \end{subfigure}
  \caption{Example pPAs $\ppa_1'$ and $\ppa_2'$.}
  \label{fig:sim_ppa1_and_ppa2}
\end{figure}

\begin{exa}\label{ex:str_sim} 
  Consider the pPAs $\ppa_1'$ and $\ppa_2'$ in \Cref{fig:sim_ppa1_and_ppa2}. 
  Let $\region = \{\valuation : \{p\} \to \reals \mid \valuation(p) \in \{\tfrac{1}{10}, \tfrac{9}{10}\}\}.$ 
  For the valuation $\valuation \in \region$ with $\valuation(p) = \tfrac{1}{10}$, the relation
$\simulationRelation_{1} = \{(s_0,t_0), (s_1,t_1)\}$ is a strong simulation (in the sense of \Cref{def:strong_sim}) between the instantiated PAs $\ppa_1'[\valuation]$ and $\ppa_2'[\valuation]$.
  For the valuation $\valuation' \in \region$ with $\valuation'(p) = \tfrac{9}{10}$, the relation $\simulationRelation_{2} = \{(s_0,t_0), (s_1,t_2)\}$ is a strong simulation between the instantiated PAs $\ppa_1'[\valuation']$ and $\ppa_2'[\valuation']$. 
  Hence, $\ppa_1' \strSimulationRegion{\region} \ppa_2'.$
\end{exa}

As the example shows, the simulation relations witnessing $\ppa_1 \strSimulationRegion{\region} \ppa_2$ generally depend on the valuation. 
We introduce a stricter notion of simulation between two pPAs that requires a \emph{single} relation to witness strong simulation for all valuations in the region. 
\begin{defi}
\label{def:strong_sim_ppa_robust}
  A relation $\simulationRelation \subseteq \stateSetOf{1} \times \stateSetOf{2}$ is a \emph{robust-strong simulation} between $\ppa_1$ and $\ppa_2$ w.r.t.\ $\region$ if, for all valuations $\valuation \in \region$, the relation $\simulationRelation$ is a strong simulation (in the sense of \Cref{def:strong_sim}) between the instantiated PAs $\ppa_1[\valuation]$ and $\ppa_2[\valuation]$.

The pPA $\ppa_2$ \emph{robust-strongly simulates} $\ppa_1$ w.r.t.\ $\region$, written $\ppa_1 \robStrSimulationRegion{\region} \ppa_2$, iff there is a robust-strong simulation between $\ppa_1$ and $\ppa_2$.
\end{defi}

\begin{cor}
     $ \ppa_1 \robStrSimulationRegion{\region} \ppa_2$
     implies
     $    \ppa_1 \strSimulationRegion{\region} \ppa_2
    $. The converse does not hold.
\end{cor}

\begin{rem}\label{rem:sim_pPA_PCTL}
Strong simulation for pPAs preserves safety PCTL formulas, similar to strong simulation for PAs. 
More precisely, for a safety PCTL formula $\Psi$, if $\ppa_1 \strSimulationRegion{\region} \ppa_2$, then for all valuations $\valuation \in \region$, $\ppa_2[\valuation] \models \Psi $ implies $\ppa_1[\valuation] \models \Psi.$ 

For the robust-strong simulation relation $\robStrSimulationRegion{\region}$, we have the following result: 
If $\ppa_1 \robStrSimulationRegion{\region} \ppa_2$ and $\Psi$ is a safety PCTL formula, then
$\ppa_2,\region \models \Psi$ implies $\ppa_1,\region \models \Psi.$
\end{rem}

\begin{exa}
\label{ex:strong_not_robust}
  Consider again the pPAs $\ppa_1'$ and $\ppa_2'$ from \Cref{fig:sim_ppa1_and_ppa2} and the region $\region = \{\valuation : \{p\} \to \reals \mid \valuation(p) \in \{\tfrac{1}{10}, \tfrac{9}{10}\}\}$. 
  As shown in \Cref{ex:str_sim}, we have $\ppa_1' \strSimulationRegion{\region} \ppa_2'$, since for each $\valuation \in \region$ there exists a strong simulation between $\ppa_1'[\valuation]$ and $\ppa_2'[\valuation]$. 
  The witnessing relations $\simulationRelation_1$ and $\simulationRelation_2$, however, depend on the valuation $\valuation$: there is no single relation $\simulationRelation$ that is a strong simulation in the sense of \Cref{def:strong_sim} for all $\valuation \in \region$. 
  Thus, $\ppa_1' \not\robStrSimulationRegion{\region} \ppa_2'$. 

Now we define a new pPA $\ppa_1''$ from $\ppa_1'$ by changing the $\altlab$-labeled transition from $s_1$ so that its successor distribution is $\mu$ with $\mu(s_0) = p$ and $\mu(s_1) = 1-p$. 
Then, the relation $\simulationRelation = \{(s_0,t_0),(s_1,t_1)\}$ is a strong simulation between $\ppa_1''[\valuation]$ and $\ppa_2'[\valuation]$ for every $\valuation \in \region$. 
We therefore obtain $\ppa_1'' \robStrSimulationRegion{\region} \ppa_2'$.
\end{exa}

\begin{restatable}[]{lem}{restatablSimPropspPA}
\label{theo:preorder_compositional_ppa}
Let $\mathbin{\trianglelefteq}_{\region}\in \{\strSimulationRegion{\region}, \robStrSimulationRegion{\region}\}$.
The relation $\mathbin{\trianglelefteq}_{\region}$ is a \emph{preorder} and  \emph{compositional}: if $\ppa_1 \mathbin{\trianglelefteq}_{\region} \ppa_2$ and $\alphabetOf{2} \subseteq \alphabetOf{1}$, then for every pPA $\ppa$ for which $\region$ is well-defined, we have $\ppa_1 \parallel \ppa \mathbin{\trianglelefteq}_{\region} \ppa_2 \parallel \ppa$.
\end{restatable}
The proof of \Cref{theo:preorder_compositional_ppa} is given in \Cref{ap:proofs_AG_sim}.
Using \Cref{theo:preorder_compositional_ppa}, we can lift the AG rule from \Cref{theo:Asym_simulation_pa} to the parametric setting and obtain the following asymmetric AG rule for pPAs.

\begin{restatable}[]{thm}{restatablSimRulepPA}
\label{theo:Asym_simulation_ppa}
  Let $\ppa_A$, $\ppa_G$ and $\ppa_i$ for $i \in \{1,2\}$ be pPAs with $\alphabetOf{A} \subseteq \alphabetOf{1}$, and let $\region_1$ and $\region_2$ be well-defined regions. 
  Let $\mathbin{\trianglelefteq} \in \{\strSimulationRegion{}, \robStrSimulationRegion{}\}$. 
  The following rule is sound and complete: 
    \begin{align*}
           \infer{
                   \ppa_1 \parallel \ppa_2 \mathbin{\trianglelefteq}_{\regionIntersectionOf{\region_1}{}{\region_2}} \ppa_G
                } 
                {
                  \deduce{    \ppa_2 \parallel \ppa_A \mathbin{\trianglelefteq}_{\region_2} \ppa_G
                        }{
                        \ppa_1 \mathbin{\trianglelefteq}_{\region_1}  \ppa_A
                }
                } 
    \end{align*}
\end{restatable}

\begin{proof}[Proof Sketch]
Completeness follows as in the PA case by instantiating $\ppa_A$ with $\ppa_1$; soundness relies on $\mathbin{\trianglelefteq}_{\region}$ being a preorder and compositional (\Cref{theo:preorder_compositional_ppa}).
The full proof is given in \Cref{ap:proofs_AG_sim}. 
\end{proof}

\section{Related Work}\label{sec:related_work}
Compositional verification has been widely studied in both probabilistic and non-probabilistic systems. 
We summarize the key contributions related to our work. 

\paragraph{AG reasoning.}
Jones' rely-guarantee method \cite{Jon83} and Pnueli's compositional proof system \cite{Pnu84} for linear temporal logic laid the foundation for AG reasoning. 
Subsequent works focused on AG rules for systems using CTL$^*$ \cite{Cla+89} and developed AG reasoning for \emph{reactive modules} \cite{Hen+98,AH+99}. 
Automated AG reasoning techniques include learning-based assumption generation~\cite{Cob+03,Pas+08}. 
More recent work has focused on circular AG reasoning \cite{Elk+18} and verification-repair approaches \cite{Fre+22}. 


\paragraph{AG reasoning for probabilistic systems.}
The initial work by de Alfaro et al.\ \cite{Alf+01} introduced AG rules for a probabilistic extension of reactive modules \cite{Hen+98,AH+99}. 
Their model is similar to PA \cite{Seg+95,Seg95}, but limited to synchronous parallel composition. 
Kwiatkowska et al.\ \cite{Kwi+10,For+11} generalized AG verification for PA, allowing more flexible parallel composition and extending AG reasoning to probabilistic safety properties. 
Their approach reduces AG verification to multi-objective model checking, as proposed by Etessami et al.\ \cite{Ete+08}. 
This line of work was further refined in \cite{Kwi+13}, enabling AG reasoning over a broader class of quantitative properties, including conjunctions over probabilistic liveness and expected rewards. 
Algorithmic learning-based assumption generation techniques \cite{Cob+03,Gup+08} were later adapted for AG reasoning in probabilistic settings \cite{Fen+10, Fen+11, LL19}. 
Other assumption generation approaches include abstraction-refinement methods \cite{Kom+12, Chat+15}, based on the CEGAR paradigm \cite{Cla+00}, and symbolic learning-based methods \cite{He+16}. 
Hampus and Nyberg~\cite{HN24} develop a theory of assume‑guarantee contracts for probabilistic cyber‑physical systems, supporting both probabilistic and non‑probabilistic behaviour, together with a deductive system to prove probabilistic properties.  

\paragraph{AG reasoning for models with uncertainty.}
The notion of bisimulation for parametric Markov models studied by Hahn et al.~\cite{Hah+11} is closely related to our robust‑strong simulation: both require a single relation that constitutes a strong simulation (in the PA sense) for every valuation in a region. 
Bouchekir and Boukala~\cite{BB18Interval} propose a learning‑based symbolic AG framework
for MDPs where assumptions are expressed as interval MDPs encoded with MTBDDs (multi-terminal BDDs). 
Compositional reasoning for interval MDPs via bisimulation has been investigated by Hashemi et al.~\cite{Has+16}. 

\paragraph{Sequential composition.}
Another recent line of research focuses on compositional reasoning of \emph{sequentially} composed MDPs: 
Junges and Spaan \cite{JS22} introduced an abstraction-refinement approach for hierarchical probabilistic models, leveraging parametric MDPs to represent sets of similar subroutines. 
Watanabe et al.\ \cite{Wat+24} exploit string diagrams to verifying sequentially composed MDPs. 
Multi-objective reachability properties are a central element in this approach.

\paragraph{Our contribution.}
We extend AG reasoning to parametric and robust PAs (and sub-classes thereof), leveraging research on parametric MDPs \cite{Jun20,Qua+16,Jun+24}, robust MDPs \cite{Badings25Thesis,NG05,Wi+13,Iye05,Su+25}, and two approach to AG verification for PAs \cite{Kwi+13,Kom+12}. 
Our framework enables to reason about monotonicity \cite{Spel23,Spe+19,Spe+21} in a compositional manner. 
To the best of our knowledge, this is the first AG-based compositional verification framework for \emph{parametric} and \emph{robust} PAs.

\section{Conclusion}\label{sec:conclusion}

\begin{figure}[t]
{\small
\begin{forest}
for tree={
font=\small, 
    grow=south,
    draw,
    rounded corners,
    align=center,
    parent anchor=south,
    child anchor=north,
    l sep=13mm,
    s sep=8mm
}
[PA with uncertainty, name=rootnode
    [pPA \\{ composition $\parallel$ }
    , name=paramnode]
    [rPA \\ { composition $\parallel$} 
    , name=robnode
        [convex rPA \\ { composition $\parallelConv$}
        , name=polynode
            [iPA \\ { composition $\parallelRel$} 
            , name=ipanode]
        ]
    ]
]
\node[below right=-0.7cm and 3.3cm of rootnode,anchor=north,font=\small]{
    \begin{tblr}{
      colspec={Q[c,wd=30mm]|[dashed]Q[c,wd=19mm]|[dashed]Q[c,wd=37mm]},
      row{1} = {ht=4mm},   
      row{2} = {ht=4mm},
      row{3-5} = {ht=23mm} 
    }
     & \SetCell[c=2]{c} {\hspace{-1.9cm}nature} & \\
     \hspace*{3cm} & {memoryless} & {memory-full} \\
     \hline[dashed]
     & \SetCell[r=2]{c=1} { \xmark \\ \Cref{sec:mless_nature_rpa} 
     }
     & \SetCell[r=1]{c=1} { \xmark \\ \Cref{{sec:non_convex_rpa}}
     }\\
     \hline[dashed]
     &  
     & \SetCell[r=1]{c=1, bg=mygreen!30} {\checkmark \\ AG proof rules \\ \medskip 
      \begin{flushleft}
     { \Cref{sec:AG_for_rpa_conv}:  MO-queries 
     }
      \end{flushleft}
      }
     \\
     \hline[dashed]
     & \SetCell[c=2]{c} { \xmark \ \Cref{{sec:interval_arithm_relaxation}} 
     } 
     & \\
    \end{tblr}
};
\node[below left=2.8cm and 1.1cm of rootnode,anchor=north,font=\small]{
    \begin{tblr}{
      colspec={|[dashed]Q[c,wd=38mm]|[dashed]},
      row{1} = {ht=23mm} 
    } 
     \hline[dashed]
     \SetCell[r=1]{c=1, bg=mygreen!30} { 
     \checkmark \\ AG proof rules \\ \medskip 
    { \begin{flushleft}
       \Cref{sec:pag}: MO-queries 
    \\ \medskip
    \Cref{sec:pag_mono}: monotonicity 
      \\ \medskip 
     \Cref{sec:AG_sim}: simulation
        \\ \medskip 
    \end{flushleft}
        }
        }
        \\
     \hline[dashed]
    \end{tblr}
};
\end{forest}
}
\caption{Overview of established AG rules for pPAs and rPAs.}\label{fig:conclOverview}
\end{figure}
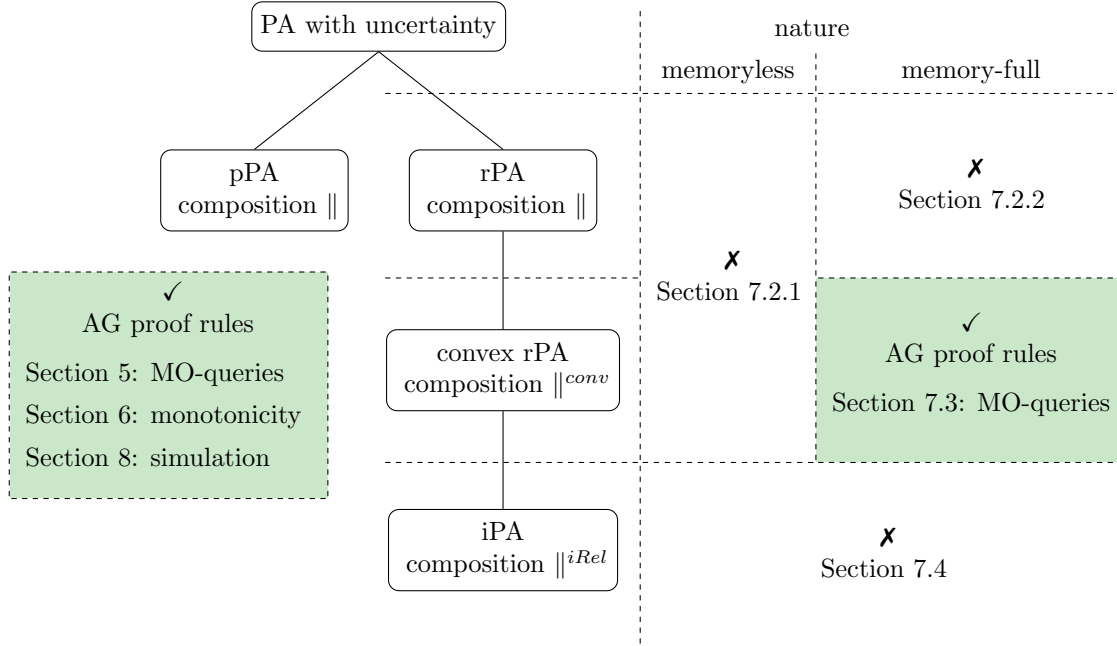 

\paragraph{Summary}
We presented an assume-guarantee framework for the compositional verification of Segala's probabilistic automata (PA) extended with uncertainty.
We studied various forms of uncertainty ranging from parameters to uncertainty sets.
Parametric PAs allow for modeling dependencies between probabilistic branching at different states; a valuation fixes all parameters once and for all.
After selecting an action in robust PA, nature selects a distribution from the uncertainty set. 
An overview of our established results is provided in \Cref{fig:conclOverview}.

Our results are primarily generalisations of the AG rules for PAs by Kwiatkowska et al.~\cite{Kwi+13}.
Our AG rules for pPAs cover safety and $\omega$-regular properties as well as (parametric) expected total rewards and their multi-objective combinations.
We also introduced compositional proof rules to reason about monotonicity.
For the robust setting, we studied AG rules for different forms of nature, and for different sub-classes of robust PAs, such as convex and interval PAs.
We showed that the asymmetric AG rule of~\cite{Kwi+13} can be lifted to convex rPAs with history-dependent nature under a convexity-preserving parallel composition operator. The majority of our results for robust PA are negative though: the asymmetric AG rule is neither applicable to memoryless nature, non-convex uncertainty sets, nor the interval relaxation of parallel composition.


Finally, 
we developed an alternative, simulation-based AG framework for parametric PAs in the style of Komuravelli et al.~\cite{Kom+12}, introducing strong and robust-strong simulation preorders and establishing sound and complete asymmetric AG rules. 

These contributions lay the theoretical foundations for modular verification of parametric PAs and robust PAs. 

\paragraph{Future work}
There are many different directions for future work.
An implementation can demonstrate the effectiveness of AG reasoning. 
This requires efficient multi-objective verification algorithms for pMDPs and rMDPs, in particular to verify the AG premises. 
Existing techniques for single-objective parametric MDPs~\cite{Jun+24} provide a good starting point, as well as robust multi-objective verification for interval MDPs~\cite{Hah+19_interval}.
Extensions towards other properties such as long-run average rewards and expected visiting times~\cite{Mer+24}, as well as investigating other models, such as partially observable MDPs (see also the preliminary works~\cite{Car+11,Wil15,Zha+16}) are of interest.
For robust PAs, additional nature semantics (e.g., Markovian or at-every-step nature~\cite{Bac+21}, see \Cref{rem:nature_semantics_variants}) as well as fairness assumptions could be investigated. 

Another direction is to study algorithmic verification of strong and robust-strong simulation for pPAs, and to develop AG rules based on alternative simulation preorders, such as weak and probabilistic simulations~\cite{Seg+95}. 
Additionally, circular AG proof rules for pPAs, following the approach of Elkader et al.~\cite{Elk+18}, could be studied.
Extensions of the simulation-based approach to robust PAs and interval PAs could be inspired by (bi)simulation notions for robust and interval models \cite{JL91,JYL01,Has+16}. 

We plan to investigate whether the assume-guarantee framework can be lifted to interleaved probabilistic I/O automata with distributed schedulers~\cite{Gir+14}, where partial-information constraints lead to undecidability of reachability and model-checking problems. 
Finally, CEGAR-style and learning-based assumption generation techniques for parametric and robust PAs are of interest, following e.g., the ideas of~\cite{Kom+12,Bob+08,Pas+08,Cha+05,Fen+11,BB18Interval}. 


\bibliographystyle{alphaurl}
\bibliography{biblio}

@article{Gir+14,
  author       = {Sergio Giro and
                  Pedro R. D'Argenio and
                  Luis Mar{\'{\i}}a Ferrer Fioriti},
  title        = {Distributed probabilistic input/output automata: Expressiveness, (un)decidability
                  and algorithms},
  journal      = {Theor. Comput. Sci.},
  volume       = {538},
  pages        = {84--102},
  year         = {2014},
  url          = {https://doi.org/10.1016/j.tcs.2013.07.017},
  doi          = {10.1016/J.TCS.2013.07.017},
  bibsource    = {dblp computer science bibliography, https://dblp.org}
}

@inproceedings{SchnitzerAP26,
  author       = {Yannik Schnitzer and
                  Alessandro Abate and
                  David Parker},
  editor       = {Sven Koenig and
                  Chad Jenkins and
                  Matthew E. Taylor},
  title        = {Efficient Solution and Learning of Robust Factored {MDP}s},
  booktitle    = {Fortieth {AAAI} Conference on Artificial Intelligence, Thirty-Eighth
                  Conference on Innovative Applications of Artificial Intelligence,
                  Sixteenth Symposium on Educational Advances in Artificial Intelligence,
                  {AAAI} 2026, Singapore, January 20-27, 2026},
  pages        = {36369--36377},
  publisher    = {{AAAI} Press},
  year         = {2026},
  doi          = {10.1609/AAAI.V40I43.40957},
}

@inproceedings{HN24,
  author       = {Anton Hampus and
                  Mattias Nyberg},
  editor       = {Tiziana Margaria and
                  Bernhard Steffen},
  title        = {A Theory of Probabilistic Contracts},
  booktitle    = {Leveraging Applications of Formal Methods, Verification and Validation.
                  Specification and Verification - 12th International Symposium, ISoLA
                  2024, Crete, Greece, October 27-31, 2024, Proceedings, Part {III}},
  series       = {Lecture Notes in Computer Science},
  volume       = {15221},
  pages        = {296--319},
  publisher    = {Springer},
  year         = {2024},
  url          = {https://doi.org/10.1007/978-3-031-75380-0\_17},
  doi          = {10.1007/978-3-031-75380-0\_17},
  bibsource    = {dblp computer science bibliography, https://dblp.org}
}

@inproceedings{JL91,
  author       = {Bengt Jonsson and
                  Kim Guldstrand Larsen},
  title        = {Specification and Refinement of Probabilistic Processes},
  booktitle    = {Proceedings of the Sixth Annual Symposium on Logic in Computer Science
                  {(LICS} '91), Amsterdam, The Netherlands, July 15-18, 1991},
  pages        = {266--277},
  publisher    = {{IEEE} Computer Society},
  year         = {1991},
  url          = {https://doi.org/10.1109/LICS.1991.151651},
  doi          = {10.1109/LICS.1991.151651},
  bibsource    = {dblp computer science bibliography, https://dblp.org}
}

@incollection{JYL01,
  author       = {Bengt Jonsson and
                  Wang Yi and
                  Kim G. Larsen},
  editor       = {Jan A. Bergstra and
                  Alban Ponse and
                  Scott A. Smolka},
  title        = {Probabilistic Extensions of Process Algebras},
  booktitle    = {Handbook of Process Algebra},
  pages        = {685--710},
  publisher    = {North-Holland / Elsevier},
  year         = {2001},
  url          = {https://doi.org/10.1016/b978-044482830-9/50029-1},
  doi          = {10.1016/B978-044482830-9/50029-1},
  bibsource    = {dblp computer science bibliography, https://dblp.org}
}

@inproceedings{Bac+21,
  author       = {Giovanni Bacci and
                  Beno{\^{\i}}t Delahaye and
                  Kim G. Larsen and
                  Anders Mariegaard},
  editor       = {Ernst{-}R{\"{u}}diger Olderog and
                  Bernhard Steffen and
                  Wang Yi},
  title        = {Quantitative Analysis of Interval {M}arkov Chains},
  booktitle    = {Model Checking, Synthesis, and Learning - Essays Dedicated to Bengt
                  Jonsson on The Occasion of His 60th Birthday},
  series       = {Lecture Notes in Computer Science},
  volume       = {13030},
  pages        = {57--77},
  publisher    = {Springer},
  year         = {2021},
  url          = {https://doi.org/10.1007/978-3-030-91384-7\_4},
  doi          = {10.1007/978-3-030-91384-7\_4},
  bibsource    = {dblp computer science bibliography, https://dblp.org}
}

@inproceedings{Has+16,
  author       = {Vahid Hashemi and
                  Holger Hermanns and
                  Lei Song and
                  K. Subramani and
                  Andrea Turrini and
                  Piotr Wojciechowski},
  editor       = {Adrian{-}Horia Dediu and
                  Jan Janousek and
                  Carlos Mart{\'{\i}}n{-}Vide and
                  Bianca Truthe},
  title        = {Compositional Bisimulation Minimization for Interval {M}arkov Decision
                  Processes},
  booktitle    = {Language and Automata Theory and Applications - 10th International
                  Conference, {LATA} 2016, Prague, Czech Republic, March 14-18, 2016,
                  Proceedings},
  series       = {Lecture Notes in Computer Science},
  volume       = {9618},
  pages        = {114--126},
  publisher    = {Springer},
  year         = {2016},
  url          = {https://doi.org/10.1007/978-3-319-30000-9\_9},
  doi          = {10.1007/978-3-319-30000-9\_9},
  bibsource    = {dblp computer science bibliography, https://dblp.org}
}

@inproceedings{Mer+25CONCUR,
  author       = {Hannah Mertens and
                  Tim Quatmann and
                  Joost{-}Pieter Katoen},
  editor       = {Patricia Bouyer and
                  Jaco van de Pol},
  title        = {Compositional Reasoning for Parametric Probabilistic Automata},
  booktitle    = {36th International Conference on Concurrency Theory, {CONCUR} 2025,
                  Aarhus, Denmark, August 26-29, 2025},
  series       = {LIPIcs},
  volume       = {348},
  pages        = {31:1--31:20},
  publisher    = {Schloss Dagstuhl - Leibniz-Zentrum f{\"{u}}r Informatik},
  year         = {2025},
  url          = {https://doi.org/10.4230/LIPIcs.CONCUR.2025.31},
  doi          = {10.4230/LIPICS.CONCUR.2025.31},
  bibsource    = {dblp computer science bibliography, https://dblp.org}
}

@article{NG05,
  author       = {Arnab Nilim and
                  Laurent El Ghaoui},
  title        = {Robust Control of {M}arkov Decision Processes with Uncertain Transition
                  Matrices},
  journal      = {Oper. Res.},
  volume       = {53},
  number       = {5},
  pages        = {780--798},
  year         = {2005},
  url          = {https://doi.org/10.1287/opre.1050.0216},
  doi          = {10.1287/OPRE.1050.0216},
  bibsource    = {dblp computer science bibliography, https://dblp.org}
}

@inproceedings{Su+25,
  author       = {Marnix Suilen and
                  Thom Badings and
                  Eline M. Bovy and
                  David Parker and
                  Nils Jansen},
  editor       = {Nils Jansen and
                  Sebastian Junges and
                  Benjamin Lucien Kaminski and
                  Christoph Matheja and
                  Thomas Noll and
                  Tim Quatmann and
                  Mari{\"{e}}lle Stoelinga and
                  Matthias Volk},
  title        = {Robust {M}arkov Decision Processes: {A} Place Where {AI} and Formal
                  Methods Meet},
  booktitle    = {Principles of Verification: Cycling the Probabilistic Landscape -
                  Essays Dedicated to Joost-Pieter Katoen on the Occasion of His 60th
                  Birthday, Part {III}},
  series       = {Lecture Notes in Computer Science},
  volume       = {15262},
  pages        = {126--154},
  publisher    = {Springer},
  year         = {2024},
  url          = {https://doi.org/10.1007/978-3-031-75778-5\_7},
  doi          = {10.1007/978-3-031-75778-5\_7},
  bibsource    = {dblp computer science bibliography, https://dblp.org}
}

@article{Wi+13,
  author       = {Wolfram Wiesemann and
                  Daniel Kuhn and
                  Ber{\c{c}} Rustem},
  title        = {Robust {M}arkov Decision Processes},
  journal      = {Math. Oper. Res.},
  volume       = {38},
  number       = {1},
  pages        = {153--183},
  year         = {2013},
  url          = {https://doi.org/10.1287/moor.1120.0566},
  doi          = {10.1287/MOOR.1120.0566},
  bibsource    = {dblp computer science bibliography, https://dblp.org}
}

@phdthesis{Badings25Thesis,
  author       = {Thom S. Badings},
  title        = {Robust Verification of Stochastic Systems: Guarantees in the Presence of Uncertainty},
  school       = {Radboud University Nijmegen, Netherlands},
  year         = {2025},
  url          = {https://hdl.handle.net/2066/317218},
  doi          = {10.54195/9789493296909},
  bibsource    = {dblp computer science bibliography, https://dblp.org}
}

@article{Iye05,
  author       = {Garud N. Iyengar},
  title        = {Robust Dynamic Programming},
  journal      = {Math. Oper. Res.},
  volume       = {30},
  number       = {2},
  pages        = {257--280},
  year         = {2005},
  url          = {https://doi.org/10.1287/moor.1040.0129},
  doi          = {10.1287/MOOR.1040.0129},
  bibsource    = {dblp computer science bibliography, https://dblp.org}
}

@phdthesis{Seg95,
  author       = {Roberto Segala},
  title        = {Modeling and verification of randomized distributed real-time systems},
  school       = {Massachusetts Institute of Technology, Cambridge, MA, {USA}},
  year         = {1995},
  url          = {https://hdl.handle.net/1721.1/36560},
  bibsource    = {dblp computer science bibliography, https://dblp.org}
}

@inproceedings{Daw04,
  author       = {Conrado Daws},
  title        = {Symbolic and Parametric Model Checking of Discrete-Time {M}arkov Chains},
  booktitle    = {{ICTAC}},
  series       = {Lecture Notes in Computer Science},
  volume       = {3407},
  pages        = {280--294},
  publisher    = {Springer},
  doi          = {10.1007/978-3-540-31862-0\_21},
  year         = {2004}
}

@inproceedings{For+11,
	author       = {Vojtech Forejt and
	Marta Z. Kwiatkowska and
	Gethin Norman and
	David Parker},
	editor       = {Marco Bernardo and
	Val{\'{e}}rie Issarny},
	title        = {Automated Verification Techniques for Probabilistic Systems},
	booktitle    = {Formal Methods for Eternal Networked Software Systems - 11th International
	School on Formal Methods for the Design of Computer, Communication
	and Software Systems, {SFM} 2011, Bertinoro, Italy, June 13-18, 2011.
	Advanced Lectures},
	series       = {Lecture Notes in Computer Science},
	volume       = {6659},
	pages        = {53--113},
	publisher    = {Springer},
	year         = {2011},
	url          = {https://doi.org/10.1007/978-3-642-21455-4\_3},
	doi          = {10.1007/978-3-642-21455-4\_3},
	bibsource    = {dblp computer science bibliography, https://dblp.org}
}

@inproceedings{Kat16,
  author    = {Joost-Pieter Katoen},
  title     = {The Probabilistic Model Checking Landscape},
  booktitle = {{LICS}},
  pages     = {31--45},
  publisher = {{ACM}},
  doi       = {10.1145/2933575.2934574},
  year      = {2016}
}

@inproceedings{FWHT15,
  author    = {Lu Feng and
               Clemens Wiltsche and
               Laura R. Humphrey and
               Ufuk Topcu},
  title     = {Controller synthesis for autonomous systems interacting with human
               operators},
  booktitle = {{ICCPS}},
  pages     = {70--79},
  publisher = {{ACM}},
  doi       = {10.1145/2735960.2735973},
  year      = {2015}
}

@article{NS06,
  author    = {Gethin Norman and
               Vitaly Shmatikov},
  title     = {Analysis of probabilistic contract signing},
  journal   = {J. Comput. Secur.},
  volume    = {14},
  number    = {6},
  pages     = {561--589},
  doi       = {10.3233/jcs-2006-14604},
  year      = {2006}
}

@article{KNP08,
  author    = {Marta Z. Kwiatkowska and
               Gethin Norman and
               David Parker},
  title     = {Using probabilistic model checking in systems biology},
  journal   = {{SIGMETRICS} Perform. Evaluation Rev.},
  volume    = {35},
  number    = {4},
  pages     = {14--21},
  doi       = {10.1145/1364644.1364651},
  year      = {2008}
}

@inproceedings{Mer+24,
  author       = {Hannah Mertens and
                  Joost{-}Pieter Katoen and
                  Tim Quatmann and
                  Tobias Winkler},
  editor       = {Bernd Finkbeiner and
                  Laura Kov{\'{a}}cs},
  title        = {Accurately Computing Expected Visiting Times and Stationary Distributions
                  in {M}arkov Chains},
  booktitle    = {Tools and Algorithms for the Construction and Analysis of Systems
                  - 30th International Conference, {TACAS} 2024, Held as Part of the
                  European Joint Conferences on Theory and Practice of Software, {ETAPS}
                  2024, Luxembourg City, Luxembourg, April 6-11, 2024, Proceedings,
                  Part {II}},
  series       = {Lecture Notes in Computer Science},
  volume       = {14571},
  pages        = {237--257},
  publisher    = {Springer},
  year         = {2024},
  url          = {https://doi.org/10.1007/978-3-031-57249-4\_12},
  doi          = {10.1007/978-3-031-57249-4\_12},
  bibsource    = {dblp computer science bibliography, https://dblp.org}
}

@inproceedings{Qua+16,
  author       = {Tim Quatmann and
                  Christian Dehnert and
                  Nils Jansen and
                  Sebastian Junges and
                  Joost{-}Pieter Katoen},
  editor       = {Cyrille Artho and
                  Axel Legay and
                  Doron Peled},
  title        = {Parameter Synthesis for {M}arkov Models: Faster Than Ever},
  booktitle    = {Automated Technology for Verification and Analysis - 14th International
                  Symposium, {ATVA} 2016, Chiba, Japan, October 17-20, 2016, Proceedings},
  series       = {Lecture Notes in Computer Science},
  volume       = {9938},
  pages        = {50--67},
  year         = {2016},
  url          = {https://doi.org/10.1007/978-3-319-46520-3\_4},
  doi          = {10.1007/978-3-319-46520-3\_4},
  bibsource    = {dblp computer science bibliography, https://dblp.org}
}

@article{Jon83,
  author       = {Cliff B. Jones},
  title        = {Tentative Steps Toward a Development Method for Interfering Programs},
  journal      = {{ACM} Trans. Program. Lang. Syst.},
  volume       = {5},
  number       = {4},
  pages        = {596--619},
  year         = {1983},
  url          = {https://doi.org/10.1145/69575.69577},
  doi          = {10.1145/69575.69577},
  bibsource    = {dblp computer science bibliography, https://dblp.org}
}

@phdthesis{Spel23,
  author       = {Jip Spel},
  title        = {Monotonicity in {M}arkov models},
  school       = {{RWTH} Aachen University, Germany},
  year         = {2023},
  url          = {https://publications.rwth-aachen.de/record/974903},
  urn          = {urn:nbn:de:101:1-2024013100425618772929},
  bibsource    = {dblp computer science bibliography, https://dblp.org}
}

@inproceedings{Seg+95,
    author       = {Roberto Segala and
                  Nancy A. Lynch},
  editor       = {Bengt Jonsson and
                  Joachim Parrow},
  title        = {Probabilistic Simulations for Probabilistic Processes},
  booktitle    = {{CONCUR} '94, Concurrency Theory, 5th International Conference, Uppsala,
                  Sweden, August 22-25, 1994, Proceedings},
  series       = {Lecture Notes in Computer Science},
  volume       = {836},
  pages        = {481--496},
  publisher    = {Springer},
  year         = {1994},
  url          = {https://doi.org/10.1007/978-3-540-48654-1\_35},
  doi          = {10.1007/978-3-540-48654-1\_35},
  bibsource    = {dblp computer science bibliography, https://dblp.org}
}

@article{Jun+24,
	author       = {Sebastian Junges and
	Erika {\'{A}}brah{\'{a}}m and
	Christian Hensel and
	Nils Jansen and
	Joost{-}Pieter Katoen and
	Tim Quatmann and
	Matthias Volk},
	title        = {Parameter synthesis for {M}arkov models: covering the parameter space},
	journal      = {Formal Methods Syst. Des.},
	volume       = {62},
	number       = {1},
	pages        = {181--259},
	year         = {2024},
	url          = {https://doi.org/10.1007/s10703-023-00442-x},
	doi          = {10.1007/S10703-023-00442-X},
	bibsource    = {dblp computer science bibliography, https://dblp.org}
}

@inproceedings{Spe+19,
  author       = {Jip Spel and
                  Sebastian Junges and
                  Joost{-}Pieter Katoen},
  editor       = {Yu{-}Fang Chen and
                  Chih{-}Hong Cheng and
                  Javier Esparza},
  title        = {Are Parametric {M}arkov Chains Monotonic?},
  booktitle    = {Automated Technology for Verification and Analysis - 17th International
                  Symposium, {ATVA} 2019, Taipei, Taiwan, October 28-31, 2019, Proceedings},
  series       = {Lecture Notes in Computer Science},
  volume       = {11781},
  pages        = {479--496},
  publisher    = {Springer},
  year         = {2019},
  url          = {https://doi.org/10.1007/978-3-030-31784-3\_28},
  doi          = {10.1007/978-3-030-31784-3\_28},
  bibsource    = {dblp computer science bibliography, https://dblp.org}
}

@book{BK08,
	author       = {Christel Baier and
	Joost{-}Pieter Katoen},
	title        = {Principles of Model Checking},
	publisher    = {{MIT} Press},
	year         = {2008},
	isbn         = {978-0-262-02649-9},
	bibsource    = {dblp computer science bibliography, https://dblp.org}
}

@phdthesis{Jun20,
	author       = {Sebastian Junges},
	title        = {Parameter synthesis in {M}arkov models},
	school       = {{RWTH} Aachen University, Germany},
	year         = {2020},
	url          = {https://publications.rwth-aachen.de/record/783179},
	urn          = {urn:nbn:de:101:1-2020110312520656944348},
	bibsource    = {dblp computer science bibliography, https://dblp.org}
}

@article{Sto+02,
	author       = {Mari{\"{e}}lle Stoelinga},
	title        = {An Introduction to Probabilistic Automata},
	journal      = {Bull. {EATCS}},
	volume       = {78},
	pages        = {176--198},
	year         = {2002},
	bibsource    = {dblp computer science bibliography, https://dblp.org}
}

@article{Kwi+13,
	author       = {Marta Z. Kwiatkowska and
	Gethin Norman and
	David Parker and
	Hongyang Qu},
	title        = {Compositional probabilistic verification through multi-objective model
	checking},
	journal      = {Inf. Comput.},
	volume       = {232},
	pages        = {38--65},
	year         = {2013},
	url          = {https://doi.org/10.1016/j.ic.2013.10.001},
	doi          = {10.1016/J.IC.2013.10.001},
	bibsource    = {dblp computer science bibliography, https://dblp.org}
}

@inproceedings{Wat+24,
  author       = {Kazuki Watanabe and
                  Clovis Eberhart and
                  Kazuyuki Asada and
                  Ichiro Hasuo},
  editor       = {Nils Jansen and
                  Sebastian Junges and
                  Benjamin Lucien Kaminski and
                  Christoph Matheja and
                  Thomas Noll and
                  Tim Quatmann and
                  Mari{\"{e}}lle Stoelinga and
                  Matthias Volk},
  title        = {Compositional Solution of Mean Payoff Games by String Diagrams},
  booktitle    = {Principles of Verification: Cycling the Probabilistic Landscape -
                  Essays Dedicated to Joost-Pieter Katoen on the Occasion of His 60th
                  Birthday, Part {III}},
  series       = {Lecture Notes in Computer Science},
  volume       = {15262},
  pages        = {423--445},
  publisher    = {Springer},
  year         = {2024},
  url          = {https://doi.org/10.1007/978-3-031-75778-5\_20},
  doi          = {10.1007/978-3-031-75778-5\_20},
  bibsource    = {dblp computer science bibliography, https://dblp.org}
}

@inproceedings{JS22,
  author       = {Sebastian Junges and
                  Matthijs T. J. Spaan},
  editor       = {Sharon Shoham and
                  Yakir Vizel},
  title        = {Abstraction-Refinement for Hierarchical Probabilistic Models},
  booktitle    = {Computer Aided Verification - 34th International Conference, {CAV}
                  2022, Haifa, Israel, August 7-10, 2022, Proceedings, Part {I}},
  series       = {Lecture Notes in Computer Science},
  volume       = {13371},
  pages        = {102--123},
  publisher    = {Springer},
  year         = {2022},
  url          = {https://doi.org/10.1007/978-3-031-13185-1\_6},
  doi          = {10.1007/978-3-031-13185-1\_6},
  bibsource    = {dblp computer science bibliography, https://dblp.org}
}

@inproceedings{Spe+21,
  author       = {Jip Spel and
                  Sebastian Junges and
                  Joost{-}Pieter Katoen},
  editor       = {Jan Friso Groote and
                  Kim Guldstrand Larsen},
  title        = {Finding Provably Optimal {M}arkov Chains},
  booktitle    = {Tools and Algorithms for the Construction and Analysis of Systems
                  - 27th International Conference, {TACAS} 2021, Held as Part of the
                  European Joint Conferences on Theory and Practice of Software, {ETAPS}
                  2021, Luxembourg City, Luxembourg, March 27 - April 1, 2021, Proceedings,
                  Part {I}},
  series       = {Lecture Notes in Computer Science},
  volume       = {12651},
  pages        = {173--190},
  publisher    = {Springer},
  year         = {2021},
  url          = {https://doi.org/10.1007/978-3-030-72016-2\_10},
  doi          = {10.1007/978-3-030-72016-2\_10},
  bibsource    = {dblp computer science bibliography, https://dblp.org}
}

@inproceedings{Zha+16,
  author       = {Xiaobin Zhang and
                  Bo Wu and
                  Hai Lin},
  title        = {Assume-guarantee reasoning framework for {MDP-POMDP}},
  booktitle    = {55th {IEEE} Conference on Decision and Control, {CDC} 2016, Las Vegas,
                  NV, USA, December 12-14, 2016},
  pages        = {795--800},
  publisher    = {{IEEE}},
  year         = {2016},
  url          = {https://doi.org/10.1109/CDC.2016.7798365},
  doi          = {10.1109/CDC.2016.7798365},
  bibsource    = {dblp computer science bibliography, https://dblp.org}
}

@article{AH+99,
  author       = {Rajeev Alur and
                  Thomas A. Henzinger},
  title        = {Reactive Modules},
  journal      = {Formal Methods Syst. Des.},
  volume       = {15},
  number       = {1},
  pages        = {7--48},
  year         = {1999},
  url          = {https://doi.org/10.1023/A:1008739929481},
  doi          = {10.1023/A:1008739929481},
  bibsource    = {dblp computer science bibliography, https://dblp.org}
}

@inproceedings{Hen+98,
  author       = {Thomas A. Henzinger and
                  Shaz Qadeer and
                  Sriram K. Rajamani},
  editor       = {Alan J. Hu and
                  Moshe Y. Vardi},
  title        = {You Assume, We Guarantee: Methodology and Case Studies},
  booktitle    = {Computer Aided Verification, 10th International Conference, {CAV}
                  '98, Vancouver, BC, Canada, June 28 - July 2, 1998, Proceedings},
  series       = {Lecture Notes in Computer Science},
  volume       = {1427},
  pages        = {440--451},
  publisher    = {Springer},
  year         = {1998},
  url          = {https://doi.org/10.1007/BFb0028765},
  doi          = {10.1007/BFB0028765},
  bibsource    = {dblp computer science bibliography, https://dblp.org}
}

@inproceedings{Cla+89,
  author       = {Edmund M. Clarke and
                  David E. Long and
                  Kenneth L. McMillan},
  title        = {Compositional Model Checking},
  booktitle    = {Proceedings of the Fourth Annual Symposium on Logic in Computer Science
                  {(LICS} '89), Pacific Grove, California, USA, June 5-8, 1989},
  pages        = {353--362},
  publisher    = {{IEEE} Computer Society},
  year         = {1989},
  url          = {https://doi.org/10.1109/LICS.1989.39190},
  doi          = {10.1109/LICS.1989.39190},
  bibsource    = {dblp computer science bibliography, https://dblp.org}
}

@inproceedings{Pas+23,
  author       = {Corina S. Pasareanu and
                  Ravi Mangal and
                  Divya Gopinath and
                  Huafeng Yu},
  editor       = {Panagiotis Katsaros and
                  Laura Nenzi},
  title        = {Assumption Generation for Learning-Enabled Autonomous Systems},
  booktitle    = {Runtime Verification - 23rd International Conference, {RV} 2023, Thessaloniki,
                  Greece, October 3-6, 2023, Proceedings},
  series       = {Lecture Notes in Computer Science},
  volume       = {14245},
  pages        = {3--22},
  publisher    = {Springer},
  year         = {2023},
  url          = {https://doi.org/10.1007/978-3-031-44267-4\_1},
  doi          = {10.1007/978-3-031-44267-4\_1},
  bibsource    = {dblp computer science bibliography, https://dblp.org}
}

@article{Pas+18,
  author       = {Corina S. Pasareanu and
                  Divya Gopinath and
                  Huafeng Yu},
  title        = {Compositional Verification for Autonomous Systems with Deep Learning
                  Components},
  journal      = {CoRR},
  volume       = {abs/1810.08303},
  year         = {2018},
  url          = {http://arxiv.org/abs/1810.08303},
  eprinttype    = {arXiv},
  eprint       = {1810.08303},
  bibsource    = {dblp computer science bibliography, https://dblp.org}
}

@inproceedings{Cal+12,
  author       = {Radu Calinescu and
                  Shinji Kikuchi and
                  Kenneth Johnson},
  editor       = {Radu Calinescu and
                  David Garlan},
  title        = {Compositional Reverification of Probabilistic Safety Properties for
                  Large-Scale Complex {IT} Systems},
  booktitle    = {Large-Scale Complex {IT} Systems. Development, Operation and Management
                  - 17th Monterey Workshop 2012, Oxford, UK, March 19-21, 2012, Revised
                  Selected Papers},
  series       = {Lecture Notes in Computer Science},
  volume       = {7539},
  pages        = {303--329},
  publisher    = {Springer},
  year         = {2012},
  url          = {https://doi.org/10.1007/978-3-642-34059-8\_16},
  doi          = {10.1007/978-3-642-34059-8\_16},
  bibsource    = {dblp computer science bibliography, https://dblp.org}
}

@inproceedings{Bou+16,
  author       = {Redouane Bouchekir and
                  Sa{\"{\i}}da Boukhedouma and
                  Mohand Cherif Boukala},
  editor       = {Yuri Merkuryev and
                  Tuncer I. {\"{O}}ren and
                  Mohammad S. Obaidat},
  title        = {Automatic Compositional Verification of Probabilistic Safety Properties
                  for Inter-organisational Workflow Processes},
  booktitle    = {Proceedings of the 6th International Conference on Simulation and
                  Modeling Methodologies, Technologies and Applications {(SIMULTECH}
                  2016), Lisbon, Portugal, July 29-31, 2016},
  pages        = {244--253},
  publisher    = {SciTePress},
  year         = {2016},
  url          = {https://doi.org/10.5220/0005978602440253},
  doi          = {10.5220/0005978602440253},
  bibsource    = {dblp computer science bibliography, https://dblp.org}
}

@article{He+16,
  author       = {Fei He and
                  Xiaowei Gao and
                  Miaofei Wang and
                  Bow{-}Yaw Wang and
                  Lijun Zhang},
  title        = {Learning Weighted Assumptions for Compositional Verification of {M}arkov
                  Decision Processes},
  journal      = {{ACM} Trans. Softw. Eng. Methodol.},
  volume       = {25},
  number       = {3},
  pages        = {21:1--21:39},
  year         = {2016},
  url          = {https://doi.org/10.1145/2907943},
  doi          = {10.1145/2907943},
  bibsource    = {dblp computer science bibliography, https://dblp.org}
}

@article{BB18Interval,
  author       = {Redouane Bouchekir and
                  Mohand Cherif Boukala},
  title        = {Learning-based symbolic assume-guarantee reasoning for {M}arkov decision
                  process by using interval {M}arkov process},
  journal      = {Innov. Syst. Softw. Eng.},
  volume       = {14},
  number       = {3},
  pages        = {229--244},
  year         = {2018},
  url          = {https://doi.org/10.1007/s11334-018-0316-7},
  doi          = {10.1007/S11334-018-0316-7},
  bibsource    = {dblp computer science bibliography, https://dblp.org}
}

@article{Chat+15,
  author       = {Krishnendu Chatterjee and
                  Martin Chmelik and
                  Przemyslaw Daca},
  title        = {{CEGAR} for compositional analysis of qualitative properties in {M}arkov
                  decision processes},
  journal      = {Formal Methods Syst. Des.},
  volume       = {47},
  number       = {2},
  pages        = {230--264},
  year         = {2015},
  url          = {https://doi.org/10.1007/s10703-015-0235-2},
  doi          = {10.1007/S10703-015-0235-2},
  bibsource    = {dblp computer science bibliography, https://dblp.org}
}

@phdthesis{Zha09,
  author       = {Lijun Zhang},
  title        = {Decision algorithms for probabilistic simulations},
  school       = {Saarland University, Saarbr{\"{u}}cken, Germany},
  year         = {2009},
  url          = {http://scidok.sulb.uni-saarland.de/volltexte/2009/2424/},
  urn          = {urn:nbn:de:bsz:291-scidok-24244},
  bibsource    = {dblp computer science bibliography, https://dblp.org}
}

@article{CV10,
  author       = {Rohit Chadha and
                  Mahesh Viswanathan},
  title        = {A counterexample-guided abstraction-refinement framework for {M}arkov
                  decision processes},
  journal      = {{ACM} Trans. Comput. Log.},
  volume       = {12},
  number       = {1},
  pages        = {1:1--1:49},
  year         = {2010},
  url          = {https://doi.org/10.1145/1838552.1838553},
  doi          = {10.1145/1838552.1838553},
  bibsource    = {dblp computer science bibliography, https://dblp.org}
}

@article{Hah+19_interval,
  author       = {Ernst Moritz Hahn and
                  Vahid Hashemi and
                  Holger Hermanns and
                  Morteza Lahijanian and
                  Andrea Turrini},
  title        = {Interval {M}arkov Decision Processes with Multiple Objectives: From
                  Robust Strategies to {P}areto Curves},
  journal      = {{ACM} Trans. Model. Comput. Simul.},
  volume       = {29},
  number       = {4},
  pages        = {27:1--27:31},
  year         = {2019},
  url          = {https://doi.org/10.1145/3309683},
  doi          = {10.1145/3309683},
  bibsource    = {dblp computer science bibliography, https://dblp.org}
}

@article{Hah+11,
  author       = {Ernst Moritz Hahn and
                  Holger Hermanns and
                  Lijun Zhang},
  title        = {Probabilistic reachability for parametric {M}arkov models},
  journal      = {Int. J. Softw. Tools Technol. Transf.},
  volume       = {13},
  number       = {1},
  pages        = {3--19},
  year         = {2011},
  url          = {https://doi.org/10.1007/s10009-010-0146-x},
  doi          = {10.1007/S10009-010-0146-X},
  bibsource    = {dblp computer science bibliography, https://dblp.org}
}

@inproceedings{Cla+00,
  author       = {Edmund M. Clarke and
                  Orna Grumberg and
                  Somesh Jha and
                  Yuan Lu and
                  Helmut Veith},
  editor       = {E. Allen Emerson and
                  A. Prasad Sistla},
  title        = {Counterexample-Guided Abstraction Refinement},
  booktitle    = {Computer Aided Verification, 12th International Conference, {CAV}
                  2000, Chicago, IL, USA, July 15-19, 2000, Proceedings},
  series       = {Lecture Notes in Computer Science},
  volume       = {1855},
  pages        = {154--169},
  publisher    = {Springer},
  year         = {2000},
  url          = {https://doi.org/10.1007/10722167\_15},
  doi          = {10.1007/10722167\_15},
  bibsource    = {dblp computer science bibliography, https://dblp.org}
}

@article{JKPW21,
  author       = {Sebastian Junges and
                  Joost{-}Pieter Katoen and
                  Guillermo A. P{\'{e}}rez and
                  Tobias Winkler},
  title        = {The complexity of reachability in parametric {M}arkov decision processes},
  journal      = {J. Comput. Syst. Sci.},
  volume       = {119},
  pages        = {183--210},
  doi          = {10.1016/J.JCSS.2021.02.006},
  year         = {2021}
}

@inproceedings{Kom+12,
  author       = {Anvesh Komuravelli and
                  Corina S. Pasareanu and
                  Edmund M. Clarke},
  editor       = {P. Madhusudan and
                  Sanjit A. Seshia},
  title        = {Assume-Guarantee Abstraction Refinement for Probabilistic Systems},
  booktitle    = {Computer Aided Verification - 24th International Conference, {CAV}
                  2012, Berkeley, CA, USA, July 7-13, 2012 Proceedings},
  series       = {Lecture Notes in Computer Science},
  volume       = {7358},
  pages        = {310--326},
  publisher    = {Springer},
  year         = {2012},
  url          = {https://doi.org/10.1007/978-3-642-31424-7\_25},
  doi          = {10.1007/978-3-642-31424-7\_25},
  bibsource    = {dblp computer science bibliography, https://dblp.org}
}

@inproceedings{LL19,
  author       = {Rui Li and
                  Yang Liu},
  title        = {Compositional Stochastic Model Checking Probabilistic Automata via
                  Symmetric Assume-Guarantee Rule},
  booktitle    = {17th {IEEE} International Conference on Software Engineering Research,
                  Management and Applications, {SERA} 2019, Honolulu, HI, USA, May 29-31,
                  2019},
  pages        = {110--115},
  publisher    = {{IEEE}},
  year         = {2019},
  url          = {https://doi.org/10.1109/SERA.2019.8886808},
  doi          = {10.1109/SERA.2019.8886808},
  bibsource    = {dblp computer science bibliography, https://dblp.org}
}

@inproceedings{Fen+11,
  author       = {Lu Feng and
                  Marta Z. Kwiatkowska and
                  David Parker},
  editor       = {Dimitra Giannakopoulou and
                  Fernando Orejas},
  title        = {Automated Learning of Probabilistic Assumptions for Compositional
                  Reasoning},
  booktitle    = {Fundamental Approaches to Software Engineering - 14th International
                  Conference, {FASE} 2011, Held as Part of the Joint European Conferences
                  on Theory and Practice of Software, {ETAPS} 2011, Saarbr{\"{u}}cken,
                  Germany, March 26-April 3, 2011. Proceedings},
  series       = {Lecture Notes in Computer Science},
  volume       = {6603},
  pages        = {2--17},
  publisher    = {Springer},
  year         = {2011},
  url          = {https://doi.org/10.1007/978-3-642-19811-3\_2},
  doi          = {10.1007/978-3-642-19811-3\_2},
  bibsource    = {dblp computer science bibliography, https://dblp.org}
}

@inproceedings{Fen+10,
  author       = {Lu Feng and
                  Marta Z. Kwiatkowska and
                  David Parker},
  title        = {Compositional Verification of Probabilistic Systems Using Learning},
  booktitle    = {{QEST} 2010, Seventh International Conference on the Quantitative
                  Evaluation of Systems, Williamsburg, Virginia, USA, 15-18 September
                  2010},
  pages        = {133--142},
  publisher    = {{IEEE} Computer Society},
  year         = {2010},
  url          = {https://doi.org/10.1109/QEST.2010.24},
  doi          = {10.1109/QEST.2010.24},
  bibsource    = {dblp computer science bibliography, https://dblp.org}
}

@article{Gup+08,
  author       = {Anubhav Gupta and
                  Kenneth L. McMillan and
                  Zhaohui Fu},
  title        = {Automated assumption generation for compositional verification},
  journal      = {Formal Methods Syst. Des.},
  volume       = {32},
  number       = {3},
  pages        = {285--301},
  year         = {2008},
  url          = {https://doi.org/10.1007/s10703-008-0050-0},
  doi          = {10.1007/S10703-008-0050-0},
  bibsource    = {dblp computer science bibliography, https://dblp.org}
}

@inproceedings{Kwi+10,
  author       = {Marta Z. Kwiatkowska and
                  Gethin Norman and
                  David Parker and
                  Hongyang Qu},
  editor       = {Javier Esparza and
                  Rupak Majumdar},
  title        = {Assume-Guarantee Verification for Probabilistic Systems},
  booktitle    = {Tools and Algorithms for the Construction and Analysis of Systems,
                  16th International Conference, {TACAS} 2010, Held as Part of the Joint
                  European Conferences on Theory and Practice of Software, {ETAPS} 2010,
                  Paphos, Cyprus, March 20-28, 2010. Proceedings},
  series       = {Lecture Notes in Computer Science},
  volume       = {6015},
  pages        = {23--37},
  publisher    = {Springer},
  year         = {2010},
  url          = {https://doi.org/10.1007/978-3-642-12002-2\_3},
  doi          = {10.1007/978-3-642-12002-2\_3},
  bibsource    = {dblp computer science bibliography, https://dblp.org}
}

@article{Ete+08,
  author       = {Kousha Etessami and
                  Marta Z. Kwiatkowska and
                  Moshe Y. Vardi and
                  Mihalis Yannakakis},
  title        = {Multi-Objective Model Checking of {M}arkov Decision Processes},
  journal      = {Log. Methods Comput. Sci.},
  volume       = {4},
  number       = {4},
  year         = {2008},
  url          = {https://doi.org/10.2168/LMCS-4(4:8)2008},
  doi          = {10.2168/LMCS-4(4:8)2008},
  bibsource    = {dblp computer science bibliography, https://dblp.org}
}

@inproceedings{Alf+01,
  author       = {Luca de Alfaro and
                  Thomas A. Henzinger and
                  Ranjit Jhala},
  editor       = {Kim Guldstrand Larsen and
                  Mogens Nielsen},
  title        = {Compositional Methods for Probabilistic Systems},
  booktitle    = {{CONCUR} 2001 - Concurrency Theory, 12th International Conference,
                  Aalborg, Denmark, August 20-25, 2001, Proceedings},
  series       = {Lecture Notes in Computer Science},
  volume       = {2154},
  pages        = {351--365},
  publisher    = {Springer},
  year         = {2001},
  url          = {https://doi.org/10.1007/3-540-44685-0\_24},
  doi          = {10.1007/3-540-44685-0\_24},
  bibsource    = {dblp computer science bibliography, https://dblp.org}
}

@article{Fre+22,
  author       = {Hadar Frenkel and
                  Orna Grumberg and
                  Corina S. Pasareanu and
                  Sarai Sheinvald},
  title        = {Assume, guarantee or repair: a regular framework for non regular properties},
  journal      = {Int. J. Softw. Tools Technol. Transf.},
  volume       = {24},
  number       = {5},
  pages        = {667--689},
  year         = {2022},
  url          = {https://doi.org/10.1007/s10009-022-00669-9},
  doi          = {10.1007/S10009-022-00669-9},
  bibsource    = {dblp computer science bibliography, https://dblp.org}
}

@article{Elk+18,
  author       = {Karam Abd Elkader and
                  Orna Grumberg and
                  Corina S. Pasareanu and
                  Sharon Shoham},
  title        = {Automated circular assume-guarantee reasoning},
  journal      = {Formal Aspects Comput.},
  volume       = {30},
  number       = {5},
  pages        = {571--595},
  year         = {2018},
  url          = {https://doi.org/10.1007/s00165-017-0436-0},
  doi          = {10.1007/S00165-017-0436-0},
  bibsource    = {dblp computer science bibliography, https://dblp.org}
}

@inproceedings{Cha+05,
  author       = {Sagar Chaki and
                  Edmund M. Clarke and
                  Nishant Sinha and
                  Prasanna Thati},
  editor       = {Kousha Etessami and
                  Sriram K. Rajamani},
  title        = {Automated Assume-Guarantee Reasoning for Simulation Conformance},
  booktitle    = {Computer Aided Verification, 17th International Conference, {CAV}
                  2005, Edinburgh, Scotland, UK, July 6-10, 2005, Proceedings},
  series       = {Lecture Notes in Computer Science},
  volume       = {3576},
  pages        = {534--547},
  publisher    = {Springer},
  year         = {2005},
  url          = {https://doi.org/10.1007/11513988\_51},
  doi          = {10.1007/11513988\_51},
  bibsource    = {dblp computer science bibliography, https://dblp.org}
}

@article{Pas+08,
  author       = {Corina S. Pasareanu and
                  Dimitra Giannakopoulou and
                  Mihaela Gheorghiu Bobaru and
                  Jamieson M. Cobleigh and
                  Howard Barringer},
  title        = {Learning to divide and conquer: applying the {L}* algorithm to automate
                  assume-guarantee reasoning},
  journal      = {Formal Methods Syst. Des.},
  volume       = {32},
  number       = {3},
  pages        = {175--205},
  year         = {2008},
  url          = {https://doi.org/10.1007/s10703-008-0049-6},
  doi          = {10.1007/S10703-008-0049-6},
  bibsource    = {dblp computer science bibliography, https://dblp.org}
}

@inproceedings{Bob+08,
  author       = {Mihaela Gheorghiu Bobaru and
                  Corina S. Pasareanu and
                  Dimitra Giannakopoulou},
  editor       = {Aarti Gupta and
                  Sharad Malik},
  title        = {Automated Assume-Guarantee Reasoning by Abstraction Refinement},
  booktitle    = {Computer Aided Verification, 20th International Conference, {CAV}
                  2008, Princeton, NJ, USA, July 7-14, 2008, Proceedings},
  series       = {Lecture Notes in Computer Science},
  volume       = {5123},
  pages        = {135--148},
  publisher    = {Springer},
  year         = {2008},
  url          = {https://doi.org/10.1007/978-3-540-70545-1\_14},
  doi          = {10.1007/978-3-540-70545-1\_14},
  bibsource    = {dblp computer science bibliography, https://dblp.org}
}

@inproceedings{Cob+03,
  author       = {Jamieson M. Cobleigh and
                  Dimitra Giannakopoulou and
                  Corina S. Pasareanu},
  editor       = {Hubert Garavel and
                  John Hatcliff},
  title        = {Learning Assumptions for Compositional Verification},
  booktitle    = {Tools and Algorithms for the Construction and Analysis of Systems,
                  9th International Conference, {TACAS} 2003, Held as Part of the Joint
                  European Conferences on Theory and Practice of Software, {ETAPS} 2003,
                  Warsaw, Poland, April 7-11, 2003, Proceedings},
  series       = {Lecture Notes in Computer Science},
  volume       = {2619},
  pages        = {331--346},
  publisher    = {Springer},
  year         = {2003},
  url          = {https://doi.org/10.1007/3-540-36577-X\_24},
  doi          = {10.1007/3-540-36577-X\_24},
  bibsource    = {dblp computer science bibliography, https://dblp.org}
}

@inproceedings{Pnu84,
  author       = {Amir Pnueli},
  editor       = {Krzysztof R. Apt},
  title        = {In Transition From Global to Modular Temporal Reasoning about Programs},
  booktitle    = {Logics and Models of Concurrent Systems - Conference proceedings,
                  Colle-sur-Loup (near Nice), France, 8-19 October 1984},
  series       = {{NATO} {ASI} Series},
  volume       = {13},
  pages        = {123--144},
  publisher    = {Springer},
  year         = {1984},
  url          = {https://doi.org/10.1007/978-3-642-82453-1\_5},
  doi          = {10.1007/978-3-642-82453-1\_5},
  bibsource    = {dblp computer science bibliography, https://dblp.org}
}

@inproceedings{Car+11,
  author       = {Luca Cardelli and
                  Kim G. Larsen and
                  Radu Mardare},
  editor       = {Luca Aceto and
                  Monika Henzinger and
                  Jir{\'{\i}} Sgall},
  title        = {Modular {M}arkovian Logic},
  booktitle    = {Automata, Languages and Programming - 38th International Colloquium,
                  {ICALP} 2011, Zurich, Switzerland, July 4-8, 2011, Proceedings, Part
                  {II}},
  series       = {Lecture Notes in Computer Science},
  volume       = {6756},
  pages        = {380--391},
  publisher    = {Springer},
  year         = {2011},
  url          = {https://doi.org/10.1007/978-3-642-22012-8\_30},
  doi          = {10.1007/978-3-642-22012-8\_30},
  bibsource    = {dblp computer science bibliography, https://dblp.org}
}

@phdthesis{Wil15,
  author       = {Clemens Wiltsche},
  title        = {Assume-guarantee strategy synthesis for stochastic games},
  school       = {University of Oxford, {UK}},
  year         = {2015},
  bibsource    = {dblp computer science bibliography, https://dblp.org}
}

@inproceedings{Kwi+11,
  author       = {Marta Z. Kwiatkowska and
                  Gethin Norman and
                  David Parker},
  editor       = {Ganesh Gopalakrishnan and
                  Shaz Qadeer},
  title        = {{PRISM} 4.0: Verification of Probabilistic Real-Time Systems},
  booktitle    = {Computer Aided Verification - 23rd International Conference, {CAV}
                  2011, Snowbird, UT, USA, July 14-20, 2011. Proceedings},
  series       = {Lecture Notes in Computer Science},
  volume       = {6806},
  pages        = {585--591},
  publisher    = {Springer},
  year         = {2011},
  url          = {https://doi.org/10.1007/978-3-642-22110-1\_47},
  doi          = {10.1007/978-3-642-22110-1\_47},
  bibsource    = {dblp computer science bibliography, https://dblp.org}
}

    \appendix
    \section{Omitted Proofs}
    \label{app:proofs}
    \subsection{Proofs of~\protect\Cref{sec:strat_projections}}
\label{app:proofs_of_strat_projections}

\projectionAlternative*
\begin{proof}
Observe that the set $\bigcup_{\ppath \in (\liftedpaths{(\ppath_i,\alpha_i,s_i)}{\pa_{3-i}})} \cyl(\ppath)$ can be partitioned into disjoint subsets $\cyl\big(\ppath, (\hat\alpha_1,\hat\alpha_2), (s_1,s_2)\big)$ with $\ppath \in (\liftedpaths{\ppath_i}{\pa_{3-i}})$, $\hat\alpha_i = \alpha_i$, and $s_{3-i} \in \stateSetOf{3-i}$.
This yields
\begin{align*}
&~\sum_{s_i \in \stateSetOf{i}} \PrOf{\pa}{\strategy}{\liftedpaths{(\ppath_i,\alpha_i,s_i)}{\pa_{3-i}}}\\
=&~ \sum_{s_i \in \stateSetOf{i}} \bigg(\sum_{\ppath \in (\liftedpaths{\ppath_i}{\pa_{3-i}})} ~\sum_{\substack{(\hat\alpha_1,\hat\alpha_2) \in \actSetOf{\parallel},\\\hat\alpha_i = \alpha_i}}~ \sum_{\substack{s_{3-i} \in \stateSetOf{3-i}}}  \PrOf{\pa}{\strategy}{\ppath,(\hat\alpha_1,\hat\alpha_2),(s_1,s_2)}\bigg)\\
=&~ \sum_{\ppath \in (\liftedpaths{\ppath_i}{\pa_{3-i}})} ~\sum_{\substack{(\hat\alpha_1,\hat\alpha_2) \in \actSetOf{\parallel},\\\hat\alpha_i = \alpha_i}} \PrOf{\pa}{\strategy}{\ppath} \cdot  \strategy(\ppath,(\hat\alpha_1,\hat\alpha_2)) \cdot \overbrace{\sum_{\substack{s_1 \in \stateSetOf{1}\\s_2 \in \stateSetOf{2}}} 
	  \transFctOf{\parallel}\big(\last{\ppath}, (\hat\alpha_1, \hat\alpha_2), (s_1, s_2)\big)}^{=1}.
\end{align*}
\end{proof}

\projectionMeasure*
\begin{proof}
       We show the statement by induction over the length of a path $\ppath_i \in \finPathsOf{\pa_i}$. 
        For $\vert \ppath_i \vert =0$, i.e., $\ppath_i =\initialOf{i} \in \finPathsOf{\pa_i}$, we have $\PrOf{\pa_i}{{\stratProjOfToValuation{\strategy}{i}{}}}{\initialOf{i}} = 1 = \PrOf{\pa}{\strategy}{\liftedpaths{\initialOf{i}}{\pa_{3-i} }}$.  
For the induction step, assume that the statement holds for $\ppath_i'\in \finPathsOf{\pa_i}$ and consider $\ppath_i =  \ppath_i', \alpha_i, s_i \in \finPathsOf{\pa_i}$.
If $\PrOf{\pa}{\strategy}{\liftedpaths{\ppath'_i}{\pa_{3-i}}} = 0$, the induction hypothesis implies $\PrOf{\pa_i}{\stratProjOfToValuation{\strategy}{i}{}}{\ppath'_i} = 0$, yielding $\PrOf{\pa_i}{\stratProjOfToValuation{\strategy}{i}{}}{\ppath_i} = 0=\PrOf{\pa}{\strategy}{\liftedpaths{\ppath_i}{\pa_{3-i}}}$.
Otherwise, we have $\PrOf{\pa}{\strategy}{\liftedpaths{\ppath'_i}{\pa_{3-i}}} > 0$ and get:
        \begin{align*}
         &~   \PrOf{\pa_i}{{\stratProjOfToValuation{\strategy}{i}{}}}{\ppath_i}\\
       ~= &~  \PrOf{\pa_i}{{\stratProjOfToValuation{\strategy}{i}{}}}{\ppath_i'} \cdot \stratProjOfToValuation{\strategy}{i}{}(\ppath_i' , \alpha_i) \cdot \transFctOf{i}(\last{\ppath_i'}, \alpha_i, s_i)\\
       = &~  \PrOf{\pa}{\strategy}{\liftedpaths{\ppath'_i}{\pa_{3-i}}} \cdot \stratProjOfToValuation{\strategy}{i}{}(\ppath_i' , \alpha_i) \cdot \transFctOf{i}(\last{\ppath_i'}, \alpha_i, s_i) \tag{By induction hypothesis}\\
       = &~  \sum_{\hat{s}_i \in \stateSetOf{i}} \PrOf{\pa}{\strategy}{\liftedpaths{(\ppath'_i,\alpha_i,\hat{s}_i)}{\pa_{3-i}}}  \cdot \transFctOf{i}(\last{\ppath_i'}, \alpha_i, s_i) \tag{By \Cref{def:projectionStrategyPA}, simplify}\\
       = &  \sum_{\ppath' \in (\liftedpaths{\ppath'_i}{\pa_{3-i}})}\, \sum_{\substack{(\hat\alpha_1,\hat\alpha_2) \in \actSetOf{\parallel},\\\hat\alpha_i = \alpha_i}}  \PrOf{\pa}{\strategy}{\ppath'} \cdot \strategy(\ppath',(\hat\alpha_1,\hat\alpha_2)) \cdot \transFctOf{i}(\last{\ppath_i'}, \alpha_i, s_i) \tag{By \Cref{lem:projection:alternative}}\\
      = &~ \PrOf{\pa}{\strategy}{\liftedpaths{(\ppath'_i,\alpha_i,s_i)}{\pa_{3-i}}}
	 = \PrOf{\pa}{\strategy}{\liftedpaths{\ppath}{\pa_{3-i}}}.
       \end{align*}
\end{proof}

\changeValuationProjection*
\begin{proof}
We show $\stratProjOfToValuation{\strategy}{1}{\valuation_1, \valuation_2} = \stratProjOfToValuation{\strategy}{1}{\valuation_1', \valuation_2}$. 
The proof for $\stratProjOfToValuation{\strategy}{2}{\valuation_1, \valuation_2} = \stratProjOfToValuation{\strategy}{2}{\valuation_1, \valuation_2'}$ is symmetric.
For simplicity, we write $\ppa[\valuation_1,\valuation_2] = \ppa_1[\valuation_1] \parallel \ppa_2[\valuation_2]$.
Let $\ppath_1 \in \finPathsOf{\ppa_1}$ and $\alpha_1 \in \actSetOf{1}$.
Our restrictions for $\valuation_1$ and $\valuation_1'$ yield $\PrOf{\ppa[\valuation_1,\valuation_2]}{\strategy}{\liftedpaths{\ppath_1}{\ppa_2}} = 0$ iff $\PrOf{\ppa[\valuation_1',\valuation_2]}{\strategy}{\liftedpaths{\ppath_1}{\ppa_2}} = 0$.
Consequently, $\PrOf{\ppa[\valuation_1,\valuation_2]}{\strategy}{\liftedpaths{\ppath_1}{\ppa_2}} = 0$ implies $\stratProjOfToValuation{\strategy}{1}{\valuation_1, \valuation_2}(\ppath_1, \alpha_1) = 0 = \stratProjOfToValuation{\strategy}{1}{\valuation_1', \valuation_2}(\ppath_1, \alpha_1)$.
Otherwise, for $\hat\valuation_1 \in \set{\valuation_1,\valuation_2}$ we have by \Cref{def:projectionStrategyPA}: 
\[
\stratProjOfToValuation{\strategy}{1}{\hat\valuation_1, \valuation_2}(\ppath_1, \alpha_1) =
\frac{\sum_{s_1 \in \stateSetOf{1}}	\PrOf{\ppa[\hat\valuation_1,\valuation_2]}{\strategy}{\liftedpaths{(\ppath_1,\alpha_1,s_1)}{\ppa_{2}}}}{\PrOf{\ppa[\hat\valuation_1,\valuation_2]}{\strategy}{\liftedpaths{\ppath_1}{\ppa_{2}}}}.
\]
We show that all transition probabilities $\transFctOf{1}(s_1,\alpha_1,s_1')[\hat\valuation_1]$ cancel out in the above fraction, i.e., $\stratProjOfToValuation{\strategy}{1}{\hat\valuation_1, \valuation_2}(\ppath_1, \alpha_1)$ is independent of the valuation $\hat\valuation_1 \in \set{\valuation_1,\valuation_2}$, which concludes the proof.

We first introduce some auxiliary notations.
For $\ppath \in \finPathsOf{\ppa}$, let $\pref{\ppath, m}  \in \finPathsOf{\ppa}$ be the prefix of $\ppath$ of length $m\le |\ppath|$. The set of minimal paths $\ppath$ with projection $\restrOfTo{\ppath}{1} = \ppath_1$ is
\[\min\set{\liftedpaths{\ppath_1}{\ppa_2}} = \bigset{\ppath \in (\liftedpaths{\ppath_1}{\ppa_2}) \mathbin{\big|} \pref{\ppath, m} \notin (\liftedpaths{\ppath_1}{\ppa_2}) \text{ for all } m<|\ppath| }.\]
Any path $\ppath = (s_1^0,s_2^0),(\alpha_1^0,\alpha_2^0),(s_1^1,s_2^1), \dots, (s_1^n,s_2^n) \in (\liftedpaths{\ppath_1}{\ppa_2})$ has a prefix $\pref{\ppath,m} \in \min\set{\liftedpaths{\ppath_1}{\ppa_2}}$ and the subsequent transitions all correspond to asynchronous steps of $\ppa_2$, i.e., $s_1^m = s_1^{m+1} = \dots = s_1^n$ and $\alpha_1^m, \alpha_1^{m+1}, \dots, \alpha_1^{n-1} \notin \actSetOf{1}$.
Consequently, if $\ppath$ is initial:
\begin{align*}
&~\PrOf{\ppa[\hat\valuation_1,\valuation_2]}{\strategy}{\ppath}\\
=&~ 
\prod_{\substack{0\le j<m\\\alpha_1^j \in \actSetOf{1}}}\! \transFctOf{1}(s_1^j,\alpha_1^j,s_1^{j+1})[\hat\valuation_1] ~\cdot
\!\prod_{\substack{0\le j<n\\\alpha_2^j \in \actSetOf{2}}}\! \transFctOf{2}(s_2^j,\alpha_2^j,s_2^{j+1})[\valuation_2] ~\cdot 
\!\prod_{\substack{0\le j<n}}\! \strategy(\pref{\ppath,j}, (\alpha_1^j,\alpha_2^j))\\
=&~
\transFctOf{1}(\ppath_1)[\hat\valuation_1] ~\cdot \!\prod_{\substack{0\le j<n\\\alpha_2^j \in \actSetOf{2}}}\! \transFctOf{2}(s_2^j,\alpha_2^j,s_2^{j+1})[\valuation_2] ~\cdot 
\!\prod_{\substack{0\le j<n}}\! \strategy(\pref{\ppath,j}, (\alpha_1^j,\alpha_2^j)),
\end{align*}
where we set 
$\transFctOf{1}(\ppath_1)[\hat\valuation_1] = \prod_{\substack{0\le j<k}}\! \transFctOf{1}(\hat{s}_1^j,\hat\alpha_1^j,\hat{s}_1^{j+1})[\hat\valuation_1] $ for $\ppath_1 = \hat{s}_1^0,\hat\alpha_1^0,\hat{s}_1^1,\dots,\hat{s}_1^k$.
This yields
\[
\frac{\PrOf{\ppa[\valuation_1,\valuation_2]}{\strategy}{\ppath}}{\transFctOf{1}(\ppath_1)[\valuation_1]}
 = \!\prod_{\substack{0\le j<n\\\alpha_2^j \in \actSetOf{2}}}\! \transFctOf{2}(s_2^j,\alpha_2^j,s_2^{j+1})[\valuation_2] ~\cdot 
 \!\prod_{\substack{0\le j<n}}\! \strategy(\pref{\ppath,j}, (\alpha_1^j,\alpha_2^j))
= \frac{\PrOf{\ppa[\valuation'_1,\valuation_2]}{\strategy}{\ppath}}{\transFctOf{1}(\ppath_1)[\valuation'_1]}.
\]
Finally, we combine the above observation with \Cref{lem:projection:alternative} to conclude
\begin{align*}
    \stratProjOfToValuation{\strategy}{1}{\valuation_1, \valuation_2}(\ppath_1, \alpha_1) 
    =&  \frac{\transFctOf{1}(\ppath_1)[\valuation_1]}{\transFctOf{1}(\ppath_1)[\valuation_1]} \cdot  \frac{\displaystyle\sum_{\ppath \in (\liftedpaths{\ppath_1}{\ppa_{2}})}~\sum_{\substack{(\alpha_1,\alpha_2) \in \actSetOf{\parallel}}} \PrOf{\ppa[\valuation_1,\valuation_2]}{\strategy}{\ppath} \cdot \strategy(\ppath,(\alpha_1,\alpha_2))}{\displaystyle\sum_{\ppath \in \min\set{\liftedpaths{\ppath_1}{\ppa_{2}}}} \PrOf{\ppa[\valuation_1,\valuation_2]}{\strategy}{\ppath}}  \\[12pt]
    =&~ \frac{\displaystyle\sum_{\ppath \in (\liftedpaths{\ppath_1}{\ppa_{2}})}~\sum_{\substack{(\alpha_1,\alpha_2) \in \actSetOf{\parallel}}} \dfrac{\PrOf{\ppa[\valuation_1,\valuation_2]}{\strategy}{\ppath}}{\transFctOf{1}(\ppath_1)[\valuation_1]} \cdot \strategy(\ppath,(\alpha_1,\alpha_2))}{\displaystyle\sum_{\ppath \in \min\set{\liftedpaths{\ppath_1}{\ppa_{2}}}} \dfrac{\PrOf{\ppa[\valuation_1,\valuation_2]}{\strategy}{\ppath}}{\transFctOf{1}(\ppath_1)[\valuation_1]}}
    =    \stratProjOfToValuation{\strategy}{1}{\valuation'_1, \valuation_2}(\ppath_1, \alpha_1).
\end{align*}

\end{proof}
    
\subsection{Proofs of \Cref{sec:verification}}
\label{app:proofs_of_verification}
We show that the verification of safety objectives reduces to maximal reachability properties in a pPA-DFA Product based on \cite[Lemma 1]{Kwi+13}. 
First, we lift the DFA product (see \cite[Definition 12]{Kwi+13}) to pPA. 
\begin{defi}
\label{def:pPA_dfa_product}
    Given a bad prefix automaton $\bpAutomatonOf{\regLang}$ with $\alphabetOf{\bpAutomatonOf{{\regLang}}} \subseteq  \alphabetOf{\ppa}$, 
    the pPA $\ppa_{\product} = (\ppa \otimes \bpAutomatonOf{\regLang})$ is defined as $\ppa_{\product} = \ppaTupleOf{\product}$, 
    where 
    \begin{itemize}
        \item $\stateSetOf{\product} = \stateSetOf{\ppa} \times \stateSetOf{\bpAutomatonOf{{\regLang}}}$, 
        \item $\initialOf{\product} = (\initialOf{\ppa}, \initialOf{\bpAutomatonOf{\regLang}})$, 
        \item $\parameterSetOf{\product} = \parameterSetOf{\ppa}$,
        \item $\actSetOf{\product} = \actSetOf{\ppa}$,
        \item for each $(s, \alpha) \in \domain(\transFctOf{\ppa})$, with $\syncOf{\ppa}(s, \alpha) = \lab \in \alphabetOf{\bpAutomatonOf{{\regLang}}}$ and $p \in \stateSetOf{\bpAutomatonOf{\regLang}}$ with $q = \transFctOf{\bpAutomatonOf{\regLang}}(p, \lab)$: 
            \[\transFctOf{\product}((s,q), \alpha) = \transFctOf{\ppa}(s, \alpha) \times \indicatorFct{}
                \quad \text{and} \quad \syncOf{\product}((s,q), \alpha) = \lab
            \]
            for each $(s, \alpha) \in \domain(\transFctOf{\ppa})$, with $\syncOf{\ppa}(s, \alpha) = \lab \not \in \alphabetOf{\bpAutomatonOf{{\regLang}}}$ and $p \in \stateSetOf{\bpAutomatonOf{\regLang}}$: 
            \[\transFctOf{\product}((s,q), \alpha) = \transFctOf{\ppa}(s, \alpha) \times \indicatorFct{p}
                \quad \text{and} \quad \syncOf{\product}((s,q), \alpha) = \lab
            \]
\end{itemize}
\end{defi}
It holds that $(\ppa {\otimes} \bpAutomatonOf{\regLang})[\valuation] = {\ppa[\valuation] \otimes\bpAutomatonOf{\regLang} }$. 
Using this, we can lift \cite[Lemma 1]{Kwi+13} to pPAs. 
\begin{lem}
    \label{theo:prob_safety_pPA_prop_reach_poduct_bijection}
    Let $\ppa$ be a pPA and let $\regLang$ be the language of a safety objective for ${\ppa}$ and let $\valuation$ be a well-defined valuation.  
    Let $\star \in \{\prt,\comp\}$. 
    \begin{itemize}	
        \item There is a bijection $f$ between the strategies
        $\strategysetOf{\ppa[\valuation]}{(\mless,)\star}$ of $\ppa$ and the strategies $\strategysetOf{\product[\valuation]}{(\mless,)\star}$ of pPA-DFA product $\product = \ppa {\otimes} \bpAutomatonOf{\regLang}$
        \item Let $bad_{\regLang} = \{ (s, q) \in \stateSetOf{\product} \mid q \in F \}$ = set of product-states that contain a final state of $\bpAutomatonOf{\regLang}$.  			
        $\strategy \in \strategysetOf{\ppa[\valuation]}{\star}$: 
        \[
            \PrOf{\ppa}{\valuation,\strategy}{\regLang} =  1-\PrOf{\product}{\valuation, f(\strategy)}{\Diamond bad_{\regLang}}
        \]
    \end{itemize}
\end{lem}	
\begin{proof}
    As stated in the proof of \cite[Lemma 1]{Kwi+13}, the statement follows by \cite[Theorem 10.51]{BK08}. 
    \cite[Theorem 10.51]{BK08} can be applied as safety objectives are prefix-closed; see \Cref{sec:explain_differences} for a detailed explanation. 
\end{proof}
\begin{rem}
    For actions that do not appear in $\regLang$ (and thus not appear in $\bpAutomatonOf{\regLang}$) the state of the automaton remains unchanged.  
    This means paths $\pi$ of $\ppa$ such that $\restrOfTo{\pi}{\regLang}$ is finit cannot reach the $bad_{\regLang}$-labeled state in $(\ppa \product \bpAutomatonOf{\regLang})$. 
    Thus, for correctness of \Cref{theo:prob_safety_pPA_prop_reach_poduct_bijection}, safety objectives need to be prefix-closed. 
\end{rem}

\restatableSafetyPartVsComplete*
\begin{proof}
    As $\region$ is well-defined, \Cref{theo:prob_safety_pPA_prop_reach_poduct_bijection} applies and 
    \Cref{lemma_safety_partial_vs_complete} follows analogously to the argument presented in \cite[Proposition 1]{Kwi+13}. 
\end{proof}

\restatableLemmaThreeNewPre*
\restatableLemmaThreeNew*
We give the proof for \Cref{theo:lemma3NewDependentPre} and \Cref{theo:lemma3NewDependent}: 
\begin{proof}   
    $\valuation$ is well-defined for $\ppa_1$ and $\ppa_2$. Thus, $\ppa_1[\valuation]$,  $\ppa_2[\valuation]$, and $(\ppa_1 \parallel \ppa_2)[\valuation] = (\ppa_1[\valuation]) \parallel (\ppa_2[\valuation])$ are PAs. 
    Then, the claims follow by \cite[Lemma 3]{Kwi+13}. 
\end{proof}

    
\subsection{Proofs of \Cref{sec:pag}}

\label{app:proofs_of_PAG}
\restatableAsymRule* 
\begin{proof}
    The proof is based on the proof provided for \cite[Theorem 1, Theorem 2]{Kwi+13}. 
    Since $\alphabetOf{\multiobjectiveQuery{(\safe)}{A}} \subseteq \alphabetOf{1}$ and $\alphabetOf{\multiobjectiveQuery{(\safe)}{G}} \subseteq \alphabetOf{\multiobjectiveQuery{(\safe)}{A}}\cup \alphabetOf{2}$ hold by assumption, it follows that 
    \begin{align*}
        \alphabetOf{\multiobjectiveQuery{(\safe)}{G}}  
        & \subseteq \alphabetOf{\multiobjectiveQuery{(\safe)}{A}} \cup \alphabetOf{2} \\
        & \subseteq \alphabetOf{1} \cup  \alphabetOf{2}  \\
        & = \alphabetOf{\ppa_1 \parallel \ppa_2}. 
    \end{align*} 
    Thus, ${\multiobjectiveQuery{(\safe)}{G}} $ is a valid multiobjective query for $\ppa_1 \parallel \ppa_2$. 
    If $\regionIntersectionOf{\region_{1}}{}{\region_2} = \emptyset$, the conclusion trivially holds. 
    We assume $\regionIntersectionOf{\region_{1}}{}{\region_2} \not = \emptyset$:  
    We show ${\ppa_1}\parallel \alphabetExtensionOfTo{\ppa_2}{\alphabetOf{\multiobjectiveQuery{(\safe)}{A}}}, \regionIntersectionOf{\region_{1}}{}{\region_2} \modelsWrt{\star} \multiobjectiveQuery{(\safe)}{G}$, 
    which directly implies that $\ppa_1 \parallel \ppa_2, \regionIntersectionOf{\region_{1}}{}{\region_2} \modelsWrt{\star} \multiobjectiveQuery{(\safe)}{G}$. 
     \begin{itemize}
         \item First, we prove correctness of the rule involving safety properties: 
        Assuming the premises hold, it follows that 
         \begin{itemize}
             \item $\ppa_1, \regionIntersectionOf{\region_{1}}{}{\region_2} \modelsWrt{\comp} \multiobjectiveQuery{\safe}{A}$, and 
             \item $\agTriple{\alphabetExtensionOfTo{\ppa_2}{\alphabetOf{\multiobjectiveQuery{\safe}{A}}}, \regionIntersectionOf{\region_{1}}{}{\region_2}  }{\prt}{\multiobjectiveQuery{\safe}{A}}{\multiobjectiveQuery{\safe}{G}}$. 
         \end{itemize}
         Let $\valuation \in \regionIntersectionOf{\region_{1}}{}{\region_2}$. 
         From $\ppa_1, \regionIntersectionOf{\region_{1}}{}{\region_2} \modelsWrt{\comp} \multiobjectiveQuery{\safe}{A}$ follows that $\ppa_1[\valuation] \modelsWrt{\comp} \multiobjectiveQuery{\safe}{A}$. 
         
         Let $\strategy \in \strategysetOf{( \ppa_1 \parallel \alphabetExtensionOfTo{\ppa_2}{\alphabetOf{\multiobjectiveQuery{(\safe)}{A}}})[\valuation]}{\comp}$. 
         The strategy $\stratProjOfToValuation{\strategy}{1}{\valuation}$ is a (partial) strategy of $\ppa_1[\valuation]$ (see \Cref{sec:strat_projections}). 
         By \Cref{lemma_safety_partial_vs_complete}: 
         \begin{align*}
             \ppa_1[\valuation] & \modelsWrt{\comp} \multiobjectiveQuery{\safe}{A} \\
             & \Rightarrow \ppa_1[\valuation], {\stratProjOfToValuation{\strategy}{1}{\valuation}} \modelsWrt{\prt} \multiobjectiveQuery{\safe}{A} \\
             &  \Rightarrow \bigl( \ppa_1 \parallel \alphabetExtensionOfTo{\ppa_2}{\alphabetOf{\multiobjectiveQuery{(\safe)}{A}}} \bigr)[\valuation], {\strategy} \modelsWrt{} \multiobjectiveQuery{\safe}{A} 
             && \text{By \Cref{lemma_3_dDependent} } \\
             & \Rightarrow { \alphabetExtensionOfTo{\ppa_2}{\alphabetOf{\multiobjectiveQuery{\safe}{A}} } }[\valuation],{ \stratProjOfToValuation{\strategy}{2}{\valuation} } \modelsWrt{}  \multiobjectiveQuery{\safe}{A}    
             && \text{By \Cref{lemma_3_fDependent}}\\
             & \Rightarrow  \bigl(\alphabetExtensionOfTo{\ppa_2}{\alphabetOf{\multiobjectiveQuery{\safe}{A}} }[\valuation]\bigr), {\stratProjOfToValuation{\strategy}{2}{\valuation}}   \modelsWrt{}  \multiobjectiveQuery{\safe}{G} 
             && \text{$ \stratProjOfToValuation{\strategy}{2}{\valuation} \in \strategysetOf{\alphabetExtensionOfTo{\ppa_2}{\alphabetOf{\multiobjectiveQuery{\safe}{A}} }}{\prt}$} \\
             & && \text{and $\agTriple{\alphabetExtensionOfTo{\ppa_2}{\alphabetOf{\multiobjectiveQuery{\safe}{A}}}, \regionIntersectionOf{\region_{1}}{}{\region_2} }{\prt}{\multiobjectiveQuery{\safe}{A}}{\multiobjectiveQuery{\safe}{G}}$} \\
             & \Rightarrow {(\ppa_1 \parallel \alphabetExtensionOfTo{\ppa_2}{\alphabetOf{\multiobjectiveQuery{(\safe)}{A}}})[\valuation]}, {\strategy} \modelsWrt{} \multiobjectiveQuery{\safe}{G} 
             && \text{By \Cref{lemma_3_fDependent}}   
         \end{align*}  
         \item Next, we prove correctness for the rule involving general properties and fair strategies: 
        Again, the premises imply that 
        \begin{itemize}
            \item ${\ppa_1, \regionIntersectionOf{\region_1}{}{\region_2} \modelsWrt{\fairWrtRegionModel{}{\decomp_1}} \multiobjectiveQuery{}{A}}$
            \item $\agTriple{\alphabetExtensionOfTo{\ppa_2}{\alphabetOf{\multiobjectiveQuery{}{A}}} , \regionIntersectionOf{\region_{1}}{}{\region_2}}{\fairWrtRegionModel{}{\decomp_2}}{\multiobjectiveQuery{}{A}}{\multiobjectiveQuery{}{G}}$
        \end{itemize}
        Let $\valuation \in \regionIntersectionOf{\region_1}{}{\region_2}$. 
        From ${\ppa_1, \regionIntersectionOf{\region_{1}}{}{\region_2} \modelsWrt{\fairWrtRegionModel{}{\decomp_1}} \multiobjectiveQuery{}{A}}$ follows $\ppa_1[\valuation] \modelsWrt{\fairWrtRegionModel{}{\decomp_1}} \multiobjectiveQuery{}{A}[\valuation]$. 

         Let $\strategy \in \strategysetOf{(\ppa_1 \parallel \ppa_2)[\valuation]}{\fairWrtRegionModel{}{\decomp_1 \cup \decomp_2}}$. 
         The strategy $\stratProjOfToValuation{\strategy}{1}{\valuation}$ is a complete strategy of $\ppa_1[\valuation]$ that is $\fairWrtRegionModel{}{\decomp_1}$ by \Cref{theo:strategy_projection_partial_fair_preserved}. 
         Thus, 
         \begin{align*}
             \ppa_1[\valuation] &  \modelsWrt{\fairWrtRegionModel{}{\decomp_1}} \multiobjectiveQuery{}{A}[\valuation]   \\
             & \Rightarrow \ppa_1[\valuation], {\stratProjOfToValuation{\strategy}{1}{\valuation}} \modelsWrt{} {\multiobjectiveQuery{}{A}[\valuation]}  \\
             &  \Rightarrow \left( \ppa_1 \parallel \alphabetExtensionOfTo{\ppa_2}{\alphabetOf{\multiobjectiveQuery{}{A}}}\right)[\valuation], {\strategy} \modelsWrt{}  {\multiobjectiveQuery{}{A}[\valuation]}
             && \text{By \Cref{lemma_3_eDependent} } \\
             & \Rightarrow {\alphabetExtensionOfTo{\ppa_2}{\alphabetOf{\multiobjectiveQuery{}{A}} }}[\valuation], { \stratProjOfToValuation{\strategy}{2}{\valuation}} \modelsWrt{}  {\multiobjectiveQuery{}{A}}[\valuation]
             && \text{By \Cref{lemma_3_fDependent}}\\
             & \Rightarrow  \alphabetExtensionOfTo{\ppa_2}{\alphabetOf{\multiobjectiveQuery{}{A}} }[\valuation], {\stratProjOfToValuation{\strategy}{2}{\valuation}}  \modelsWrt{}  {\multiobjectiveQuery{}{G}}[\valuation]
             && \text{$ \stratProjOfToValuation{\strategy}{2}{\valuation} \in \strategysetOf{{\alphabetExtensionOfTo{\ppa_2}{\alphabetOf{\multiobjectiveQuery{}{A}}}}[\valuation]}{\fairWrtRegionModel{}{\decomp_2}}$} \\
             & && \text{and $\agTriple{\alphabetExtensionOfTo{\ppa_2}{\alphabetOf{\multiobjectiveQuery{}{A}}},\regionIntersectionOf{\region_1}{}{\region_2}}{\fairWrtRegionModel{}{\decomp_2}}{\multiobjectiveQuery{}{A}}{\multiobjectiveQuery{}{G}}$} \\
             & \Rightarrow {\left(\ppa_1 \parallel \alphabetExtensionOfTo{\ppa_2}{\alphabetOf{\multiobjectiveQuery{}{A}}}\right)[\valuation]}, {\strategy} \modelsWrt{} \multiobjectiveQuery{}{G}[\valuation]
             && \text{By \Cref{lemma_3_fDependent}}   
         \end{align*}  

     \end{itemize}
 \end{proof}

\restatableCircRule*
\begin{proof}
    The proof is based on the proof of \cite[Theorem 5]{Kwi+13}. 
    We show the proof rule on the right. 
    For the rule on the left the proof works analogously. 
    
    Since $\alphabetOf{\multiobjectiveQuery{}{A}_2} \subseteq \alphabetOf{\multiobjectiveQuery{}{A}_1} \cup \alphabetOf{1} $ 
    and $\alphabetOf{\multiobjectiveQuery{}{G}} \subseteq  \alphabetOf{\multiobjectiveQuery{}{A}_2} \cup \alphabetOf{2}$, it follows that $\alphabetOf{\multiobjectiveQuery{}{G}}  \subseteq \alphabetOf{1} \cup \alphabetOf{2}$. 
    Thus, ${\multiobjectiveQuery{}{G}} $ is a valid multiobjective query for $\ppa_1 \parallel \ppa_2$. 

    If $\regionIntersectionOf{\region_{1}}{\region_2}{\region_3} = \emptyset$, the conclusion trivially holds. 
    We assume $\regionIntersectionOf{\region_{1}}{\region_2}{\region_3} \not = \emptyset$. 
    
    We show that $\alphabetExtensionOfTo{\ppa_1}{\alphabetOf{\multiobjectiveQuery{}{A}_1}}\parallel \alphabetExtensionOfTo{\ppa_2}{\alphabetOf{\multiobjectiveQuery{}{A}_2}}, \regionIntersectionOf{\region_{1}}{\region_2}{\region_3} \modelsWrt{\fairWrtRegionModel{}{\decomp_1 \cup \decomp_2 \cup \decomp_3}} \multiobjectiveQuery{}{G}$ holds,  
    which directly implies $\ppa_1 \parallel \ppa_2, \regionIntersectionOf{\region_{1}}{\region_2}{\region_3} \modelsWrt{\fairWrtRegionModel{}{\decomp_1 \cup \decomp_2 \cup \decomp_3}} \multiobjectiveQuery{}{G}$. 
    
    From the premises, it follows that: 
    \begin{itemize}
        \item $\agTriple{\alphabetExtensionOfTo{\ppa_1}{\alphabetOf{\multiobjectiveQuery{}{A}_{1}}}, \regionIntersectionOf{\region_1}{\region_2}{\region_3} }{\fairWrtRegionModel{}{\decomp_1 }}{\multiobjectiveQuery{}{A}_{1}}{\multiobjectiveQuery{}{A}_2} $,
        \item $ \agTriple{\alphabetExtensionOfTo{\ppa_2}{\alphabetOf{\multiobjectiveQuery{}{A}_{2}}}, \regionIntersectionOf{\region_1}{\region_2}{\region_3} }{\fairWrtRegionModel{}{ \decomp_2 }}{\multiobjectiveQuery{}{A}_{2}}{\multiobjectiveQuery{}{G}} $, and
        \item $\ppa_2, \regionIntersectionOf{\region_1}{\region_2}{\region_3} \modelsWrt{\fairWrtRegionModel{}{\decomp_3}} \multiobjectiveQuery{}{A}_1$.
    \end{itemize}
    Let $\valuation \in \regionIntersectionOf{\region_1}{\region_2}{\region_3}$. 
    From $\ppa_2, \regionIntersectionOf{\region_1}{\region_2}{\region_3} \modelsWrt{\fairWrtRegionModel{}{\ppa_2}} \multiobjectiveQuery{}{A}_1$ follows 
    $\ppa_2[\valuation] \modelsWrt{\fairWrtRegionModel{}{\decomp_3}}  \multiobjectiveQuery{}{A}_1[\valuation]$. 
    Together with $\alphabetOf{\multiobjectiveQuery{}{A}_1} \subseteq \alphabetOf{2}$, and $\decomp_3 \in \setOfdecompsOf{\alphabetOf{2}}$, we obtain that: 
    \begin{align}\label{eq:circproof}
        \alphabetExtensionOfTo{\ppa_2}{\alphabetOf{\multiobjectiveQuery{}{A}_2}}[\valuation] \modelsWrt{\fairWrtRegionModel{}{\decomp_3}} \multiobjectiveQuery{}{A}_1[\valuation]
    \end{align} 
    Let $\strategy \in \strategysetOf{(\alphabetExtensionOfTo{\ppa_1}{\alphabetOf{\multiobjectiveQuery{}{A}_1}} \parallel \alphabetExtensionOfTo{\ppa_2}{\alphabetOf{\multiobjectiveQuery{}{A}_2}})[\valuation]}{\fairWrtRegionModel{}{\decomp_1 \cup \decomp_2 \cup \decomp_3}}$. 
    From \Cref{theo:strategy_projection_partial_fair_preserved} and $\decomp_3 \in \setOfdecompsOf{\alphabetOf{2}} \subseteq \setOfdecompsOf{\alphabetOf{2} \cup \alphabetOf{\multiobjectiveQuery{}{A}_2}}$, follows that the strategy $\stratProjOfToValuation{\strategy}{2}{\valuation}$ is a $\fairWrtRegionModel{}{\decomp_3}$ strategy of $\alphabetExtensionOfTo{\ppa_2}{\alphabetOf{\multiobjectiveQuery{}{A}_2}}[\valuation]$. 
    Thus, by \Cref{eq:circproof} we obtain:
    \begin{align*}
        \alphabetExtensionOfTo{\ppa_2}{\alphabetOf{\multiobjectiveQuery{}{A}_2}}[\valuation]&  \modelsWrt{\fairWrtRegionModel{}{\decomp_3}} \multiobjectiveQuery{}{A}_1[\valuation] \\
        & \Rightarrow \alphabetExtensionOfTo{\ppa_2}{\alphabetOf{\multiobjectiveQuery{}{A}_2}}[\valuation], {\stratProjOfToValuation{\strategy}{2}{\valuation}} \modelsWrt{} {\multiobjectiveQuery{}{A}_1[\valuation]}\\
        %
        &  \Rightarrow \left(\alphabetExtensionOfTo{\ppa_1}{\alphabetOf{\multiobjectiveQuery{}{A}_1}}\parallel {\alphabetExtensionOfTo{\ppa_2}{\alphabetOf{\multiobjectiveQuery{}{A}_2}}}\right)[\valuation], {\strategy} \modelsWrt{}  \multiobjectiveQuery{}{A}_1[\valuation]
        && \text{By \Cref{lemma_3_eDependent} } \\
        %
        & \Rightarrow {\alphabetExtensionOfTo{\ppa_1}{\alphabetOf{\multiobjectiveQuery{}{A}_1} }[\valuation]}, { \stratProjOfToValuation{\strategy}{1}{\valuation}} \modelsWrt{}  \multiobjectiveQuery{}{A}_1[\valuation]
        && \text{By \Cref{lemma_3_fDependent}}
    \end{align*}
    We have that 
    $\stratProjOfToValuation{\strategy}{1}{\valuation}$ is a $\fairWrtRegionModel{}{\decomp_1}$ strategy of ${\alphabetExtensionOfTo{\ppa_1}{\alphabetOf{\multiobjectiveQuery{}{A}_1} }}[\valuation]$. 
    As a consequence, it holds that $\agTriple{\alphabetExtensionOfTo{\ppa_1}{\alphabetOf{\multiobjectiveQuery{}{A}_1}}, \region_1 \cap \region_2 \cap \region_3 }{\fairWrtRegionModel{}{\decomp_1}}{\multiobjectiveQuery{}{A}_1}{\multiobjectiveQuery{}{A}_2}$ from which follows that 
    \begin{align*}
            \alphabetExtensionOfTo{\ppa_1}{\alphabetOf{\multiobjectiveQuery{}{A}_1}}[\valuation]&, { \stratProjOfToValuation{\strategy}{1}{\valuation}} \modelsWrt{}  {\multiobjectiveQuery{}{A}_2[\valuation]} \\
        %
        & \Rightarrow {\left( \alphabetExtensionOfTo{\ppa_1}{\alphabetOf{\multiobjectiveQuery{}{A}_1}}\parallel \alphabetExtensionOfTo{\ppa_2}{\alphabetOf{\multiobjectiveQuery{}{A}_2}}\right)[\valuation]}, {\strategy} \modelsWrt{} {\multiobjectiveQuery{}{A}_2[\valuation]}
        && \text{By \Cref{lemma_3_fDependent}}   \\
        %
        %
        & \Rightarrow {\alphabetExtensionOfTo{\ppa_2}{\alphabetOf{\multiobjectiveQuery{}{A}_2} }[\valuation] }, {\stratProjOfToValuation{\strategy}{2}{\valuation} } \modelsWrt{}  \multiobjectiveQuery{}{A}_2[\valuation]
        && \text{By \Cref{lemma_3_fDependent}}\\
    \end{align*}  
    We apply the third premise and obtain that 
         $\agTriple{\alphabetExtensionOfTo{\ppa_2}{\alphabetOf{\multiobjectiveQuery{}{A}_2}}, \region_1 \cap \region_2 \cap \region_3 }{\fairWrtRegionModel{}{\decomp_2}}{\multiobjectiveQuery{}{A}_2}{\multiobjectiveQuery{}{G}}$. Then, as $\stratProjOfToValuation{\strategy}{2}{\valuation}$ is a $\fairWrtRegionModel{}{\decomp_2 }$ strategy of ${\alphabetExtensionOfTo{\ppa_2}{\alphabetOf{\multiobjectiveQuery{}{A}_2} }}[\valuation]$ it follows that:  
            \begin{align*}
            \alphabetExtensionOfTo{\ppa_2}{\alphabetOf{\multiobjectiveQuery{}{A}_2} }[\valuation]&, { \stratProjOfToValuation{\strategy}{2}{\valuation}} \modelsWrt{}  {\multiobjectiveQuery{}{G}[\valuation]}  \\
            %
            &  \Rightarrow {{\left(\alphabetExtensionOfTo{\ppa_1}{\alphabetOf{\multiobjectiveQuery{}{A}_1}}\parallel \alphabetExtensionOfTo{\ppa_2}{\alphabetOf{\multiobjectiveQuery{}{A}_2}}\right)[\valuation]}}, {\strategy} \modelsWrt{} \multiobjectiveQuery{}{G}[\valuation]
            && \text{By \Cref{lemma_3_fDependent}}   
        \end{align*}

\end{proof}

    \subsection{Proofs of \Cref{sec:pag_mono}}
\label{app:proofs_of_pag_mono}

\restatablMonoRule*
\begin{proof}
We continue the proof from the main paper, which already established the left rule for partial strategies. For fair strategies the proof is similar. However, to deduce that 
        \begin{align*}
            {\solutionFctMdpObjective{{\alphabetExtensionOfTo{\ppa_i}{\alphabetOf{}}},{{\stratProjOfToValuation{\strategy}{i}{\valuation_+, \valuation}}}}{}}{}(\valuation) 
            & \leq  {\solutionFctMdpObjective{{\alphabetExtensionOfTo{\ppa_i}{\alphabetOf{}}},{{\stratProjOfToValuation{\strategy}{i}{\valuation_+, \valuation}}}}{}}{}(\valuation_+),
        \end{align*}
        we need to ensure that we can apply $\monotonicOnRegionParameter{\budarrow}{\solutionFctMdpObjective{\alphabetExtensionOfTo{\ppa_i}{\alphabetOf{}}}{}}{p}{\region_i}{\fairWrtRegionModel{}{\decomp_i}}$. 
        Thus, we need to show that 
        ${\stratProjOfToValuation{\strategy}{i}{\valuation_+, \valuation}} \in  \strategysetOf{\alphabetExtensionOfTo{\ppa_i}{\alphabetOf{}}[{\valuation_+}]}{\fairWrtRegionModel{}{\decomp_i}}$ and 
        ${\stratProjOfToValuation{\strategy}{i}{\valuation_+, \valuation}} \in  \strategysetOf{\alphabetExtensionOfTo{\ppa_i}{\alphabetOf{}}[{\valuation}]}{\fairWrtRegionModel{}{\decomp_i}}$.

        The strategy $\strategy$ of $\alphabetExtensionOfTo{\ppa_1}{\alphabetOf{}} \parallel \alphabetExtensionOfTo{\ppa_2}{\alphabetOf{}}$ is $\fairWrtRegionModel{}{\decomp_1 \cup \decomp_2}$ w.r.t.\ the valuations $\valuation, \text{ and }\valuation_+ \in \regionIntersectionOf{\region_1}{}{\region_2}$. 
        By \Cref{theo:strategy_projection_partial_fair_preserved} follows that the projection  
        ${\stratProjOfToValuation{\strategy}{1}{\valuation, \valuation}}$ is in $\strategysetOf{\alphabetExtensionOfTo{\ppa_1}{\alphabetOf{}}[{\valuation}]}{\fairWrtRegionModel{}{\decomp_1}}$  
        and ${\stratProjOfToValuation{\strategy}{2}{\valuation_+, \valuation_+}}$ is in $\strategysetOf{\alphabetExtensionOfTo{\ppa_2}{\alphabetOf{}}[{\valuation_+}]}{\fairWrtRegionModel{}{\decomp_2}}$. 
        Since $\valuation$ and $\valuation_+$ are graph-preserving, we can apply \Cref{theo:change_valuation_projection} and thus, 
        ${\stratProjOfToValuation{\strategy}{1}{\valuation_+, \valuation}} \in  \strategysetOf{\alphabetExtensionOfTo{\ppa_1}{\alphabetOf{}}[{\valuation}]}{\fairWrtRegionModel{}{\decomp_1}}$  
        and ${\stratProjOfToValuation{\strategy}{2}{\valuation_+, \valuation}} \in  \strategysetOf{\alphabetExtensionOfTo{\ppa_2}{\alphabetOf{}}[{\valuation_+}]}{\fairWrtRegionModel{}{\decomp_2}}$. 
        Since $\valuation$ and $\valuation_+$ are graph-preserving for $\alphabetExtensionOfTo{\ppa_i}{\alphabetOf{}}$, we have 
        $\strategysetOf{\alphabetExtensionOfTo{\ppa_i}{\alphabetOf{}}[{\valuation}]}{\fairWrtRegionModel{}{\decomp_i}} = \strategysetOf{\alphabetExtensionOfTo{\ppa_i}{\alphabetOf{}}[{\valuation_+}]}{\fairWrtRegionModel{}{\decomp_i}}$ 
        from which we deduce that  
        ${\stratProjOfToValuation{\strategy}{1}{\valuation_+, \valuation}} \in  \strategysetOf{\alphabetExtensionOfTo{\ppa_1}{\alphabetOf{}}[{\valuation_+}]}{\fairWrtRegionModel{}{\decomp_1}}$  
        and ${\stratProjOfToValuation{\strategy}{2}{\valuation_+, \valuation}} \in  \strategysetOf{\alphabetExtensionOfTo{\ppa_2}{\alphabetOf{}}[{\valuation}]}{\fairWrtRegionModel{}{\decomp_2}}$. 
\end{proof} 
    
\subsection{Proofs of \Cref{sec:strong_robust_sim_pPA}}
\label{ap:proofs_AG_sim}
We first introduce an auxiliary result~\cite[Lemma~4.1.2]{Zha09} that is needed for \Cref{theo:preorder_compositional_ppa}: 
\begin{prop}\label{theo:helper_distr_sim_transtiv}
  Let $\mu_i \in \dist{\stateSetOf{i}}{}$ for $i \in \{1,2,3\}$ and let
  $\simulationRelation{}_1 \subseteq \stateSetOf{1} \times \stateSetOf{2}$,
  $\simulationRelation{}_2 \subseteq \stateSetOf{2} \times \stateSetOf{3}$ be relations such that
  $\mu_1 \sqsubseteq_{\simulationRelation_{1}} \mu_2$ and $\mu_2 \sqsubseteq_{\simulationRelation_{2}} \mu_3$.
  Then $\mu_1 \sqsubseteq_{\simulationRelation} \mu_3$, for $\simulationRelation = \{ (s_1,s_3) \mid \exists s_2 \in \stateSetOf{2}: (s_1,s_2) \in \simulationRelation{}_1 \wedge (s_2,s_3) \in \simulationRelation{}_2 \}.$
\end{prop}

\restatablSimPropspPA*

\begin{proof}
Let $\ppa_1,\ppa_2,\ppa_3$ be pPAs and let $\region$ be a region that is well-defined for all three pPAs.
We first show the properties for $\strSimulationRegion{\region}$ and then for $\robStrSimulationRegion{\region}$. 

\medskip
\noindent
\emph{Case $\trianglelefteq = \strSimulationRegion{}$.}  
\begin{itemize}
    \item \emph{Preorder.}
    For reflexivity, fix $\ppa_1$ and $\valuation \in \region$. 
    By \Cref{theo:preorder_compositional}, we have $\ppa_1[\valuation] \strSimulationRegion{} \ppa_1[\valuation]$ for all $\valuation$, hence $\ppa_1 \strSimulationRegion{\region} \ppa_1$.
    
    For transitivity, assume $\ppa_1 \strSimulationRegion{\region} \ppa_2$ and $\ppa_2 \strSimulationRegion{\region} \ppa_3$. 
    By definition, for each $\valuation \in \region$ we have $\ppa_1[\valuation] \strSimulationRegion{} \ppa_2[\valuation]$ and $\ppa_2[\valuation] \strSimulationRegion{} \ppa_3[\valuation]$, and transitivity of strong simulation on PAs (\Cref{theo:preorder_compositional}) yields $\ppa_1[\valuation] \strSimulationRegion{} \ppa_3[\valuation]$ for all $\valuation \in \region$. 
    Thus $\ppa_1 \strSimulationRegion{\region} \ppa_3$.
        
    \item 
    \emph{Compositionality.}
    Assume $\ppa_1 \strSimulationRegion{\region} \ppa_2$ and $\alphabetOf{2} \subseteq \alphabetOf{1}$, and let $\ppa$ be a pPA for which $\region$ is well-defined. 
    For every $\valuation \in \region$, we have $\ppa_1[\valuation] \strSimulationRegion{} \ppa_2[\valuation]$ by assumption, and by compositionality of strong simulation on PAs (\Cref{theo:preorder_compositional}),
    \[ 
    \ppa_1[\valuation] \parallel \ppa[\valuation] \strSimulationRegion{} \ppa_2[\valuation] \parallel \ppa[\valuation].
    \]
    Hence, $\ppa_1 \parallel \ppa \strSimulationRegion{\region} \ppa_2 \parallel \ppa$.
\end{itemize}

\medskip
\noindent
\emph{Case $\trianglelefteq = \strSimulationRegion{}$.}  
\begin{itemize}
 \item 
\emph{Preorder.}
Reflexivity follows by taking the identity relation $\simulationRelation_{id} = \{(s,s) \mid s \in \stateSetOf{1}\}$. This is a witness relation for $\ppa_1 \robStrSimulationRegion{\region} \ppa_1$. 


For transitivity, assume $\ppa_1 \robStrSimulationRegion{\region} \ppa_2$ witnessed by the robust-strong simulation $\simulationRelation_1 \subseteq  \stateSetOf{1} \times  \stateSetOf{2}$ and $\ppa_2 \robStrSimulationRegion{\region} \ppa_3$ witnessed by $\simulationRelation_2 \subseteq  \stateSetOf{2} \times  \stateSetOf{3}$ . 
Define $\simulationRelation \subseteq \stateSetOf{2} \times  \stateSetOf{3}$ as follows: 
\[
\simulationRelation = \{ (s_1,s_3) \mid \exists s_2 \in \stateSetOf{2}: (s_1,s_2) \in \simulationRelation_1 \text{ and } (s_2,s_3) \in \simulationRelation_2 \}.
\] 


We show that $\simulationRelation$ is a strong simulation for $\ppa_1[\valuation]$ and $\ppa_3[\valuation]$ for all $\valuation \in \region$ by checking the clauses of \Cref{def:strong_sim}. 
Fix an arbitrary $\valuation \in \region$.

    \begin{itemize}
      \item Since $(\initialOf{1},\initialOf{2}) \in \simulationRelation_1$ and $(\initialOf{2},\initialOf{3}) \in \simulationRelation_2$, we have $(\initialOf{1},\initialOf{3}) \in \simulationRelation$ by definition of $\simulationRelation$.
      \item Let $(s_1,s_3) \in \simulationRelation$ and let $(s_1,\alpha_1) \in \domain(\transFctOf{1})$ be a transition of $\ppa_1[\valuation]$. 
            We show there is a transition $(s_3,\alpha_3) \in \domain(\transFctOf{3})$ such that
    \[
    \syncOf{1}(s_1,\alpha_1) = \syncOf{3}(s_3,\alpha_3)
              \quad\text{ and }\quad
              \transFctOf{1}(s_1,\alpha_1)[\valuation] \sqsubseteq_{\simulationRelation} \transFctOf{3}(s_3,\alpha_3)[\valuation].
    \]
    
            Since $(s_1,s_3) \in \simulationRelation$, there exists $s_2$ with $(s_1,s_2) \in \simulationRelation_1$ and $(s_2,s_3) \in \simulationRelation_2$. 
            Because $\simulationRelation_1$ is a robust-strong simulation, for the transition $(s_1,\alpha_1)$ of $\ppa_1[\valuation]$ there exists a matching transition $(s_2,\alpha_2) \in \domain(\transFctOf{2})$ of $\ppa_2[\valuation]$ such that
    \[
    \syncOf{1}(s_1,\alpha_1) = \syncOf{2}(s_2,\alpha_2)
              \quad\text{ and }\quad
              \transFctOf{1}(s_1,\alpha_1)[\valuation] \sqsubseteq_{\simulationRelation_1} \transFctOf{2}(s_2,\alpha_2)[\valuation].
    \]
    
            Similarly, since $\simulationRelation_2$ is a robust-strong simulation and $(s_2,s_3) \in \simulationRelation_2$, there exists a matching transition $(s_3,\alpha_3) \in \domain(\transFctOf{3})$ of $\ppa_3[\valuation]$ such that
    \[
    \syncOf{2}(s_2,\alpha_2) = \syncOf{3}(s_3,\alpha_3)
              \quad\text{ and }\quad
              \transFctOf{2}(s_2,\alpha_2)[\valuation] \sqsubseteq_{\simulationRelation_2} \transFctOf{3}(s_3,\alpha_3)[\valuation].
    \]
    
            Combining the equalities yields $\syncOf{1}(s_1,\alpha_1) = \syncOf{3}(s_3,\alpha_3)$. 
            Moreover, by \Cref{theo:helper_distr_sim_transtiv} and the reasoning above, we obtain 
    $
    \transFctOf{1}(s_1,\alpha_1)[\valuation]
                \sqsubseteq_{\simulationRelation}
              \transFctOf{3}(s_3,\alpha_3)[\valuation].
    $
    \end{itemize}

\item \emph{Compositionality.}
Assume $\ppa_1 \robStrSimulationRegion{\region} \ppa_2$ via a relation $\simulationRelation$, and let $\ppa$ be a pPA for which $\region$ is well-defined. 
Define 
\[ 
\simulationRelation_{\parallel} = \{ ((s_1,s),(s_2,s)) \mid s \in \stateSetOf{\ppa} \text{ and } (s_1,s_2) \in \simulationRelation \}.
\] 

For each $\valuation \in \region$, the relation $\simulationRelation$ is a strong simulation between $\ppa_1[\valuation]$ and $\ppa_2[\valuation]$. 
By compositionality of strong simulation on PAs (\Cref{theo:preorder_compositional}), $\simulationRelation_{\parallel}$ is a strong simulation between the composed PAs $(\ppa_1 \parallel \ppa)[\valuation]$ and $(\ppa_2 \parallel \ppa)[\valuation]$. 
Since this holds for all $\valuation \in \region$, we obtain $\ppa_1 \parallel \ppa \robStrSimulationRegion{\region} \ppa_2 \parallel \ppa$.

\end{itemize}

\end{proof}

\restatablSimRulepPA*
\begin{proof}
The proof is analogous to the PA case in~\cite[Theorem~1]{Kom+12}. 
Completeness follows directly by instantiating the assumption component with $\ppa_A = \ppa_1$ and commutativity of the parallel composition $\parallel$. 

For soundness, assume the premises hold and fix $\mathbin{\trianglelefteq} \in \{\strSimulationRegion{}, \robStrSimulationRegion{}\}$. 
From $\ppa_1 \mathbin{\trianglelefteq}_{\region_1} \ppa_A$ and $\alphabetOf{A} \subseteq \alphabetOf{1}$, compositionality of $\mathbin{\trianglelefteq}_{\region_1}$ (\Cref{theo:preorder_compositional_ppa}) yields
\[ 
\ppa_1 \parallel \ppa_2 \mathbin{\trianglelefteq}_{\regionIntersectionOf{\region_1}{}{\region_2}} 
\ppa_A \parallel \ppa_2 = \ppa_2 \parallel \ppa_A.
\]
Together with the second premise $\ppa_2 \parallel \ppa_A \mathbin{\trianglelefteq}_{\region_2} \ppa_G$ and transitivity of $\mathbin{\trianglelefteq}_{\regionIntersectionOf{\region_1}{}{\region_2}}$ (see \Cref{theo:preorder_compositional_ppa}), we obtain
\[
\ppa_1 \parallel \ppa_2 \mathbin{\trianglelefteq}_{\regionIntersectionOf{\region_1}{}{\region_2}} \ppa_G.
\]
\end{proof} 
    
\subsection{Proofs of~\protect\Cref{sec:studying_ag_for_rpa_semantics}}
\label{ap:proof_for_AG_for_rPA}

We provide the technical details underlying the results of \Cref{sec:AG_for_rpa_conv}. 
We establish an equivalence between convex rPAs and their PA-reductions (cf.~\Cref{defi_PA_reduction_rpa_paths}) with respect to safety multi-objective queries and AG triples. 
Then, we use these equivalence results to prove the assume–guarantee rules for convex rPAs under the convexity-preserving composition operator from \Cref{sec:AG_for_rpa_conv}. 

We define the PA-reduction of an (in)finite path $\ppath = s_0, \alpha_0, s_1, \alpha_1, \dots$ of $\rpa$, with states $s_i \in \stateSetOf{}$ and actions $\alpha_i \in \actSetOf{}$, as the (possibly uncountable) set of PA-paths
\[ 
\paReductionOf{\ppath}  = \{ 
    s_0, (\alpha_0, \mu_0), s_1, (\alpha_1, \mu_1), \ldots \mid \forall i \text{ with } 0 \leq i < \vert \ppath \vert: 
    \mu_i \in \gen{\transFctOf{}(s_i, \alpha_i)}  \}. 
\]
We lift this notation to sets $X$ of (in)finite paths by setting
\[ 
\paReductionOf{X}
  =
  \bigcup_{\ppath \in X} \paReductionOf{\ppath}.
\]

The models ${\paReductionOf{\rpa}}$ and ${\rpa}$ are equivalent in the sense that they satisfy exactly the same safety multi-objective queries and AG triples. 
This equivalence is obtained by the following lemmas:
\begin{itemize}
  \item \Cref{forallpaexistsrpa} shows that the behaviour of a convex rPA over-approximates the behaviour of its PA-reduction, in the sense that for every strategy on ${\paReductionOf{\rpa}}$ there exists a nature--strategy pair on ${\rpa}$ inducing the same probability measure.
  \item \Cref{theo:forallrPAexistsPA} shows the converse, namely that the behaviour of the PA-reduction over-approximates that of the rPA, by constructing for every nature--strategy pair on ${\rpa}$ a corresponding strategy on ${\paReductionOf{\rpa}}$ inducing the same probability measure. 
\end{itemize}
Taken together, these results establish the equivalence between ${\rpa}$ and ${\paReductionOf{\rpa}}$, which is needed for reducing AG proof rules for rPAs to the AG framework for PAs by Kwiatkowska et al.~\cite{Kwi+13}.
Similar results for polytopic rMDPs and single-objective optimisation for a stochastic game reduction appear, e.g., in~\cite{NG05,Iye05,Wi+13}.
Here we obtain an analogous result for convex (not necessarily polytopic) rPAs in a cooperative setting and for safety multi-objective queries. 

We formalize the relationship between the sets of paths of $\rpa$ and its PA-reduction whose traces lie in a given language $\regLang$.
\begin{lem}\label{eq:paths_equiv}
Let $\rpa$ be an rPA and let $\regLang \subseteq \alphabetOf{}^\infty$ be a language over $\alphabetOf{}$. Then, 
\[ \{ \infpath \in \infPathsOf{\paReductionOf{\rpa}}{} \mid  \restrOfTo{\traceOf{\infpath}}{\alphabetOf{}} \in \regLang  \} = 
    \paReductionOf{{\{\infpath \in \infPathsOf{{\rpa}}{} \mid \restrOfTo{\traceOf{\infpath}}{\alphabetOf{}} \in \regLang \}}}. 
    \] 
\end{lem}
\begin{proof}
We prove both inclusions separately.
\begin{itemize}
  \item For the ``$\subseteq$"-direction, let $ \infpath' \in \{\infpath \in \infPathsOf{\paReductionOf{\rpa}}{} \mid \restrOfTo{\traceOf{\infpath}}{\alphabetOf{}} \in \regLang \}. $ 
    By definition of the transition function $\transFctOf{\paReductionOf{\rpa}}$ of the PA-reduction $\paReductionOf{\rpa}$, there exists a path $\infpath'' \in \infPathsOf{\rpa}{}$ such that $\infpath' \in \paReductionOf{\infpath''}$. 
    Moreover, by the definition of the labeling function $\syncOf{\paReductionOf{\rpa}}$, the trace over $\alphabetOf{}$ is preserved, i.e., 
    \[ 
    \restrOfTo{\traceOf{\infpath'}}{\alphabetOf{}}
    =
    \restrOfTo{\traceOf{\infpath''}}{\alphabetOf{}}.\]
  Since $\restrOfTo{\traceOf{\infpath'}}{\alphabetOf{}} \in \regLang$ holds by assumption, it follows that  $\restrOfTo{\traceOf{\infpath''}}{\alphabetOf{}} \in \regLang$. 
  Hence $\infpath'' \in \infPathsOf{\rpa}{}$ with $\restrOfTo{\traceOf{\infpath''}}{\alphabetOf{}} \in \regLang$, and $\infpath' \in \paReductionOf{\infpath''}$. 
  This implies 
  $\infpath' \in \paReductionOf{\{\infpath \in \infPathsOf{\rpa}{} \mid \restrOfTo{\traceOf{\infpath}}{\alphabetOf{}} \in \regLang\}}.$

\item For the ``$\supseteq$"-direction, let $\infpath' \in \paReductionOf{\{\infpath \in \infPathsOf{\rpa}{} \mid \restrOfTo{\traceOf{\infpath}}{\alphabetOf{}} \in \regLang\}}.$
By definition of the PA-reduction on sets of paths, we have 
\[ \infpath'  \in \bigcup_{\substack{\infpath \in \infPathsOf{{\rpa}}{} \\ \restrOfTo{\traceOf{\infpath}}{\alphabetOf{}} \in \regLang }  }\paReductionOf{\infpath}.\]
This means there is 
$\infpath'' \in \infPathsOf{\rpa}{}$ such that $\restrOfTo{\traceOf{\infpath''}}{\alphabetOf{}} \in \regLang $ and $\infpath' \in \paReductionOf{\infpath''}$. 
Again, by the definition of the labeling $\syncOf{\paReductionOf{\rpa}}$ of the PA-reduction, we have 
 \[ 
 \restrOfTo{\traceOf{\infpath'}}{\alphabetOf{}}=\restrOfTo{\traceOf{\infpath''}}{\alphabetOf{}}.
\]
Since $\restrOfTo{\traceOf{\infpath''}}{\alphabetOf{}} \in \regLang$, we obtain that also $\restrOfTo{\traceOf{\infpath'}}{\alphabetOf{}} \in \regLang$. 
Thus, $\infpath'  \in \{ \infpath \in \infPathsOf{\paReductionOf{\rpa}}{} \mid  \restrOfTo{\traceOf{\infpath}}{\alphabetOf{}} \in \regLang  \}$. 
\end{itemize}
Since both inclusions hold, the two sets are equal, which concludes the proof.
\end{proof}

The following lemma directly implies that, if ${\rpa}$ satisfies a (multi-objective) property or AG triple, then $\paReductionOf{\rpa}$ does so as well. 
\begin{lem}\label{forallpaexistsrpa}
  Let $\star \in \{\prt,\comp\}$ and let $\rpa$ be a convex rPA. 
  For every strategy $\widehat{\strategy} \in \strategysetOf{\paReductionOf{\rpa}}{\star}$ there exist a nature $\nature \in \natureSetOf{\rpa}{}$ and a strategy $\strategy \in \strategysetOf{\rpa}{\star}$ such that, for every language $\regLang \subseteq{\alphabetOf{\rpa}}$: 
\[ 
\PrOf{\paReductionOf{\rpa}}{\widehat{\strategy}}{\regLang}
    =
    \PrOf{\rpa}{\nature,\strategy}{\regLang}.
\]
\end{lem}
\begin{proof}
Let $\widehat{\strategy} \in \strategysetOf{\paReductionOf{\rpa}}{\star}$.
We construct a strategy $\strategy$ and a nature $\nature$ for $\rpa$ such that the induced probability measures coincide on all languages over $\alphabetOf{}$. 
    \begin{rem} 
        The sets $\paReductionOf{\ppath}$ may be uncountable.
        However, by construction, all transition distributions $\transFctOf{\paReductionOf{\rpa}}(s,(\alpha,\mu)) = \mu$ and all strategy choices $\widehat{\strategy}(\widehat{\ppath})$ are discrete distributions and thus have at most countable support. 
        The same will hold for $\strategy(\ppath)$ and $\nature(\ppath,\alpha)$ defined below. 
        Hence the induced (sub-)probability measures on infinite paths of $\paReductionOf{\rpa}$ and $\rpa$ are well defined via the standard cylinder-set construction.  
     \end{rem}

We define $\strategy \in \strategysetOf{\rpa}{\star}$ such that, for $\ppath \in \finPathsOf{\rpa}{}$ and $\alpha \in \actSetOf{}$: 
    \begin{align*}
     \strategy(\ppath,\alpha) =   
     \begin{cases}
        \displaystyle
			 \frac{ \sum_{s \in \stateSetOf{}} \PrOf{\paReductionOf{\rpa}}{\widehat{\strategy}}{\paReductionOf{ \ppath, \alpha,s} } }{
             \PrOf{\paReductionOf{\rpa}}{\widehat{\strategy}}{\paReductionOf{ \ppath}} } & \text{if } {\PrOf{\paReductionOf{\rpa}}{\widehat{\strategy}}{\paReductionOf{ \ppath}} } > 0 \\
			0 & \text{otherwise.}
		\end{cases}    
    \end{align*}

    Intuitively, $\strategy(\ppath,\alpha)$ is the conditional probability of choosing action $\alpha$ after $\ppath$ under the PA-strategy $\widehat{\strategy}$, given that $\ppath$ has positive probability under $\widehat{\strategy}$. 
    If $\widehat{\strategy}$ is a complete strategy, then $\strategy$ is complete as well. 

We define $\nature \in \natureSetOf{\rpa}{}$ by setting, for $\ppath \in \finPathsOf{\rpa}{}$ and $\alpha \in \actSetOf{}$, $\nature(\ppath,\alpha) = \mu$, where 
     \begin{align*}
    \mu = 
    \begin{cases}
        \displaystyle
			 \frac{ \PrOf{\paReductionOf{\rpa}}{\widehat{\strategy}}{\paReductionOf{ \ppath, \alpha, s} } }{  \sum_{s' \in \stateSetOf{}} \PrOf{\paReductionOf{\rpa}}{\widehat{\strategy}}{\paReductionOf{ \ppath, \alpha, s'}} }& \text{if } { \sum_{s' \in \stateSetOf{}} \PrOf{\paReductionOf{\rpa}}{\widehat{\strategy}}{\paReductionOf{ \ppath, \alpha, s'}} } > 0 \\
			\text{arbitrary } \mu' \in \transFctOf{}(\last{\ppath}, \alpha)    & \text{otherwise.}
		\end{cases}   
     \end{align*}  
     Intuitively, in the first case $\nature(\ppath,\alpha)$ chooses a distribution equal to the normalised weighted sum of distributions consistent with the choice of $\widehat{\strategy}$ after $(\ppath,\alpha)$. 
     Such a $\mu$ always exists since each uncertainty set $\transFctOf{}(\last{\ppath},\alpha)$ is convex and non-empty. 
     In the second case the value of $\nature(\ppath,\alpha)$ is irrelevant, as the corresponding prefix has zero probability under $\widehat{\strategy}$. 
     
     We show that for every $\ppath \in \finPathsOf{\rpa}{}$,
\begin{align}\label{eq:prob_equal} 
  \PrOf{\paReductionOf{\rpa}}{\widehat{\strategy}}{\paReductionOf{\ppath}}
  =
  \PrOf{\rpa}{\nature,\strategy}{\ppath}.
\end{align}
The proof is by induction on the length of $\ppath \in\finPathsOf{\rpa}$. 
For $\vert \ppath \vert = 0$, i.e., $\ppath = \initialOf{}$, we have $\paReductionOf{\ppath} = \{\initialOf{}\}$ and hence,
\[
  \PrOf{\paReductionOf{\rpa}}{\widehat{\strategy}}{\{\initialOf{}\}}
= 1 = \PrOf{\rpa}{\nature,\strategy}{\initialOf{}}.
\]
For the induction step, assume \Cref{eq:prob_equal} holds for $\ppath' \in \finPathsOf{\rpa}{}$ and consider $\ppath = \ppath',\alpha,s \in \finPathsOf{\rpa}{}$.
If $\PrOf{\rpa}{\nature,\strategy}{\ppath'} = 0$, then the induction hypothesis implies $\PrOf{\paReductionOf{\rpa}}{\widehat{\strategy}}{\paReductionOf{\ppath'}} = 0$, and since no extension of a zero-probability prefix has positive probability, we have  
 $\PrOf{\paReductionOf{\rpa}}{\widehat{\strategy}}{ {\paReductionOf{\ppath}}} = \PrOf{\rpa}{\nature,\strategy}{\ppath}=0$. 
 %
Otherwise, $\PrOf{\rpa}{\nature,\strategy}{\ppath'} > 0$ and thus also $\PrOf{\paReductionOf{\rpa}}{\widehat{\strategy}}{\paReductionOf{\ppath'}} > 0$ by the induction hypothesis.
We then obtain: 
         \begin{align*}
                     &~   \PrOf{\rpa}{\nature,\strategy}{\ppath}\\
                   ~= &~  \PrOf{\rpa}{\nature,\strategy}{\ppath'} \cdot \strategy(\ppath' , \alpha) \cdot \nature(\ppath', \alpha)(s)\\
                   = &~  \PrOf{\paReductionOf{\rpa}}{\widehat{\strategy}}{\paReductionOf{\ppath'}} 
                   \cdot  \strategy(\ppath' , \alpha) \cdot \nature(\ppath', \alpha)(s) \tag{By induction hypothesis}\\
                   = &~  
                    \PrOf{\paReductionOf{\rpa}}{\widehat{\strategy}}{\paReductionOf{\ppath'}} 
                    \cdot 
                     \frac{  \sum_{s' \in \stateSetOf{}} \PrOf{\paReductionOf{\rpa}}{\widehat{\strategy}}{\paReductionOf{ \ppath', \alpha,s'} } }{\PrOf{\paReductionOf{\rpa}}{\widehat{\strategy}}{\paReductionOf{ \ppath'}} } 
                     \cdot 
                    \frac{ \PrOf{\paReductionOf{\rpa}}{\widehat{\strategy}}{\paReductionOf{ \ppath', \alpha, s} } }{ \sum_{s' \in \stateSetOf{}}  \PrOf{\paReductionOf{\rpa}}{\widehat{\strategy}}{\paReductionOf{ \ppath', \alpha, s'} }}
                   \\
                    = &~  
                     \PrOf{\paReductionOf{\rpa}}{\widehat{\strategy}}{\paReductionOf{ \ppath', \alpha} } 
                     \cdot 
                    \frac{ \PrOf{\paReductionOf{\rpa}}{\widehat{\strategy}}{\paReductionOf{ \ppath', \alpha, s} } }{\PrOf{\paReductionOf{\rpa}}{\widehat{\strategy}}{\paReductionOf{ \ppath', \alpha}} } 
                   \\
                  = &~  
                      \PrOf{\paReductionOf{\rpa}}{\widehat{\strategy}}{\paReductionOf{ \ppath', \alpha, s}} 
        \end{align*}
        This proves~\Cref{eq:prob_equal}. 

    Let $\regLang$ be a language over $\alphabetOf{} \subseteq \alphabetOf{\rpa}$.
    Using \Cref{eq:paths_equiv} and \Cref{eq:prob_equal}, we obtain 
    \begin{align*}
        \PrOf{\paReductionOf{\rpa}}{\widehat{\strategy}}{\regLang} & = 
        \PrOf{\paReductionOf{\rpa}}{\widehat{\strategy}}{\{ \infpath \in \infPathsOf{\paReductionOf{\rpa}}{} \mid \restrOfTo{\traceOf{\infpath}}{\alphabetOf{}} \in \regLang  \}}  \\
        &= \PrOf{\paReductionOf{\rpa}}{\widehat{\strategy}}{  \paReductionOf{\{ \infpath \in \infPathsOf{\rpa}{} \mid \restrOfTo{\traceOf{\infpath}}{\alphabetOf{}} \in \regLang  \}} } \tag{by  \Cref{eq:paths_equiv}} \\
        & = \PrOf{\rpa}{\nature,\strategy}{\{ \infpath \in \infPathsOf{\rpa}{} \mid \restrOfTo{\traceOf{\infpath}}{\alphabetOf{}} \in \regLang  \}} \tag{by  \Cref{eq:prob_equal}}\\
        & = \PrOf{\rpa}{\nature,\strategy}{\regLang} 
    \end{align*}
\end{proof}

The following result implies that, if $\paReductionOf{\rpa}$ satisfies a (multi-objective) property or AG triple, then $\rpa$ does so as well.
\begin{lem}\label{theo:forallrPAexistsPA}
  Let $\star \in \{\prt,\comp\}$ and let $\rpa$ be a convex rPA. 
  For every nature $\nature \in \natureSetOf{\rpa}{}$ and every strategy $\strategy \in \strategysetOf{\rpa}{\star}$ there exists 
  a strategy $\widehat{\strategy} \in \strategysetOf{\paReductionOf{\rpa}}{\star}$ such that, for every language $\regLang$ over $\alphabetOf{} \subseteq \alphabetOf{\rpa}$,
    \[   \PrOf{\rpa}{\nature,\strategy}{ \regLang} = \PrOf{\paReductionOf{\rpa}}{\widehat{\strategy}}{\regLang}.
        	\]

\end{lem}

\begin{proof}
We show the claim for the case that $\transFctOf{}(s,\alpha) = \convh{\extr{\transFctOf{}(s,\alpha)}}$, for all 
$s \in \stateSetOf{}$, and $\alpha\in \actSetOf{}$. 
When, $\transFctOf{}(s,\alpha) \not = \convh{\extr{\transFctOf{}(s,\alpha)}}$, for some $s \in \stateSetOf{}$, and $\alpha\in \actSetOf{}$ the proof works similarly---without without decompising the distributions chosen by nature into extreme points. 

Let $\nature \in \natureSetOf{\rpa}{}$ and ${\strategy} \in \strategysetOf{\rpa}{\star}$. 
We construct a strategy $\widehat{\strategy} \in \strategysetOf{\paReductionOf{\rpa}}{\star}$ for $\paReductionOf{\rpa}$: 
For each finite path $\widehat{\ppath} \in \finPathsOf{\paReductionOf{\rpa}}{}$ and each action $(\alpha,\mu) \in \actSetOf{\paReductionOf{\rpa}}$, let $\ppath$ be the (unique) finite path of $\rpa$ such that $\widehat{\ppath} \in \paReductionOf{\ppath}$.

Since $\rpa$ is convex and $\stateSetOf{}$ is finite, for every $\nature(\ppath,\alpha) \in \transFctOf{\rpa}(\last{\ppath},\alpha)$, there is a finite convex decomposition into extreme points of $\transFctOf{\rpa}(\last{\ppath},\alpha)$: There exist a finite set $E^{\ppath,\alpha} \subseteq \extr{\transFctOf{\rpa}(\last{\ppath},\alpha)}$ and weights $w^{\ppath,\alpha}_{\mu'} \in [0,1]$ for $\mu' \in E^{\ppath,\alpha}$ such that

\[ 
\sum_{\mu' \in E^{\ppath,\alpha}} w^{\ppath,\alpha}_{\mu'} = 1
  \quad\text{ and }\quad 
  \nature(\ppath,\alpha)  =
  \sum_{\mu' \in E^{\ppath,\alpha}} w^{\ppath,\alpha}_{\mu'} \cdot \mu'.
\]
We extend these weights to all $\mu' \in \extr{\transFctOf{\rpa}(\last{\ppath},\alpha)}$ by setting $w^{\ppath,\alpha}_{\mu'}$ to $0$ if $\mu' \notin E^{\ppath,\alpha}$.

Then, we define $\widehat{\strategy}$ such that 
\[ 
\widehat{\strategy}(\widehat{\ppath},(\alpha,\mu)) = \strategy(\ppath,\alpha) \cdot w^{\ppath,\alpha}_{\mu}.
\]
Then, $\widehat{\strategy}(\widehat{\ppath})$ is always a discrete (countably supported) (sub-)distribution over enabled transitions. 
Additionally, if $\strategy$ is a complete strategy, then so is $\widehat{\strategy}$. 

We show that, for every finite path $\ppath \in \finPathsOf{\rpa}{}$: 
\begin{align}\label{eq:prob_equal_2} 
  \PrOf{\paReductionOf{\rpa}}{\widehat{\strategy}}{\paReductionOf{\ppath}}
  =
  \PrOf{\rpa}{\nature,\strategy}{\ppath}.
\end{align}
Let $\ppath = s_0,\alpha_0,\dots,s_n \in \finPathsOf{\rpa}{}$. 
The set $\paReductionOf{\ppath}$ can be uncountable, since each $\transFctOf{\rpa}(s_i,\alpha_i)$ may be uncountable. 
However, all $\widehat{\ppath} \in \paReductionOf{\ppath}$ are finite, and for each finite path $\widehat{\ppath}'$ the strategy $\widehat{\strategy}(\widehat{\ppath}')$ is a discrete (sub)distribution. 
Therefore, only countably many $\widehat{\ppath} \in \paReductionOf{\ppath}$ have positive probability under $\widehat{\strategy}$. 
    We denote this countable subset by 
  \[
  \paReductionOfStrategy{\ppath}{\widehat{\strategy}}
  = 
  \{ \widehat{\ppath}  \in \paReductionOf{\ppath}  \mid \PrOf{\paReductionOf{\rpa}}{\widehat{\strategy}}{\widehat{\ppath}} >0 \} .
  \]
Then, $
    \PrOf{\paReductionOf{\rpa}}{\widehat{\strategy}}{\paReductionOf{\ppath}}    = 
    \PrOf{\paReductionOf{\rpa}}{\widehat{\strategy}}{\paReductionOfStrategy{\ppath}{\widehat{\strategy}}}  $  
and hence,  
     \begin{align*}
    \PrOf{\paReductionOf{\rpa}}{\widehat{\strategy}}{\paReductionOf{\ppath}}   & = 
    \PrOf{\paReductionOf{\rpa}}{\widehat{\strategy}}{\paReductionOfStrategy{\ppath}{\widehat{\strategy}}} 
    \\ & =
    \sum_{\substack{ \widehat{\ppath}  \in \paReductionOfStrategy{\ppath}{\widehat{\strategy}} \\ \widehat{\ppath} = s_0, (\alpha_0, \mu_0),  \dots, s_n } } \PrOf{\paReductionOf{\rpa}}{\widehat{\strategy}} {\widehat{\ppath}}  \\
    &=  \sum_{\substack { \widehat{\ppath}  \in \paReductionOfStrategy{\ppath}{\widehat{\strategy}} \\ \widehat{\ppath} = s_0, (\alpha_0, \mu_0),  \dots, s_n } }
      \iverson{s_0 = \initialOf{}} 
	   \cdot \prod_{ i= 0}^{n-1} \widehat{\strategy}( \widehat{\ppath}[0,i], (\alpha_i, \mu_i)) \cdot \transFctOf{\paReductionOf{\rpa}}(s_{i}, (\alpha_{i}, \mu_i), s_{i+1}) \\
    &=  \sum_{\substack { \widehat{\ppath}  \in \paReductionOfStrategy{\ppath}{\widehat{\strategy}} \\ \widehat{\ppath} = s_0, (\alpha_0, \mu_0),  \dots, s_n } }
      \iverson{s_0 = \initialOf{}} 
	   \cdot \prod_{ i= 0}^{n-1} \widehat{\strategy}( \widehat{\ppath}[0,i], (\alpha_i, \mu_i)) \cdot \mu_i(s_{i+1}) \\
    &=  \sum_{\substack { \widehat{\ppath}  \in \paReductionOfStrategy{\ppath}{\widehat{\strategy}} \\ \widehat{\ppath} = s_0, (\alpha_0, \mu_0),  \dots, s_n } }
      \iverson{s_0 = \initialOf{}} 
	   \cdot \prod_{ i= 0}^{n-1} \strategy(\ppath[0,i], \alpha_i) \cdot w^{{\ppath[0,i]},\alpha_i}_{\mu_i} \cdot \mu_i(s_{i+1})  \tag{By definition of $\widehat{\strategy}$ } \\
   & =     \iverson{s_0 = \initialOf{}}
     \prod_{i=0}^{n-1} \left(\strategy(\ppath[0,i],\alpha_i)
       \sum_{\mu_i \in E^{\ppath[0,i],\alpha_i}}
         w^{\ppath[0,i],\alpha_i}_{\mu_i} \cdot \mu_i(s_{i+1}) \right) 
     \tag{By definition of $\nature$ } \\
    \\
        &=    \iverson{s_0 = \initialOf{}}\cdot  \prod_{i=0}^{n-1} 
	    \strategy(\ppath[0,i], \alpha_i) \cdot \nature(\ppath[0,i], \alpha_i)(s_{i+1})  \\
    %
    & = \PrOf{\rpa}{\nature,\strategy}{\ppath} 
\end{align*}
This proves~\Cref{eq:prob_equal_2}.

Let $\regLang$ be a language over $\alphabetOf{}$.
Using \Cref{eq:paths_equiv} and \eqref{eq:prob_equal_2}, we obtain 
\begin{align*}
    \PrOf{\paReductionOf{\rpa}}{\widehat{\strategy}}{\regLang} & = 
    \PrOf{\paReductionOf{\rpa}}{\widehat{\strategy}}{
    \{ \infpath \in \infPathsOf{\paReductionOf{\rpa}}{}  \mid \restrOfTo{\traceOf{\infpath}}{\alphabetOf{}} \in \regLang \}}  \\
    &= \PrOf{\paReductionOf{\rpa}}{\widehat{\strategy}}{  \paReductionOf{\{ \infpath \in \infPathsOf{\rpa}{} \mid \restrOfTo{\traceOf{\infpath}}{\alphabetOf{}} \in \regLang  \}} }  && \text{by  \Cref{eq:paths_equiv}} \\
    & = \PrOf{\rpa}{\nature,\strategy}{\{ \infpath \in \infPathsOf{\rpa}{} \mid \restrOfTo{\traceOf{\infpath}}{\alphabetOf{}} \in \regLang  \}} && \text{by  \Cref{eq:prob_equal_2}}\\
    & = \PrOf{\rpa}{\nature,\strategy}{\regLang} 
\end{align*}
\end{proof}

\restatablePAConvConvPA* 

\begin{proof}
We prove the claim for the case that every uncertainty set is the convex hull of its extreme points. 
The other case works analogously without decomposing distributions into convex combinations of extreme points. 

We show the claim by providing a \emph{combined bisimulation}
\cite[Definition~5.7]{Sto+02} between
$\paReductionOf{\rpa_1} \parallel \paReductionOf{\rpa_2}$ and
$\paReductionOf{\rpa_1\parallelConv\rpa_2}$.
A combined bisimulation allows matching a single transition by a finite
convex combination of transitions with the same label; hence bisimilar PAs satisfy the same (multi-objective) queries. 

By construction of the PA-reduction and of parallel composition, both PAs
$\paReductionOf{\rpa_1} \parallel \paReductionOf{\rpa_2}$ and
$\paReductionOf{\rpa_1\parallelConv\rpa_2}$ have the same set of states
$\stateSetOf{1} \times \stateSetOf{2}$, the same initial state
$(\initialOf{1},\initialOf{2})$, where $\stateSetOf{i}$ is the state set of
$\rpa_i$ and the same set of labels. 
Consider the relation
\[
  \simulationRelation = \{ (s,s) \mid s \in \stateSetOf{1} \times \stateSetOf{2} \}.
\]
We show that $\simulationRelation$ is a combined bisimulation. 

Fix a state $s = (s_1,s_2)$ and a label $\lab$ of $\paReductionOf{\rpa_1} \parallel \paReductionOf{\rpa_2}$ and $\paReductionOf{\rpa_1 \parallelConv \rpa_2}.$ 
In $\paReductionOf{\rpa_1} \parallel \paReductionOf{\rpa_2}$, any
successor distribution from $s$ with label $\lab$ is a product 
distributions $\mu_1 \times \mu_2$ where $\mu_i$ is an extreme
point of the corresponding uncertainty set of $\rpa_i$. 
In $\rpa_1 \parallelConv \rpa_2$, the uncertainty set at $s$ and label
$\lab$ is defined as the convex hull of precisely these product 
distributions, i.e., the PA-reduction $\paReductionOf{\rpa_1 \parallelConv \rpa_2}$ then chooses a generator of this convex hull as its set of successor distributions at $s$ with label $\lab$. 

Thus every product distribution $\mu_1 \times \mu_2$ that appears as an
$\lab$-transition in $\paReductionOf{\rpa_1} \parallel \paReductionOf{\rpa_2}$ lies
in the convex hull of the successor distributions of
$\paReductionOf{\rpa_1\parallelConv\rpa_2}$ and can therefore be written as
a finite convex combination of those. 

Conversely, every transition from $s$ 
with label $\lab$ of $\paReductionOf{\rpa_1\parallelConv\rpa_2}$ is distribution that is in the convex hull of the product distributions and hence is a finite convex combination of transitions in $\paReductionOf{\rpa_1} \parallel \paReductionOf{\rpa_2}$. 

This is exactly the condition of combined bisimulation. 
The argument for non-synchronising labels is analogous. 

\end{proof}

We now show AG proof rules for the compositional verification of rPAs from \Cref{sec:AG_for_rpa_conv}. 
In the remainder of this section, we fix convex rPAs $\rpa_i$ with alphabets $\alphabetOf{i}$. 
\restatableAsymRuleRPA* 
\begin{proof}

  The proof is based on the reduction of convex rPAs to (possibly uncountably branching) PAs, described in \Cref{defi_PA_reduction_rpa_paths}. 
  
  We first show the asymmetric proof rule on the left.  
  Assume that the premises hold, i.e., 
  \[ 
\rpa_1 \modelsWrt{\comp} \multiobjectiveQuery{}{A}
    \quad \text{ and }
    \quad \agTriple{\alphabetExtensionOfTo{\rpa_2}{\alphabetOf{\multiobjectiveQuery{}{A}}}}{\prt}{\multiobjectiveQuery{}{A}}{\multiobjectiveQuery{}{G}}.
\]
  By \Cref{forallpaexistsrpa}, the corresponding PA-reductions satisfy
\[ 
\paReductionOf{\rpa_1} \modelsWrt{\comp} \multiobjectiveQuery{}{A}\quad\text{ and }
\quad \agTriple{\alphabetExtensionOfTo{\paReductionOf{\rpa_2}}{\alphabetOf{\multiobjectiveQuery{}{A}}}}{\prt}{\multiobjectiveQuery{}{A}}{\multiobjectiveQuery{}{G}}.
\]
  For the latter, we exploit that both orders of applying PA-reduction reduction and alphabet extension extension yield PAs with the same self-loops and therefore equivalent satisfaction of safety mo-queries and AG triples, i.e., the alphabet extension of the PA-reduction $\alphabetExtensionOfTo{\paReductionOf{\rpa_2}}{\alphabetOf{\multiobjectiveQuery{}{A}}}$ and the PA-reduction of the extension 
                              $\paReductionOf{ \alphabetExtensionOfTo{\rpa_2}{\alphabetOf{\multiobjectiveQuery{}{A}}}}$ 
  satisfy the same safety mo-queries and AG triples. 
  The premises above are exactly those of the asymmetric AG rule for PAs
  from~\cite[Theorem~1]{Kwi+13}, so we obtain 
  \[ \paReductionOf{\rpa_1} \parallel \paReductionOf{\rpa_2}
      \modelsWrt{\comp} \multiobjectiveQuery{}{G}.
      \]
By \Cref{thm:pa_conv_eq_conv_pa} follows that $\paReductionOf{\rpa_1 \parallelConv \rpa_2} \modelsWrt{\star} \multiobjectiveQuery{}{G}$. 
  Finally, by \Cref{theo:forallrPAexistsPA}, this result lifts back to rPAs, and we obtain 
  \[ \rpa_1 \parallelConv \rpa_2 \modelsWrt{\comp} \multiobjectiveQuery{}{G}.\]

We show the circular proof rule on the right. 
        The proof proceeds analogously to that of \Cref{theo_rpa_asym_rule}. 
        Assume that the premises of the circular rule hold for the convex rPAs $\rpa_1$ and $\rpa_2$, i.e.,
  \[ 
    \agTriple{\alphabetExtensionOfTo{\rpa_1}{\alphabetOf{\multiobjectiveQuery{}{A}_1}}}{\prt}
             {\multiobjectiveQuery{}{A}_1}{\multiobjectiveQuery{}{A}_2}, \quad
\agTriple{\alphabetExtensionOfTo{\rpa_2}{\alphabetOf{\multiobjectiveQuery{}{A}_2}}}{\prt}
             {\multiobjectiveQuery{}{A}_2}{\multiobjectiveQuery{}{G}}, \quad \text{ and } \quad 
             \rpa_2 \modelsWrt{\comp} \multiobjectiveQuery{}{A}_1.
\]
  By \Cref{forallpaexistsrpa}, the corresponding premises hold for the PA-reductions:
  \[ 
    \agTriple{\alphabetExtensionOfTo{\paReductionOf{\rpa_1}}{\alphabetOf{\multiobjectiveQuery{}{A}_1}}}{\prt}
             {\multiobjectiveQuery{}{A}_1}{\multiobjectiveQuery{}{A}_2}, \quad
\agTriple{\alphabetExtensionOfTo{\paReductionOf{\rpa_2}}{\alphabetOf{\multiobjectiveQuery{}{A}_2}}}{\prt}
             {\multiobjectiveQuery{}{A}_2}{\multiobjectiveQuery{}{G}}, \quad \text{ and } \quad 
             \paReductionOf{\rpa_2} \modelsWrt{\comp} \multiobjectiveQuery{}{A}_1.
\]
  Here we again use that the alphabet extension of the PA-reductions and the PA-reductions of the extension satisfy the same safety mo-queries and AG triples. 
  These are exactly the premises of the circular AG rule for PAs from
  \cite[Theorem~5]{Kwi+13}, which yields
  \[ 
\paReductionOf{\rpa_1} \parallel \paReductionOf{\rpa_2}
      \modelsWrt{\comp} \multiobjectiveQuery{}{G}.\]

      By \Cref{thm:pa_conv_eq_conv_pa} follows that $\paReductionOf{\rpa_1 \parallelConv \rpa_2} \modelsWrt{\star} \multiobjectiveQuery{}{G}$. 
  Finally, by \Cref{theo:forallrPAexistsPA}, this result lifts back to rPAs, and we obtain 
  \[ \rpa_1 \parallelConv \rpa_2 \modelsWrt{\comp} \multiobjectiveQuery{}{G}.\]

    \end{proof}


    \section{Partial vs.\ Complete Strategies}\label{app:part_vs_complete}
\Cref{app:equivalence_partial_complete_measures} section establishes that verifying probabilistic, and reward objectives under partial strategies in a given pPA is equivalent to verifying them under complete strategies in an extended pPA. 
This result is extended to multiobjective queries in \Cref{app:equivalence_partial_complete_qmo}. 

\subsection{Equivalence w.r.t.\ Probabilistic and Reward Objectives}\label{app:equivalence_partial_complete_measures}
Next, we define the pPA $\ppa_{\tau}$ which extends the pPA $\ppa$ by introducing transitions from each state to a new sink state which are labeled by a fresh alphabet symbol $\tau$.  
We show that the class of partial strategies of a pPA $\ppa$ is equivalent to the complete strategies of $\ppa_{\tau}$. 
\begin{defi}\label{def:ppaTau}
	Given a pPA $\ppa = \ppaTupleOf{\ppa}$, we define the pPA $\ppa_{\tau}= (\stateSetOf{\ppa} \cupdot \{s_{\tau}\}, \initialOf{\ppa},\parameterSetOf{\ppa}, \actSetOf{\ppa} \cupdot \{\tau\} , \transFctOf{\ppa_{\tau}}, \syncOf{\ppa_{\tau}})$, where
	\begin{itemize}
	\item for each $(s, \alpha) \in \domain(\transFctOf{\ppa})$: 
		$\transFctOf{\ppa_{\tau}}(s, \alpha) = \transFctOf{\ppa}(s, \alpha) 
			\text{ and } 
			\syncOf{\ppa_{\tau}}(s, \alpha) = \syncOf{\ppa}(s, \alpha),$
	\item for each $s \in \stateSetOf{\ppa}$: 
		$\transFctOf{\ppa_{\tau}}(s, \tau) = \indicatorFct{s_{\tau}} 
			\text{ and } 
			\syncOf{\ppa_{\tau}}(s, \tau) = \tau. $
	\end{itemize}
\end{defi}
By definition of $\ppa_{\tau}$ follows, that a valuation is well-defined (graph-preserving) for pPA $\ppa$ iff it is well-defined (graph-preserving) for $\ppa_\tau$. 
From the following theorem follows that verifying objectives w.r.t\ partial strategies in $\ppa$ is equivalent to  verifying objectives w.r.t\ complete strategies in $\ppa_{\tau}$, 
i.e., $\ppa, \region \modelsWrt{(\mless,)\prt} \generalPredicate \Leftrightarrow  \ppa_{\tau}, \region \modelsWrt{(\mless,)\comp}$. 
\Cref{theo_prob_reward_partial_vs_complete} is based on the proof of \cite[Proposition 2]{Kwi+13}. 
However, we modified the construction such that the memoryless property is preserved. 
\begin{restatable}{lem}{restatableProbRewardPartVsComplete}
	\label{theo_prob_reward_partial_vs_complete} 
	Let $\regLang$ be a language and $\rewFct$ be a (parametric) reward function over $\alphabetOf{}$ with $\tau \not \in \alphabetOf{}$. 
	\begin{itemize} 
	\item For any strategy $\strategy \in \strategysetOf{\ppa}{(\mless,)\prt}$ there is a strategy $\strategy_{\tau} \in \strategysetOf{\ppa_{\tau} }{(\mless,)\comp}$ 
		such that for any valuation $\valuation$ that is well-defined for $\ppa$ (and $\rewFct$) 
		\[
			\PrOf{\ppa}{\valuation,\strategy}{\regLang}  = \PrOf{\ppa_{\tau}}{\valuation, \strategy_{\tau}}{\regLang}  \quad \text{ and } \quad
			\ExpTot{\ppa}{\valuation,\strategy}{\rewFct} = \ExpTot{\ppa_{\tau}}{\valuation, \strategy_{\tau}}{\rewFct}.
		\]
	\item  
		For any strategy $\strategy_{\tau} \in \strategysetOf{\ppa_{\tau}}{(\mless,)\comp}$ there is a strategy $\strategy \in  \strategysetOf{\ppa }{(\mless,)\prt}$ 
		such that for any valuation $\valuation$ that is well-defined for $\ppa$ (and $\rewFct$) 
		\[
			\PrOf{\ppa_{\tau}}{\valuation, \strategy_{\tau}}{\regLang}  = \PrOf{\ppa}{\valuation,\strategy}{\regLang}
			 \quad  \text{ and } \quad 
			\ExpTot{\ppa_{\tau}}{\valuation, \strategy_{\tau}}{\rewFct} = \ExpTot{\ppa}{\valuation, \strategy}{\rewFct}.
		\]
	\end{itemize}
\end{restatable}
Intuitively, all ``unused'' probabilities of a partial strategy of $\ppa$ are redirected to a sink state by the complete strategy of $\ppa_{\tau}$. 
In \cite[proposition 2]{Kwi+13}, $\tau$-labeled self-loops are introduced. 
However, the resulting complete strategy requires memory. 
\begin{proof}
       First, given a strategy $\strategy \in \strategysetOf{\ppa}{(\mless,)\prt}$, we construct a strategy $\strategy_{\tau} \in \strategysetOf{\ppa_{\tau} }{(\mless,)\comp}$ as follows: 
       We can apply the reasoning provided in the proof of \cite[proposition 2]{Kwi+13} to our construction: 
        Given a finite path $\pi$ of $\ppa_{\tau}$, and $\alpha \in \actSetOf{\ppa_\tau}$: 
        \begin{align*}
            \strategy_{\tau}(\pi, \alpha) = 
            \begin{cases}
                \strategy(\pi, \alpha) & \text{if  $\pi \in \finPathsOf{\ppa}{}$ and $\alpha \in \actSetOf{\ppa}$} \\
                1- \sum_{\alpha \in \actSetOf{\ppa}} \strategy(\pi, \alpha) & \text{if $\pi \in \finPathsOf{\ppa}{}$ and } \alpha= \tau \\
                1 & \text{if $\pi \not  \in \finPathsOf{\ppa}{}$ and } \alpha= \tau \\
                0 & \text{otherwise}
            \end{cases}
        \end{align*}
        Now, for any valuation $\valuation$ that is well-defined for $\ppa$ (and $\rewFct$), we have $\PrOf{\ppa}{\valuation,\strategy}{\regLang}  = \PrOf{\ppa_{\tau}}{\valuation, \strategy_{\tau}}{\regLang}$ and $\ExpTot{\ppa}{\valuation,\strategy}{\rewFct} = \ExpTot{\ppa_{\tau}}{\valuation, \strategy_{\tau}}{\rewFct}.$
        Note that---in contrast to \cite[proposition 2]{Kwi+13}---the strategy $\strategy_{\tau}$ is memoryless if $\strategy$ is a memoryless strategy. 

        Second, given a strategy $\strategy_{\tau} \in \strategysetOf{\ppa_{\tau}}{(\mless,)\comp}$, we construct a strategy $\strategy \in \strategysetOf{\ppa}{(\mless,)\prt}$ as follows: 
        For ${\pi} \in \finPathsOf{\ppa}{} \subseteq \finPathsOf{\ppa_{\tau}}{}$ and $\alpha \in \actSetOf{\ppa}$ it holds that 
			$\strategy({\pi}, \alpha)  = \strategy_{\tau}({\pi}, \tau )$. 
        Again, for any valuation $\valuation$ that is well-defined for $\ppa$ (and $\rewFct$) we have $\PrOf{\ppa}{\valuation,\strategy}{\regLang}  = \PrOf{\ppa_{\tau}}{\valuation, \strategy_{\tau}}{\regLang}$ and $\ExpTot{\ppa}{\valuation,\strategy}{\rewFct} = \ExpTot{\ppa_{\tau}}{\valuation, \strategy_{\tau}}{\rewFct}$.
        Additionally, if $\strategy_{\tau}$ is memoryless, then $\strategy$ is also memoryless. 
\end{proof}

\subsection{Equivalence w.r.t.\ Multiobjective Queries}\label{app:equivalence_partial_complete_qmo}
As a consequence of \Cref{theo_prob_reward_partial_vs_complete}, we directly obtain \cite[Proposition 2]{Kwi+13} for pPA and parametric mo-queries: 
\begin{cor}
	\label{lemma_moq_partial_via_complete}
	Let $\ppa$ be a pPA and let $\multiobjectiveQuery{}{X}$ be a (parametric) multiobjective query. 
	Construct pPA $\ppa_{\tau}$ as in \Cref{def:ppaTau}. 
	$\alphabetOf{\multiobjectiveQuery{}{X}} \subseteq \alphabetOf{\ppa}$, where $\tau \not \in \alphabetOf{\ppa} \cup \alphabetOf{\multiobjectiveQuery{}{X}}$ . 
	Let $\region$ be graph-preserving well-defined for $\ppa$ (i.e., also for $\ppa_{\tau}$) and well-defined for $\multiobjectiveQuery{}{X}$. 
	Then, 
	\begin{align*}
		& \exists \strategy \in \strategysetOf{\ppa}{(\mless,)\prt}: \exists \valuation \in \region :  \ppa[\valuation], {\strategy} \modelsWrt{} \multiobjectiveQuery{}{X}[\valuation]
        \\ & \quad \Leftrightarrow \quad  \exists \strategy \in  \strategysetOf{\ppa_{\tau} }{(\mless,)\comp}: \exists \valuation \in \region :  \ppa_{\tau}[\valuation], {\strategy} \modelsWrt{}  \multiobjectiveQuery{}{X}[\valuation]
	\end{align*}

\end{cor}

\subsection{Equivalence w.r.t.\ Monotonicity}\label{app:equivalence_partial_complete_mono}
As a direct consequence of \Cref{theo_prob_reward_partial_vs_complete}, we obtain that monotonicity for partial strategies w.r.t\ general properties is equivalent to monotonicity for complete strategies in $\ppa_{\tau}$ (see \Cref{def:ppaTau}). 
\begin{restatable}{cor}{restatablePartialCompleteMono}\label{theo:reducing_partial_mono_to_complete}
    Let $\region$ be a region that is well-defined for pPA $\ppa$ and let $\budarrow \in \{\buparrow, \bdownarrow \}$. 
	For the solution functions $\solutionFctMdpObjective{\ppa}{}{}$ and $\solutionFctMdpObjective{\ppa_{\tau}}{}{}$ for $\ppa$ and $\ppa_{\tau}$ with $\tau \not \in \alphabetOf{\solutionFctMdpObjective{\ppa}{}{}}$ holds 
    $\monotonicOnRegionParameter{\budarrow}{\solutionFctMdpObjective{\ppa}{}}{p}{\region}{(\mless,)\prt}$ iff 
 	$\monotonicOnRegionParameter{\budarrow}{\solutionFctMdpObjective{{\ppa_{\tau}}}{}}{p}{\region}{(\mless,)\comp}$. 	
\end{restatable}
\Cref{theo:reducing_partial_mono_to_complete} also holds for local monotonicity.

    \section{Further AG Proof Rules for pPAs}\label{app:pag_extension}
We extend the remaining proof rules from \cite{Kwi+13} to pPA. This includes \Cref{theo:pag_conj_rule}, which enables reasoning about the conjunction of multi-objective queries, as well as \Cref{theo:pag_asymN_rule}, an asymmetric rule for more than two components that extends \Cref{theo:pag_asym_rule}. 
Additionally, we introduce \Cref{theo:pag_interleaving_rule} for handling interleaving components and \Cref{theo:pag_reward_rule}, which supports modular reasoning about reward structures. 

When the premises include $\ppa, \region  \modelsWrt{\star} \multiobjectiveQuery{}{A}$ we adhere to the following assumptions:
\begin{itemize}
    \item $\multiobjectiveQuery{}{A}$ is a (parametric) mo-query for $\ppa$, i.e., $\multiobjectiveQuery{}{A}$ is defined over the alphabet  $\alphabetOf{\multiobjectiveQuery{}{A}} \subseteq \alphabetOf{\ppa}$, and
    \item $\region$ is a region that is well-defined for $\ppa$ (and $\multiobjectiveQuery{}{A}$). 
\end{itemize}
Similarly, if an AG triple $\agTriple{\alphabetExtensionOfTo{\ppa}{\alphabetOf{\multiobjectiveQuery{}{A}}}, \region}{\star}{\multiobjectiveQuery{}{A}}{\multiobjectiveQuery{}{G}}$ occurs in the premises, we assume:
\begin{itemize}
    \item $\multiobjectiveQuery{}{G}$ is a (parametric) mo-query for $\alphabetExtensionOfTo{\ppa}{\alphabetOf{\multiobjectiveQuery{}{A}}}$, i.e., $\multiobjectiveQuery{}{G}$ is defined over the alphabet $\alphabetOf{\multiobjectiveQuery{}{G}} \subseteq \alphabetOf{\ppa} \cup \alphabetOf{\multiobjectiveQuery{}{A}}$\footnote{$\multiobjectiveQuery{(\safe)}{A} \not \subseteq \alphabetOf{\ppa}$ is possible.}, and 
    \item the region $\region$ is well-defined for $\ppa$ ($\multiobjectiveQuery{}{A}$, and $\multiobjectiveQuery{}{G}$). 
\end{itemize}
Safety mo-queries are marked with a superscript ``$\safe$''. 

To reason about the conjunction of multi-objective queries, we generalize \cite[Theorem 3]{Kwi+13} for pPA.
When safety-related multi-objective queries are combined through conjunction, the resulting query remains a safety multi-objective query. 
\begin{restatable}[Conjunction]{thm}{restatableConjRule}\label{theo:pag_conj_rule} 
   Let $\region_{\regionIntersection} = \regionIntersectionOf{\region_1}{}{\region_2}$, $\multiobjectiveQuery{}{A}_{\land} = \multiobjectiveQuery{}{A}_{1} \land \multiobjectiveQuery{}{A}_{2}$, and $\multiobjectiveQuery{}{G}_{\land} = \multiobjectiveQuery{}{G}_{1} \land \multiobjectiveQuery{}{G}_{2}$.

    \begin{tabularx}{\linewidth}{p{0.45\linewidth}| p{0.45\linewidth}} 
                $
                \infer{
                    \agTriple{\alphabetExtensionOfTo{\ppa}{\alphabetOf{\multiobjectiveQuery{\safe}{A}_{\land}} }, \region_{\regionIntersection} }{\prt}{\multiobjectiveQuery{\safe}{A}_{\land}}{\multiobjectiveQuery{\safe}{G}_{\land}} 
                } 
                {
                    \deduce{
                        \agTriple{\alphabetExtensionOfTo{\ppa}{\alphabetOf{\multiobjectiveQuery{\safe}{A}_{2}}}, \region_2 }{\prt}{\multiobjectiveQuery{\safe}{A}_{2}}{\multiobjectiveQuery{\safe}{G}_2} 
                    }{
                        \agTriple{\alphabetExtensionOfTo{\ppa}{\alphabetOf{\multiobjectiveQuery{\safe}{A}_{1}}}, \region_1 }{\prt}{\multiobjectiveQuery{\safe}{A}_{1}}{\multiobjectiveQuery{\safe}{G}_1} 
                    }
                    } 
                $
        & 
        $
        \infer{
            \agTriple{\alphabetExtensionOfTo{\ppa}{\alphabetOf{\multiobjectiveQuery{}{A}_{\land}} }, \region_{\regionIntersection} }{\fairWrtRegionModel{}{\decomp_1 \cup \decomp_2}}{\multiobjectiveQuery{}{A}_{\land}}{\multiobjectiveQuery{}{G}_{\land}} 
        } 
        {
            \deduce{
                \agTriple{\alphabetExtensionOfTo{\ppa}{\alphabetOf{\multiobjectiveQuery{}{A}_{2}}}, \region_2 }{\fairWrtRegionModel{}{\decomp_2}}{\multiobjectiveQuery{}{A}_{2}}{\multiobjectiveQuery{}{G}_2} 
            }{
                \agTriple{\alphabetExtensionOfTo{\ppa}{\alphabetOf{\multiobjectiveQuery{}{A}_{1}}}, \region_1 }{\fairWrtRegionModel{}{\decomp_1}}{\multiobjectiveQuery{}{A}_{1}}{\multiobjectiveQuery{}{G}_1} 
            }
            } 
        $
        \\ %
        \\
        &  
        for $\decomp_i \in \setOfdecompsOf{\alphabetOf{\ppa} \cup \alphabetOf{\multiobjectiveQuery{}{A}_i}}$. \\
    \end{tabularx}
\end{restatable}
\begin{proof}
    We show the statement for safety properties. For general properties and fair strategies, the proof is analogous. 
    First, we observe that ${\multiobjectiveQuery{\safe}{G}_{\land}}$ is a valid multiobjective query for $\alphabetExtensionOfTo{\ppa}{\alphabetOf{\multiobjectiveQuery{}{A}_{\land}}}$ as 
    $\alphabetOf{\multiobjectiveQuery{\safe}{G}_i} \subseteq \alphabetOf{\ppa} \cup \alphabetOf{\multiobjectiveQuery{\safe}{A}_i} = \alphabetOf{\alphabetExtensionOfTo{\ppa}{\alphabetOf{\multiobjectiveQuery{}{A}_{\land}}}}$ holds. 
    Second, assume the premises hold. 
    As $\alphabetOf{\multiobjectiveQuery{\safe}{G}_i} \subseteq \alphabetOf{\ppa} \cup \alphabetOf{\multiobjectiveQuery{\safe}{A}_i}$,  
    it holds that for $i \in \{1,2\}$: 
        \[ 
            \agTriple{\alphabetExtensionOfTo{\ppa}{\alphabetOf{\multiobjectiveQuery{\safe}{A}_1 \land \multiobjectiveQuery{\safe}{A}_2}}, \region_i}{\prt}{\multiobjectiveQuery{\safe}{A}_i}{\multiobjectiveQuery{\safe}{G}_i}
        \]
	From this it follows that 
        \[ 
            \agTriple{\alphabetExtensionOfTo{\ppa}{\alphabetOf{\multiobjectiveQuery{\safe}{A}_\land}}, \regionIntersectionOf{\region_1}{}{\region_2} }{\prt}{\multiobjectiveQuery{\safe}{A}_i}{\multiobjectiveQuery{\safe}{G}_i}
        \]
    As in the proof of \cite[Theorem 3]{Kwi+13}, by the definition of pAG triples (\Cref{def_pag_triple}), 
    we observe, that for $\valuation \in \regionIntersectionOf{\region_1}{}{\region_2}$ and $\strategy \in \strategysetOf{\alphabetExtensionOfTo{\ppa}{\alphabetOf{\multiobjectiveQuery{\safe}{A}_\land}}}{\prt}$ it holds that 
	\begin{align*}
	 & \bigwedge_{i\in\{1,2\}} \left(\alphabetExtensionOfTo{\ppa}{\alphabetOf{\multiobjectiveQuery{\safe}{A}_1 \land \multiobjectiveQuery{\safe}{A}_2}} [\valuation], {\strategy} \models \multiobjectiveQuery{\safe}{A}_i  \quad  \rightarrow  \quad  \alphabetExtensionOfTo{\ppa}{\alphabetOf{\multiobjectiveQuery{\safe}{A}_1 \land \multiobjectiveQuery{\safe}{A}_2}}[\valuation], {\strategy} \models \multiobjectiveQuery{\safe}{G}_i \right)  \\
	  \Leftrightarrow\quad  
      & \bigwedge_{i\in\{1,2\}} \left(\alphabetExtensionOfTo{\ppa}{\alphabetOf{\multiobjectiveQuery{\safe}{A}_1 \land \multiobjectiveQuery{\safe}{A}_2}} [\valuation], {\strategy} \not\models \multiobjectiveQuery{\safe}{A}_i  \quad  \lor \quad  \alphabetExtensionOfTo{\ppa}{\alphabetOf{\multiobjectiveQuery{\safe}{A}_1 \land \multiobjectiveQuery{\safe}{A}_2}}[\valuation], {\strategy} \models \multiobjectiveQuery{\safe}{G}_i \right) \\
	  \Rightarrow \quad &
	    \alphabetExtensionOfTo{\ppa}{\alphabetOf{\multiobjectiveQuery{\safe}{A}_\land}} [\valuation],{\strategy} \not\models \multiobjectiveQuery{\safe}{A}_\land  \quad \lor \quad
	    \alphabetExtensionOfTo{\ppa}{\alphabetOf{\multiobjectiveQuery{\safe}{A}_\land}}[\valuation],{\strategy} \models \multiobjectiveQuery{\safe}{G}_\land \\
	\Leftrightarrow \quad &	\alphabetExtensionOfTo{\ppa}{\alphabetOf{\multiobjectiveQuery{\safe}{A}_\land}} [\valuation], {\strategy} \models \multiobjectiveQuery{\safe}{A}_\land  \quad  \rightarrow \quad  \alphabetExtensionOfTo{\ppa}{\alphabetOf{\multiobjectiveQuery{\safe}{A}_\land}}[\valuation], {\strategy} \models \multiobjectiveQuery{\safe}{G}_\land
\end{align*}  
Thus, 
\begin{align*}
	{\agTriple{\alphabetExtensionOfTo{\ppa}{\alphabetOf{\multiobjectiveQuery{\safe}{A}_\land}}, \regionIntersectionOf{\region_1}{}{\region_2} }{\prt}{ \multiobjectiveQuery{\safe}{A}_\land}{ \multiobjectiveQuery{\safe}{G}_\land}}
\end{align*}
    
\end{proof} 
By combining \Cref{theo:pag_asym_rule} and \Cref{theo:pag_conj_rule}, we obtain the 
following analogue to \cite[Theorem 4]{Kwi+13}---the asymmetric proof rule  for $n>2$ components---for pPA.  %
\begin{restatable}[Asymmetric-N]{thm}{restatableAsymNRule}\label{theo:pag_asymN_rule}
    Let $\region_{\regionIntersection}  = \regionIntersectionOf{\region_1}{\dots}{\region_n}$, and $ \ppa =  \ppa_1 \parallel \dots \parallel \ppa_n$. 

    \begin{tabularx}{\linewidth}{p{0.45\linewidth}|p{0.45\linewidth}} 
                $
                \infer{
                    \ppa, \region_{\regionIntersection} \modelsWrt{\comp} \multiobjectiveQuery{\safe}{G}
                } 
                {
                    \deduce{
                        \agTriple{\alphabetExtensionOfTo{\ppa_n}{\alphabetOf{\multiobjectiveQuery{\safe}{A}_{n-1}}}, \region_{n} }{\prt}{\multiobjectiveQuery{\safe}{A}_{n-1}}{\multiobjectiveQuery{\safe}{G}} 
                    }{
                        \deduce{
                            \vdots
                            }{
                            \deduce{
                                \agTriple{\alphabetExtensionOfTo{\ppa_2}{\alphabetOf{\multiobjectiveQuery{\safe}{A}_{1}}}, \region_2 }{\prt}{\multiobjectiveQuery{\safe}{A}_{1}}{\multiobjectiveQuery{\safe}{A}_2} 
                                }{
                                    \deduce{
                                        \ppa_1, \region_1 \modelsWrt{\comp} \multiobjectiveQuery{\safe}{A}_1
                                        }{}
                                }
                        }
                }} 
                $
        & 
            $
                \infer{
                    \ppa, \region_{\regionIntersection} \modelsWrt{\fairWrtRegionModel{}{
                    \decomp_1 \cup  \dots \cup \decomp_n}} \multiobjectiveQuery{}{G}
                }
                {\deduce{
                    \agTriple{\alphabetExtensionOfTo{\ppa_n}{\alphabetOf{\multiobjectiveQuery{}{A}_{n-1}}}, \region_{n}}{\fairWrtRegionModel{}{\decomp_n}}{\multiobjectiveQuery{}{A}_{n-1}}{\multiobjectiveQuery{}{G}} 
                    }{
                        \deduce{
                            \vdots
                        }{
                            \deduce{
                                \agTriple{\alphabetExtensionOfTo{\ppa_2}{\alphabetOf{\multiobjectiveQuery{}{A}_{1}}}, \region_2 }{\fairWrtRegionModel{}{\decomp_2}}{\multiobjectiveQuery{}{A}_{1}}{\multiobjectiveQuery{}{A}_2} 
                            }{
                                \deduce{
                                    \ppa_1, \region_1 \modelsWrt{\fairWrtRegionModel{}{\decomp_1}} \multiobjectiveQuery{}{A}_1
                                }{}
                            }
                        }
                    }
                }
             $
        \\ %
        \\
        &  
       for $\decomp_1 \in \setOfdecompsOf{\alphabetOf{\ppa_1}}$ and $\decomp_i \in \setOfdecompsOf{\alphabetOf{\ppa_i}\cup \alphabetOf{\multiobjectiveQuery{}{A}_{i-1}}}$, for $i>1$.  
    \end{tabularx}
    
\end{restatable}
\begin{proof}
    Analogously to \cite[Theorem 4]{Kwi+13}, the result follows by a repeated application of \Cref{theo:pag_asym_rule} and \Cref{theo:pag_conj_rule}. 
\end{proof}

Next, we extend \cite[Theorem 6]{Kwi+13}\textemdash{}a rule for reasoning about properties with non-synchronized actions.  
\begin{restatable}[Interleaving]{thm}{restatableInterleavingRule}\label{theo:pag_interleaving_rule}
Let $\region_{\regionIntersection} = \regionIntersectionOf{\region_1}{}{\region_2}$, $\multiobjectiveQuery{}{A}_{\land} = \multiobjectiveQuery{}{A}_{1} \land \multiobjectiveQuery{}{A}_{2}$, $\ppa = \alphabetExtensionOfTo{(\ppa_1 \parallel \ppa_2)}{\alphabetOf{\multiobjectiveQuery{}{A}_{\land}}}$ and $(\alphabetOf{\ppa_1} \cup \alphabetOf{\multiobjectiveQuery{}{A}_1} )\cap (\alphabetOf{\ppa_2} \cup \alphabetOf{\multiobjectiveQuery{}{A}_2}) = \emptyset$.

    \begin{tabularx}{\linewidth}{p{0.45\linewidth}|p{0.45\linewidth}} 
         $
                \infer{
                    \agTriple{{ \ppa}, \region_{\regionIntersection}}{\prt}{\multiobjectiveQuery{\safe}{A}_{\land} }{\probPredicate{\substack{\geq p_1+p_2 \\ \ - p_1\cdot p_2}}{\regLang_1^{\safe} \cup \regLang_2^{\safe}}}
                } 
                {
                    \deduce{
                        \agTriple{\alphabetExtensionOfTo{\ppa_2}{\alphabetOf{\multiobjectiveQuery{\safe}{A}_{2}}}, \region_2 }{\prt}{\multiobjectiveQuery{\safe}{A}_{2}}{{\probPredicate{\geq p_2}{\regLang_2^{\safe}} }} 
                    }{
                                    \deduce{
                                        \agTriple{\alphabetExtensionOfTo{\ppa_1}{\alphabetOf{\multiobjectiveQuery{\safe}{A}_{1}}}, \region_1 }{\prt}{\multiobjectiveQuery{\safe}{A}_{1}}{\probPredicate{\geq p_1}{\regLang_1^{\safe}}  }
                                        }{}
                }} 
             $
        & 
       $ 
            \infer{
                \agTriple{{ \ppa}, \region_{\regionIntersection}}{{\fairWrtRegionModel{}{ \decomp_1 \cup \decomp_2 }  }  }{\multiobjectiveQuery{}{A}_{\land} }{\probPredicate{ \substack{\sim p_1+p_2 \\ \  - p_1\cdot p_2}}{\regLang_1^{} \cup \regLang_2^{}}}
            } 
            {
                \deduce{
                    \agTriple{\alphabetExtensionOfTo{\ppa_2}{\alphabetOf{\multiobjectiveQuery{}{A}_{2}}}, \region_2 }{\fairWrtRegionModel{}{\decomp_2}}{\multiobjectiveQuery{}{A}_{2}}{{\probPredicate{\sim p_2}{\regLang_2^{}} }} 
                }{
                                \deduce{
                                    \agTriple{\alphabetExtensionOfTo{\ppa_1}{\alphabetOf{\multiobjectiveQuery{}{A}_{1}}}, \region_1 }{\fairWrtRegionModel{}{\decomp_1}}{\multiobjectiveQuery{}{A}_{1}}{\probPredicate{\sim p_1}{\regLang_1^{}}  }
                                    }{}
            }} 
        $
      \\ %
      \\
        &  
        for $\sim \in \{ <, \leq, >, \geq\}$ and 
       $\decomp_i \in \setOfdecompsOf{\alphabetOf{\ppa_i} \cup \alphabetOf{\multiobjectiveQuery{}{A}_{i}}}$. 
    \end{tabularx}
\end{restatable}

\begin{proof}
    The proof is based on the proof of \cite[Theorem 6]{Kwi+13}. 
    First, we have that ${\probPredicate{ \substack{\sim p_1+p_2 \\ \  - p_1\cdot p_2}}{\regLang_1^{(\safe)} \cup \regLang_2^{(\safe)}}}$ is a valid query for $\ppa$, since $\alphabetOf{{\probPredicate{ \substack{\sim p_1+p_2 \\ \  - p_1\cdot p_2}}{\regLang_1^{(\safe)} \cup \regLang_2^{(\safe)}}}} \subseteq \alphabetOf{\ppa} $.  
    We show the proof of correctness of the rule involving general predicates and fair strategies. 
    For safety properties the proof works analogously. 

    If $\regionIntersectionOf{\region_{1}}{}{\region_2} = \emptyset$, the conclusion trivially holds. 
    We assume $\regionIntersectionOf{\region_{1}}{}{\region_2} \not = \emptyset$. 
    We show that $\agTriple{{\alphabetExtensionOfTo{\ppa_1}{\alphabetOf{\multiobjectiveQuery{}{A}_1}}\parallel \alphabetExtensionOfTo{\ppa_2}{\alphabetOf{\multiobjectiveQuery{}{A}_2}}}, \region_{\regionIntersection}}{{\fairWrtRegionModel{}{ \decomp_1 \cup \decomp_2 }  }  }{\multiobjectiveQuery{}{A}_{\land} }{\probPredicate{ \substack{\sim p_1+p_2 \\ \  - p_1\cdot p_2}}{\regLang_1^{} \cup \regLang_2^{}}}$
    holds, which implies that 
    $\agTriple{{ \ppa}, \region_{\regionIntersection}}{{\fairWrtRegionModel{}{ \decomp_1 \cup \decomp_2 }  }  }{\multiobjectiveQuery{}{A}_{\land} }{\probPredicate{ \substack{\sim p_1+p_2 \\ \  - p_1\cdot p_2}}{\regLang_1^{} \cup \regLang_2^{}}}$ holds. 

	Let $\valuation$ be a valuation in $\regionIntersectionOf{\region_1}{}{\region_2}$ and $\strategy$ be a $\fairWrtRegionModel{}{\decomp_1 \cup  \decomp_2}$ strategy of 
	$(\alphabetExtensionOfTo{\ppa_1}{\alphabetOf{\multiobjectiveQuery{}{A}_1}}\parallel \alphabetExtensionOfTo{\ppa_2}{\alphabetOf{\multiobjectiveQuery{}{A}_2}})[\valuation]$. 
	We show	${(\alphabetExtensionOfTo{\ppa_1}{\alphabetOf{\multiobjectiveQuery{}{A}_1}}\parallel \alphabetExtensionOfTo{\ppa_2}{\alphabetOf{\multiobjectiveQuery{}{A}_2}})}[\valuation], {\strategy}\models  \multiobjectiveQuery{}{A}_1 \land \multiobjectiveQuery{}{A}_1[ \valuation]  \text{ implies }
    {(\alphabetExtensionOfTo{\ppa_1}{\alphabetOf{\multiobjectiveQuery{}{A}_1}}\parallel \alphabetExtensionOfTo{\ppa_2}{\alphabetOf{\multiobjectiveQuery{}{A}_2}})}[\valuation], {\strategy} \models  \probPredicate{\substack{\sim p_1+p_2 \\ \  - p_1\cdot p_2}}{\regLang_1 \cup \regLang_2}$. 
	
	Assume ${(\alphabetExtensionOfTo{\ppa_1}{\alphabetOf{\multiobjectiveQuery{}{A}_1}}\parallel \alphabetExtensionOfTo{\ppa_2}{\alphabetOf{\multiobjectiveQuery{}{A}_2}})}[\valuation], {\strategy} \models  \multiobjectiveQuery{}{A}_1 \land \multiobjectiveQuery{}{A}_1[\valuation]$ holds. 
			
    We have $\agTriple{\alphabetExtensionOfTo{\ppa_i}{\multiobjectiveQuery{}{A}_i}, \regionIntersectionOf{\region_1}{}{\region_2} }{\fairWrtRegionModel{}{\decomp_i} }{\multiobjectiveQuery{}{A}_i}{\probPredicate{\sim p_i}{\regLang_i}}$ for $i \in \{1,2\}$ by the premises.  
    Thus, 
	\begin{align}
		& {(\alphabetExtensionOfTo{\ppa_1}{\alphabetOf{\multiobjectiveQuery{}{A}_1}}\parallel \alphabetExtensionOfTo{\ppa_2}{\alphabetOf{\multiobjectiveQuery{}{A}_2}})}[\valuation], {\strategy} \models   \multiobjectiveQuery{}{A}_i[ \valuation ] \nonumber \\
		& \Rightarrow \alphabetExtensionOfTo{\ppa_i }{\alphabetOf{\multiobjectiveQuery{}{A}_i}}[\valuation], {\stratProjOfToValuation{\strategy}{i}{\valuation}} \models   \multiobjectiveQuery{}{A}_i[ \valuation ]
		&&  \text{By \Cref{lemma_3_fDependent}} \nonumber \\
		& \Rightarrow 
		 \alphabetExtensionOfTo{\ppa_i}{\alphabetOf{\multiobjectiveQuery{}{A}_i}}[{\valuation}], \stratProjOfToValuation{\strategy}{i}{\valuation} \models   \probPredicate{\sim p_i}{\regLang_i}
		 && \text{$ \stratProjOfToValuation{\strategy}{i}{\valuation} \in \strategysetOf{{\alphabetExtensionOfTo{\ppa_i}{\alphabetOf{\multiobjectiveQuery{}{A}_i} }}[\valuation]}{\fairWrtRegionModel{}{\decomp_i}}$} \nonumber\\
		 & && \text{and $\agTriple{\alphabetExtensionOfTo{\ppa_i}{\alphabetOf{\multiobjectiveQuery{}{A}_i}}, \regionIntersectionOf{\region_1}{}{\region_2} }{\fairWrtRegionModel{}{\decomp_i}}{\multiobjectiveQuery{}{A}_i}{\probPredicate{\sim p_i}{\regLang_i}}$} \nonumber \\
		 & \Rightarrow \PrOf{\alphabetExtensionOfTo{\ppa_i}{\alphabetOf{\multiobjectiveQuery{}{A}_i}}}{{\valuation},\stratProjOfToValuation{\strategy}{i}{\valuation}}{\regLang_i}  \sim p_i && \label{eq:proof_async_quant_G_sat} 
	\end{align}
    We have $(\alphabetOf{\ppa_1} \cup \alphabetOf{\multiobjectiveQuery{}{A}_1} )\cap (\alphabetOf{\ppa_2} \cup \alphabetOf{\multiobjectiveQuery{}{A}_2}) = \emptyset$. 
    Thus,
    \begin{align*}
    &	\PrOf{{(\alphabetExtensionOfTo{\ppa_1}{\alphabetOf{\multiobjectiveQuery{}{A}_1}}\parallel \alphabetExtensionOfTo{\ppa_2}{\alphabetOf{\multiobjectiveQuery{}{A}_2}})}}{\valuation,\strategy}{\{
            \pi \in \infPathsOf{{(\alphabetExtensionOfTo{\ppa_1}{\alphabetOf{\multiobjectiveQuery{}{A}_1}}\parallel \alphabetExtensionOfTo{\ppa_2}{\alphabetOf{\multiobjectiveQuery{}{A}_2}})}[\valuation]} \mid \restrOfTo{\traceOf{\pi}}{\alphabetOf{\regLang_1 \cup \regLang_2}}  \not \in \regLang_1 \cup \regLang_2 \}} \\
        %
        %
        %
        &	\quad  =  \prod_{i \in \{1,2\}}	
            \PrOf{\alphabetExtensionOfTo{\ppa_i}{\alphabetOf{\multiobjectiveQuery{}{A}_i }}}{\valuation,\stratProjOfToValuation{\strategy}{i}{\valuation}}{  \{ \pi \in \infPathsOf{\alphabetExtensionOfTo{\ppa_i }{\alphabetOf{\multiobjectiveQuery{}{A}_i}}[\valuation]} \mid \restrOfTo{\traceOf{\pi}}{\alphabetOf{\regLang_i}}  \not \in \regLang_i \} } 
        \\
        &	\quad  =  \left(1-\PrOf{\alphabetExtensionOfTo{\ppa_1}{\alphabetOf{\multiobjectiveQuery{}{A}_1}}}{{\valuation},\stratProjOfToValuation{\strategy}{1}{\valuation}}{\regLang_1} \right) \cdot
            \left(1 - \PrOf{\alphabetExtensionOfTo{\ppa_2}{2}}{{\valuation},\stratProjOfToValuation{\strategy}{2}{\valuation}}{\regLang_2}  \right)
    \end{align*}
    Then, by \Cref{eq:proof_async_quant_G_sat}: 
    \begin{align*}
        & \PrOf{{(\alphabetExtensionOfTo{\ppa_1}{\alphabetOf{\multiobjectiveQuery{}{A}_1}}\parallel \alphabetExtensionOfTo{\ppa_2}{\alphabetOf{\multiobjectiveQuery{}{A}_2}})}}{\valuation,\strategy}{\regLang_1 \cup \regLang_2} 
        \\
        & = 1- 	\PrOf{{(\alphabetExtensionOfTo{\ppa_1}{\alphabetOf{\multiobjectiveQuery{}{A}_1}}\parallel \alphabetExtensionOfTo{\ppa_2}{\alphabetOf{\multiobjectiveQuery{}{A}_2}})}}{\valuation,\strategy}{ \{
            \pi \in \infPathsOf{(\alphabetExtensionOfTo{\ppa_1}{\alphabetOf{\multiobjectiveQuery{}{A}_1}}\parallel \alphabetExtensionOfTo{\ppa_2}{\alphabetOf{\multiobjectiveQuery{}{A}_2}})[\valuation]}  \mid \restrOfTo{\traceOf{\pi}}{\alphabetOf{\regLang_1 \cup \regLang_2}}  \not \in \regLang_1 \cup \regLang_2 \}} \\
        & =1- \left(1-\underbrace{\PrOf{\alphabetExtensionOfTo{\ppa_1}{\alphabetOf{\multiobjectiveQuery{}{A}_1}}}{{\valuation},\stratProjOfToValuation{\strategy}{1}{\valuation}}{\regLang_1}}_{\sim p_1} \right) \cdot
        \left(1 - \underbrace{\PrOf{\alphabetExtensionOfTo{\ppa_2}{\alphabetOf{\multiobjectiveQuery{}{A}_2}}}{{\valuation},\stratProjOfToValuation{\strategy}{2}{\valuation}}{\regLang_2}}_{\sim p_2} \right) \\
        & \sim p_1 +p_2 -p_1\cdot p_2
    \end{align*}
\end{proof}
Lastly, we lift \cite[Theorem 7]{Kwi+13} to pPA, which reasons about the sum of two (parametric) reward functions $\rewFct_1$, $\rewFct_2$ for $\ppa_1$ and $\ppa_2$, respectively, in the composition $\ppa_1 \parallel \ppa_2$. 
The sum of $\rewFct_1$ and $\rewFct_2$ is reward function $(\rewFct_1 + \rewFct_2) $ on $\alphabetOf{\rewFct_1} \cup  \alphabetOf{\rewFct_2} $ over the parameter set $\parameterSetOf{\rewFct_1} \cup \parameterSetOf{\rewFct_2}$. 
\begin{align*}
    (\rewFct_1 + \rewFct_2)(a) = 
    \begin{cases}
        \rewFct_1(a) + \rewFct_2(a), & \text{for } a \in \alphabetOf{\rewFct_1} \cap \alphabetOf{\rewFct_2} \\
        \rewFct_1(a), & \text{for }a \in \alphabetOf{\rewFct_1} \setminus \alphabetOf{\rewFct_2} \\
        \rewFct_2(a), & \text{for }a \in \alphabetOf{\rewFct_2} \setminus \alphabetOf{\rewFct_1}. 
    \end{cases}
\end{align*}
\begin{restatable}[Reward Sum]{thm}{restatableRewardRule}\label{theo:pag_reward_rule}

   Let $\region_{\regionIntersection} = \regionIntersectionOf{\region_1}{}{\region_2}$, $\multiobjectiveQuery{}{A}_{\land} = \multiobjectiveQuery{}{A}_{1} \land \multiobjectiveQuery{}{A}_{2}$, and $\ppa = \alphabetExtensionOfTo{(\ppa_1 \parallel \ppa_2)}{\alphabetOf{\multiobjectiveQuery{}{A}_{\land}}}$. 

   $
    \infer{
            \agTriple{{ \ppa}, \region_{\regionIntersection}}{{\fairWrtRegionModel{}{ \decomp_1 \cup \decomp_2 }  }  }{\multiobjectiveQuery{}{A}_{\land} }{\expPredicate{\sim r_1 +r_2 }{\rewFct_1 + \rewFct_2}}
        }
        {\deduce{
                \agTriple{\alphabetExtensionOfTo{\ppa_2}{\multiobjectiveQuery{}{A}_2}, \region_2 }{\fairWrtRegionModel{}{  \decomp_2} }{\multiobjectiveQuery{}{A}_2}{\expPredicate{\sim r_2}{\rewFct_2}}
            }
            {
                \agTriple{\alphabetExtensionOfTo{\ppa_1}{\multiobjectiveQuery{}{A}_1}, \region_1 }{\fairWrtRegionModel{}{ \decomp_1} }{\multiobjectiveQuery{}{A}_1}{\expPredicate{\sim r_1}{\rewFct_1}}
            }
        }
    $

        for  $\sim \in \{ <, \leq, >, \geq\}$ and 
       $\decomp_i \in \setOfdecompsOf{\alphabetOf{\ppa_i}\cup \alphabetOf{\multiobjectiveQuery{}{A}_{i}}}$. 
\end{restatable}
\begin{proof}
        The proof is similar to the proof of \cite[Theorem 7]{Kwi+13}. 
        Since $\alphabetOf{\expPredicate{\sim r_1 + r_2}{\rewFct_1 + \rewFct_2}} \subseteq \alphabetOf{\ppa} $,  
        We have that ${\expPredicate{\sim r_1 + r_2}{\rewFct_1 + \rewFct_2}}$ is a valid query for $\ppa$, 
        If $\regionIntersectionOf{\region_{1}}{}{\region_2} = \emptyset$, the conclusion trivially holds. 
        We assume $\regionIntersectionOf{\region_{1}}{}{\region_2} \not = \emptyset$. 

        We show that $\agTriple{{\alphabetExtensionOfTo{\ppa_1}{\alphabetOf{\multiobjectiveQuery{}{A}_1}}\parallel \alphabetExtensionOfTo{\ppa_2}{\alphabetOf{\multiobjectiveQuery{}{A}_2}}}, \region_{\regionIntersection}}{{\fairWrtRegionModel{}{ \decomp_1 \cup \decomp_2 }  }  }{\multiobjectiveQuery{}{A}_{\land} }{{\expPredicate{\sim r_1 + r_2}{\rewFct_1 + \rewFct_2}}}$
        holds, which implies that 
        $\agTriple{{ \ppa}, \region_{\regionIntersection}}{{\fairWrtRegionModel{}{ \decomp_1 \cup \decomp_2 }  }  }{\multiobjectiveQuery{}{A}_{\land} }{{\expPredicate{\sim r_1 + r_2}{\rewFct_1 + \rewFct_2}}}$ holds. 
    
        Let $\valuation$ be a valuation in $\regionIntersectionOf{\region_1}{}{\region_2}$ and $\strategy$  a $\fairWrtRegionModel{}{\decomp_1 \cup  \decomp_2 }$ strategy of 
        $\alphabetExtensionOfTo{\ppa_1}{\alphabetOf{\multiobjectiveQuery{}{A}_1}}\parallel \alphabetExtensionOfTo{\ppa_2}{\alphabetOf{\multiobjectiveQuery{}{A}_2}}[\valuation]$. 
        Assume 
        $ {(\alphabetExtensionOfTo{\ppa_1}{\alphabetOf{\multiobjectiveQuery{}{A}_1}}\parallel \alphabetExtensionOfTo{\ppa_2}{\alphabetOf{\multiobjectiveQuery{}{A}_2}})}[\valuation],{\strategy} \models  (\multiobjectiveQuery{}{A}_1 \land \multiobjectiveQuery{}{A}_2)[\valuation ]$. 
        We show ${{\alphabetExtensionOfTo{\ppa_1}{\alphabetOf{\multiobjectiveQuery{}{A}_1}}\parallel \alphabetExtensionOfTo{\ppa_2}{\alphabetOf{\multiobjectiveQuery{}{A}_2}}}}[\valuation],{\strategy} \models {\expPredicate{\sim r_1 +r_2 }{\rewFct_1 + \rewFct_2}}[\valuation]$. 
        The premises imply that $\agTriple{\alphabetExtensionOfTo{\ppa_i}{\multiobjectiveQuery{}{A}_i}, \regionIntersectionOf{\region_1}{}{\region_2} }{\fairWrtRegionModel{}{\decomp_i} }{\multiobjectiveQuery{}{A}_i}{\expPredicate{\sim r_i}{\rewFct_i}}$ holds 
        for $i \in \{1,2\}$. 
        Thus, 
        \begin{align}
            & {(\alphabetExtensionOfTo{\ppa_1}{\alphabetOf{\multiobjectiveQuery{}{A}_1}}\parallel \alphabetExtensionOfTo{\ppa_2}{\alphabetOf{\multiobjectiveQuery{}{A}_2}})}[\valuation],{\strategy} \models   \multiobjectiveQuery{}{A}_i[ \valuation ] \nonumber \\
            & \Rightarrow \alphabetExtensionOfTo{\ppa_i }{\alphabetOf{\multiobjectiveQuery{}{A}_i}}[\valuation],{\stratProjOfToValuation{\strategy}{i}{\valuation} } \models   \multiobjectiveQuery{}{A}_i[ \valuation ]
            &&  \text{By \Cref{lemma_3_fDependent}} \nonumber \\
            & \Rightarrow 
            \alphabetExtensionOfTo{\ppa_i }{\alphabetOf{\multiobjectiveQuery{}{A}_i}}[{\valuation}],{\stratProjOfToValuation{\strategy}{i}{\valuation} } \models  {\expPredicate{\sim r_i  }{\rewFct_i }}[\valuation]
            && \text{$ \stratProjOfToValuation{\strategy}{i}{\valuation} \in \strategysetOf{{\alphabetExtensionOfTo{\ppa_i}{\alphabetOf{\multiobjectiveQuery{}{A}_i} }}[\valuation]}{\fairWrtRegionModel{}{\decomp_i}}$} \nonumber\\
            & && \text{as $\agTriple{\alphabetExtensionOfTo{\ppa_i}{\alphabetOf{\multiobjectiveQuery{}{A}_i}}, \regionIntersectionOf{\region_1}{}{\region_2} }{\fairWrtRegionModel{}{\decomp_i}}{\multiobjectiveQuery{}{A}_i}{\expPredicate{\sim r_i  }{\rewFct_i }}$} \nonumber 
        \end{align}
        It follows that 
        \begin{align}	\ExpTot{\alphabetExtensionOfTo{\ppa_i}{\alphabetOf{\multiobjectiveQuery{}{A}_i}}}{{\valuation},\stratProjOfToValuation{\strategy}{i}{\valuation}}{\rewFct_i}  \sim r_i && \label{eq:proof_rew_quant_G_sat} 
        \end{align}
        Thus, 
        \begin{align*}
            & \ExpTot{{(\alphabetExtensionOfTo{\ppa_1}{\alphabetOf{\multiobjectiveQuery{}{A}_1}}\parallel \alphabetExtensionOfTo{\ppa_2}{\alphabetOf{\multiobjectiveQuery{}{A}_2}})}}{\valuation,\strategy}{\rewFct_1+\rewFct_2}  \\
            & = \int_{\pi } (\rewFct_1 + \rewFct_2)[\valuation](\pi)  \,d 	\PrOf{{(\alphabetExtensionOfTo{\ppa_1}{\alphabetOf{\multiobjectiveQuery{}{A}_1}}\parallel \alphabetExtensionOfTo{\ppa_2}{\alphabetOf{\multiobjectiveQuery{}{A}_2}})}}{\valuation,\strategy}{}
            && 
            \text{By definition}
            \\
            &  = \sum_{i \in \{1,2\}} \int_{\pi } (\rewFct_i[\valuation])(\pi)  \,d 	\PrOf{{(\alphabetExtensionOfTo{\ppa_1}{\alphabetOf{\multiobjectiveQuery{}{A}_1}}\parallel \alphabetExtensionOfTo{\ppa_2}{\alphabetOf{\multiobjectiveQuery{}{A}_2}})}}{\valuation,\strategy}{} 
            &&   \\
            & = \sum_{i \in \{1,2\}} \ExpTot{{(\alphabetExtensionOfTo{\ppa_1}{\alphabetOf{\multiobjectiveQuery{}{A}_1}}\parallel \alphabetExtensionOfTo{\ppa_2}{\alphabetOf{\multiobjectiveQuery{}{A}_2}})}}{\valuation,\strategy}{\rewFct_i}  && 
            \text{By definition} \\
                & = \sum_{i \in \{1,2\}} \ExpTot{\left( \alphabetExtensionOfTo{\ppa_1}{\alphabetOf{\multiobjectiveQuery{}{A}_1}} \parallel \alphabetExtensionOfTo{\ppa_2}{\alphabetOf{\multiobjectiveQuery{}{A}_2}}
                    \right)}{\valuation,\strategy}{\rewFct_i}  
                &&  \\
                &  = \sum_{i \in \{1,2\}} \underbrace{\ExpTot{ \alphabetExtensionOfTo{\ppa_i}{\alphabetOf{\multiobjectiveQuery{}{A}_i}} 
                }{\valuation,\stratProjOfToValuation{\strategy}{i}{\valuation}}{\rewFct_i} }_{\sim r_i}
                && \text{By \Cref{lemma_3_dDependent}}\\
                & \sim r_1 +r_2  &&\text{ by \Cref{eq:proof_rew_quant_G_sat}}
        \end{align*}
        Thus, we have shown  
        \[
            {(\alphabetExtensionOfTo{\ppa_1}{\alphabetOf{\multiobjectiveQuery{}{A}_1}}\parallel \alphabetExtensionOfTo{\ppa_2}{\alphabetOf{\multiobjectiveQuery{}{A}_2}})}[\valuation],{\strategy} \models {\expPredicate{\sim r_1 +r_2 }{\rewFct_1 + \rewFct_2}}[\valuation]
        \]
    \end{proof}

    \section{Further AG Proof Rules for rPAs}
\label{ap:rPA_conv_rules}
We lift several additional AG proof rules from the PA setting of~\cite{Kwi+13} to convex rPAs under the convex composition operator~$\parallelConv$. 
In particular, \Cref{theo:rpa_conj_rule} enables reasoning about the conjunction of multi-objective queries, and \Cref{theo:rPA_asymN_rule} provides an asymmetric rule for more than two components, extending the two-component rule \Cref{theo_rpa_asym_rule}. 
In addition, \Cref{theo:rpa_interleaving_rule} handles interleaving components (i.e., components without synchronising labels). 
Again, we restrict to safety properties and leave general properties w.r.t.\ fair satisfaction to future work.

When the premises include $\rpa  \modelsWrt{\star} \multiobjectiveQuery{}{A}$ we adhere to the following assumptions:
\begin{itemize}
    \item $\multiobjectiveQuery{}{A}$ is a safety mo-query for $\rpa$, i.e., $\multiobjectiveQuery{}{A}$ is defined over the alphabet  $\alphabetOf{\multiobjectiveQuery{}{A}} \subseteq \alphabetOf{\rpa}$ 
\end{itemize}
Similarly, if an AG triple $\agTriple{\alphabetExtensionOfTo{\rpa}{\alphabetOf{\multiobjectiveQuery{}{A}}}}{\star}{\multiobjectiveQuery{}{A}}{\multiobjectiveQuery{}{G}}$ occurs in the premises, we assume:
\begin{itemize}
    \item $\multiobjectiveQuery{}{G}$ is a safety mo-query for $\alphabetExtensionOfTo{\rpa}{\alphabetOf{\multiobjectiveQuery{}{A}}}$, i.e., $\multiobjectiveQuery{}{G}$ is defined over the alphabet $\alphabetOf{\multiobjectiveQuery{}{G}} \subseteq \alphabetOf{\rpa} \cup \alphabetOf{\multiobjectiveQuery{}{A}}$\footnote{$\multiobjectiveQuery{(\safe)}{A} \not \subseteq \alphabetOf{\rpa}$ is possible.}
\end{itemize}

To reason about the conjunction of multi-objective queries, we generalise \cite[Theorem 3]{Kwi+13} for PAs (see \Cref{theo:pag_conj_rule} for pPAs) to rPAs. 
\begin{thm}[Conjunction]\label{theo:rpa_conj_rule} 
    Let $\rpa$ be a convex rPA and let $\multiobjectiveQuery{}{A}_{\land} = \multiobjectiveQuery{}{A}_{1} \land \multiobjectiveQuery{}{A}_{2}$, and $\multiobjectiveQuery{}{G}_{\land} = \multiobjectiveQuery{}{G}_{1} \land \multiobjectiveQuery{}{G}_{2}$.

     \[
                \infer{
                    \agTriple{\alphabetExtensionOfTo{\rpa}{\alphabetOf{\multiobjectiveQuery{}{A}_{\land}} }}{\prt}{\multiobjectiveQuery{}{A}_{\land}}{\multiobjectiveQuery{}{G}_{\land}} 
                } 
                {
                    \deduce{
                        \agTriple{\alphabetExtensionOfTo{\rpa}{\alphabetOf{\multiobjectiveQuery{}{A}_{2}}}}{\prt}{\multiobjectiveQuery{}{A}_{2}}{\multiobjectiveQuery{}{G}_2} 
                    }{
                        \agTriple{\alphabetExtensionOfTo{\rpa}{\alphabetOf{\multiobjectiveQuery{}{A}_{1}}}}{\prt}{\multiobjectiveQuery{}{A}_{1}}{\multiobjectiveQuery{}{G}_1} 
                    }
                    } 
            \]
        
\end{thm}
\begin{proof}

Assume that the premises hold. 
As in the proof of \Cref{theo_rpa_asym_rule}, we first apply the PA-reduction for convex rPAs defined in $\Cref{defi_PA_reduction_rpa_paths}$. 
And use \Cref{forallpaexistsrpa} to obtain corresponding AG triples on the PA-reduction $\paReductionOf{\rpa}$. 
In particular, the premises imply that the PA-reduction satisfies the premises of the conjunction rule for PAs from \cite[Theorem~3]{Kwi+13}. 
  Hence we obtain
  \[ \agTriple{\paReductionOf{\alphabetExtensionOfTo{\rpa}
                              {\alphabetOf{\multiobjectiveQuery{}{A}_\land}}}}{\prt}
             {\multiobjectiveQuery{}{A}_\land}{\multiobjectiveQuery{}{G}_\land}.
    \]
  Then, by \Cref{theo:forallrPAexistsPA}, we deduce 
  \[
\agTriple{\alphabetExtensionOfTo{\rpa}{\alphabetOf{\multiobjectiveQuery{}{A}_\land}}}{\prt}
             {\multiobjectiveQuery{}{A}_\land}{\multiobjectiveQuery{}{G}_\land}.
\]
\end{proof}

We now lift the asymmetric AG rule to systems with more than two components, yielding an Asym-$N$ rule for convex rPAs. 
This is the rPA analogue of the corresponding rule for
PAs in \cite[Theorem~4]{Kwi+13} and its pPA version in \Cref{theo:pag_asymN_rule}.
\begin{thm}[Asymmetric-N]\label{theo:rPA_asymN_rule}
    Let $\rpa_1, \dots, \rpa_n$ be convex rPAs and $ \rpa =  \rpa_1 \parallelConv \dots \parallelConv \rpa_n$. 
                \[
                \infer{
                    \rpa \modelsWrt{\comp} \multiobjectiveQuery{}{G}
                } 
                {
                    \deduce{
                        \agTriple{\alphabetExtensionOfTo{\rpa_n}{\alphabetOf{\multiobjectiveQuery{}{A}_{n-1}}} }{\prt}{\multiobjectiveQuery{}{A}_{n-1}}{\multiobjectiveQuery{}{G}} 
                    }{
                        \deduce{
                            \vdots
                            }{
                            \deduce{
                                \agTriple{\alphabetExtensionOfTo{\rpa_2}{\alphabetOf{\multiobjectiveQuery{}{A}_{1}}} }{\prt}{\multiobjectiveQuery{}{A}_{1}}{\multiobjectiveQuery{}{A}_2} 
                                }{
                                    \deduce{
                                        \rpa_1\modelsWrt{\comp} \multiobjectiveQuery{}{A}_1
                                        }{}
                                }
                        }
                }} 
                \]
    \end{thm}
    \begin{proof}
    Assume that the premises hold. 
    As in the proof of \Cref{theo_rpa_asym_rule}, we first apply the PA-reduction for convex rPAs as given in \Cref{defi_PA_reduction_rpa_paths}. 
   Then, by \Cref{forallpaexistsrpa} this yields corresponding premises for the asymmetric-N AG rule for PAs from \cite[Theorem~4]{Kwi+13}. 
   More precisely, 
     \[
     \paReductionOf{\rpa_1} \modelsWrt{\comp} \multiobjectiveQuery{}{A}_1, \ 
       \agTriple{\alphabetExtensionOfTo{\paReductionOf{\rpa_2}}{\alphabetOf{\multiobjectiveQuery{}{A}_{1}}} }{\prt}{\multiobjectiveQuery{}{A}_{1}}{\multiobjectiveQuery{}{A}_2} , 
       \ \dots,\text{ and } \ 
       \agTriple{\alphabetExtensionOfTo{\paReductionOf{\rpa_n}}{\alphabetOf{\multiobjectiveQuery{}{A}_{n-1}}} }{\prt}{\multiobjectiveQuery{}{A}_{n-1}}{\multiobjectiveQuery{}{G}}.
    \]
    By \cite[Theorem~4]{Kwi+13}, we have  
    $\paReductionOf{\rpa_1} \parallel \dots \parallel \paReductionOf{\rpa_n}
      \modelsWrt{\comp} \multiobjectiveQuery{}{G}.$ 
      Then, by \Cref{thm:pa_conv_eq_conv_pa} lifted to multiple components, we obtain that 
      \[
        \paReductionOf{\rpa_1 \parallelConv \dots \parallelConv \rpa_n} \modelsWrt{\comp} \multiobjectiveQuery{}{G}.
        \]
  Finally, by \Cref{theo:forallrPAexistsPA}, this result lifts back to rPAs, and we obtain 
  \[ 
  \rpa_1 \parallelConv \dots \parallelConv \rpa_n \modelsWrt{\comp} \multiobjectiveQuery{}{G}.
  \] 
    \end{proof}

Lastly, we consider reasoning about properties whose alphabets do not share any common labels, i.e., about interleaving behaviour. 
The following rule is the rPA analogue of
\cite[Theorem~6]{Kwi+13} for PAs (cf.\ \Cref{theo:pag_interleaving_rule} for pPAs).
As above, we restrict to to safety properties: The languages $\regLang_1$ and $\regLang_2$ are safety languages, and their union $\regLang_1 \cup
\regLang_2$ is again a safety language. 
\begin{thm}[Interleaving]\label{theo:rpa_interleaving_rule}
Let $\rpa_1,$ and $ \rpa_2$ be convex rPAs such that $(\alphabetOf{\rpa_1} \cup \alphabetOf{\multiobjectiveQuery{}{A}_1} )\cap (\alphabetOf{\rpa_2} \cup \alphabetOf{\multiobjectiveQuery{}{A}_2}) = \emptyset$, $\multiobjectiveQuery{}{A}_{\land} = \multiobjectiveQuery{}{A}_{1} \land \multiobjectiveQuery{}{A}_{2}$, and $\rpa = \alphabetExtensionOfTo{(\rpa_1 \parallelConv \rpa_2)}{\alphabetOf{\multiobjectiveQuery{}{A}_{\land}}}$. 
         \[
                \infer{
                    \agTriple{{\rpa}}{\prt}{\multiobjectiveQuery{}{A}_{\land} }{\probPredicate{\substack{\geq p_1+p_2 \\ \ - p_1\cdot p_2}}{\regLang_1 \cup \regLang_2}}
                } 
                {
                    \deduce{
                        \agTriple{\alphabetExtensionOfTo{\rpa_2}{\alphabetOf{\multiobjectiveQuery{}{A}_{2}}}}{\prt}{\multiobjectiveQuery{}{A}_{2}}{{\probPredicate{\geq p_2}{\regLang_2} }} 
                    }{
                                    \deduce{
                                        \agTriple{\alphabetExtensionOfTo{\rpa_1}{\alphabetOf{\multiobjectiveQuery{}{A}_{1}}}}{\prt}{\multiobjectiveQuery{}{A}_{1}}{\probPredicate{\geq p_1}{\regLang_1}  }
                                        }{}
                }} 
             \]
 \end{thm}
     \begin{proof}
     Assume that the premises hold. 
     Analogously to the proof of \Cref{theo_rpa_asym_rule}, we first apply the PA-reduction for convex rPAs (\Cref{defi_PA_reduction_rpa_paths}) to obtain corresponding premises for the PA-reduction $\paReductionOf{\rpa_1} \parallel \paReductionOf{\rpa_2}$ with respect to the interleaving rule for PAs from \cite[Theorem~6]{Kwi+13}.
    \[ 
     \agTriple{\alphabetExtensionOfTo{\paReductionOf{\rpa_1}}{\alphabetOf{\multiobjectiveQuery{}{A}_{1}}}}{\prt}{\multiobjectiveQuery{}{A}_{1}}{\probPredicate{\geq p_1}{\regLang_1}  } \ \text{ and } \  
        \agTriple{\alphabetExtensionOfTo{\paReductionOf{\rpa_2}}{\alphabetOf{\multiobjectiveQuery{}{A}_{2}}}}{\prt}{\multiobjectiveQuery{}{A}_{2}}{{\probPredicate{\geq p_2}{\regLang_2} }}. 
    \]
     By \cite[Theorem~6]{Kwi+13}, and by lifting the reasoning in \Cref{thm:pa_conv_eq_conv_pa} to AG-triples, we can  
     deduce that 
     $\agTriple{\paReductionOf{\rpa}}{\prt}{\multiobjectiveQuery{}{A}_{\land} }{\probPredicate{\substack{\geq p_1+p_2 \\ \ - p_1\cdot p_2}}{\regLang_1 \cup \regLang_2}}$. 
     Then, by \Cref{theo:forallrPAexistsPA}, we obtain 
     \[ \agTriple{{\rpa}}{\prt}{\multiobjectiveQuery{}{A}_{\land} }{\probPredicate{\substack{\geq p_1+p_2 \\ \ - p_1\cdot p_2}}{\regLang_1 \cup \regLang_2}}\] 
    \end{proof} 
    
    \section{Differences to Kwiatkowska et al. (2013)---Explained}
\label{sec:explain_differences}
The safety properties in \cite{Kwi+13} are languages over infinite words, whereas we allow finite words. 
This leads to the following problems in \cite{Kwi+13}: 
\begin{itemize}
	\item  \cite[Lemma 1]{Kwi+13} is inconsistent because $\regLang$ is not prefix-closed: 
	In \cite{Kwi+13}, safety properties are defined over infinite words: 
	$\regLang = \{ w \in  \alphabetOf{}^{{\color{myred}{\omega}}} \mid \text{no prefix of $w$ is in } \regLang_{err} \}$. 
	However, the probability of a language is defined as: \[\PrOf{\pa}{\strategy}{{  \regLang}} = \PrOf{\pa}{\strategy}{\{ \pi \in \infPathsOf{\pa}{} \mid {\color{myred} \restrOfTo{\traceOf{\pi}}{\alphabetOf{\regLang}} \in \regLang} \}}.\] 
	This definition does not properly account for bad prefixes. 
	Specifically, for a path $\pi$ of a PA $\pa$ and a safety $\regLang$ for $\pa$, the fact that $\restrOfTo{\traceOf{\pi}}{\alphabetOf{\regLang}} \not \in \regLang$ does not imply that $\restrOfTo{\traceOf{\pi}}{\alphabetOf{\regLang}}$ is accepted by the bad-prefix automaton of $\regLang$.
	For example, if \(\alphabetOf{\regLang} \subset \alphabetOf{\pa}\), then \(\restrOfTo{\traceOf{\pi}}{\alphabetOf{\regLang}}\) may be a finite prefix of a word in \(\regLang\) without being a bad prefix. 
	Since \(\regLang\) is not prefix-closed, this leads to situations where $\pi \not \in \{ \pi \in \infPathsOf{\pa}{} \mid \restrOfTo{\traceOf{\pi}}{\alphabetOf{\regLang}} \in \regLang \},$ even though \(\restrOfTo{\traceOf{\pi}}{\alphabetOf{\regLang}}\) is not a bad prefix. 
	\item This inconsistency propagates to \cite[Proposition 1]{Kwi+13} and thus also to the asymmetric AG-rule in \cite[Theorem 1]{Kwi+13}, both of which rely on \cite[Lemma 1]{Kwi+13}.
\end{itemize}
By allowing both finite and infinite words, our definition ensures that safety properties are inherently prefix-closed, resolving these issues: 
$\regLang = \{ w \in  \alphabetOf{}^{{\color{myred}{\infty}}} \mid \text{no prefix of $w$ is in } \regLang_{err} \}$. 

	\begin{figure}[ht] 
		\centering 
		\begin{subfigure}[t]{.2\textwidth} 
			\centering 
			\begin{tikzpicture}[mdp]
	\node[ps, init=left] (s0)  {$s_0$};
	
	\node[ps, below=1.2 of s0] (s1)  {$s_1$};

	\path[ptrans]
	
	(s0) edge[] node[pos=0.5,right] {\tact{\lab}} node[pos=0.5,dist] (d0a) {} node[pos=0.1,right] {\tprob{}} (s1)
	
	(s1) edge[loop right] node[pos=0.5,below] {\tact{\altaltlab}} node[dist] (d1b) {} node[pos=0.1,right] {\tprob{}} node[pos=0.2,below] {} (s1)
	
	;

\end{tikzpicture}
			\caption{PA $\pa_1$}\label{fig:simple_paM1}
		\end{subfigure}
		\begin{subfigure}{.44\textwidth} 
			\centering 
			\begin{subfigure}[t]{.45\textwidth} 
				\centering 
				\renewcommand{\extended}{false}
				\begin{tikzpicture}[mdp]
	\node[ps, init=left] (t0)  {$t_0$};

	\path[ptrans]
	
	(t0) edge[loop right] node[pos=0.5,below] {{\ifthenelse{\equal{\extended}{true}}{\tactext{\lab,\altaltlab},}{}}\tact{\altlab}} node[dist] (d1b) {} node[pos=0.1,below] {\tprob{}} node[pos=0.1,below] {} (t0)
	;
	
\end{tikzpicture}
				\caption{PA $\pa_2$}\label{fig:simple_paM2}
			\end{subfigure}\vspace{0.5cm}
			
			\begin{subfigure}[t]{.45\textwidth} 
				\centering 
				\renewcommand{\extended}{true}
				\begin{tikzpicture}[mdp]
	\node[ps, init=left] (t0)  {$t_0$};

	\path[ptrans]
	
	(t0) edge[loop right] node[pos=0.5,below] {{\ifthenelse{\equal{\extended}{true}}{\tactext{\lab,\altaltlab},}{}}\tact{\altlab}} node[dist] (d1b) {} node[pos=0.1,below] {\tprob{}} node[pos=0.1,below] {} (t0)
	;
	
\end{tikzpicture}
				\caption{PA $\alphabetExtensionOfTo{\pa_2}{{\{\lab,\altaltlab\}}}$}\label{fig:simple_paM2_extended}
			\end{subfigure} 
		\end{subfigure} 
		\begin{subfigure}[t]{.3\textwidth} 
			\centering
			\renewcommand{\extended}{false}
			\begin{tikzpicture}[mdp]
	\node[pswide, init=left] (00)  {$s_0, t_0$};	
	
	\node[pswide, below=1.1 of 00] (10)  {$s_1, t_0$};
	
	\path[ptrans]
	
	(00) edge[loop right] node[pos=0.5,below right] {\tact{\altlab}} node[dist] (d1b) {} node[pos=0.15,below] {\tprob{}} node[pos=0.25,below] {} (00)
	
	(10) edge[loop right] node[pos=0.5,below right] {\tact{\altlab},\tact{\altaltlab}} node[dist] (d1b) {} node[pos=0.15,below] {\tprob{}} node[pos=0.25,below] {} (10)
	
	(00) edge[] node[pos=0.5,right] {\tact{\lab}} node[pos=0.5,dist] (d0a) {} node[pos=0.25,below] {\tprob{}} (10)
	
	;
\end{tikzpicture}
			\caption{PA $\pa_1 \parallel \pa_2$}\label{fig:simple_pas_composed}
		\end{subfigure} 
		\caption{PAs used in \Cref{ex:error_AGsafety_rule_1}} 
	\end{figure}

We present a counterexample to \cite[Lemma 1]{Kwi+13} and \cite[Theorem 1]{Kwi+13}, for safety properties containing only infinite words and $\multiobjectiveQuery{\safe}{G}  \not = \alphabetOf{\pa_1 \parallel \pa_2}$ (a similar counterexample exists if $\multiobjectiveQuery{\safe}{A}  \not = \alphabetOf{\pa_1}$). 
\begin{exa}\label{ex:error_AGsafety_rule_1}

		Consider the safety property $\regLang = \{ w \in \{\lab,\altaltlab\}^{\omega} \mid \vert w \vert_{\lab} \leq 1 \}$, i.e., $\alphabetOf{\regLang} = {\lab,\altaltlab}$.  
		Importantly, $\regLang$ consists only of infinite words.
		Now, consider the PAs $\pa_1$ and $\pa_2$ in \Cref{fig:simple_paM1,fig:simple_paM2} 
		and their composition $\pa_1 \parallel \pa_2$ in \Cref{fig:simple_pas_composed}. 
		We define the complete strategy $\strategy$ that always chooses action $b$ in $(s_0, t_0)$ with probability 1. 
		Then, $\PrOf{\pa_1 \parallel \pa_2}{\strategy}{\regLang} = 0$ since $\restrOfTo{b^{\omega}}{\alphabetOf{\regLang}} = \epsilon \not \in \regLang$. 
		Thus, we obtain $\pa_1 \parallel \pa_2  \not \modelsWrt{\comp} \multiobjectiveQuery{\safe}{G}$. 
		However, by \cite[Lemma 1]{Kwi+13}, it would follow that $\pa_1 \parallel \pa_2  \modelsWrt{\comp} \multiobjectiveQuery{\safe}{G}$ as there is no strategy $\strategy'$  such that 
        \[
             1- \PrOf{(\pa_1 \parallel \pa_2) {\otimes} \bpAutomatonOf{\regLang}}{\strategy'}{\Diamond bad_{\regLang}} = 1= \PrOf{\pa_1 \parallel \pa_2}{\strategy}{\regLang}. 
        \]
		Since $\PrOf{\pa_1 \parallel \pa_2}{\strategy}{\regLang} = 0$, there cannot be a bijection, contradicting \cite[Lemma 1]{Kwi+13}. 

		This also provides a counterexample to \cite[Theorem 1]{Kwi+13}.
		Let $\multiobjectiveQuery{\safe}{A} = \multiobjectiveQuery{\safe}{G}  = \{ \probPredicate{\geq 1}{\regLang} \}$.    
		For $\pa_1$ in \Cref{fig:simple_paM1}, we have that $\pa_1 \modelsWrt{\comp} \multiobjectiveQuery{\safe}{A}$. 
		Additionally, for $\alphabetExtensionOfTo{\pa_2}{\alphabetOf{\multiobjectiveQuery{\safe}{A}}}$ in \Cref{fig:simple_paM2_extended}, we obtain $\agTriple{\alphabetExtensionOfTo{\pa_2}{\alphabetOf{\regLang}}}{\prt}{\multiobjectiveQuery{\safe}{A}}{\multiobjectiveQuery{\safe}{G}}$, which trivially holds, as it is a tautology. 
		By \cite[Theorem 1]{Kwi+13} this implies that $\pa_1 \parallel \pa_2  \modelsWrt{\comp} \multiobjectiveQuery{\safe}{G}$.  
\end{exa}
\cite[Lemma 1, Theorem 1]{Kwi+13} yields results we would expect using our definition of safety properties, which consider both finite and infinite words.

\end{document}